\documentclass[11pt]{article}
\usepackage[T1]{fontenc}
\usepackage[utf8]{inputenc}
\usepackage{amsmath,amsthm}
\IfFileExists{newtxtext.sty}{\usepackage{newtxtext,newtxmath}}{\usepackage{txfonts}}
\usepackage[letterpaper,margin=1in]{geometry}
\usepackage{longtable,arydshln}
\usepackage{tikz}
\usetikzlibrary{arrows.meta,positioning,calc}
\tikzset{
  >={Stealth[length=2.2mm]},
  every node/.style={inner sep=1.5pt},
  lab/.style={font=\normalsize},
  arr/.style={->,line width=0.45pt}
}
\usepackage[round,authoryear]{natbib}
\usepackage{etoolbox}
\usepackage[colorlinks=true,linkcolor=blue!50!black,citecolor=green!40!black,urlcolor=blue!50!black]{hyperref}

\numberwithin{equation}{section}

\renewenvironment{abstract}{\small\begin{center}\bfseries Abstract\end{center}\noindent\ignorespaces}{\par}

\theoremstyle{plain}
\newtheorem{theorem}{Theorem}[section]
\newtheorem{lemma}{Lemma}[section]
\newtheorem{proposition}{Proposition}[section]
\newtheorem{corollary}{Corollary}[section]
\theoremstyle{definition}
\newtheorem{definition}{Definition}[section]
\newtheorem{notation}{Notation}[section]
\newtheorem{example}{Example}[section]
\newtheorem{examples}[example]{Examples}
\newtheorem{remark}{Remark}[section]
\newtheorem{remarks}[remark]{Remarks}
\newtheorem{digression}{Digression}[section]
\newtheorem{discussion}{Discussion}[section]

\DeclareUnicodeCharacter{21D2}{\ensuremath{\Rightarrow}}
\DeclareUnicodeCharacter{21D4}{\ensuremath{\Leftrightarrow}}
\DeclareUnicodeCharacter{22EE}{\ensuremath{\vdots}}
\DeclareUnicodeCharacter{25A1}{\ensuremath{\square}}
\DeclareUnicodeCharacter{25B3}{\ensuremath{\triangle}}
\DeclareUnicodeCharacter{03A3}{\ensuremath{\Sigma}}
\DeclareUnicodeCharacter{2297}{\ensuremath{\otimes}}

\begin{document}
\title{Discrete-time polynomial systems with inputs and outputs:
  structure, reachability, observability, and minimal realizations}
\author{Eduardo D. Sontag\\
  Northeastern University\\
Departments of Electrical and Computer Engineering and Bioengineering}
\date{}
\maketitle

\begin{abstract}

We study the realization problem for discrete-time input/output
polynomial systems. These are formalized using tools from commutative
algebra and algebraic geometry as systems whose state spaces are
algebraic varieties, or more abstractly the set of $k$-points of an
affine $k$-scheme, where $k$ is an arbitrary infinite field.
The input/output behaviors of such systems are described by
``polynomial response maps'' in which outputs are polynomial functions
of past inputs.
The main results show that every polynomial response map admits a
canonical (quasi-reachable and algebraically observable) realization,
which is unique up to isomorphism.
The key to the approach is to linearize dynamics by considering a
``dual'' system in which states are functions defined on states.
Finite dimensionality of the canonical realization, and its
polynomiality, are characterized in terms of the space, algebra, and
field of observables of the map, as well as in terms of algebraic
input/output difference equations, and by a Jacobian rank criterion.
A particular subclass consists of the maps that we call ``bounded,''
defined by the property that their degree in the past inputs is
uniformly bounded.
Bounded maps are shown to be finitely realizable if and only if they
are realizable by finite-dimensional state-affine systems, whose
theory in turn reduces to that of rational formal power series.
We also study the lattice of quasi-reachable realizations of a given
map, including normal realizations.

\end{abstract}

\newpage
\tableofcontents
\section*{Notations}
\addcontentsline{toc}{section}{Notation}
The page numbers refer to the place where the notation is introduced.

\begin{longtable}{@{}p{0.24\textwidth}p{0.6\textwidth}r@{}}
$\mathrm{rad}_{k} A$ & $k$-radical of $A$ & \pageref{op:17}\\
$\iota$ & inclusion of $A$ in $k^{X(A)}$ & \pageref{op:18}\\
$X(A)$ & $k$-homomorphisms $A \rightarrow k$ & \pageref{op:18}\\
$V(S)$ & solution set of $S$ & \pageref{op:24}\\
$V^{B}(S)$ & idem with equations in $B$ & \pageref{op:25}\\
$I(Z)$ & annihilator of $Z$ & \pageref{op:25}\\
$A(X)$ & polynomial functions on $X$ & \pageref{op:29}\\
$A(g)$ & transpose of $g$ & \pageref{op:31}\\
$\operatorname{trdeg}_{B} A$ & transcendence degree over $B$ & \pageref{op:38}\\
$L^{*}$ & finite sequences of elements of $L$ & \pageref{op:42}\\
$|\alpha|$ & length of $\alpha$ & \pageref{op:42}\\
$\|\alpha\|$ & weight of $\alpha$ & \pageref{op:42}\\
$\Delta$ & proper sequences & \pageref{op:42}\\
$\Psi$ & algebra of formal Volterra series & \pageref{op:45}\\
$\operatorname{deg} \psi$ & degree of series & \pageref{op:45}\\
$\operatorname{supp} \psi$ & support of series & \pageref{op:47}\\
$\epsilon_{t}$ & homomorphism $\Psi \rightarrow A\left(U^{t}\right)$ & \pageref{op:49}\\
$\Omega$ & input space & \pageref{op:50}\\
$U$ & input-value set & \pageref{op:50}\\
$U[z]$ & finitely nonzero sequences & \pageref{op:51}\\
$\delta_{s}$ & concatenation & \pageref{op:52}\\
$f^{\Omega}$ & induced map on input space & \pageref{op:53}\\
$\mathbb{V}$ & sequences zero in the past & \pageref{op:54}\\
$\sigma_{V}$ & shift on $\mathbb{V}$ & \pageref{op:54}\\
$\Gamma$ & output space & \pageref{op:54}\\
$\mathbf{f}$ & i/o map for response $f$ & \pageref{op:55}\\
$\Sigma$ & system & \pageref{op:56}\\
$x^{\#}$ & initial state & \pageref{op:56}\\
$h^{w}$ & observable induced by $w$ & \pageref{op:56}\\
$h^{\Gamma}$ & observability map & \pageref{op:56}\\
$P^{(t)}$ & iterate of transition map & \pageref{op:56}\\
$g, g_{t}$ & reachability map, restriction & \pageref{op:56}\\
$X_{t}$ & $t$-step reachable set & \pageref{op:56}\\
$\Sigma_{a c}(f)$ & abstractly canonical realization of $f$ & \pageref{op:59}\\
$\Sigma_{\text {free }}(f)$ & free realization & \pageref{op:61}\\
$\Sigma_{Q}$ & quasi-reachable subsystem & \pageref{op:63}\\
$\mathcal{L}(\Sigma)$ & observation space of $\Sigma$ & \pageref{op:64}\\
$\mathcal{A}(\Sigma)$ & observation algebra of $\Sigma$ & \pageref{op:64}\\
$\mathcal{Q}(\Sigma)$ & observation field of $\Sigma$ & \pageref{op:64}\\
$\Sigma^{\text {obs }}$ & observable ``quotient'' & \pageref{op:65}\\
$\mathcal{L}_{f}, \mathcal{A}_{f}, \mathcal{Q}_{f}$ & observation space, etc., of $f$ & \pageref{op:69}\\
$\mathcal{L}_{f}^{R}, \mathcal{A}_{f}^{R}, \mathcal{Q}_{f}^{0}, \ldots$ & reachability and observability chains & \pageref{op:71}\\
$\operatorname{Obs}(x)$ & observation class & \pageref{op:76}\\
$\mathcal{L}_{f}^{K}, \ldots$ & extended observables & \pageref{op:85}\\
$J_{n}(f)$ & $n$-th Jacobian of $f$ & \pageref{op:89}\\
$\mathcal{B}(\varphi)$ & behavior matrix & \pageref{op:102}\\
$\operatorname{tdeg} f$ & total degree & \pageref{op:109}\\
$\mathrm{QR}(f)$ & quasi-reachable lattice & \pageref{op:113}\\
$\mathrm{QC}(f)$ & quasi-canonical lattice & \pageref{op:121}\\
$\mathrm{AO}(f), \mathrm{RD}(f)$ & subsets of $\mathrm{QR}(f)$ & \pageref{op:123}\\
$\mathrm{NOR}(f)$ & normal realizations & \pageref{op:127}\\
$\bar{\Sigma}$ & integral closure of $\Sigma$ & \pageref{op:128}\\
$r(A)$ & smallest number of generators & \pageref{op:139}\\
$C_{n}, C_{*}, C_{\infty}$ & algebras of input spaces & \pageref{op:146}\\
$\Omega^{\prime}$ & generalized input space & \pageref{op:147}\\
\end{longtable}

\newpage

\section{Introduction}\label{intro}

This paper is a revised version of my PhD thesis, which was completed
at the University of Florida in December 1976, fifty years ago as of
this writing, under the supervision of Rudolf E.~Kalman. The material
in the thesis was never published as a journal paper. Parts of it were
reproduced in typewritten form in Volume~13 of the Springer-Verlag
series \emph{Lecture Notes in Control and Information Sciences}
(1979). The thesis and lecture notes are very hard to read, due to the
old typesetting, hand-drawn diagrams, old-style numbering of
definitions and results, not enough displays for equations, and poor
quality of scans. In addition, they contain several typos and errors in
formulas and bibliographical references.  Perhaps due to these issues,
the work has not been much cited, and the connection to currently
popular ideas such as the ``linearization by duality'' representation
of input/output systems has been missed.

The purpose of this paper, then, is to provide a streamlined
exposition, correcting typos and errors, drawing modern diagrams, and
updating citations.
In a final Section~\ref{sec:addendum}, however, we have included an addendum with some
connections to related work done after 1976, including a few of my own later papers.
This is of course incomplete, and is intended merely as a sample.
It is unfortunate that real algebraic geometry had not yet been developed at the time. Several results here are stated only for algebraically closed fields, or for $k=\mathbb{R}$ with ad hoc arguments (see Theorems~\ref{thm:3.14} and~\ref{thm:4.6}, and the remarks in Section~\ref{sec:30}); with the modern theory, I would have been able to make them considerably more interesting.

As alluded to above, among the reasons for preparing this clean
version is that the ``Koopman'' approach to dynamics, introduced by
Koopman in 1931, has become very popular in control theory. In that
approach, a nonlinear system is linearized by representing its
dynamics through a dual system, whose states are observables, that is,
functions on the original state space. The same approach was taken in
this work, which was in turn strongly influenced by Kalman's ideas,
except that here the approach is algebraic instead of
functional-analytic: the observables are polynomial functions, and
they are organized into vector spaces, algebras, and fields rather
than into spaces of functions with a topology. Instead of states, the
basic objects here are the observation space $\mathcal{L}_{f}$, the
observation algebra $\mathcal{A}_{f}$, and the observation field
$\mathcal{Q}_{f}$ of a response map $f$ (Section~\ref{sec:12}), on
which the transition map acts through its transpose. Realizations are
built from these spaces of observables; in particular, when the
observation space is finite dimensional the dual representation is
linear and it leads to the state-affine realizations studied in Section~\ref{ch:5}.
Section~\ref{sec:addendum} discusses Koopman-based work on systems with inputs and exact finite-dimensional embeddings.

\paragraph{Co-transitions.}
To make the connection explicit, consider a $k$-system $\Sigma$ with state space $X$, input-value space $U=k^{m}$, and transition map $P: X \times U \rightarrow X$ (Definition~\ref{def:8.1}). Every polynomial map $g: X_{1} \rightarrow X_{2}$ has a transpose
\[
A(g): A\left(X_{2}\right) \rightarrow A\left(X_{1}\right), \qquad A(g)(\varphi)=\varphi \circ g
\]
(Definition~\ref{def:3.6}), and polynomial maps correspond bijectively to algebra homomorphisms in the opposite direction (Lemma~\ref{lem:3.7} and Corollary~\ref{cor:3.8}). Since $A(X \times U)=A(X)\left[T_{1}, \ldots, T_{m}\right]$, where $T_{1}, \ldots, T_{m}$ are the coordinate functions of the input, the transpose of the transition map is
\[
A(P): A(X) \rightarrow A(X)\left[T_{1}, \ldots, T_{m}\right], \qquad (A(P) \varphi)(x, u)=\varphi(P(x, u)).
\]
This is the ``co-transition'' map, as it is called in Section~\ref{sec:23}. It is linear in the observable $\varphi$ (indeed, it is an algebra homomorphism), and it depends polynomially on the input. For each fixed input value $u$, the map $\varphi \mapsto \varphi \circ P(\cdot, u)$ is precisely the Koopman operator of the dynamics $x \mapsto P(x, u)$. The components of the output map are the simplest observables, and the observation algebra $\mathcal{A}(\Sigma)$ is the smallest subalgebra of $A(X)$ that contains them and is invariant under the co-transitions, in the sense that
\[
A(P)(\mathcal{A}(\Sigma)) \subseteq \mathcal{A}(\Sigma)\left[T_{1}, \ldots, T_{m}\right]
\]
(see \eqref{eq:10.5} and Lemma~\ref{lem:10.6} for this invariance, and the analogous statement for the observation space $\mathcal{L}(\Sigma)$).

\paragraph{The canonical realization as a dual object.}
The canonical realization of a response map is built by the same mechanism, but starting from the observables instead of from a state space. For the ``free'' realization, whose state space is the input space $\Omega$, the observables are the Volterra series, and the dynamics (concatenation of inputs) is defined as the dual of the homomorphism
\[
\theta_{1}: \Psi \rightarrow \Psi \otimes k\left[S_{1}\right]: \psi \mapsto \psi\left(S_{1}\right)
\]
(see \eqref{eq:6.7} and \eqref{eq:6.8}). The canonical realization $\Sigma_{f}$ is obtained by restricting this co-transition to the observation algebra $\mathcal{A}_{f}$ of $f$, and by taking as state space the set $X\left(\mathcal{A}_{f}\right)$ of $k$-points of $\mathcal{A}_{f}$; its transition map is $X\left(\theta_{1}\right)$, with $\theta_{1}: \mathcal{A}_{f} \rightarrow \mathcal{A}_{f}\left[S_{1}\right]$ (Section~\ref{sec:14}). In this sense, states are recovered from observables, rather than the other way around. When the observation space $\mathcal{L}_{f}$ is finite dimensional, the co-transitions act linearly on it, with a polynomial dependence on the input, and this is what leads to the state-affine realizations of Section~\ref{ch:5}; in that case, the state space of the span-canonical realization is in a natural duality with $\mathcal{L}_{f}$ (Proposition~\ref{prop:20.15}).

AI tools were employed in order to read the typewritten manuscript,
reformat it in a modern, \LaTeX-friendly manner, redraw diagrams,
correct typos, and perform other minor editing.
The material in the addendum was heavily assisted by AI tools.
In the original manuscript, in contrast, a sophisticated HI (human intelligence) tool assisted tremendously in editing the manuscript, both in mathematics content and language: Professor Kalman was tireless in red-lining my manuscript with suggestions (a skill which I inherited and then practiced on my own students); the quality of his work was far superior to the best current AI!

The rest of this introduction reproduces the preface of the Lecture
Notes, which in turn was a minor edit of the preface for the thesis
(University of Florida, 1976):

The past 20 years have witnessed the emergence of a new area of application-oriented mathematics and engineering, that of Mathematical System Theory. This new field has achieved significant advances during its relatively short existence, in particular with respect to the control and observation of finite-dimensional linear dynamical systems, whose theory is by now widely known and applied. Perhaps the central concept in this latter theory is that of \textit{realization}, the basic problem of studying what are the possible internal structures (i.e., sets of evolution equations) giving
 rise to an observed external behavior (i.e., input/output map, impulse response, transfer function, etc.). In one way or another, either implicitly or explicitly, realization theory---together with its associated concepts of reachability and observability (and variations of these like controllability and reconstructibility)---permeate most methods and results in linear system theory. When dealing with \textit{nonlinear} systems, however, the question of realization is only now beginning to be studied. Besides its intrinsic interest, it is reasonable to expect in view of the above remarks that a nonlinear realization theory may eventually derive analogous benefits to the design and analysis of more general systems.

The present work is an attempt to attack the realization problem for a wide class of discrete-time nonlinear behaviors. In choosing an appropriate class of behaviors, one should of course strive for a class which is general enough to accommodate many examples of interest while at the same time having sufficient structure to allow for the application of useful mathematical tools. Thus the extreme (set-theoretic) case of automata theory and "general system theory", although providing much of the intuition and philosophy of the approach, does not by itself constitute the right level of generality from a more applied viewpoint. The most important type of nonlinearity, when no strong threshold effects or other discontinuities are dominant, is given by multiplicative effects. This gives rise to the notion of a \textit{polynomial} input/output map, in which present output values are sums and products of past input values.

The present work is based upon the premise that the natural tools for the study of the structural-algebraic properties (in particular, realization theory) of polynomial input/output maps are provided by \textit{algebraic geometry} and \textit{commutative algebra}, perhaps as much as linear algebra provides the natural tools for studying linear systems. The results obtained until now, and the problems and directions of research suggested, seem to indicate that this premise is indeed correct. Although (or rather, \textit{because}) the theory is clearly far from complete, it seems appropriate to present its main lines in an expository way, with the hope that it will generate additional research. Since algebra-geometric concepts and tools are rather new in the context of system theory, a rather detailed discussion is included of some basic algebraic definitions and results, in a terminology geared towards the intended applications. In this sense, the present work can be seen dually as an essentially self-contained introduction to some areas of basic algebraic geometry, illustrated through system-theoretic applications (Hilbert's basis theorem to finite-time observability, dimension theory to minimal realizations, Zariski's Main Theorem to uniqueness of canonical realizations, etc.) In order to keep the level elementary (in particular, not utilizing sheaf-theoretic concepts) certain ideas like nonaffine varieties are used only implicitly (e.g., quasi-affine as open sets in affine varieties) or in technical parts of a few proofs, and the terminology is similarly simplified (e.g., "polynomial map" instead of "scheme morphism restricted to $k$-points", or "$k$-space" instead of "$k$-points of an affine $k$-scheme"). Hopefully, the reader will be sufficiently motivated by the methods and results to deepen his/her knowledge of algebraic geometry through the study of any of various existing purely mathematical texts.

This work deals only with discrete-time systems, and no attempt is made to treat systems evolving in continuous-time. This reflects a bias of the author, due in part to the influence of the present microprocessor revolution, and the new possibilities that this opens up for digital control. Associated with this, it is at present not uncommon to model physical systems (and even more, economic and biological ones) via difference equations, sometimes as "sampled" continuous-time processes. It is clear, however, that some future applications will depend also on a deeper understanding than is now possible of the interplay between the notions of continuous and discrete-time systems.

Section~\ref{ch:1} summarizes the problems and main results in an intuitive and relatively nontechnical way. After the algebraic preliminaries of Section~\ref{ch:2}, the next two sections develop an abstract realization theory and study various finiteness conditions. Section~\ref{ch:5} treats a class of systems which are suggested naturally by the general framework in the particular case of a certain invariant (the observation space) being finite-dimensional; these systems turn out to include those types for which realization theories had been developed by various authors, and a general realization algorithm is presented, which restricts to the various known procedures. The next section studies the class of realizations of a fixed input/output map, while the last one deals with generalizations, further examples and remarks, and a discussion of open problems. 

This work is largely based on the doctoral dissertation submitted by the author to the University of Florida in 1976, under the supervision of Professor R. E. Kalman. Professor Kalman provided much of the encouragement and arranged for the long-term financial support which made that and other research possible. Furthermore, his early intuition of the system-theoretic relevance of algebraic geometry and rational power series had an obvious influence on this work. The main direct motivation for the research into the topics discussed here was given by joint work with Y. Rouchaleau \citep{sontag1975discrete}. A number of other people had an important influence, either directly or indirectly through the discussion of closely related topics; in particular, S. Eilenberg, M. Fliess, M. Hazewinkel, E. W. Kamen, M. Heymann, and S. Mitter.
This research was supported in part by U.S. Army Grant DAAG29-76-G-0203 and U.S. Air Force Grant AFOSR 76-3034 through the Center for Mathematical System Theory, University of Florida.


\section{The Systems Setup}\label{ch:1}\label{op:1}

In the present work we study the problem of realization of \textit{polynomial input/output maps}. In this introduction we restrict ourselves to \textit{shift-invariant, scalar} input/output maps defined over infinite fields, in order to present the definitions and results in a simple way. The development in the main text proceeds in greater generality.

We choose an infinite field $k$, which will be fixed throughout this section.

Let $\mathbb{S}$ denote the set of all sequences of elements of $k$ indexed by the integers and with support bounded on the left. In other words, $u(\cdot)$ in $\mathbb{S}$ is a function $u(\cdot): \mathbb{Z} \rightarrow k$ for which there exists an integer $t_{0}$ such that $u(t)=0$ for all $t<t_{0}$. A (scalar) \textit{input/ output map} is then a map $f: \mathbb{S} \rightarrow \mathbb{S}: u(\cdot) \mapsto y(\cdot)$; $f$ is (strictly) \textit{causal} when, for all $t$ in $\mathbb{Z}$, the output $y(t)$ depends only on values $u(j)$ of inputs for $j<t$; $f$ is \textit{shift-invariant} if $u(\cdot) \mapsto y(\cdot)$ implies $(\sigma u)(\cdot) \mapsto(\sigma y)(\cdot)$, where $\sigma$ is the shift operator defined by $(\sigma u)(t):=u(t-1)$.

These concepts are standard. One of the contributions of the present work is the introduction of a new notion, that of a "polynomial" input/output map. Informally, this means that $y(t)$ is a polynomial function of the past input values $u(j), j<t$.

In order to rigorously define polynomial input/output maps we need the concept of a \textit{Volterra series} $\psi$ : this is a formal power series in denumerably many variables $\xi_{1}, \xi_{2}, \xi_{3}, \ldots$ such that $\psi$ is of finite degree in each variable separately. A causal, shift-invariant map is uniquely determined by specifying the dependence of $y(0)$ upon values of the input $u(t)$ for $t<0$. So we can now say, more precisely, that a causal, shift-invariant input/output map $f$ is \textit{polynomial} iff there exists a Volterra series $\psi_{f}$ such that the output $y(0)$ due to an input sequence $u(\cdot)$ is obtained by substituting $u(-j)$ for $\xi_{j}$ into $\psi_{f}$ and
 evaluating the expression thus obtained. This evaluation is well-defined because there are only finitely many nonzero $u(t)$ (by definition of $\mathbb{S}$), and because $\psi_{f}$ is a polynomial in each finite subset of variables.

\label{op:2}%
The present formalism is able to represent a wide variety of behaviors.

For example, consider $f_{1}$, defined by $\psi_{f_{1}}=a_{1} \xi_{1}+a_{2} \xi_{2}+\ldots$. Then an input $u(\cdot)$ produces an output
\[
\sum_{j=0}^{\infty} a_{j} u(-j) .
\]
Thus $f_{1}$ corresponds to a linear system with impulse response sequence $a_{1}, a_{2}, \ldots$.

Another example is $\psi_{f_{2}}=\sum \xi_{j_{1}} \xi_{j_{2}} \ldots \xi_{j_{r}}$, the sum running over all $j_{1}<j_{2}<\ldots<j_{r}$ and all $r \geq 0$. Then $f_{2}$ corresponds to adding all possible products of past inputs.

The main problem we are interested in is the following. What are natural internal (i.e., state-space) representations for polynomial input/output maps? We are of course interested in representations with a certain amount of algebraic, geometric, and/or topological structure; otherwise the above question could be trivially answered via the "Nerode realization" method of automata theory. Further, we want to use our results to \textit{infer} possible internal properties of a given "black box"; so the choice of structure should be directly related to properties of \textit{polynomial} input/output maps.

Polynomial systems constitute a class of systems whose defining maps are always polynomial. A \textit{polynomial system} $\Sigma$ (provisional definition) has

(a) $X=k^{n}$ as its state-space ($n=$ integer);

(b) state-transitions given by simultaneous first order difference equations
\[
x(t+1)=P(x(t), u(t)),
\]
where $x(t)=\left(x_{1}(t), \ldots, x_{n}(t)\right)$ and $P=\left(P_{1}, \ldots, P_{n}\right)$ is a polynomial function $k^{n+1} \rightarrow k^{n}$;

(c) an output map
\label{op:3}%
\[
y(t)=h(x(t)),
\]
where $h$ is a polynomial in $n$ variables; and

(d) an initial state $x^{\#}$ which is an equilibrium state for the zero input:
\[
P\left(x^{\#}, 0\right)=x^{\#} .
\]
(The constraint on $x^{\#}$ to be an equilibrium state is dictated by our restriction to shift-invariant input/output maps; the specific choice of 0 as "equilibrium input" is just a matter of choice of coordinates in the input space.) Let us denote by $P$ also the recursive extension of $P$ to sequences of inputs, i.e.
\[
P\left(x, v_{1}, \ldots, v_{n+1}\right):=P\left(P\left(x, v_{1}, \ldots, v_{n}\right), v_{n+1}\right) .
\]
Then $\Sigma$ defines an input/output map $f_{\Sigma}: u(\cdot) \mapsto y(\cdot)$ by the rule
\[
y(t):=h\left(P\left(x^{\#}, u\left(t_{0}\right), u\left(t_{0}+1\right), \ldots, u(t-1)\right)\right),
\]
where $t_{0}<t$ is any integer for which $u(\tau)=0$ if $\tau<t_{0}$. Then $f_{\Sigma}$ is clearly a polynomial input/output map because it is defined as a composition of polynomials. We may in fact exhibit $\psi_{f_{\Sigma}}$ directly by the rule that the coefficient of a monomial $\xi_{i}^{i_{1}} \ldots \xi_{t}^{i_{t}}$ should be equal to the coefficient of the same monomial in the polynomial $h\left(P\left(x^{\#}, \xi_{t}, \ldots, \xi_{1}\right)\right)$.

Thus we have defined a large class of systems whose input/output maps are polynomial. Such systems are appealing from both mathematical and system-theoretic reasons, because they can be realized by finite interconnections of adders, multipliers, amplifiers, and delay lines. In order to get a reasonably complete and general theory, however, it is necessary to go beyond polynomial systems. A larger class of systems, called $k$-\textit{systems}, arises when we study the problem of obtaining "canonical" realizations of input/output maps. The theory to be developed will show that $k$-systems provide the right amount of
\label{op:4}generality for studying realizations of polynomial maps. We now motivate their introduction.

One of our main objectives is to obtain realizations which are "natural" or "canonical" in the sense of not depending on any information not implied by the input/output behavior. The class of candidates to be considered should have some fixed structure (like polynomial systems) so that the canonical system is recoverable just from the knowledge of its external behavior. The approach which has been highly successful with automata and linear systems consists in trying to construct realizations which are as "minimal" or "irredundant" as possible; see for instance \citet[Chapters 7 and 10]{kalman1969topics} and \citet[Chapters 3, 12, and 16]{eilenberg1974automata}. We shall adopt such a viewpoint here, beginning with polynomial systems, and we shall see how we are forced to introduce more general systems.

Let us consider the two-dimensional system
\[
\Sigma_{o}=\left\{\begin{array}{l}
x_{1}(t+1)=x_{1}(t)+u(t), \\
x_{2}(t+1)=x_{1}(t) x_{2}(t)+x_{1}(t)+x_{2}(t), \\
y(t)=x_{2}(t),
\end{array}\right.
\]
with initial state 0. It is easy to see that it is possible to reach from 0 any state $\binom{a}{b}$ in $k^{2}$, using in fact inputs of length not greater than two. There are, however, redundant states which behave identically in the sense that they cannot be distinguished by input/ output experiments. They are of the form $\binom{a}{-1}$. Any other states can be pairwise distinguished from the data ($x_{2}, x_{1} x_{2}+x_{1}+x_{2}$) resulting of the observation of the output at two consecutive instants. In order to obtain a system with no unobservable states, we must identify the states $\binom{a}{-1}$ for all $a$ in $k$ and we must then define appropriate "polynomial" transitions, compatible with the original $P$, on the quotient set thus obtained. To have a well-defined notion of "polynomial map" we must first endow our quotient set with a suitable notion of "coordinate system", i.e. we need to define in it a geometric structure.

\label{op:5}%
But this structure may not correspond to a polynomial system.

It turns out that the input/output map $f_{\Sigma_{0}}$ of the above system $\Sigma_{0}$ admits no observable polynomial realization. This remains true for the weaker question of existence of polynomial realizations for which we only require the property of distinguishable reachable states. In other words, it is in general impossible to embed the "Nerode realization" of an input/output map in a polynomial system, even if $f$ is the input/ output map of a polynomial system.

The natural algebraic-geometric way to proceed consists in introducing the notion of ($k$-points of) an \textit{affine} $k$-\textit{scheme}, or, as we shall say for short, a $k$-\textit{space}. Such a space consists of a topological space $X$ together with an algebra of \textit{polynomial functions} on $X$ (a distinguished family of continuous functions on $X$ subject to appropriate axioms). In particular, the spaces $k^{n}$ become $k$-spaces when endowed with the "Zariski topology", whose closed sets correspond to subsets of $k^{n}$ defined by polynomial equations; the polynomial functions on the $k$-space $k^{n}$ are the usual polynomial functions in $n$ variables. Thus our previous choice of state-spaces furnishes an (easy) example of $k$-spaces. Given two $k$-spaces $X_{1}, X_{2}$ there is a well-defined concept of \textit{polynomial map} $P: X_{1} \rightarrow X_{2}$; these are precisely those maps which when composed with the polynomial functions on $X_{2}$ give polynomial functions on $X_{1}$.

A $k$-\textit{system} $\Sigma$ is then defined by letting the state set $X_{\Sigma}$ be an arbitrary $k$-space and letting the transition map $X_{\Sigma} \times k \rightarrow X_{\Sigma}$ and the output function $X_{\Sigma} \rightarrow k$ be polynomial maps. The fundamental observation is that \textit{the input/output maps of} $k$-\textit{systems are polynomial}.

Conversely, \textit{each polynomial input/output map can be realized by some} $k$-\textit{system}.
 This fact follows rather trivially once that $k$-spaces have been recognized as the proper state spaces. The proof relies on turning the space of input sequences into a $k$-space $\Omega$ in such a way that the notion of (polynomial) input/output map becomes precisely that of a polynomial map between $k$-spaces.

\label{op:6}%
Having established $k$-systems as the class of systems to be considered, we return to the problem that motivated the introduction of $k$-spaces in the first place, namely, the existence of "observable" realizations.

We shall prove that in the new class of systems it is always possible to "reduce" a given system to one all of whose states can be distinguished by input/output experiments. Nevertheless, this does not settle the question of observability. It was noticed already in \citet{sontag1975discrete} that there exist input/output maps having realizations $\Sigma_{1}, \Sigma_{2}$ both of which are reachable and observable but such that $\Sigma_{1}$ and $\Sigma_{2}$ are nonisomorphic (as $k$-spaces). In fact, the example given in the above reference has $X_{1}=k$ while $X_{2}$ is a curve with a singularity, a very different kind of $k$-space. The difficulty lies in the concept of observability itself. This notion is usually defined by the intuitive requirement that different states be "distinguishable by processing the input/output data". The precise notion in this context is that different states should be distinguishable by an \textit{algebraic} processing of the input/ output data. This point of view leads to the definition of \textit{algebraic observability}, which turns out to be the proper notion in our context.

The next step in our program for obtaining a "canonical" realization of a given input/output map is to construct an observable realization all of whose states are reachable from the initial state. Here we run into a new problem: the reachable set of an arbitrary system is not necessarily a $k$-space. For instance, let us consider a two-dimensional system with transitions defined by
\[
\begin{aligned}
& x_{1}(t+1)=u(t), \\
& x_{2}(t+1)=x_{2}(t) u(t)+x_{2}(t)+u(t),
\end{aligned}
\]
and zero initial state. The reachable set fails to contain the points $\binom{-1}{b}, \quad b \neq-1$.

This difficulty can be easily eliminated. Our ultimate goal is not to obtain reachable and observable realizations but rather to construct
\label{op:7}"natural" realizations. It is therefore enough to observe that (for continuity reasons), the dynamical properties of the reachable part of a system $\Sigma$ uniquely determine the dynamical properties of the closure (in the topology of the $k$-space $X_{\Sigma}$) of the set of reachable states. (In the above example the closure corresponds to the whole plane.) We shall say that $\Sigma$ is \textit{quasi-reachable} if the closure of the reachable states is $X_{\Sigma}$. The closure of the subset of reachable states is always a $k$-space invariant under the action of inputs. So a quasi-reachable realization can always be obtained from an arbitrary realization. If we begin with a polynomial system, the closure of the reachable set is a very special type of $k$-space namely, an \textit{algebraic variety}. It is natural therefore to generalize our preliminary definition of polynomial systems to include the case in which $X_{\Sigma}$ is a variety (not necessarily $k^{n}$). In other words, a \textit{polynomial system is given by a finite set of simultaneous polynomial difference equations together with a set of polynomial constraints on the state variables}.

We shall say that a $k$-system is \textit{canonical} if it is quasi-reachable and algebraically observable. One of the main results of this work is then: \textit{Every input/output map} $f$ \textit{admits a canonical realization} $\Sigma_{f}$ \textit{and any other canonical realization of} $f$ \textit{is isomorphic to} $\Sigma_{f}$. We have thus attained our goal of determining a natural class of state representations for polynomial input/output maps.

The result on existence and uniqueness of canonical realizations must be complemented by a discussion of finiteness conditions. In principle, there is of course no guarantee that the state-space $X_{f}$ of $\Sigma_{f}$ is in any sense "finite dimensional".

We have chosen the "transcendence degree" notion of dimension out of the many possible definitions of dimension of $k$-spaces. The \textit{dimension} of a system $\Sigma$ is then the dimension of $X_{\Sigma}$. Informally, the dimension of $\Sigma$ counts the "degrees of freedom" in the state space. In the particular case of polynomial systems the dimension is what one would intuitively expect. For instance, if $X_{\Sigma}$ is the "cusp", given by $\left\{\left(x_{1}, x_{2}\right) \in k^{2} \mid x_{1}^{3}=x_{2}^{2}\right\}$, then $\operatorname{dim} \Sigma=1$.

\label{op:8}%
We shall say that a given system $\Sigma$ is \textit{almost polynomial} when $X_{\Sigma}$ can be obtained as a "quotient" of some space $k^{n}$ (the terminology "quotient" is not quite precise here, since we shall have to admit in general the existence of some points besides those representing the equivalence classes of points of $k^{n}$). The name "almost polynomial" is due to the fact that in this case $X_{\Sigma}$ can be expressed as a union of a variety and a lower-dimensional subset.

A central result in this context is: \textit{The input/output map} $f$ \textit{has a finite-dimensional realization if and only if} $\Sigma_{f}$ \textit{is an almost-polynomial system if and only if} $f$ \textit{satisfies an algebraic difference equation}, i.e., if and only if there exists an integer $s$ and a polynomial $E$ in $2 s+1$ variables such that
\[
E(y(t), y(t-1), \ldots, y(t-s), u(t-1), \ldots, u(t-s))=0,
\]
for \textit{all} input/output pairs $u(\cdot), y(\cdot)$. Up to constant multiples there is a unique \textit{irreducible} equation $E=0$ of minimal order $s$ satisfied by $f$. We shall also prove that, if $f$ satisfies some algebraic difference equation, then $f$ satisfies as well an equation $\hat{E}=0$ \textit{linear in} $y(t)$.

As a simple illustration of the above results, let $f:=f_{\Sigma_{0}}$, where $\Sigma_{0}$ is the system, introduced before,
\[
\Sigma_{o}=\left\{\begin{array}{l}
x_{1}(t+1)=x_{1}(t)+u(t), \\
x_{2}(t+1)=x_{1}(t) x_{2}(t)+x_{1}(t)+x_{2}(t), \\
y(t)=x_{2}(t),
\end{array}\right.
\]
with zero initial state. The canonical state-space $X_{f}$ is, as a set, the union of the singleton $\{*\}$ and the subset $U:=\left\{x_{2} \neq-1\right\}$ of $k^{2}$. The transition and output maps of $\Sigma_{f}$ are those induced by the projection $T: k^{2} \rightarrow X_{f}: x \mapsto x$ if $x$ is in $U, x \mapsto *$ if $x_{2}=-1$. (These are polynomial maps for a suitable $k$-space structure on $X_{f}$.) The irreducible equation of minimal order satisfied by $f$ is
\label{op:9}%
\[
\begin{aligned}
& {[y(t-2)+1][y(t)+1]-[y(t-1)+1]^{2}-} \\
& -[y(t-1)+1][y(t-2)+1] u(t-2)=0
\end{aligned}
\]
In the "classical" case of linear systems it is well known that a system is canonical iff it is a minimal-dimensional realization of its input/output map. This result does not generalize directly to the present situation. A counterexample is given by the system (with $X=k^{1}$)
\[
\Sigma=\left\{\begin{array}{l}
x(t+1)=u(t), \quad x^{\#}=0, \\
y(t)=x^{2}(t) .
\end{array}\right.
\]
Clearly, $\Sigma$ is not canonical, because all pairs of states $\{a,-a\}$ are indistinguishable. However $\Sigma$ is minimal, since it has dimension 1.

The proper treatment of the above minimality question is through the concept of weakly canonical realizations. We shall say that $\Sigma$ is \textit{weakly canonical} when it is quasi-reachable and (in a sense to be made precise) "almost all" states are indistinguishable of only \textit{finitely many} other states. The example in the previous paragraph is therefore weakly canonical, since in fact each state is indistinguishable of only \textit{one} other state. Let $k$ be either the field of real numbers or an algebraically closed field. We prove that a \textit{realization} $\Sigma$ \textit{of a polynomial input/output map} $f$ \textit{is of minimal dimension among all realizations of} $f$ \textit{if and only if} $\Sigma$ \textit{is weakly canonical}. Over any field $k$, canonical realizations are minimal.

The question of deciding when $\Sigma_{f}$ is in fact a \textit{polynomial} system (i.e. $X_{f}$ is a special kind of $k$-space: a variety) can be answered theoretically via the introduction of the \textit{observation algebra} $A_{f}$ of the input/output map. This is a $k$-algebra which is canonically associated to any given $f$. We prove that $\Sigma_{f}$ is a \textit{polynomial system precisely when} $A_{f}$ \textit{is finitely generated as a} $k$-\textit{algebra}
. Further, the smallest $n$ for which $\Sigma_{f}$ can be embedded in a system of $n$ simultaneous polynomial difference equations is equal to the minimal
\label{op:10}possible cardinality of sets of generators of $A_{f}$. Unless $A_{f}$ is isomorphic to a polynomial ring, $n$ is not equal to the dimension of $\Sigma_{f}$. We shall also prove that $\Sigma_{f}$ \textit{is a polynomial system when} $f$ \textit{satisfies an input/output equation of the type}
\[
a(u(t-1), \ldots, u(t-s)) y(t)^{r}+\hat{E}=0,
\]
where $y(t)$ appears in $\hat{E}$ with degree less than $r$. Thus if, for instance, $f$ is known to satisfy a regression equation, the realization theory of $f$ can in principle be carried out without introducing the concept of $k$-spaces. Even in this special case, however, the general theory is needed in order to understand the meaning of the special hypothesis.

One of the main results is valid for input/output maps defined over fields $k$ which contain the rational numbers. The result states that $f$ \textit{has a finite realization if and only if the Jacobian matrices in a certain sequence} $J_{1}(f), J_{2}(f), \ldots$ \textit{have a uniformly bounded rank}. For a trivial example, we point out that when $f$ is linear the matrix $J_{n}(f)$ is precisely the $n$-th principal minor of the behavior (Hankel) matrix of $f$.

All the results presented up to this point are proved later for \textit{multivariable} polynomial input/output maps, for which both the inputs and outputs are vector-valued.

Proofs of the preceding results use tools of algebraic geometry. In other words, we use the "theory of polynomials" in the study of arbitrary polynomial input/output maps.

The second part of this work deals with a broad class of \textit{bounded} (polynomial) input/output maps, whose study can be "linearized". This linearization permits us to obtain sharper statements. Furthermore, $\Sigma_{f}$ will again be a \textit{polynomial} (not arbitrary $k$-) system.

Bounded maps $f$ are defined as follows. Recall that $\psi_{f}$ has a finite degree $d_{j}$ in each variable $\xi_{j}$. We say that $f$ is \textit{bounded} when the degrees $d_{j}$ are bounded independently of $j$. In other words, there
\label{op:11}exists an integer $d$ such that no input is raised to a power higher than $d$. There are no restrictions on products between inputs at different instants and/or different channels. It is at first surprising that the concept of bounded map includes as particular cases all those families of maps for which a satisfactory realization theory has been developed in the past. For instance, \textit{linear} systems, \textit{internally-bilinear} systems (\citealp{brockett1972algebraic}, \citealp{isidori1973realization}, \citealp{isidori1973direct,isidori1974new}, \citealp{fliess1973sur,fliess1975outil}, \citealp{dalessandro1974realization}, and others) give rise to bounded maps. (Internally-bilinear systems are those whose internal map is bilinear in the state and input and whose output map is linear. No products of inputs at same instants are performed by such systems, so $d:=1$ bounds all $d_{j}$.) \textit{Multilinear input/output maps} (\citealp{kalman1968lectures,kalman1979realization}) are also included. (Such maps allow products of inputs only in different channels, so that $d:=1$ is again a bound.)

We prove that \textit{if a bounded input/output map is at all realizable by a finite dimensional $k$-system, then it is also realizable by an (observable) state-affine system}. The latter are (polynomial) systems with $X_{\Sigma}=k^{n}$ whose defining equations take the special form
\[
\begin{aligned}
& x(t+1)=F(u(t)) x(t)+G(u(t)), \quad x^{\#}=0, \\
& y(t)=H x(t) .
\end{aligned}
\]
where $F(\cdot)$ and $G(\cdot)$ are polynomial matrices and $H$ is a linear map. The characteristic feature of state-affine systems is the linear occurrence of the state variable.

The above realizability result establishes state-affine systems as a very useful and natural class of systems with respect to bounded maps. Input/output maps realizable by state-affine systems (equivalently, finitely realizable bounded maps) are precisely those whose "observation space" (a \textit{linear} space directly associated to the map) is finite-dimensional. These
 and other results indicate that state-affine systems play an "approximation" role in the discrete theory similar to the role
\label{op:12}of internally-bilinear systems in the continuous-time context (see for instance, \citealp{fliess1974matrices,fliess1975outil} and \citealp{sussmann1975semigroup}).

We then restrict our attention to realizations by state-affine systems. Canonical realizations can now be obtained (for bounded maps) without recourse to $k$-spaces. In fact, it is now natural to define \textit{span-canonical} state-affine systems as observable systems such that the \textit{linear span} of the reachable states is the full state-space $k^{n}$. We then prove that \textit{span-canonical realizations of a given bounded finitely realizable} $f$ \textit{always exist}. Further, \textit{any two such realizations can be related by a linear change of coordinates in the state-space}. Finally, \textit{a realization is span canonical if and only if its dimension} $n$ \textit{is smallest possible among all state-affine realizations of the same input/ output map}.

The above-mentioned results are proved by first associating to the bounded map $f$ the \textit{exponent} formal power series $\varphi_{f}$ obtained directly from the Volterra series $\psi_{f}$. As opposed to $\psi_{f}$, the exponent series is a power series in \textit{noncommutative} variables. The transformation $\psi_{f} \mapsto \varphi_{f}$ permits the explicit consideration of dynamics. We then remark that state-linear realizations are in a one-to-one correspondence with \textit{representations} of $\varphi_{f}$. (The concept of representation of a noncommutative power series was introduced by \citealp{schutzenberger1961definition} as a generalization of automata-theoretic ideas, and has been rediscovered since by many authors, notably in the context of stochastic automata. Representations have been called \textit{sequential systems} by \citealp{carlyle1971realizations} and \textit{linear-space automata} by \citealp{turakainen1972minimization}. A fairly complete account of representations, also called "automata with multiplicities", may be found in \citealp{eilenberg1974automata}. The notion of representation which we use is in fact a minor variation of that in the literature.) The idea of associating representations to systems is not totally original, since an analogous method was used by \citet{fliess1973sur} to study the special case of internally-bilinear systems. We give a brief but self-contained exposition of those results on representations which are relevant to our work.

\label{op:13}%
An interesting observation is that, under the above one-to-one correspondence, span-reachability and observability for state-affine systems corresponds precisely to (automata-theoretic) reachability and observability for representations. Realizability of $f$ can then be studied via the \textit{behavior (Hankel) matrix} $\mathcal{B}(f)$ of $\varphi_{f}$, and minimal state-affine realizations can be obtained by operating on $\mathcal{B}(f)$, using the methods developed for representations by \citet{fliess1972sur,fliess1975outil}. We give such a realization procedure, which generalizes and unifies known algorithms for linear and for the various kinds of bilinear systems. An interpretation of $\mathcal{B}(f)$ follows from the remark that the observation space of $f$ is isomorphic to the row space of $B(f)$.

We sharpen the result on algebraic difference equations by proving that a \textit{bounded map is finitely realizable if and only if it satisfies an input/output difference equation which is linear in the output}. This is a new result even in the (very special) cases of internally-bilinear systems and multilinear input/output maps.

\citet{schutzenberger1961definition} gave a generalization to power series of Kleene's theorem: A \textit{language} $L$ \textit{is recognizable by a finite automaton if and only if} $L$ \textit{can be described by a regular expression}. This generalization can be applied to $\varphi_{f}$ via the above correspondence between state-affine systems and representations. The conclusion is that $f$ \textit{has a state-affine realization if and only if} $\varphi_{f}$ \textit{is rational} i.e., if and only if $\varphi_{f}$ can be obtained from polynomials by a finite number of additions, multiplications, and inversions. As a consequence, it becomes possible to apply the standard calculus of interconnections of automata (see, for instance, \citealp{eilenberg1974automata}) to find $\varphi_{f}$ (and therefore $\psi_{f}$) from any state-affine realization of $f$, and, viceversa, to construct realizations given rational expressions for $\varphi_{f}$.

We also define the subclass of \textit{finite} maps $f$, corresponding to the restriction that the \textit{total degree} of $\psi_{f}$ should be finite. We show that the span canonical realization of such maps can be decomposed as a cascade of linear systems and memory-free nonlinearities. The existence of such decompositions characterizes finite maps.

\label{op:14}%
Returning to the case of general polynomial response maps, we study the class $\mathrm{QR}(f)$ of quasi-reachable realizations of a fixed $f$. Under the natural ordering induced by simulation $\left(\Sigma_{1}\right.$ \textit{simulates} $\Sigma_{2}$ when there is a dominating $k$-system morphism from $\Sigma_{1}$ into $\Sigma_{2}$), $\mathrm{QR}(f)$ turns out to be a complete lattice. The minimal element of $\mathrm{QR}(f)$ is $\Sigma_{f}$, and the largest element is the realization having the input space $\Omega$ as its state-space.
 The join in $\mathrm{QR}(f)$ of $\Sigma_{1}$ and $\Sigma_{2}$ is a subsystem of a parallel connection of $\Sigma_{1}$ and $\Sigma_{2}$ (a fibre product). The lattice operations permit constructing (sometimes simpler) realizations from given ones. The relevance of $\mathrm{QR}(f)$ lies mainly in the understanding of the relationships that hold among different realizations, and also in the development of alternative realization theories. For example, some authors use a different definition of "canonical", as "initial (not necessarily \textit{algebraically}) observable realization". A theory using this alternative definition will be easily derived from the consideration of the order properties of the subset of observable realizations. Other subsets (in fact, sublattices) of interest are also studied. Using $\operatorname{AR}(f)$ permits obtaining further insight also into the role of arbitrary $k$-systems as a "completion" of the subset of polynomial systems. Moreover, it also allows the construction of counterexamples to the existence of polynomial canonical realizations even if "canonical" is interpreted differently than quasi-reachable and algebraically observable (for example, the above alternative, or as "final quasi-reachable").

Another application of $\mathrm{QR}(f)$ will be in the study of \textit{normal} realizations of $f$. (The notion of normality is closely tied in algebraic geometry with that of nonsingularity; in fact, both coincide in dimension one.) We shall construct a complete lattice of normal realizations of $f$, and shall obtain a normalization of any element of $\mathrm{QR}(f)$. (For example, a system whose state-space is a cusp will have as its normalization a system whose state-space is a line.) Normality permits proving a strong version of the uniqueness theorem for canonical realizations: Two \textit{abstractly canonical} (i.e., reachable and [not necessarily algebraically] observable) \textit{normal} polynomial realizations are necessarily \textit{isomorphic}. In particular, returning to the "naive" definition of polynomial system with $X=k^{n}$,
\label{op:15}these are always normal, so any two such canonical realizations of a given $f$ must be equal up to a polynomial change of coordinates. The proof uses in part some well-known but nontrivial algebraic-geometric facts.

A number of results provide necessary and/or sufficient conditions for $\Sigma_{f}$ being a polynomial system (among them: finitely generated observation algebra, existence of integral or recursive difference equations, $f$ bounded). In many applications these conditions hold directly; for example, it is usual to define i/o maps via "autoregressive" (i.e., recursive) equations, while other problems give rise to bounded maps: internally-bilinear $f$ with nuclear reactor and population models (see e.g. \citealp{mohler1973bilinear}), multilinear $f$ in image processing (e.g., \citealp{kamen1979relationship}), finite $f$ in some stochastic filtering contexts (e.g., \citealp{marcus1979discrete}). In other applications, an approximation of the original problem may result in these conditions being true (for example, disregarding higher-order harmonics corresponding to a periodic input, for systems with "mild" nonlinearities).

In the above context, $k$-systems may be seen just as a technical tool which facilitates the study of polynomial systems, which can be implemented in turn by sets of simultaneous polynomial difference equations. There are cases, however, in which $\Sigma_{f}$ may not be polynomial, even if it admits a polynomial (noncanonical) realization. In fact, this was the original motivation for introducing more general systems. In those cases, it becomes of interest to find a way of somehow "programming" explicitly the resulting $k$-system. This will be accomplished in the last section, resulting in a description for $\Sigma_{f}$ in terms of locally rational transition and output maps in finitely many variables. Some remarks are also included there on the topic of determining a bound for the number of equations needed to represent $\Sigma_{f}$ when this is polynomial.

Also in the last section, we shall briefly discuss generalizations to arbitrary $k$-spaces of input and output values, and to nonequilibrium initial states. The first generalization allows the inclusion of algebraic constraints, for example, for $k=$ reals, the restriction to inputs of a fixed magnitude. The second allows treating i/o maps for which the dependence itself of present outputs on past inputs is allowed to change in time. The work closes with some remarks on other results and open problems and suggestions for further research.

\label{op:16}%
\section{Algebraic Preliminaries}\label{ch:2}

In this section we shall briefly discuss some basic notions of algebraic geometry which are used in the sequel. The main object to be introduced is \textit{the set of} $k$-\textit{points of an affine} $k$-\textit{scheme} ($k=$ field); we shall simply call this object a "$k$-space".

The study of $k$-spaces is \textit{per se} not included in standard texts in algebraic geometry; usually one studies instead the set of all points of a scheme and then tries to deduce special properties of the $k$-points. For instance, the study of finitely generated reduced schemes over the reals $\mathbb{R}$, i.e. the study of solutions of polynomial equations with real coefficients
\begin{equation*}
P_{i}\left(x_{1}, \ldots, x_{n}\right)=0, i=1, \ldots, r, \tag{*}
\end{equation*}
focuses on the \textit{complex} solutions $\left(x_{1}
, \ldots, x_{n}\right)$ in $\mathbb{C}^{n}$ of (*) rather than on the real solutions. This approach has proved highly appealing, since statements concerning the set of complex solutions do not have to elucidate certain exceptional or degenerate cases. In fact, it is customary to proceed a step further and embed the corresponding problem in \textit{projective} space.

To infer the nature of the set of $k$-points from the properties of the entire scheme is not always a straightforward matter; it may involve, in fact, nonalgebraic (e.g. differential-geometric) arguments.

For purposes of this exposition we have adopted the procedure of defining $k$-spaces directly. We shall give here the definitions and the main results needed later. With the exception of some trivial statements, no proofs will be given for those facts for which a precise reference is available (and given). There is unfortunately no single source for the results quoted. We rely mainly on \citet{bourbaki1972commutative} and \citet{dieudonne1974cours}. Except for some matters of style and emphasis, no original contributions appear in this section.

\label{op:17}%
\subsection{$k$-Reduced Algebras}\label{sec:1}

Let $k$ denote an arbitrary but infinite field, to be fixed throughout the discussion. Recall that a (\textit{commutative) $k$-algebra} (or, simply, an \textit{algebra}) is a pair ($A, \varphi$), where $A$ is a commutative ring with identity and $\varphi: k \rightarrow A$ is a ring homomorphism with $\varphi(1)=1$. We shall denote such algebras by the corresponding ring $A$ and identify $k$ with its (isomorphic) image $\varphi(k)$. Thus the scalar product r.a, $r \in k$, $a \in A$, is the multiplication in $A$. The field $k$ may always be viewed as a $k$-algebra, with $\varphi=$ identity.

A \textit{homomorphism of} $k$-\textit{algebras} $\mu: A \rightarrow B$ will mean a homomorphism whose restriction to $\varphi(k)=k$ is the identity.

We adopt the following notation conventions:

(i) the first few upper-case Latin letters $A, B, C, \ldots$ denote $k$-algebras;

(ii) $A \otimes_{k} B$ or simply $A \otimes B$ is the tensor product algebra;

(iii) if $A$ is an integral domain, then $Q(A)$ denotes the quotient field of $A$;

(iv) $k\left[T_{1}, \ldots, T_{r}\right]$ denotes the ring of polynomials in $r$ variables over $k$; when $r=1$ we write simply $k[T]$;

(v) "homomorphism" will \textit{always} mean $k$-algebra homomorphism;

(vi) Hom (A, B) denotes the set of all homomorphisms $A \rightarrow B$.

\begin{definition}\label{def:1.1}
A k-\textit{ideal M of a $k$-algebra A is the kernel of a homomorphism} $A \rightarrow k$. \textit{The} $k$-\textit{radical} $\operatorname{rad}_{k} A$ \textit{of} $A$ \textit{is the intersection of all} $k$-\textit{ideals of} $A$.
\end{definition}

Let $M$ be a $k$-ideal of $A$. Since $A$ is a $k$-algebra, $A / M \simeq k$ is a field, so $M$ is maximal, but \textit{not every maximal ideal of} $A$ \textit{is a} $k$-\textit{ideal}. For instance, let $k=\mathbb{R}$ and $A=\mathbb{R}[T]$. Then the ideal $M$ generated by $x^{2}+1$ is maximal (because $x^{2}+1$ is an irreducible polynomial) but $A / M \cong \mathbb{C} \neq \mathbb{R}$. In the particular case in which $A$ is
\label{op:18}finitely generated and k is algebraically closed, all maximal ideals are $k$-ideals; this is a consequence of Hilbert's Nullstellensatz; see \citet[V.3.3, Proposition 2]{bourbaki1972commutative}.

\textit{There is a bijective correspondence between} $k$-\textit{ideals and homomorphisms} $A \rightarrow k$. Indeed, let $\mu, v: A \rightarrow k$ and suppose that ker $\mu=$ ker $\nu$. Take any $x$ in $A$. Since $\mu(x)$ is in $k \subset A$ and $\mu$ is the identity on $k, x-\mu(x)$ belongs to ker $\mu=$ ker $v$. Thus $0=\nu(x-\mu(x))=\nu(x)-\mu(x)$. So the maps $\mu$ and $\nu$ are equal.

Let $X$ be a set. The set $k^{X}$ of all functions $X \rightarrow k$ is a $k$-algebra under the pointwise operations, the constant functions constituting the subring isomorphic to $k$. A subalgebra of $k^{X}$ is called an \textit{algebra of functions} on $X$.

\begin{lemma}\label{lem:1.2}
The following statements are equivalent:
\begin{enumerate}
\item[(a)] $\operatorname{rad}_{k} A=\{0\}$.
\item[(b)] A is isomorphic to an algebra of functions.
\end{enumerate}
\end{lemma}

\begin{proof}
(b) implies (a). Let $A$ be identified with a subalgebra of $k^{X}$. For each $x$ in $X$, the evaluation map
\[
e_{x}: k^{X} \rightarrow k: \varphi \mapsto \varphi(x),
\]
restricted to $A$ is a homomorphism, hence ker $e_{X} \mid A$ is a $k$-ideal. Clearly then $\operatorname{rad}_{k} A \subseteq \cap\left\{\right.$ ker $e_{x}, x$ in $\left.X\right\}=\{0\}$.

(a) implies (b). Let $X(A)$ denote the set of all homomorphisms $\mu: A \rightarrow k$. Define $\iota: A \rightarrow k^{X(A)}$ by evaluation:
\[
\iota(a)(\mu):=\mu(a)
\]
Then $\iota$ is a homomorphism; moreover, it is in fact one-to-one. To prove this, assume that $\iota(a)=0$, i.e. $\iota(a)(\mu)=\mu(a)=0$ for some $a$ in $A$ and for all $\mu$ in $X(A)$. Then $a$ is in the kernel of every $\mu: A \rightarrow k$, i.e. in the $k$-radical of $A$, which is 0 by hypothesis. So $a=0$. Therefore $A \simeq \iota(A)$.

\label{op:19}%
Let us observe the duality implicit in the preceding arguments. Homomorphisms $A \rightarrow k$ may be viewed as \textit{points} on which elements of $A$ act by evaluation.

\textit{This duality is fundamental in algebraic geometry}.
\end{proof}

\begin{definition}\label{def:1.3}
\textit{An algebra satisfying the conditions in} Lemma~\ref{lem:1.2} \textit{is called a} $k$-\textit{reduced algebra}.
\end{definition}

For every $k$-algebra $A$, the quotient ring $A / \operatorname{rad}_{k} A$ is $k$-reduced.

An important fact related to Definition~\ref{def:1.3} is Hilbert's Nullstellensatz, which can be phrased as follows: A \textit{finitely generated algebra A over an algebraically closed field} $k$ \textit{is} $k$-\textit{reduced if and only if} $A$ \textit{has no nonzero nilpotents}.

A recent generalization of this celebrated result (\citealp{dubois1967nullstellensatz}, \citealp{dubois1970algebraic}) is the following: A \textit{finitely generated algebra} $A$ \textit{over a maximally ordered field} $k$ (e.g. $k=\mathbb{R}$) is $k$-\textit{reduced if and only if for any} $x_{i}$ \textit{in} $A$ \textit{the relation} $\sum_{i=1}^{n} x_{i}^{2}=0$ \textit{implies} $x_{i}=0$ for all $i$.

The main idea in what follows is to view the elements of a $k$-reduced algebra $A$ as functions on $X(A)$. Take $x$ in $X(A)$, $a$ in $A$. Viewed as a function $\iota(a)$ on $X(A)$, $a$ has the value $\iota(a)(x)=x(a)$ at $x$. Except when discussing certain delicate points, we shall therefore identify $A$ with its image under $\iota$ and so we shall write $\iota(a)$ as $a$ and $\iota(a)(x)$ as $a(x)$. The fact that $x$ is a homomorphism means that the algebra operations in $A$ are now
 represented as pointwise operations: $(a b)(x)=a(x) b(x)$.

It is worth keeping in mind the following:

\begin{example}\label{ex:1.4}
Let $A:=k\left[T_{1}, \ldots, T_{n}\right]$. Since $k$ is infinite, a polynomial $a\left(T_{1}, \ldots, T_{n}\right)$ can be identified with the polynomial function $k^{n} \rightarrow k:\left(x_{1}, \ldots, x_{n}\right) \mapsto a\left(x_{1}, \ldots, x_{n}\right)$. Thus, by Lemma~\ref{lem:1.2}, A is $k$-reduced. A homomorphism $\mu: A \rightarrow k$ is completely determined by giving values of
\label{op:20}$\mu\left(T_{1}\right), \ldots, \mu\left(T_{n}\right)$ in $k$. Conversely, for any choice of $\left(x_{1}, \ldots, x_{n}\right)$ in $k^{n}$ there is a homomorphism $\mu: A \rightarrow k$ defined by $\mu\left(T_{i}\right)=x_{i}$. Thus one identifies $X(A)$ with $k^{n}$ and $\iota: A \rightarrow k^{X(A)}$ with the assignment: polynomial $\mapsto$ polynomial function.
\end{example}

\begin{lemma}[Canonical factorizations]\label{lem:1.5}
Let $\tau: A \rightarrow B$, where $B$ is $k$-reduced.
\begin{enumerate}
\item[(i)] Then there exists a factorization
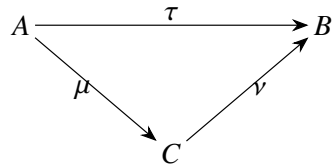
\[
\begin{tikzpicture}[x=1cm,y=1cm]
\node (A) at (0,0) {$A$};
\node (B) at (4,0) {$B$};
\node (C) at (2,-1.7) {$C$};
\draw[arr] (A) -- node[above,lab] {$\tau$} (B);
\draw[arr] (A) -- node[left,lab] {$\mu$} (C);
\draw[arr] (C) -- node[right,lab] {$\nu$} (B);
\end{tikzpicture}
\]
where $\mu$ is onto, $\nu$ is one-to-one, and $C$ is $k$-reduced.
\item[(ii)] Further, let two factorizations $A \xrightarrow{\mu_{1}} C_{1} \xrightarrow{\nu_{1}} B$ and $A \xrightarrow{\mu_{2}} C_{2} \xrightarrow{\nu_{2}} B$ of $\tau$ be given, with $\mu_{1}$ onto and $\nu_{2}$ one-to-one. Then there exists a unique $\eta: C_{1} \rightarrow C_{2}$ such that the following diagram commutes:
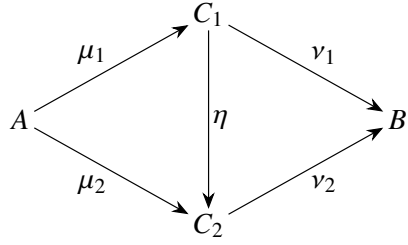
\[
\begin{tikzpicture}[x=1cm,y=1cm]
\node (A) at (0,0) {$A$};
\node (C1) at (2.5,1.4) {$C_1$};
\node (C2) at (2.5,-1.4) {$C_2$};
\node (B) at (5,0) {$B$};
\draw[arr] (A) -- node[above left,lab] {$\mu_1$} (C1);
\draw[arr] (A) -- node[below left,lab] {$\mu_2$} (C2);
\draw[arr] (C1) -- node[right,lab] {$\eta$} (C2);
\draw[arr] (C1) -- node[above right,lab] {$\nu_1$} (B);
\draw[arr] (C2) -- node[below right,lab] {$\nu_2$} (B);
\end{tikzpicture}
\]
\end{enumerate}
\end{lemma}

\begin{proof}
The existence of factorizations ($\mu, C, \nu$) and of the map $\eta$ are elementary algebraic facts. We must only prove that any such $C$ is $k$-reduced. But $C \simeq \nu(C) \subseteq B$, and $B$ is an algebra of functions. So $C$ is also an algebra of functions.
\end{proof}

The above lemma plays an important role in characterizing canonical realizations, analogous to the role of its linear variant, Zeiger's lemma (\citealp[Chapter 10, Lemma 6.2]{kalman1969topics}), in linear system theory.

\begin{remark}\label{rem:1.6}
\textit{Products of} $k$-\textit{reduced algebras are} $k$-\textit{reduced}. Indeed, assume that $A_{j}$ is a subalgebra of $k^{X_{j}}$ for each $j$ in $J$. Then $\prod_{J} A_{j}$ is an algebra of functions on the disjoint union of the $X_{j}$.
\end{remark}

\label{op:21}%
It follows from the definition of the tensor product ⊗ that for any two homomorphisms $\mu: A \rightarrow C$ and $v: B \rightarrow C$ there exists a unique homomorphism
\begin{equation}\label{eq:1.7}
\mu \otimes v: A \otimes B \rightarrow C: \sum a_{i} \otimes b_{i} \mapsto \sum \mu\left(a_{i}\right) v\left(b_{i}\right) .
\end{equation}

\begin{lemma}\label{lem:1.8}
Let A, B, C be $k$-algebras. Then:
\begin{enumerate}
\item[(a)] The assignment $(\mu, \nu) \rightarrow \mu \otimes v$ establishes a bijection between Hom ($A, C$) $\times \operatorname{Hom}(B, C)$ and $\operatorname{Hom}(A \otimes B, C)$.
\item[(b)] $X(A \otimes B)$ is naturally identified with $X(A) \times X(B)$.
\item[(c)] If $A, B$ are $k$-reduced, so is $A \otimes B$.
\end{enumerate}
\end{lemma}

\begin{proof}
(a) Write $j_{1}: A \rightarrow A \otimes B: a \mapsto a \otimes 1$ and $j_{2}: B \rightarrow A \otimes B: b \mapsto 1 \otimes b$. We then define the inverse of the assignment $(\mu, \nu) \mapsto \mu \otimes \nu$ as $\operatorname{Hom}(A \otimes B, C) \rightarrow \operatorname{Hom}(A, C) \times \operatorname{Hom}(B, C): \gamma \mapsto\left(\gamma \circ j_{1}, \gamma \circ j_{2}\right)$.

(b) Apply (a) to $C:=k$.

(c) Since $X(A \otimes B)$ has been identified to $X(A) \times X(B)$, it is enough to prove the following: if $(\mu, \nu)(c)=0$ for some $c$ in $A \otimes B$ and for all ($\mu, v$) in $X(A) \times X(B)$, then $c=0$. Assume that $c \neq 0$, and express $c$ as a finite \textit{sum} $\sum a_{i} \otimes b_{i}$ with the $b_{i}$ linearly independent over $k$. For any $\mu$ in $X(A)$, consider the element $b:=\sum \mu\left(a_{i}\right) b_{i}$ of $B$. For any $y$ in $X(B)$, $\nu(b)=\sum \mu\left(a_{i}\right) \nu\left(b_{i}\right)=(\mu, \nu)(c)=0$. Since $B$ is $k$-reduced, the function $b=0$. The $b_{i}$ being linearly independent, it follows that $\mu\left(a_{i}\right)=0$ for all $i$. Since $A$ is $k$-reduced and $\mu$ was arbitrary, $a_{i}=0$ for all $i$. So $c=0$.
\end{proof}

\begin{remark}\label{rem:1.9}
Recall that, in any category, the \textit{coproduct} of a family $\left\{T_{\lambda}, \lambda \in \Lambda\right\}$ of objects is defined as an object $T=\frac{\pi}{\Lambda} T$ and morphisms $v_{\lambda}: T_{\lambda} \rightarrow T$ which satisfy: for any object $\Theta$ and morphisms $\theta_{\lambda}: T_{\lambda} \rightarrow \Theta$ there exists a unique morphism $v: T \rightarrow \Theta$ such that $v \circ v_{\lambda}=\theta_{\lambda}$ for all $\lambda$. An equivalent way of expressing the properties Lemma~\ref{lem:1.8}(a) and Lemma~\ref{lem:1.8}(c) is to say that $A \otimes B$ (together with the inclusions
\label{op:22}$a \mapsto a \otimes 1, b \mapsto 1 \otimes b)$ is the coproduct of $A$ and $B$ in the category $\mathbf{Alg}_{k}^{\text {red }}$ of $k$-reduced algebras, having ($k$-algebra) homomorphisms as its morphisms.
\end{remark}

By induction, the coproduct of any \textit{finite} family $A_{1}, \ldots, A_{n}$ in $\mathbf{Alg}_{k}^{\text {red }}$ is $A_{1} \otimes \ldots \otimes A_{n}$. \textit{In the case of the category} $\mathbf{Alg}_{k}$ \textit{of all} $k$-\textit{algebras} (as in all \textit{algebraic} categories), \textit{an arbitrary} (not finite) \textit{coproduct exists}. In fact, this coproduct can be obtained by the following direct-limit construction. Consider first the disjoint union of all $\left\{\otimes_{\Lambda_{0}} A_{\lambda}, \Lambda \supseteq \Lambda_{0}=\right.$ finite $\}$. Then, for each pair $\Lambda_{0} \subseteq \Lambda_{0}^{\prime}$, identify $\otimes_{\Lambda_{0}} A_{\lambda}$ with the image of the inclusion morphism $\otimes_{\Lambda_{0}} A_{\lambda} \rightarrow \otimes_{\Lambda_{0}^{\prime}} A_{\lambda}$ obtained by adding coordinates equal to 1 in the positions $\lambda$ in $\Lambda_{0}^{\prime}-\Lambda_{0}$. The algebra $A$ obtained from this construction is the coproduct of the $A_{\lambda}$ in $\mathbf{Alg}_{k}$. \textit{If all} $A_{\lambda}$ \textit{are} $k$-\textit{reduced, so is} $A$. Indeed, by the construction just sketched, an element $x$ of $A$ is in some (finite) tensor product $\otimes_{\Lambda_{0}} A_{\lambda}$, and every homomorphism $\otimes_{\Lambda_{0}} A_{\lambda} \rightarrow k$ extends to a homomorphism $A \rightarrow k$ (just define $A_{\lambda} \rightarrow k$ arbitrarily if $\lambda$ is not in $\Lambda_{0}$). Assume now that $x$ is in the kernel of all homomorphisms $A \rightarrow k$. Then $x$ is in the kernel of all homomorphisms $\otimes_{\Lambda_{0}} A_{\lambda} \rightarrow k$, so by Remark~\ref{rem:2.8}(c), $x=0$. Therefore $A$ is $k$-reduced, and \textit{a fortiori} $A$ is the coproduct of the $A_{\lambda}$ in $\mathbf{Alg}_{k}^{\text {red }}$. Thus \textit{arbitrary coproducts exist in the category of} $k$-\textit{reduced algebras}. Moreover, the construction shows that $A=\coprod A_{\lambda}$ includes all the $A_{\lambda}$, and that $A$ is \textit{generated by the} $k$-\textit{algebras} $A_{\lambda}$. Finally, observe that the categorical definition of coproduct given above, applied to $\Theta:=k$, shows that the set $X(A)$ of morphisms $A \rightarrow k$ is identified, through composition with the inclusion homomorphisms $A_{\lambda} \rightarrow A$, with the set $\Pi_{\Lambda} X\left(A_{\lambda}\right)$ of families of homomorphisms $A_{\lambda} \rightarrow k$.

\begin{definition}\label{def:1.10}
\textit{Let} $A$ \textit{be an algebra of functions on a set} $X$. \textit{Take} $x, y$ \textit{in} $X$. \textit{Then} $A$ \textit{separates} $x$ \textit{and} $y$ \textit{iff there exists an} $a$ \textit{in A such that} $a(x) \neq a(y)$.
\end{definition}

\begin{definition}\label{def:1.11}
\textit{Let} $A$ \textit{be a subalgebra of the} $k$-\textit{reduced algebra} $B$. \textit{Then} $A$ is \textit{maximally separating with respect to} $B$ \textit{iff there is no} $C$ \textit{such that} $A \varsubsetneqq C \subseteq B$, \textit{and} $C$ \textit{separates the same points of} $X(B)$ \textit{as} $A$.
\end{definition}

\label{op:23}%
We shall later make use of the following

\begin{example}\label{ex:1.12}
Let $A:=k\left[T_{1} T_{2}^{n}, n \geq 0\right]$ be the subalgebra of $k\left[T_{1}, T_{2}\right]$ generated by all the monomials $T_{1} T_{2}^{n}, n \geq 0$. Then $A$ \textit{is maximally separating with respect to} $B:=k\left[T_{1}, T_{2}\right]$. Note that $A$ identifies all points of the form $\left(0, x_{2}\right)$ and separates any other pair of points in $X(B)=k^{2}$. Assume that there is a $C$ as in Definition~\ref{def:1.11}. Take $P\left(T_{1}, T_{2}\right)=\sum a_{i j} T_{1}^{i} T_{2}^{j}$ in $C$. Since $C$ separates no more points than $A, P\left(0, x_{2}\right)=P(0,0)$ for all $x_{2}$. Therefore $\sum a_{0 j} T_{2}^{j}$ is a constant polynomial. Since $k$ is infinite, $P$ has no terms in $T_{2}^{j}$, $j>0$. So $P=a_{00}+\sum_{i>0} a_{i j} T_{1}^{i-1}\left(T_{1} T_{2}^{j}\right)$ is in $A$.
\end{example}

Recall (\citealp[Chapter V]{bourbaki1972commutative}) that the algebra $A$ is \textit{integral} over the subalgebra $B$ iff every $a$ in $A$ satisfies an equation
\[
a^{n}+b_{1} a^{n-1}+\ldots+b_{n}=0,
\]
for some $b_{j}, j=1, \ldots, n$ in $B$ and for some $n \geq 0$.

A $k$-algebra $A$ is \textit{finitely generated} iff there exists a finite subset $\left\{a_{1}, \ldots, a_{s}\right\}$ of $A$ such that each element of $A$ can be expressed as a finite combination of $a_{1}, \ldots, a_{s}$ using sums, products, and multiplication by elements of $k$.

\begin{lemma}\label{lem:1.13}
Let the $k$-algebra $A$ be an integral domain, and assume that $A$ is integral over a subalgebra $B$.
\begin{enumerate}
\item[(a)] If $B$ is a finitely generated $k$-algebra and $Q(A)$ is a finitely generated field extension of $Q(B)$, then $A$ is a finite B-module and so also finitely generated as a $k$-algebra.
\item[(b)] If $A$ is a finitely generated $k$-algebra then $B$ is also finitely generated.
\end{enumerate}
\end{lemma}

\begin{proof}
(a) This is an easy consequence of \citet[V.3.2, Theorem 2]{bourbaki1972commutative}.

\label{op:24}%
(b) Let $a_{1}, \ldots, a_{n}$ generate $A$. Each $a_{i}$ satisfies an integral equation with coefficients $b_{i j}$ in $B$; call $C \subseteq B$ the $k$-algebra generated by all the $b_{i j}$. Then $A$ is integral over $C$, and so $B$ is integral over $C$. By ($a$), $B$ is finitely generated.
\end{proof}

\subsection{The Zariski Topology}\label{sec:2}

For the rest of this section, unless the contrary is explicitly stated, $A, B, C, \ldots$ will always denote $k$-\textit{reduced} algebras.

Recall that $X(A)$ denotes the set of all homomorphisms $\mu: A \rightarrow k$.

We now introduce an operator $V$ which assigns a subset $V(S)$ of $X(A)$ to each subset $S$ of $A$. It is defined as
\begin{equation}\label{eq:2.1}
V(S):=\{x \text{ in } X(A) \mid a(x)=0 \text{ for all } a \text{ in } S\} .
\end{equation}

Thus $V(S)$ is the set of solutions of the simultaneous equations $a(x)=0$, $a$ in $S$. Since, by convention $a(x)$ means $x(a)$, where $x: A \rightarrow k$ is a homomorphism, we can give the equivalent definition
\begin{equation}\label{eq:2.2}
V(S):=\{x: A \rightarrow k \mid \text{ ker } x \supseteq S\} .
\end{equation}

\begin{proposition}\label{prop:2.3}
The operator $V$ satisfies the following properties:
\begin{enumerate}
\item[(a)] $V(S)=X(A)$ if and only if $S=\{0\}$.
\item[(b)] $V(A)=\varnothing$.
\item[(c)] $S \subseteq T$ implies $V(T) \subseteq V(S)$.
\item[(d)] Let $\langle S\rangle$ denote the ideal of $A$ generated by $S$. Then $V(S)=V(\langle S\rangle)$.
\item[(e)] $V\left(\bigcup\left\{S_{\lambda}, \lambda\right.\right.$ in $\left.\left.\Lambda\right\}\right)=\bigcap\left\{V\left(S_{\lambda}\right), \lambda\right.$ in $\left.\Lambda\right\}$.
\item[(f)] Let $I . J$ denote the product of the ideals $I, J$ of $A$, i.e. the ideal generated by $\{a b$, $a$ in $I$, $b$ in $J\}$. Then
\[
V(I) \cup V(J)=V(I \cap J)=V(I . J) .
\]
\end{enumerate}
\end{proposition}

\label{op:25}%

\begin{proof}
(a), (b), (c), (d), and (e) are easy consequences of the definition of $V$.

We now prove (f). By (c), $I \cap J \subseteq I$ and $I \cap J \subseteq J$ imply $V(I) \cup V(J) \subseteq V(I \cap J)$. Similarly, I.J $\subseteq I \cap J$ implies $V(I \cap J) \subseteq V(I . J)$. Consequently, it will be sufficient to show that if $x$ is not in $V(I) \cup V(J)$ then $x$ is not in $V(I . J)$. But if $x$ belongs to neither $V(I)$ nor $V(J)$, there are $a$ in $I$, $b$ in $J$ such that $a(x) \neq 0, b(x) \neq 0$. So $(a . b)(x)=a(x) b(x) \neq 0$. Since $a b \in I . J, x$ is not in $V(I . J)$.

It follows from (a), (b), (e), and (f) above that $\{V(S), S \subseteq A\}$ \textit{is the family of closed sets for a topology on} $X(A)$, called the \textit{Zariski topology}. Therefore we shall henceforth refer to sets of the type $V(S), S \subseteq A$, as \textit{closed sets}.

Occasionally it is convenient to define closed subsets of $X(A)$ in an indirect manner. Let $A, B$ be any $k$-reduced algebras, and consider the tensor product $A \otimes B$. Let $S$ be any subset of $A \otimes B$. For any $s=\sum a_{i} \otimes b_{i}$ in $S$ and any $x$ in $X(A)$ let $s(x)$ be the element $\sum a_{i}(x) b_{i}$ of $B$.
\end{proof}

\begin{lemma}\label{lem:2.4}
$V^{B}(S):=\{x \in X(A) \mid s(x)=0$ for all $s$ in $S\}$ is a closed subset of $X(A)$.
\end{lemma}

\begin{proof}
Let $\left\{b_{\lambda}, \lambda\right.$ in $\left.\Lambda\right\}$ be a basis for $B$ as a $k$-vector space. Each $s$ can be expressed in the form $\sum a_{\lambda, s} \otimes b_{\lambda}$ (finite sum), so $s(x)=\sum a_{\lambda, s}(x) b_{\lambda}$. Since $\left\{b_{\lambda}\right\}$ is linearly independent,
\[
V^{B}(S)=V\left(\left\{a_{\lambda, s} \text { in } A, s \text { in } S, \lambda \text { in } \Lambda\right\}\right) . \qedhere
\]
\end{proof}

We now define an operator $I$ which associates a subset $I(Z)$ (in fact, an ideal) of $A$ to each subset $Z$ of $X(A)$. It is defined as
\begin{equation}\label{eq:2.5}
I(Z):=\{a \text{ in } A \mid a(x)=0 \text{ for all } x \text{ in } Z\} .
\end{equation}

Thus $I(Z)$ is the annihilator of $Z$; it is evidently an ideal of $A$. We can define $I(Z)$ equivalently via homomorphisms, which gives
\label{op:26}%
\begin{equation}\label{eq:2.6}
I(Z):=\bigcap\{\operatorname{ker} x, x \text{ in } Z\} .
\end{equation}

\begin{proposition}\label{prop:2.7}
The operator I has the following properties:
\begin{enumerate}
\item[(a)] $I(X(A))=\{0\}$.
\item[(b)] $Z_{1} \subseteq Z_{2} \Rightarrow I\left(Z_{2}\right) \subseteq I\left(Z_{1}\right)$.
\item[(c)] $Z \subseteq V(I(Z)), \quad S \subseteq I(V(S))$.
\item[(d)] $I\left(\bigcup\left\{Z_{\lambda}, \lambda\right.\right.$ in $\left.\left.\Lambda\right\}\right)=\bigcap\left\{I\left(Z_{\lambda}\right), \lambda\right.$ in $\left.\Lambda\right\}$.
\item[(e)] $I(V(S))$ is the intersection of all the $k$-ideals containing $S$.
\item[(f)] $V(I(Z))=\left\{x \text { in } X(A) \mid \operatorname{ker} x \supseteq \bigcap_{z \in Z} \operatorname{ker} z\right\}$.
\end{enumerate}
\end{proposition}

\begin{proof}
(a), (b), (c), and (d) are easy consequences from \eqref{eq:2.5}, \eqref{eq:2.6}.

(e) Observe that, by \eqref{eq:2.6},
\[
I(V(S))=\bigcap\{\operatorname{ker} x, \quad x \text { in } V(S)\},
\]
and by \eqref{eq:2.2},
\[
x \text { is in } V(S) \text { iff } S \subseteq \operatorname{ker} x .
\]

(f) Similar to (e).
\end{proof}

\begin{remark}\label{rem:2.8}
Using Proposition~\ref{prop:2.3} and Proposition~\ref{prop:2.7} it is easy to verify the following facts:
\begin{enumerate}
\item[(i)] $I(V(I(Z)))=I(Z)$,
\item[(ii)] $V(I(V(Z)))=V(Z)$,
\item[(iii)] the Zariski closure of any $Z \subseteq X(A)$ is $V(I(Z))$, and
\item[(iv)] an ideal $I$ is of the form $I(Z)$ if and only if $I=I(V(I))$.
\end{enumerate}
Ideals as in (iv) are called \textit{closed}; this terminology comes from regarding $S \mapsto I(V(S))$ as an \textit{algebraic (not topological) closure operator}.
\label{op:27}%
One could have deduced (i)-(iv) from the fact that, by Proposition~\ref{prop:2.3}(c), Proposition~\ref{prop:2.7}(b) and Proposition~\ref{prop:2.7}(c), the pair $\{V, I\}$ constitutes a \textit{duality} or \textit{Galois connection} (\citealp[§ 51]{kurosh1963general}). Finally, it also follows that $\{V, I\}$ \textit{establish an inclusion-reversing bijection between closed ideals of} $A$ \textit{and closed subsets of} $X(A)$.
\end{remark}

\begin{lemma}\label{lem:2.9}
For any ideal $I$ of $A, X(A / I)$ can be naturally identified with $V(I)$; therefore the canonical map $\pi: A \rightarrow A / I$ can be naturally identified with the restriction map $a \mapsto a \mid V(I)$, where $a$ is viewed as a function on $X(A)$. Furthermore, $I$ is closed if and only if $A / I$ is $k$-reduced.
\end{lemma}

\begin{proof}
From elementary algebraic considerations, the injection $X(A / I) \rightarrow X(A): y \mapsto y \circ \pi$ shows that there is a bijection between functions $y$ in $X(A / I)$ and those functions $x=y \circ \pi$ in $X(A)$ for which $x \mid I=0$, that is, $I \subseteq$ ker $x$. These are precisely the $x$ in $V(I)$.

To prove the second part, note that $k$-ideals of $A / I$ correspond via $\pi$ to those $k$-ideals of $A$ which contain $I$. Thus $\mathrm{rad}_{k} A / I=0$ precisely when $I$ is the intersection of all the $k$-ideals containing it. Applying Proposition~\ref{prop:2.7}(e) to Remark~\ref{rem:2.8}(iv) gives the proof.
\end{proof}

We need the following elementary topological

\begin{definition}\label{def:2.10}
\textit{A topological space} $Z$ \textit{is irreducible iff} $Z$ \textit{is not the union of two proper closed subsets, in other words, iff}
\[
Z=Z_{1} \cup Z_{2}, \quad Z_{i} \text { closed, implies } Z=Z_{1} \text { or } Z=Z_{2} .
\]
Clearly, $Z$ is irreducible iff any two nonempty open sets of $Z$ have a nonempty intersection, i.e. iff \textit{any open subset of} $Z$ \textit{is dense}. Therefore, \textit{irreducibility permits the use of local intuition and methods} (= arguments about neighborhoods) \textit{in the proof of global statements}.
\end{definition}

To apply the above concepts in our context, we must study more closely the Zariski topology in the spaces $X(A)$. This topology will
\label{op:28}in general not be Hausdorff ($=$ different points having disjoint neighborhoods), but it is true that \textit{each point of} $X(A)$ \textit{is a closed set}. Indeed, by $Proposition~\ref{prop:2.7}(f), \overline{\{x\}}=\{z \mid \operatorname{ker} z \supseteq \operatorname{ker} x\}$; since ker $x$ is a maximal ideal, it follows that $\overline{\{x\}}=\{x\}$ as wanted.

A set $Z \subseteq X(A)$ can be given the subspace topology; thus irreducibility of $Z$ is well defined. In particular, it is easy to see that a closed subset $V$ of $X(A)$ is irreducible iff
\begin{equation}\label{eq:2.11}
V \subseteq V_{1} \cup V_{2}, \ V_{i} \text{ both closed in } X(A), \text{ implies } V \subseteq V_{1} \text{ or } V \subseteq V_{2} .
\end{equation}

Recall that an ideal $P$ of $A, P \neq A$, is \textit{prime} iff, for any ideals $J_{1}, J_{2}, J_{1} \cdot J_{2} \subseteq P$ implies $J_{1} \subseteq P$ or $J_{2} \subseteq P$. Equivalently, $P$ is prime iff, for any $a, b$ in $A$, $a b$ in $P$ implies either $a$ is in $P$ or $b$ is in $P$; see \citet[II.1.1]{bourbaki1972commutative}. We then have:

\begin{lemma}\label{lem:2.12}
Let $V$ be a closed subset of $X(A)$. Then $V$ is irreducible if and only if $I(V)$ is prime.
\end{lemma}

\begin{proof}
["only if"] Assume $J_{1} \cdot J_{2} \subseteq I(V)$. Then, by Proposition~\ref{prop:2.3}(c), $V=V(I(V)) \subseteq V\left(J_{1}\right) \cup V\left(J_{2}\right)$. Since $V$ is irreducible, $V \subseteq V\left(J_{1}\right)$ or $V \subseteq V\left(J_{2}\right)$. Thus by \eqref{eq:2.6}(c,b) $J_{1} \subseteq I\left(V\left(J_{1}\right)\right) \subseteq I(V)$ or $J_{2} \subseteq I\left(V\left(J_{2}\right)\right) \subseteq I(V)$.

["if"] Assume $V=V_{1} \cup V_{2}, V_{i}$ closed. Then by Proposition~\ref{prop:2.7}(d) $I\left(V_{1}\right) \cdot I\left(V_{2}\right) \subseteq I\left(V_{1}\right) \cap I\left(V_{2}\right)=I\left(V_{1} \cup V_{2}\right)=I(V)$. This and $I(V)=$ prime imply $I\left(V_{1}\right) \subseteq I(V)$ or $I\left(V_{2}\right) \subseteq I(V)$. Thus $V=V I(V) \subseteq V I\left(V_{1}\right)=V_{1}$ or $V \subseteq V_{2}$.

In particular, $X(A)$ is irreducible iff $A$ is an integral domain.

It follows from Lemma~\ref{lem:2.12} that $V \mapsto I(V)$ and $I \mapsto V(I)$ establish an inclusion-reversing bijection between irreducible subsets of $X(A)$ and prime ideals of $A$.
\end{proof}

\label{op:29}%
\subsection{$k$-Spaces}\label{sec:3}

A $k$-space is the abstract version of a space of the type $X(A)$ :

\begin{definition}\label{def:3.1}
A $k$-\textit{space} $(X, A(X))$, \textit{or simply} $X$, \textit{is a set} $X$ and a $k$-\textit{algebra} $A(X)$ \textit{of functions} $X \rightarrow k$ \textit{such that the map} $X \rightarrow X(A(X)): x \mapsto e_{x} \quad(=$ \textit{evaluation at} $x)$ \textit{is bijective}. \textit{The elements of} $A(X)$ \textit{are the polynomial functions on} $X$.
\end{definition}

The terminology "polynomial functions" is motivated by Example \ref{ex:1.4}.

We always consider a $k$-space $X$ as a topological space, with topology induced from the Zariski topology on $X(A)$ by the bijection $X \simeq X(A(X))$. We shall see presently that all $k$-spaces are essentially of the type $(X(A), \iota(A))$, where $A$ is $k$-reduced and $\iota$ is the map introduced in Lemma~\ref{lem:1.2}.

\begin{definition}\label{def:3.2}
\textit{Let} $X, Y$ \textit{be} $k$-\textit{spaces. A map} $f: X \rightarrow Y$ \textit{is polynomial iff for each polynomial function} $b: Y \rightarrow k$ \textit{in} $A(Y)$ \textit{the composition} $b \circ f: X \rightarrow k$ \textit{is a polynomial function in} $A(X)$.
\end{definition}

We shall use $X, Y, Z$ to indicate $k$-spaces; $f: X \rightarrow Y$ will always mean a polynomial map.

It can be trivially verified that $k$-spaces as objects, together with polynomial maps as morphisms, constitute a category. In this category we have the following

\begin{lemma}\label{lem:3.3}
Every $k$-space is isomorphic to a $k$-space of the type $(X(A), \iota(A))$, where $A$ is $k$-reduced and $\iota$ is the map introduced in Lemma~\ref{lem:1.2}.
\end{lemma}

\begin{proof}[Sketch of proof]
Given a $k$-space $(X, A(X))$, let $A:=A(X)$. An easy check shows that $X \rightarrow X(A): x \mapsto e_{x}$ is the required isomorphism of $k$-spaces.
\end{proof}

Since all results will be stated up to isomorphism, we are justified by Lemma~\ref{lem:3.3} when carrying out our proofs about $k$-spaces only for those of type $(X(A), \iota(A))$. We continue to use our convention of identifying $\iota(A)$ with $A$ (so $A(X(A))=A$) when there is no danger of confusion.

\label{op:30}%
A major rôle is played by the $k$-spaces introduced by the following

\begin{definition}\label{def:3.4}
A $k$-\textit{space} $X$ \textit{is a variety iff} $A(X)$ \textit{is a finitely generated $k$-algebra}.
\end{definition}

Let $a_{1}, \ldots, a_{n}$ generate the $k$-algebra $A$. Generation is equivalent to ontoness of the $k$-algebra homomorphism $\mu: k\left[T_{1}, \ldots, T_{n}\right] \rightarrow A$ defined by the assignment $T_{i} \mapsto a_{i}$. By Example~\ref{ex:1.4} and Lemma~\ref{lem:2.9} we see that $(X(A), \iota(A))$ is isomorphic to $(V, B)$, where $V$ is the set of points $x$ in $k^{n}$ which satisfy all equations $f(x)=0, f$ in ker $\mu$, and where the functions in $B$ are the restrictions to $V$ of the polynomial functions $k^{n} \rightarrow k$. There are many possible representations ($V, B$) for each $A$, depending upon the choice of generators $a_{1}, \ldots, a_{n}$. Thus, (affine) \textit{varieties are the coordinate-free versions of those subsets of} $k^{n}$ \textit{defined by polynomial equations}. When studying varieties, it is usually simpler to deal with representations of the type ($V, B$). Note that if $\left(V_{1}, B_{1}\right)$ and $\left(V_{2}, B_{2}\right)$ are (concrete) varieties, where $V_{1} \subseteq k^{n}$ and $V_{2} \subseteq k^{m}$, a map $f: V_{1} \rightarrow V_{2}$ is polynomial precisely when $f$ is defined by an m-vector of n-variable polynomials.

\begin{digression}\label{dig:3.5}
Let $k=\mathbb{R}$, the reals. Then varieties have a natural topology induced from their embedding in Euclidean space with the usual topology. This topology is finer than the Zariski topology (for instance the only proper Zariski-closed subsets of $\mathbb{R}^{1}$ are finite sets). More generally, given any normed field $k$, we may define a \textit{strong topology} on $k$-spaces by choosing as a basis of open sets all finite intersections of sets of the type $f^{-1}(N)$, for all polynomial functions $f$ and all open sets $N$ in $k$, i.e. the coarsest topology in $X$ for which all $f$ in $A(X)$ become continuous for the \textit{normed} topology of $k$. (See for instance \citealp[Chapter 7]{shafarevich1975basic} for the case of varieties over $\mathbb{C}$.) For the purpose of this work, realization theory, we are mainly interested in (global) questions of structure; thus we shall use only the Zariski topology, even in the cases $k=\mathbb{R}$ or $\mathbb{C}$, except for the proof of some technical facts on varieties over $\mathbb{R}$ in Section~\ref{sec:4}.
\end{digression}

\label{op:31}%

\begin{definition}\label{def:3.6}
\textit{Let} $g: X_{1} \rightarrow X_{2}$. \textit{The transpose of} $g$ \textit{is the homomorphism} $A(g): A\left(X_{2}\right) \rightarrow A\left(X_{1}\right)$ \textit{defined by}
\[
A(g)(f):=f \circ g .
\]
\end{definition}

\begin{lemma}\label{lem:3.7}
Fix two $k$-spaces $X_{1}, X_{2}$. Then the assignment $A: g \mapsto A(g)$ establishes a bijection between polynomial maps $X_{1} \rightarrow X_{2}$ and $k$-algebra homomorphisms $A\left(X_{2}\right) \rightarrow A\left(X_{1}\right)$.
\end{lemma}

We shall write $X(\mu): X(B) \rightarrow X(A)$ for the polynomial map corresponding to the homomorphism $\mu: A \rightarrow B$.

\begin{proof}
The problem is to define an inverse $X$ of the transpose functor $A$. Let the $k$-spaces $X_{j}$ be $\left(X\left(A_{j}\right), \iota\left(A_{j}\right)\right), j=1,2$ and $\iota_{j}: A_{j} \rightarrow k^{X_{j}}$ the canonical maps. Let $\alpha: \iota_{2}\left(A_{2}\right) \rightarrow \iota_{1}\left(A_{1}\right)$ and take $x$ in $X$. Since $e_{x}: \iota_{1}\left(A_{1}\right) \rightarrow k$ is a homomorphism, $e_{x} \circ \alpha \circ \iota_{2}: A_{2} \rightarrow k$ is also a homomorphism. Now define
\[
X(\alpha): X\left(A_{1}\right) \rightarrow X\left(A_{2}\right): x \mapsto e_{x} \circ \alpha \circ \iota_{2} .
\]
It is easy to verify that $X(A(g))=g$ and $A(X(\alpha))=\alpha$ for all $g: X\left(A_{1}\right) \rightarrow X\left(A_{2}\right)$ and all $\alpha: \iota_{2}\left(A_{2}\right) \rightarrow \iota_{1}\left(A_{1}\right)$.
\end{proof}

\begin{corollary}\label{cor:3.8}
The category of $k$-spaces is dual (arrows reversed) to the category of $k$-reduced $k$-algebras.
\end{corollary}

The above duality allows the translation of constructions and statements about algebras into (dual) statements about $k$-spaces, and vice versa. For instance, Lemma~\ref{lem:1.8} says that the categorical product $X \times Y$ of two $k$-spaces $X, Y$ is the $k$-space $X(A(X) \otimes A(Y))$ and that the underlying set of this $k$-space is the cartesian product $X \times Y$. By induction, $X\left(A_{1}\right) \times X\left(A_{2}\right) \times \ldots \times X\left(A_{n}\right)=X\left(A_{1} \otimes \ldots \otimes A_{n}\right)$. And, in particular, $(X(k[T]))^{n}$ ($n$-th fold power) coincides with $X\left(k\left[T_{1}, \ldots, T_{n}\right]\right)=k^{n}$; see Example \ref{ex:1.4}. This also shows that the notation $k^{n}$ is consistent with products in the category of $k$-spaces. (Note that, in particular, $k^{0}=X(k)=$ one point, say $\{0\}$).

\label{op:32}%
As an example of the transpose construction, consider a function $f: X \rightarrow k$. Since the transpose $A(f): k[T] \rightarrow A(X)$ is a $k$-algebra homomorphism, $A(f)$ is determined by $A(f)(T)=T \circ f$. Since $T$ is the identity map on $k$,
\begin{equation}
A(f)(T)=f . \label{eq:3.9}
\end{equation}
So the transpose of $f$ is the map $P(T) \mapsto P(f)$.

We now relate various properties of polynomial maps to properties of their transposes.

\begin{definition}\label{def:3.10}
A \textit{polynomial map} $f: X \rightarrow Y$ \textit{is dominating iff} $\overline{f(X)}=Y ; f$ \textit{is a closed immersion iff} $f$ \textit{can be factored as} $g_{1} \circ g_{2}$, \textit{where} $g_{2}$ \textit{is an isomorphism} $X \cong V$ \textit{and} $g_{1}$ \textit{is the inclusion map} $V \rightarrow Y$, \textit{for some closed subset} $V$ of $Y$.
\end{definition}

\begin{lemma}\label{lem:3.11}
Let $\alpha: A \rightarrow B$ and denote $f:=X(\alpha): X(B) \rightarrow X(A)$. Then
\begin{enumerate}
\item[(a)] $f^{-1}(V(S))=V(\alpha(S))$ for any $S \subseteq A$.
\item[(b)] $f$ is continuous.
\item[(c)] $\overline{f(V(I))}=V\left(\alpha^{-1}(I)\right)$ for any closed ideal $I$ of $B$.
\item[(d)] $f$ is dominating if and only if $\alpha$ is one-to-one.
\item[(e)] $f$ is a closed immersion if and only if $\alpha$ is onto.
\end{enumerate}
\end{lemma}

\begin{proof}
(a) $x$ is in $f^{-1}(V(S))$
\[
\begin{aligned}
& \text { iff } f(x) \text { is in } V(S) \text {, } \\
& \text { iff } a(f(x))=0 \text { for all a in } S \text {, } \\
& \text { iff } \alpha(a)(x)=0 \text { for all a in } S \text {, } \\
& \text { iff } x \text { is in } V(\alpha(S)) \text {. }
\end{aligned}
\]
(b) All closed sets in $X(A)$ are by definition of the form $V(S)$. By (a), pre-images of closed sets are closed.

(c) We first prove that $\alpha^{-1}(I)=I(f(V(I)))$; in fact the following statements are equivalent:
\label{op:33}%
\[
\begin{aligned}
& \text { a is in } I(f(V(I))), \\
& a(f(x))=0 \text { for all } x \text { in } V(I), \\
& \alpha(a) \text { is in } I(V(I))=I \quad(I=\text { closed! }), \\
& a \text { belongs to } \alpha^{-1}(I) .
\end{aligned}
\]
Therefore $V\left(\alpha^{-1}(I)\right)=V(I(f(V(I))))=\overline{f(V(I))}$, as required.

(d) Applying (c) to $I=\{0\}, V\left(\alpha^{-1}(I)\right)=\overline{f(X)}$. So $f$ is dominating iff $V(\operatorname{ker} \alpha)=Y$, which by Definition~\ref{def:1.11}(a) is equivalent to ker $\alpha=\{0\}$.

(e) Dualizing Definition~\ref{def:3.10}, $f$ is a closed embedding iff the transpose homomorphism $\alpha$ factors as $\beta_{2} \circ \beta_{1}$, where $\beta_{1}$ is a homomorphism $B \rightarrow B / I$ for some ideal $I$ and $\beta_{2}$ is an isomorphism. Such factorizations exist precisely when $\alpha$ is onto.
\end{proof}

Both dominating maps and closed immersions will play important roles in our treatment of realization theory. We emphasize some intuitive aspects of these concepts through the following

\begin{discussion}\label{disc:3.12}
It follows from (e) above that $\alpha$ onto implies that $f$ is one-to-one. The converse is false. For an easy example, consider $X=Y=k:=\mathbb{R}, f_{1}(x):=x^{3}$. Then, $f_{1}$ is one-to-one, but $A\left(f_{1}\right): k[T] \rightarrow k[T]: T \mapsto T^{3}$ is not onto ($T$ is not in the image). The problem does not lie in the fact that $\mathbb{R}$ is not algebraically closed: for any field $k$ we may consider $f_{2}: k \rightarrow k^{2}: x \mapsto\left(x^{2}, x^{3}\right)$; $f_{2}$ is one-to-one but $A\left(f_{2}\right): k\left[T_{1}, T_{2}\right] \rightarrow k[T]: T_{1} \mapsto T^{2}, T_{2} \mapsto T^{3}$ has image $k\left[T^{2}, T^{3}\right] \neq k[T]$.

A variation of the last example provides a \textit{bijective map} $f_{3}$ \textit{which is not an isomorphism}. Indeed, consider the "cusp"
\[
Y:=\left\{(x, y) \text { in } k^{2} \mid x^{3}=y^{2}\right\} .
\]
Then
\begin{equation}\label{eq:3.13}
f_{3}: k \rightarrow Y: x \mapsto\left(x^{2}, x^{3}\right) ,
\end{equation}
\label{op:34}is bijective. But $f$ is not an isomorphism, because, by the equivalence of categories \eqref{eq:1.7}, $f_{3}$ is an isomorphism iff $A\left(f_{3}\right)$ is an isomorphism. But $A(Y) \simeq k\left[T^{2}, T^{3}\right] \neq k[T]$; see \citet[1.1, Example 5]{dieudonne1974cours}. Intuitively, we cannot expect to have any isomorphism between $k$ and $Y$ because the curve $Y$ has a singularity (at the origin) while the line $k$ has none.

It is not difficult to prove that, in the category of $k$-spaces, monomorphism $=$ one-to-one and epimorphism $=$ dominating. A dominating map is in general not onto. This is illustrated over $k=\mathbb{R}$ by $X=Y:=\mathbb{R}, f(x):=x^{2}$; this is a dominating map because the smallest Zariski closed set containing the nonnegative reals is all of $\mathbb{R}$. In the particular case when $k$ is algebraically closed and $X$ is an irreducible variety, a dominating $f$ becomes almost onto, in the sense that $f(X)$ contains a Zariski open subset of $Y$; see Theorem~\ref{thm:3.14} below. Thus in this particular case $f(X)$ is all of $Y$ except at most for a subset of "lower dimension" (to be made precise later). Moreover, it can be proved that, when $k=\mathbb{C}$ and $Y$ is given the strong topology Definition~\ref{def:3.4}, $f(X)$ has a nowhere dense complement.
\end{discussion}

We remarked above that the image of a polynomial map is in general not a closed set. When $f: X \rightarrow Y$ is a polynomial map between two varieties, one can sometimes characterize $f(X)$ as a set defined by polynomial equalities \textit{and} inequalities. A \textit{constructible} subset $C$ of a variety $X$ is a finite union of sets of the type $U \cap F$, where $U$ is open and $F$ is closed. In other words, $C$ is in the Boolean algebra generated by the Zariski topology of $X$. When $k$ is a real-closed field (\citealp[VI.2]{jacobson1964lectures}), like $k=\mathbb{R}$, we define a \textit{real-constructible} set $C$ as a finite union of sets $U \cap F$, where $F$ is closed and $U$ is of the type $\{x$ in $X \mid f(x)<0\}$ for some polynomial function $f$.

\begin{theorem}\label{thm:3.14}
Let $X, Y$ be varieties and let $f: X \rightarrow Y$. Then
\begin{enumerate}
\item[(a)] If $k$ is algebraically closed and $C$ is a constructible subset of $X$ (e.g. $C=X$), then $f(C)$ is a constructible subset of $Y$.
\label{op:35}%
\item[(b)] If $k$ is a real-closed field (e.g. $k=\mathbb{R}$) and $C$ is a real-constructible subset of $X$ (e.g. $C=X$), then $f(C)$ is a real-constructible subset of $Y$.
\item[(c)] If $f$ is dominating, $Y$ is irreducible, and $k$ is algebraically closed, $f(X)$ contains a (Zariski) open subset of $Y$.
\item[(d)] If $f$ is dominating, $Y$ is irreducible and $k=\mathbb{R}$, then $f(X)$ contains a set open in the strong topology of $Y$.
\end{enumerate}
\end{theorem}

\begin{proof}
(a) This is the well-known Chevalley's theorem; see for instance \citet[Chapter 4, Corollary to Proposition 14]{dieudonne1974cours}.

(b) This statement is essentially the \textit{generalized Sturm's Theorem} due to Tarski and Seidenberg; see \citet[VI.10]{jacobson1964lectures}. (The usual statement of the Tarski-Seidenberg result requires that the coefficients of $f$ be rational, so that $f(C)$ can be constructed algorithmically. However, the proof itself does not depend upon this requirement; see \citealp[footnote in page 366]{seidenberg1954new}.)

To prove (c) and (d), write $f(X)=\bigcup_{\text {finite }} U_{i} \cap F_{i}$ as in the definition of constructibles and real-constructibles, with the $F_{i}$ closed and the $U_{i}$ Zariski-open or, when $k=\mathbb{R}$, of the type $\{f(x)<0\}$, so open in the strong topology. If $F_{i} \neq Y$ for all $i$, then, by irreducibility of $Y$, $\overline{f(X)}=\bigcup F_{i} \neq Y$, contradicting domination of $f$. So some $F_{i}=Y$, and $f(X)$ contains $U_{i}$.
\end{proof}

An important type of dominating map arises in the following

\begin{definition}\label{def:3.15}
\textit{The principal open set defined by} $a \in A$ \textit{is}
\[
D(a):=\{x \text { in } X(A) \mid a(x) \neq 0\} .
\]
\textit{The principal open sets constitute a basis for the Zariski topology}. Indeed, the complement of any closed set $V(S)$ is the union of the $D(a)$, $a$ in $S$. For simplicity, let $A$ be an integral domain. Denote by $a^{-1} A$ the algebra $A \subseteq a^{-1} A \subseteq Q(A)$ consisting of all $b / a^{n}, b$ in $A, n \geq 0$. Take any $\alpha: a^{-1} A \rightarrow k$. Then $\beta:=\alpha \mid A: A \rightarrow k$ satisfies $\beta(a) \neq 0$. Conversely, if $\beta: A \rightarrow k$ and $\beta(a) \neq 0$ then the rule
\label{op:36}$\alpha\left(b / a^{n}\right):=\beta(b) / \beta(a)^{n}$ defines the (unique) $\alpha$ extending $\beta$. Therefore $D(a)$ is the image of the map
\begin{equation}\label{eq:3.16}
X\left(a^{-1} A\right) \rightarrow X(A) ,
\end{equation}
dual to the inclusion $A \subseteq a^{-1} A$. This map is both one-to-one and dominating and establishes a homeomorphism between $X\left(a^{-1} A\right)$ and $D(a)$; see \citet[II.4.3, Corollary to Proposition 13]{bourbaki1972commutative}. Local arguments about $k$-spaces are often simplified by restricting attention to principal open sets.
\end{definition}

We shall be especially interested in "quotients" of varieties:

\begin{definition}\label{def:3.17}
\textit{A} $k$-\textit{space} $X$ \textit{is an almost-variety iff there exists a variety} $\hat{X}$ \textit{and a dominating polynomial map} $f: \hat{X} \rightarrow X$.
\end{definition}

By definition of "dominating", this means that $X$ has a dense subset $f(\hat{X})$ consisting of equivalence classes of elements of $X$. Let $f: \hat{X} \rightarrow X$ as above. By Hilbert's basis theorem we may write $\hat{X}=\hat{X}_{1} \cup \ldots \cup \hat{X}_{r}$, where the $\hat{X}_{i}$ are irreducible closed sets; see \citet[III.2.10, Corollary 3 of Theorem 3]{bourbaki1972commutative}. From the definition of irreducibility, it is easy to prove that continuous images \textit{and} closure of irreducible sets are irreducible. Hence $X_{i}:=\overline{f\left(\hat{X}_{i}\right)}$ is an irreducible closed subset for each $i$. Since $f \mid \hat{X}_{i}: \hat{X}_{i} \rightarrow X_{i}$ is dominating and each $\hat{X}_{i}$ is a variety, $X$ can be written as a finite union of irreducible almost-varieties.

The next lemma shows that every irreducible almost-variety has an open (hence dense) subset which is a variety, justifying the terminology "almost variety".

\begin{lemma}\label{lem:3.18}
Let $X$ be an irreducible almost-variety. Then there exists a principal open set $D(a)$ of $X$ which is a variety, i.e. such that $a^{-1} A(X)$ is finitely generated.
\end{lemma}

\begin{proof}
Let $f: \hat{X} \rightarrow X$ be dominating with $\hat{X}$ an irreducible variety. Then $A(f): A(X) \rightarrow A(\hat{X})$ embeds $A(X)$ in $B:=A(\hat{X})$; we
\label{op:37}identify $A(X)$ with its image $A \subseteq B$. By \citet[V.3.1, Corollary 1]{bourbaki1972commutative} there is an $a$ in $A$ and $T_{1}, \ldots, T_{r}$ in $B$ such that the finitely generated algebra $a^{-1} B$ is integral over the polynomial ring $a^{-1} A\left[T_{1}, \ldots, T_{r}\right]$. By Lemma~\ref{lem:1.13} $a^{-1} A\left[T_{1}, \ldots, T_{r}\right]$ is a finitely generated algebra. Let $P_{1}\left(T_{1}, \ldots, T_{r}\right), \ldots, P_{m}\left(T_{1}, \ldots, T_{r}\right)$ generate $a^{-1} A\left[T_{1}, \ldots, T_{r}\right]$. Then $P_{1}(0, \ldots, 0), \ldots, P_{m}(0, \ldots, 0)$ generate $a^{-1} A$. Thus $a^{-1} A$ is finitely generated
\end{proof}

We shall find in Section~\ref{ch:4} that the canonical state-spaces are in general almost-varieties. In particular we shall give an input/output map whose canonical state-space is that given by the following

\begin{example}\label{ex:3.19}
Let $A:=k\left[T_{1} T_{2}^{n}, n \geq 0\right]$. Since $A \subseteq k\left[T_{1}, T_{2}\right]$, $X(A)$ is an almost-variety. Take $a:=T_{1}$. Then $a^{-1} A=k\left[T_{1}, T_{2}, T_{1}^{-1}\right]$ is a finitely generated algebra. As a variety, $X\left(a^{-1} A\right)$ can be represented by the set of solutions $(x, y, z)$ in $k^{3}$ of the equation $x z-1=0$. Note that $X\left(a^{-1} A\right)$ can be also naturally viewed as the principal open set $x_{1} \neq 0$ in $k^{2}$. A point $x$ of $X(A)$ not in $D(a)$ satisfies $x\left(T_{1}\right)=T_{1}(x)=0$ and (as we now show) it also satisfies $x\left(T_{1} T_{2}^{n}\right)=0$ for all $n$. This is not trivial, since the "proof" $x\left(T_{1} T_{2}\right)=x\left(T_{1}\right) \cdot x\left(T_{2}\right)=0 \cdot x\left(T_{2}\right)=0$ is fallacious: $x\left(T_{2}\right)$ is not defined, because $T_{2}$ is not in $A$. One way to prove the statement is by means of the theorem of extension of places (see \citealp[VI.2.4, Proposition 13]{bourbaki1972commutative}). This theorem implies that, if there is no extension of $x$ such that $x\left(T_{2}\right)$ is defined, then there does exist an overring of $A$ containing $T_{2}^{-1}$ and such that $x\left(T_{2}^{-1}\right)=0$ (the values of the extension of $x$, though, are not necessarily in $k$). But if this is the case, then $x\left(T_{1} T_{2}^{n}\right)=x\left(T_{1} T_{2}^{n+1}\right) x\left(T_{2}^{-1}\right)=0$ for all $n$, as wanted. Therefore the complement of $D(a)$ consists of just one point and $X(A)$ is the disjoint union of the variety $X\left(a^{-1} A\right)$ and this one extra point.
\end{example}

\subsection{Dimension}\label{sec:4}

There are various possible notions of dimension for a $k$-space $X$. All these notions coincide if $X$ is a variety and $k$ is algebraically
\label{op:38}closed, but this is not true for more general $X$ or $k$. Thus our choice will be to some extent arbitrary, to be justified by the results.

Let $A$ be an overring of $B$ having no zero divisors. Recall that elements $\ell_{1}, \ldots, \ell_{n}$ of $A$ are \textit{algebraically dependent} over $B$ iff there exists a nonzero polynomial $P$ in $B\left[T_{1}, \ldots, T_{n}\right]$ such that $P\left(\ell_{1}, \ldots, \ell_{n}\right)=0$. When no such $P$ exists, $\ell_{1}, \ldots, \ell_{n}$ are algebraically \textit{independent}. We review some elementary properties of algebraic dependence. They can be found, for instance, in \citet[II.12]{zariski1958commutative}.

An arbitrary subset $L$ of $A$ is \textit{algebraically independent} iff every finite subset of $L$ is. A \textit{transcendence basis} for $A$ (over $B$) is a maximal algebraically independent set $L \subseteq A$; in other words, if $s$ is not in $L$, then $L \cup\{s\}$ is algebraically dependent. All transcendence bases have the same cardinality, the \textit{transcendence degree} $\operatorname{trdeg}_{B} A$ of $A$ over $B$. When $B=k$, we denote $\operatorname{trdeg}_{k} A$ just by $\operatorname{trdeg} A$. When $\operatorname{trdeg}_{B} A=0, A$ is \textit{algebraic} over $B$.

The notion of transcendence degree, which is based on "dependence", is analogous to that of dimension of vector spaces.

A transcendence basis for $A$ over $B$ can be extracted out of any system of generators of $A$ over $B$. In particular, when $A$ is a finitely generated $B$-algebra, $\operatorname{trdeg}_{B} A$ is finite. If both $A, B$ are $k$-algebras, then
\begin{equation}\label{eq:4.1}
\operatorname{trdeg} A=\operatorname{trdeg} B+\operatorname{trdeg}_{B} A .
\end{equation}

\begin{lemma}[{\citealp[II.2, Theorems 28 and 29]{zariski1958commutative}}]\label{lem:4.2}
Let $A, B$ be integral domains, and let $\varphi: A \rightarrow B$ be onto. Then $\operatorname{trdeg} B \leq \operatorname{trdeg} A$; if both are finite then equality can only hold when $\varphi$ is an isomorphism.
\end{lemma}

When the $k$-algebra $A$ is not an integral domain, one can still define trdeg $A$, making use of the integral domains $\{A / P$, $P$ prime ideal of $A\}$; just let
\label{op:39}%
\begin{equation}\label{eq:4.3}
\operatorname{trdeg} A:=\sup \{\operatorname{trdeg} A / P, \quad P \text{ prime ideal of } A\} .
\end{equation}

By Lemma~\ref{lem:4.2}, this definition is consistent with the basic definition for integral domains.

We define a notion of \textit{dimension} for a $k$-space by setting
\begin{equation}\label{eq:4.4}
\operatorname{dim} X:=\operatorname{trdeg} A(X) .
\end{equation}

Observe that an \textit{almost-variety} $X$ \textit{is always finite-dimensional}, since then $A(X)$ is included in a finitely-generated algebra.

A combinatorial consequence of the definition of dimension is:

\begin{lemma}\label{lem:4.5}
Let $\operatorname{dim} X=n$ be finite. Let
\[
V_{0} \subsetneqq V_{1} \subsetneqq \cdots \subsetneqq V_{s},
\]
be a chain of irreducible closed subsets of $X$. Then $s \leq n$.
\end{lemma}

\begin{proof}
Clear by Lemma~\ref{lem:2.12} and induction on Lemma~\ref{lem:4.2}.
\end{proof}

Let $f: X \rightarrow Y$ be a polynomial map. Since $f$ is continuous and points of $Y$ are closed, the \textit{fibers} $f^{-1}(y), y$ in $Y$, are closed subsets of $X$. Therefore each fiber is a $k$-space and as such has a well-defined dimension.

We shall see below how to generalize the dimension formulas for linear maps ($\operatorname{dim} X-\operatorname{dim}$ ker $f=\operatorname{dim} f(X)$) to the polynomial context. This will follow from a general result on the structure of polynomial maps.

The next theorem summarizes a number of relevant results in a form convenient for our purposes.

\begin{theorem}\label{thm:4.6}
Let $X, Y$ be two irreducible almost-varieties, with $\operatorname{dim} X=n, \operatorname{dim} Y=m$. Let $f: X \rightarrow Y$ be a dominating polynomial map. Then there exists an integer $s \geq 0$, irreducible varieties $X_{1}, Y_{1}$ and polynomial maps $j_{X}: X_{1} \rightarrow X, j_{Y}: Y_{1} \rightarrow Y, f_{1}: X_{1} \rightarrow Y_{1}$ and $g: X_{1} \rightarrow Y_{1} \times k^{n-m}$ such that
\begin{enumerate}
\item[(a)] \label{op:40}$j_{X}$ [respectively $j_{Y}$] identifies $X_{1}$ [respectively $Y_{1}$] with a principal open set of $X$ [respectively $Y$] as in \eqref{eq:3.16}.
\item[(b)] The following diagram commutes:
\[
\begin{tikzpicture}[x=1cm,y=1cm]
\node (X) at (0,1.5) {$X$};
\node (Y) at (0,-1.2) {$Y$};
\node (X1) at (3,1.5) {$X_1$};
\node (Y1) at (3,-1.2) {$Y_1$};
\node (P) at (6.6,0.15) {$Y_1\times k^{n-m}$};
\draw[arr] (X1) -- node[above,lab] {$j_X$} (X);
\draw[arr] (X) -- node[left,lab] {$f$} (Y);
\draw[arr] (X1) -- node[left,lab] {$f_1$} (Y1);
\draw[arr] (X1) -- node[above right,lab] {$g$} (P);
\draw[arr] (P) -- node[below right,lab] {$\mathrm{pr}_1$} (Y1);
\draw[arr] (Y1) -- node[below,lab] {$j_Y$} (Y);
\end{tikzpicture}
\]
where $\mathrm{pr}_{1}$ denotes the projection in the first factor.
\item[(c)] $g$ is dominating and the sets $g^{-1}(z)$ have cardinality at most $s$ for each $z$ in $Y_{1} \times k^{n-m}$.
\item[(d)] If $k$ is algebraically closed, then $g$ is onto and each $g^{-1}(z)$ has exactly $s$ elements.
\item[(e)] If $k=\mathbb{R}$ then $X_{1}, Y_{1}$ are differentiable manifolds and the differentials $d g$ and $d f_{1}$ have full rank at every point.
\item[(f)] If either $k=\mathbb{R}$ or $k$ is algebraically closed,
\[
\operatorname{dim} f_{1}^{-1}\left(f_{1}(x)\right)=n-m,
\]
for each $x$ in $X_{1}$.
\end{enumerate}
\end{theorem}

\begin{proof}
Let $A:=A(Y)$ be identified through $A(f)$ to a subalgebra of $B:=A(X)$. By Lemma~\ref{lem:3.18}, there exist $a$ in $A$ and $b$ in $B$ such that $a^{-1} A$ and $b^{-1} B$ are finitely generated. Then $a^{-1} A \subseteq a^{-1} B \subseteq a^{-1} b^{-1} B$. By \citet[V.1.5]{bourbaki1972commutative}, there is some $s$ in $a^{-1} A$ such that $s^{-1}\left(a^{-1} A\right)$ is integrally closed in its quotient field. Replacing if necessary $a$ by $s a$, we may assume that $a^{-1} A$ is integrally closed. By \citet[V.3.1, Corollary 1]{bourbaki1972commutative}, there exist $T_{1}, \ldots, T_{n-m}$ algebraically independent elements of $(a b)^{-1} B$ such that $(a b)^{-1} B$ is integral over $A^{\prime}:=a^{-1} A\left[T_{1}, \ldots, T_{n-m}\right]$. Since $a^{-1} A$ is integrally closed, $A^{\prime}$ is also integrally closed; see \citet[V.1.3, Corollary 2]{bourbaki1972commutative}.

\label{op:41}%
Let $s$ be the separable degree of the algebraic field extension $Q(B): Q(A^{\prime})$. By \citet[V.2.3, Remark 3]{bourbaki1972commutative}, the fibers of the canonical map $X\left((a b)^{-1} B\right) \rightarrow X\left(A_{1}\right)$ have at most $s$ elements. Let $X_{1}:=X\left((a b)^{-1} B\right), Y_{1}:=X\left(a^{-1} A\right)$ and let $j_{X}, j_{Y}, f_{1}, g$ be respectively the maps dual to the inclusions $B \subseteq(a b)^{-1} B, A \subseteq a^{-1} A, a^{-1} A \subseteq(a b)^{-1} B$, $A^{\prime} \subseteq(a b)^{-1} B$. Then (a), (b), and (c) hold by construction.

If $k$ is algebraically closed, there exists an open set $V$ in $Y_{1} \times k^{n-m}$ such that $g^{-1}(z)$ has precisely $s$ elements for each $z$ in $V$; see \citet[Chapter 5, Proposition 6]{dieudonne1974cours}. Let $U \subseteq g^{-1}(V)$ be a principal open set. By Theorem~\ref{thm:3.14}, $f_{1}(U)$ contains a principal open set $W$. Let $\hat{U} \subseteq f_{1}^{-1}(W) \cap U$ be a principal open set. Then $g(\hat{U}) \subseteq V$ implies that, for any $z$ in $Y_{1} \times k^{n-m}$, either $g^{-1}(z) \cap \hat{U}$ is empty or it has exactly $s$ elements. Replacing $X_{1}$ by $\hat{U}$ and $Y_{1}$ by $W$ gives (d); $g$ is surjective by \citet[IV.4, Corollary 1]{dieudonne1974cours}.

Let $k=\mathbb{R}$. We replace $X_{1}, Y_{1}$ by nonsingular principal open subsets; see, for example, \citet[Theorem 12.12]{brocker1975differentiable}. We may further replace $X_{1}$ by a principal open set in which $d g$ has maximal rank $p$. By Theorem~\ref{thm:3.14}, $g\left(X_{1}\right)$ contains a strong open set of $Y_{1} \times k^{n-m}$, so by Sard's theorem (see \citealp[Theorem 2.11]{brocker1975differentiable}), $p=n$. Similarly with $f_{1}$. Thus (e) follows.

(f) For $k$ algebraically closed, see \citet[Chapter 4, Theorem 2]{dieudonne1974cours}; for $k=\mathbb{R}$ this follows from (e) and \citet[Theorem 1.9]{brocker1975differentiable}.
\end{proof}

\label{op:42}%
\section{Realization Theory}\label{ch:3}

We investigate in this section the general realization theory of an m-input, p-output polynomial response. Before doing this, we develop the formalism of Volterra series and prove some simple facts to be used later in the study of finiteness conditions.

Throughout this work, $k$ denotes an \textit{infinite} field; all parameters belong to $k$. The assumption that $k$ is infinite is merely a technical convenience, permitting the identification of polynomials and polynomial functions.

Both $m$ and $p$ (number of input and output channels, respectively) are positive integers, arbitrary but fixed throughout the discussion.

\subsection{Volterra Series}\label{sec:5}

We shall use the following notations:

\begin{quote}
$\mathbb{N}:=$ set of nonnegative integers.\\
$\mathbb{N}^{m}:=$ set of column $m$-vectors over $\mathbb{N}$.\\
$\mathbb{N}^{m \times t}:=$ set of $m \times t$ matrices over $\mathbb{N}$.\\
$\left(\mathbb{N}^{m}\right)^{*}:=$ set of all finite sequences of elements of $\mathbb{N}^{m}$,\\
\hspace*{3em}including the empty sequence denoted by $\Lambda$.
\end{quote}

If $\alpha$ is in $\left(\mathbb{N}^{m}\right)^{*}, \alpha=\alpha_{1} \ldots \alpha_{t} \neq \Lambda$ then $|\alpha|:=$ \textit{length} of $\alpha:=t ; \quad\|\alpha\|:=$ \textit{weight} of $\alpha:=\alpha_{1}+\ldots+\alpha_{t} ; \quad|\Lambda|=\|\Lambda\|:=0 ; \alpha_{i j}:=$ element in $i$-th row of column vector $\alpha_{j}, j=1, \ldots, t$, $i=1, \ldots, m$.

$\Delta:=$ set of \textit{proper} sequences: $\alpha$ in $\left(\mathbb{N}^{m}\right)^{*}$ belongs to $\Delta$ iff $\alpha=\Lambda$ or $\alpha=\alpha_{1} \cdots \alpha_{t}$ with $\alpha_{t} \neq 0$.

If $\alpha=\alpha_{1} \ldots \alpha_{t}$ and $\beta=\beta_{1} \ldots \beta_{s}$ are in $\left(\mathbb{N}^{m}\right)^{*}, \alpha \beta:=$ \textit{concatenation} of $\alpha$ and $\beta:=\alpha_{1} \ldots \alpha_{t} \beta_{1} \ldots \beta_{s}$.

If, say, $t \leq s$, then $\alpha+\beta:=$ \textit{sum} of $\alpha$ and $\beta:=\gamma_{1} \ldots \gamma_{s}$, where $\gamma_{i}:=\alpha_{i}+\beta_{i}$ for $i=1, \ldots, t$ (addition of $\alpha_{i}, \beta_{i}$ is rowwise in $\mathbb{N}^{m}$) and $\gamma_{i}:=\beta_{i}$ if $i=t+1, \ldots$, s. Similarly if $t>s$.

\label{op:43}%
Since an $\alpha=\alpha_{1} \ldots \alpha_{t}$ in $\left(\mathbb{N}^{m}\right)^{*}$ is a sequence of columns, we may regard $\alpha$ as an $m \times t$ matrix. Thus we may, and shall, make the following identification:
\[
\left(\mathbb{N}^{m}\right)^{*}=\bigcup_{t \geq 0} \mathbb{N}^{m \times t} .
\]
Under this identification, concatenation of $\alpha$ and $\beta$ is the same thing as formation of the block matrix $[\alpha \vdots \beta]$; addition is simply addition of matrices (augmented by zeroes to the right if necessary). Observe that the notation $\alpha_{i j}$ is consistent with the matrix interpretation.

Each of the operations, concatenation and addition, make $\left(\mathbb{N}^{m}\right)^{*}$ into a \textit{monoid}; in both cases $\Lambda$ is the identity. In both cases $\Delta$ is a \textit{submonoid}, i.e. if $\alpha$ and $\beta$ are both in $\Delta$, then both $\alpha \beta$ and $\alpha+\beta$ are in $\Delta$. We shall denote by $(\Delta, \cdot)$ and $(\Delta,+)$ the two monoids thus obtained. Both monoids will play a central role in our theory. The monoid ($\Delta,+$) is used in defining "polynomial" and ($\Delta, \cdot$) is used in obtaining finiteness conditions.

Let $\xi_{i j}, i=1, \ldots, m, j=1,2,3, \ldots$ denote denumerably many (distinct) indeterminates, and for each $j$ let $\xi_{j}$ denote the subset $\xi_{1 j}, \ldots, \xi_{m j}$. Let $\alpha=\alpha_{1} \ldots \alpha_{t}$ be in $\left(\mathbb{N}^{m}\right)^{*}$ and define
\[
\xi^{\alpha}:=\xi_{1}^{\alpha_{1}} \ldots \xi_{t}^{\alpha_{t}}
\]
where each $\xi_{j}^{\alpha_{j}}$ is itself a monomial
\[
\xi_{j}^{\alpha_{j}}:=\xi_{1 j}^{\alpha_{1 j}} \ldots \xi_{m j}^{\alpha_{m j}},
\]
with $\xi_{i j}^{0}=1$ for any $i, j$. We interpret $\xi^{\Lambda}$ as 1 . A formal power series $\psi$ is an infinite formal combination of the monomials $\xi^{\alpha}$ with coefficients in $k$. Since $\xi^{\alpha}=\xi^{\alpha 0}=\xi^{\alpha 00}=\ldots$, care must be taken not to count each $\xi^{\alpha}$ more than once. This is the reason for introducing $\Delta$. Thus, a \textit{formal power series in the} $\xi_{i j}$ \textit{with coefficients in} $k$ is a map
\label{op:44}%
\[
\psi: \Delta \rightarrow k: \alpha \mapsto \psi_{\alpha},
\]
denoted also as
\begin{equation}
\psi\left(\xi_{1}, \xi_{2}, \ldots\right)=\sum_{\alpha \text { in } \triangle} \psi_{\alpha} \xi^{\alpha}, \psi_{\alpha} \text { in } k . \label{eq:5.1}
\end{equation}
The set of formal power series can be made into a $k$-vector space by defining term-wise the addition of two power series $\psi, \hat{\psi}$ and the multiplication by scalars $r$ in $k$ :
\[
(\psi+r \hat{\psi})_{\alpha}:=\psi_{\alpha}+r \hat{\psi}_{\alpha} \text { for all } \alpha .
\]
Using the structure of the monoid ($\Delta,+$), we may also induce a \textit{convolution product} among series, which extends the multiplication of monomials defined as $\xi^{\alpha} \cdot \xi^{\beta}:=\xi^{\alpha+\beta}$. This extension is well-defined because ($\Delta,+$) is \textit{locally finite}, i.e. each $\alpha$ in $\Delta$ can be split in only finitely many ways as $\alpha_{1}+\alpha_{2}, \alpha_{i}$ in $\Delta$. The convolution of the power series $\psi_{1}, \psi_{2}$ is then defined globally by the formula:
\begin{equation}
(\psi \hat{\psi})_{\alpha}:=\sum_{\beta+\gamma=\alpha} \psi_{\beta} \hat{\psi}_{\gamma} \text { for all } \alpha \text { in } \Delta . \label{eq:5.2}
\end{equation}
The set of all formal power series forms a $k$-algebra when endowed with the operations of scalar product, sum, and convolution product. In fact, it is easy to prove that this algebra has no zero divisors.

We intend to derive response maps by evaluating power series for particular values of the $\xi_{i j}$. Thus we want to restrict our attention to a suitable class of series so that evaluation at arbitrary input values is welldefined. In accordance with related investigations in the literature, we shall call these \textit{Volterra} series.

Let $\psi$ be a formal power series. For each $\xi_{i j}$, $\psi$ may be rearranged into a power series in $\xi_{i j}$ whose coefficients are power series in the other variables. We call $\psi$ a (\textit{formal}) \textit{Volterra series} (\textit{over} $k$) iff, after each such rearrangement, $\psi$ becomes a \textit{polynomial} in $\xi_{i j}$.

\label{op:45}%
In other words, $\psi$ is a Volterra series precisely when there exist integers $d_{i j}$ such that any $\xi_{i j}$ appearing in $\psi$ has exponent $\leq d_{i j}$. The smallest bound $d_{i j}$ for the exponents of $\xi_{i j}$ is the \textit{degree} $\operatorname{deg}_{i j} \psi$ of $\psi$ in $\xi_{i j}$ (if $\xi_{i j}$ does not appear in $\psi, \operatorname{deg}_{i j} \psi:=-\infty$). Thus $\psi$ \textit{is a Volterra series if and only if} $\operatorname{deg}_{i j} \psi<\infty$ \textit{for all} $i$, $j$.

For example,
\[
\psi_{1}:=\xi_{11}^{2}+\xi_{11}^{2} \xi_{12}^{4}+\xi_{11}^{2} \xi_{12}^{4} \xi_{13}^{8}+\ldots+\xi_{11}^{2} \ldots \xi_{1 n}^{2^{n}}+\ldots
\]
is a Volterra series, and $\operatorname{deg}_{1 j} \psi_{1}=2^{j}$ (if $m>1$ then also $\operatorname{deg}_{i j} \psi_{1}=-\infty$ for $i=2, \ldots, m$);
\[
\psi_{2}:=\xi_{11}+\xi_{12}+\xi_{12}^{2}+\ldots+\xi_{1 n}+\ldots+\xi_{1 n}^{n}+\ldots
\]
is also a Volterra series, with $\operatorname{deg}_{1 j} \psi_{2}=j ;$ but
\[
\psi_{3}:=\xi_{11}+\xi_{11}^{2}+\xi_{11}^{3}+\ldots+\xi_{11}^{n}+\ldots
\]
is \textit{not} a Volterra series.

Since finiteness of the $\operatorname{deg}_{i j}$ is preserved under the algebra operations, the set of \textit{all} Volterra series is a $k$-algebra. This $k$-algebra is an integral domain, since the algebra of all power series does not have any zero divisors.

\begin{notation}\label{nota:5.3}
$\Psi_{k}$, or just $\Psi$, is the $k$-algebra of (formal) Volterra series over $k$. (For each $m$, a different $\Psi$.)
\end{notation}

\begin{definition}\label{def:5.4}
\textit{The degree of a Volterra series} $\psi$ \textit{is}
\[
\operatorname{deg} \psi:=\sup _{i, j}\left\{\operatorname{deg}_{i j} \psi\right\} \leq \infty .
\]
Thus $\operatorname{deg} \psi_{1}=\operatorname{deg} \psi_{2}=\infty$ for the above examples while, on the other hand,
\[
\operatorname{deg}\left(a_{1} \xi_{11}+a_{2} \xi_{12}+\ldots+a_{n} \xi_{1 n}+\ldots\right)=1 .
\]
\label{op:46}%
A column $p$-vector of Volterra series can be obviously regarded as a power series in the $\xi_{i j}$ with coefficients in $k^{p}$, via the identification
\[
\psi=\left(\begin{array}{c}
\psi^{(1)} \\
\vdots \\
\psi^{(p)}
\end{array}\right)=\sum_{\alpha \text { in } \Delta}\left(\begin{array}{c}
\psi_{\alpha}^{(1)} \\
\vdots \\
\psi_{\alpha}^{(p)}
\end{array}\right) \xi^{\alpha} .
\]
The definition of degree can be obviously generalized to the vector case via
\[
\operatorname{deg}_{i j} \psi:=\max \left\{\operatorname{deg}_{i j} \psi^{(1)}, \ldots, \operatorname{deg}_{i j} \psi^{(p)}\right\} .
\]
We let $\Psi^{p}$ denote the set of all vector Volterra series.
\end{definition}

Volterra series with $\operatorname{deg} \psi<\infty$ will be studied in detail in Section~\ref{ch:5}. An important tool in that study will be the concept of \textit{exponent series}, which we now introduce.

The concept of time-shift is incorporated into the context of Volterra series through a product of Volterra series which is based upon the monoid ($\Delta, \cdot$). We denote by $\psi \cdot \hat{\psi}$ the Volterra series defined by
\begin{equation}
(\psi \cdot \hat{\psi})_{\alpha}:=\sum_{\beta \gamma=\alpha} \psi_{\beta} \hat{\psi}_{\gamma} \text { for all } \alpha \text { in } \Delta . \label{eq:5.5}
\end{equation}
Note that $\psi \cdot \hat{\psi}$ is \textit{not} the same as the convolution product $\psi \hat{\psi}$ defined by \eqref{eq:5.2}. A change of notation is useful at this stage. Instead of writing $\psi=\sum \psi_{\alpha} \xi^{\alpha}$, we shall use the notation
\begin{equation}
\sum \psi_{\alpha} \alpha, \label{eq:5.6}
\end{equation}
and call the expression \eqref{eq:5.6} the \textit{exponent series} $\varphi$ \textit{associated to} $\psi$. Thus $\varphi$ is just a different notation for the same mathematical object $\psi$. If $\varphi, \hat{\varphi}$ are associated to $\psi, \hat{\psi}$, we denote $\varphi \hat{\varphi}:=\psi \cdot \hat{\psi}$. With these notations, the product \eqref{eq:5.5} can now be expressed simply as a linear extension of the multiplication among indeterminates:
\begin{equation}
\varphi \hat{\varphi}=\left(\sum \varphi_{\alpha} \alpha\right)\left(\sum \hat{\varphi}_{\beta} \beta\right)=\sum_{\alpha} \sum_{\beta} \varphi_{\alpha} \hat{\varphi}_{\beta} \alpha \beta . \label{eq:5.7}
\end{equation}

\label{op:47}%
Exponent series provide a new way of expressing the condition "deg $\psi<\infty$ ". For this, let
\begin{equation}\label{eq:5.8}
\operatorname{supp} \psi=\operatorname{supp} \varphi:=\left\{\alpha \text{ in } \Delta \mid \psi_{\alpha} \neq 0\right\},
\end{equation}
be the \textit{support} of $\psi$ (or of its associated exponent series).

\begin{definition}\label{def:5.9}
\textit{The support of} $\psi$ \textit{is finitely generated iff there exists a finite subset} $J=J_{\psi}$ \textit{of} $\mathbb{N}^{m}$ \textit{such that}
\begin{equation}\label{eq:5.10}
\operatorname{supp} \psi \subseteq \Delta_{J}:=J^{*} \cap \Delta ,
\end{equation}
i.e., \textit{each column} $\alpha_{j}$ \textit{of} $\alpha=\alpha_{1} \ldots \alpha_{t}$ \textit{is in} $J$, \textit{for each} $\alpha$ \textit{in} $\operatorname{supp} \psi$.
\end{definition}

In terms of the exponent series $\varphi$ associated to $\psi$, Definition~\ref{def:5.9} \textit{means that} $\varphi$ \textit{is a power series in the finitely many (noncommuting) variables in} $J$. Since for any integer $d$ there are only finitely many vectors in $\mathbb{N}^{m}$ with all entries $\leq d$, we have the following trivial

\begin{lemma}\label{lem:5.11}
deg $\psi<\infty$ if and only if supp $\psi$ is finitely generated.
\end{lemma}

The fact that $\operatorname{deg} \psi<\infty$ is equivalent to $\varphi$ being a series in finitely many variables will be exploited in Section~\ref{ch:5}.

We now return to our study of arbitrary Volterra series. Their introduction was motivated by the need of evaluating series at arbitrary input values. We now study these evaluations. In fact, we study a more general type of operation on Volterra series.

Let $K$ be an overring of $k$, and suppose given an infinite family $r=\left\{r_{i j}, i=1, \ldots, m, j=1,2, \ldots\right\}$ of elements of $K$. As before, we introduce the shorthand notation $r_{j}=r_{1 j}, \ldots, r_{m j}$ for each $j$, and, for each $\alpha=\alpha_{1} \ldots \alpha_{t}$ in $\Delta$,
\[
r^{\alpha}:=r_{1}^{\alpha_{1}} \ldots r_{t}^{\alpha_{t}}:=r_{11}^{\alpha_{11}} \ldots r_{m t}^{\alpha_{m t}}
\]
(products in the ring $K$).

\label{op:48}%
The definition of a Volterra series $\psi$ as a power series with all $\operatorname{deg}_{i j} \psi<\infty$ is clearly equivalent to the requirement that $\psi$ be a polynomial when expressed as a series in each finite subset of variables $\xi_{1}, \ldots, \xi_{t}$, for any fixed $t$. Thus we may write
\begin{equation}\label{eq:5.12}
\psi\left(\xi_{1}, \xi_{2}, \ldots\right)=\sum_{\substack{\alpha \text { in } \Delta \\|\alpha| \leq t}} \zeta_{\alpha}\left(\xi_{t+1}, \ldots\right) \xi^{\alpha} ,
\end{equation}
and $\zeta_{\alpha}=0$ for all but finitely many $\alpha$. Let us make the substitutions $\xi_{i j} \mapsto r_{i j}$ in \eqref{eq:5.12}, for $i=1, \ldots, m$ and $j=1, \ldots, t$, performing the resulting products between the $r_{i j}$. We obtain
\begin{equation}\label{eq:5.13}
\psi\left(r_{1}, \ldots, r_{t}, \xi_{t+1}, \ldots\right)=\sum_{\substack{\alpha \text { in } \Delta \\|\alpha| \leq t}} \zeta_{\alpha}\left(\xi_{t+1}, \ldots\right) r^{\alpha} .
\end{equation}

Since, as we already remarked, the sum in \eqref{eq:5.13} is finite, a further evaluation $\xi_{i j} \mapsto 0$, $i=1, \ldots, m$, $j>t$, results in a finite linear combination
\begin{equation}\label{eq:5.14}
\psi\left(r_{1}, \ldots, r_{t}, 0, \ldots\right)=\sum \zeta_{\alpha}(0, \ldots) r^{\alpha}, \zeta_{\alpha}(0, \ldots) \text{ in } k .
\end{equation}

When $K=k$, \eqref{eq:5.14} is then in $k$, the result of substituting a finite "input" sequence into the Volterra series. In general we obtain an element of $K$, and the assignment
\begin{equation}
\Psi \rightarrow K: \psi \mapsto \psi\left(r_{1}, \ldots, r_{t}, 0, \ldots\right), \label{eq:5.15}
\end{equation}
is clearly a $k$-algebra homomorphism.

We may instead apply to \eqref{eq:5.13} the further substitutions $\xi_{i j} \mapsto \xi_{i, j-t}, i=1, \ldots, m, j>t$, to obtain
\begin{equation}
\psi\left(r_{1}, \ldots, r_{t}, \xi_{1}, \xi_{2}, \ldots\right)=\sum \zeta_{\alpha}\left(\xi_{1}, \ldots\right) r^{\alpha} . \label{eq:5.16}
\end{equation}
Since \eqref{eq:5.16} is a finite $K$-combination of Volterra series over $k \subseteq K$, we may regard \eqref{eq:5.16} as a Volterra series with coefficients in $K$; this justifies the notation $\psi\left(r_{1}, \ldots, r_{t}\right)\left(\xi_{1}, \xi_{2}, \ldots\right)$ or just
\label{op:49}%
\begin{equation}\label{eq:5.17}
\psi\left(r_{1}, \ldots, r_{t}\right) ,
\end{equation}
instead of $\psi\left(r_{1}, \ldots, r_{t}, \xi_{1}, \ldots\right)$. An alternative is to view $\psi\left(r_{1}, \ldots, r_{t}\right)$ as an element of $\Psi \otimes_{k} K$. The assignment
\begin{equation}\label{eq:5.18}
\Psi \rightarrow \Psi \otimes_{k} K: \psi \mapsto \psi\left(r_{1}, \ldots, r_{t}\right) ,
\end{equation}
is clearly a $k$-algebra homomorphism. Moreover, \eqref{eq:5.18} is an isomorphism when the $r_{i j}$ are algebraically independent over $\Psi$, since it just amounts to a relabeling of variables.

In order to state a technical lemma to be used later, we need the following

\begin{notation}\label{nota:5.19}
Let $t \geq 0$ be an integer. Then
\[
\epsilon_{t}: \Psi \rightarrow k\left[\xi_{1}, \ldots, \xi_{t}\right],
\]
is the homomorphism $\psi \mapsto \psi\left(\xi_{1}, \ldots, \xi_{t}, 0, \ldots\right)$.
\end{notation}

Since $k\left[\xi_{1}, \ldots, \xi_{t}\right]$ is an integral domain, ker $\epsilon_{t}$ is a prime ideal. Also,
\begin{equation}\label{eq:5.20}
\bigcap_{t \geq 0} \text{ ker } \epsilon_{t}=\{0\} .
\end{equation}

Let $s \leq t$. Then there exists an onto homomorphism
\[
\epsilon_{s, t}: k\left[\xi_{1}, \ldots, \xi_{t}\right] \rightarrow k\left[\xi_{1}, \ldots, \xi_{s}\right],
\]
obtained by setting $\xi_{i j}=0$ for $i=1, \ldots, m$ and $j=s+1, \ldots, t$. By definition of $\epsilon_{t}$,
\begin{equation}\label{eq:5.21}
\epsilon_{s, t} \circ \epsilon_{t}=\epsilon_{s} .
\end{equation}

Let $B$ be a $k$-subalgebra of $\Psi$ and write $B_{t}:=\epsilon_{t}(B)$. Then $R_{t}:=$ ker $\epsilon_{t} \mid B=B \cap$ ker $\epsilon_{t}$ is a prime ideal of $B$ and $\cap_{t} R_{t} \subseteq \bigcap_{t}$ ker $\epsilon_{t}=\{0\}$. We prove a technical

\begin{lemma}\label{lem:5.22}
trdeg $B=\sup _{t \geq 0}\left\{\operatorname{trdeg} B_{t}\right\}$.
\end{lemma}

\label{op:50}%

\begin{proof}
By Lemma~\ref{lem:4.2}, trdeg $B \geq \sup \left\{\operatorname{trdeg} B_{t}\right\}$. Thus it will be enough to prove that $\sup \left\{\right.$ trdeg $\left.B_{t}\right\}=n<\infty$ implies $B \simeq B_{t}$ for all large $t$. Since each $\epsilon_{s, t} \mid B_{t}: B_{t} \rightarrow B_{s}$ is onto, the integers trdeg $B_{t}$ form an ascending chain, bounded above by $n$. So trdeg $B_{r}=\operatorname{trdeg} B_{r+1}=\ldots$ for some $r$. By Lemma~\ref{lem:4.2}, $\epsilon_{r, t} \mid B_{t}$ is an isomorphism for each $t \geq r$. So by \eqref{eq:5.21}, $R_{t}=R_{r}$. Therefore $R_{r}=\cap R_{t}=\{0\}$ and $\epsilon_{t} \mid B: B \simeq B_{t}$ for all $t \geq r$.
\end{proof}

\subsection{Construction of $\Omega$ and $\Gamma$}\label{sec:6}

We now begin to define response maps.

\begin{definition}\label{def:6.1}
\textit{The input space} $\Omega$ \textit{is the} $k$-\textit{space} $X(\Psi)$.
\end{definition}

To verify that $\Omega$ is well-defined according to the setup developed in Section~\ref{ch:2}, we remark that $\Psi$ is $k$-reduced. Indeed, ker $\epsilon_{t}$ is a closed ideal for each $t$, since $k\left[\xi_{1}, \ldots, \xi_{t}\right]$ is $k$-reduced. So by \eqref{eq:5.20} the ideal $\{0\}$ is closed, i.e. $\Psi$ is $k$-reduced.

\begin{definition}\label{def:6.2}
\textit{The space of input values is}
\[
U:=k^{m} .
\]
The algebra of polynomial functions on the $k$-space $U$ is $k\left[T_{1}, \ldots, T_{m}\right]$. Therefore the algebra $A\left(U^{t}\right)$ of polynomial functions on the $t$-fold product of $k$-spaces $U^{t}=U \times \ldots \times U$ is the ring of polynomials in $m t$ variables. So we may denote $A\left(U^{t}\right)$ by $k\left[\xi_{1}, \ldots, \xi_{t}\right]$. We adopt the notational convention of writing the sequences ($u_{1}, u_{2}, \ldots, u_{t}$) in $U^{t}$ in an inverted order ($u_{t}, \ldots, u_{1}$). Thus, the coordinate function $\xi_{i j}$ acts on elements of $U^{t}$ by $\xi_{i j}\left(u_{s}, \ldots, u_{1}\right):=u_{i j}=i$-th entry of $j$-th vector counted from the right (for instance, if $m=1$, then $\xi_{12}(0,1,0,0,0)=0$, $\xi_{14}(0,1,0,0,0)=1$). This notation will be consistent with the following interpretation: ($u_{t}, \ldots, u_{1}$) represents an input sequence such that $u_{j}$ is the input at time $1-j$.
\end{definition}

\label{op:51}%
Using this notation, $\epsilon_{t}$ gives rise to a closed immersion $i_{t}:=X\left(\epsilon_{t}\right): U^{t} \rightarrow \Omega$. Similarly, each $i_{s, t}:=X\left(\epsilon_{s, t}\right): U^{s} \rightarrow U^{t}$ is a closed embedding, mapping sequences $\left(u_{s}, \ldots, u_{1}\right)$ of $U^{s}$ into $\left(0, \ldots, 0, u_{s}, \ldots, u_{2}, u_{1}\right)$ in $U^{t}$. By \eqref{eq:5.21} the following diagram commutes:

\begin{equation}
\begin{tikzpicture}[x=1cm,y=1cm]
\node (Us) at (0,0) {$U^s$};
\node (Ut) at (4,0) {$U^t$};
\node (Om) at (2,-2) {$\Omega$};
\draw[arr] (Us) -- node[above,lab] {$i_{s,t}$} (Ut);
\draw[arr] (Us) -- node[above left,lab] {$i_s$} (Om);
\draw[arr] (Ut) -- node[above right,lab] {$i_t$} (Om);
\end{tikzpicture} \label{eq:6.3}
\end{equation}
Therefore the sets $i_{t}\left(U^{t}\right), t \geq 0$, form an ascending chain in $\Omega$ whose union (limit) may be identified with the set of all infinite sequences ($\ldots, u_{n}, \ldots, u_{2}, u_{1}$) with finitely many nonzero entries; the rule
\begin{equation}\label{eq:6.4}
\left(\ldots, 0, u_{t}, \ldots, u_{1}\right) \sim u_{t} z^{t-1}+\ldots+u_{1} ,
\end{equation}
permits identifying this union with the set of polynomials over some symbol, say $z$, with coefficients in $U=k^{m}$. Thus we have the

\begin{notation}\label{nota:6.5}
$U[z]:=\bigcup_{t \geq 0} i_{t}\left(U^{t}\right)$.
\end{notation}

We shall denote by (0) the sequence with $u_{t}=0$ for all $t$.

\begin{lemma}\label{lem:6.6}
$\overline{U[z]}=\Omega$.
\end{lemma}

\begin{proof}
Clear by Proposition~\ref{prop:2.7}(f) and \eqref{eq:5.20}. \hfill $\square$

\textit{We may regard} $\Omega$ \textit{as a "completion" of} $U[z]$. \textit{This will allow us to endow} $U[z]$ \textit{with the geometric structure carried by a subset of the $k$-space} $\Omega$.

There will be no danger in identifying $U^{t}$ with $i_{t}\left(U^{t}\right)$, so that we may think of $U^{t}$ as the closed subset of $U[z]$ corresponding to the sequences $\left(\ldots, 0, u_{t}, \ldots, u_{1}\right)$.

Now let $K$ be the field obtained from $k$ by adjunction of denumerably many new indeterminates $S_{i j}, i=1, \ldots, m, j=1,2, \ldots$. Applying \eqref{eq:5.18} with $r_{i j}:=S_{i j}$, we define
\label{op:52}%
\begin{equation}\label{eq:6.7}
\theta_{t}: \Psi \rightarrow \Psi \otimes k\left[S_{1}, \ldots, S_{t}\right]: \psi \mapsto \psi\left(S_{1}, \ldots, S_{t}\right) .
\end{equation}

By duality, there is a polynomial map
\begin{equation}
\delta_{t}:=X\left(\theta_{t}\right): \Omega \times U^{t} \rightarrow \Omega . \label{eq:6.8}
\end{equation}
Thus $\delta_{t}$ is the map whose transpose $A\left(\delta_{t}\right)=\theta_{t}$. We claim that the polynomial action of $U^{t}$ upon $\Omega$ given by $\delta_{t}$ is the natural extension to $\Omega$ of the "concatenation" maps
\begin{align}
\delta_{s, t}: U^{s} \times U^{t} \rightarrow U^{s+t}:\left(\left(u_{s}, \ldots, u_{1}\right),\right. & \left.\left(\hat{u}_{t}, \ldots, \hat{u}_{1}\right)\right) \mapsto  \label{eq:6.9}\\
& \mapsto\left(u_{s}, \ldots, u_{1}, \hat{u}_{t}, \ldots, \hat{u}_{1}\right) . \notag\end{align}
In other words, we have:
\end{proof}

\begin{lemma}\label{lem:6.10}
The following diagram commutes for each $s$, $t$:

\[
\begin{tikzpicture}[x=1cm,y=1cm]
\node (Ust) at (0,0) {$U^s\times U^t$};
\node (Usp) at (5,0) {$U^{s+t}$};
\node (OmUt) at (0,-2.5) {$\Omega\times U^t$};
\node (Om) at (5,-2.5) {$\Omega$};
\draw[arr] (Ust) -- node[above,lab] {$\delta_{s,t}$} (Usp);
\draw[arr] (Ust) -- node[left,lab] {$i_s\times 1$} (OmUt);
\draw[arr] (OmUt) -- node[below,lab] {$\delta_t$} (Om);
\draw[arr] (Usp) -- node[right,lab] {$i_{s+t}$} (Om);
\end{tikzpicture}
\]
\end{lemma}

\begin{proof}
By duality, it is necessary and sufficient to verify that the following diagram commutes:

\[
\begin{tikzpicture}[x=1cm,y=1cm]
\node (Psi) at (0,0) {$\Psi$};
\node (Psist) at (7,0) {$\Psi[S_1,\ldots,S_t]$};
\node (K1) at (0,-3) {$k[\xi_1,\ldots,\xi_{s+t}]$};
\node (K2) at (7,-3) {$k[\xi_1,\ldots,\xi_s][S_1,\ldots,S_t]$};
\draw[arr] (Psi) -- node[above,lab] {$A(\delta_t)$} (Psist);
\draw[arr] (Psi) -- node[left,lab] {$\epsilon_{s+t}$} (K1);
\draw[arr] (Psist) -- node[right,lab] {$\epsilon_s[S_1,\ldots,S_t]$} (K2);
\draw[arr] (K1) -- node[above,lab] {$A(\delta_{s,t})$} (K2);
\end{tikzpicture}
\]
Note that, in the coordinates displayed, $A\left(\delta_{s, t}\right)$ is given by $\xi_{i j} \mapsto S_{i j}$ if $j=1, \ldots, t$, and by $\xi_{i j} \mapsto \xi_{i, j-t}$ if $j>t$. Thus the diagram commutes by definition of $\epsilon_{s}, \epsilon_{s+t}, \theta_{t}$.
\end{proof}

In view of Lemma~\ref{lem:6.10}, there will be no danger in denoting the operations $\delta_{s, t}$, as well as $\delta_{t}$, for all $s$ and $t$, simply by concatenation
\label{op:53}%
\begin{equation}\label{eq:6.11}
\omega v:=\delta_{t}(\omega, v) \text{ for } \omega \text{ in } \Omega, v \text{ in } U^{t} .
\end{equation}

Let $f$ be a polynomial function $\Omega \rightarrow k$. Since $A(\Omega)=\iota(\Psi) \simeq \Psi$, clearly $f$ can be identified with a Volterra series $\psi=\psi_{f}$. So, by Corollary~\ref{cor:3.8}, $\psi \mid U^{t}=\psi \circ i_{t}$,
\[
A\left(\psi \circ i_{t}\right)(T)=\left(A\left(i_{t}\right) \circ A(\psi)\right)(T)=A\left(i_{t}\right)(\psi)=\epsilon_{t}(\psi) ;
\]
thus (with the notation in \eqref{eq:5.14})
\begin{equation}\label{eq:6.12}
\psi \mid U^{t}:\left(u_{t}, \ldots, u_{1}\right) \mapsto \psi\left(u_{1}, \ldots, u_{t}, 0, \ldots\right) .
\end{equation}

Since $U[z]$ is the increasing union of the $U^{t}$, a map $U[z] \rightarrow X$ is specified by its restrictions to the $U^{t}$. In particular, take $\psi$ in $\Psi$. Then by \eqref{eq:6.12} the value of $\psi \mid U[z]$ at $\left(\ldots, u_{n}, \ldots, u_{1}\right)$ is obtained by evaluating the power series $\psi$ at $\xi_{i j}:=u_{i j}=i$-th row of $u_{j}$. This evaluation is well defined because almost all $u_{j}$ are zero. A continuous function is already determined by its values in a dense subset, so by Lemma~\ref{lem:6.6} the assignment $\psi \mapsto \psi \mid U[z]$ is one-to-one.

Thus the following mild abuse of terminology is justified:

\begin{definition}\label{def:6.13}
A \textit{polynomial map} $\ell: U[z] \rightarrow X$, \textit{where} $X$ \textit{is a} $k$-\textit{space, is the restriction of a polynomial map} $\ell^{\Omega}: \Omega \rightarrow X$.
\end{definition}

The gist of the introduction of $\Omega$ is that the abstract set of input sequences $U[z]$ can now be exhibited as a dense subset of the $k$-space $\Omega$ and thus $U[z]$ is itself endowed (by restriction) with coordinates and polynomial functions.

Thus the polynomial functions $f: U[z] \rightarrow k$ are in a bijective correspondence with Volterra series $\psi$, via evaluation of $\psi$ at $\xi_{i j}:=u_{i j}$. More generally, a polynomial map $f: U[z] \rightarrow k^{p}$ is uniquely determined by the functions $\pi_{j} \circ f: U[z] \rightarrow k$, where $\pi_{j}, j=1, \ldots, p$ are the natural projections $k^{p} \rightarrow k$. So polynomial maps $f: U[z] \rightarrow k^{p}$ are in a bijective correspondence with $\Psi^{p}$, the ordered $p$-tuples of Volterra series.

\label{op:54}%

\begin{definition}\label{def:6.14}
\textit{The space of output values is} $Y:=k^{p}$.
\end{definition}

Thus,
\begin{equation}
\Psi^{p} \approx \text { polynomial maps } U[z] \rightarrow Y . \label{eq:6.15}
\end{equation}
Finally, we define
\[
U^{*}:=\bigcup_{t \geq 0} U^{t} \quad (\textit{disjoint} \text { union). }
\]
This set should not be confused with $U[z]$, the set of finitely nonzero sequences, which was obtained as a quotient set of $U^{*}$, via the identifications $\left(u_{t}, \ldots, u_{1}\right) \sim\left(0, \ldots, 0, u_{t}, \ldots, u_{1}\right)$. The unique element of $U^{0}$ is denoted $(\emptyset)$.

\begin{definition}\label{def:6.16}
\textit{The output space} $\Gamma$ \textit{is the set} $Y^{U^{*}}$ \textit{of all maps} $U^{*} \rightarrow Y$.
\end{definition}

Thus an element of $\Gamma$ is an $U^{*}$-indexed sequence of elements of $Y$. By Remark~\ref{rem:1.9} and Corollary~\ref{cor:3.8}, $\Gamma$ is a $k$-space, the product of the $k$-space $Y$ with itself, $U^{*}$ times.

\subsection{Abstract Response Maps and Systems}\label{sec:7}

For any vector space $V$ over $k$ we consider the set of sequences $\mathbb{Z} \rightarrow V$ with support bounded on the left:
\[
\mathbb{V}:=\left\{u: \mathbb{Z} \rightarrow V \mid\left(\exists \tau_{u}\right)\left(u(t)=0 \text { if } t<\tau_{u}\right)\right\} .
\]
The \textit{shift} operator $\sigma=\sigma_{V}: \mathbb{V} \rightarrow \mathbb{V}$ is defined by
\[
(\sigma u)(t):=u(t+1) \quad \text { for all } t \text { in } \mathbb{Z} .
\]
In particular we shall call
\begin{align*}
&\mathbb{U}:=\text { set of input sequences; } \\
&\mathbb{Y}:=\text { set of output sequences. }
\end{align*}

\label{op:55}%

\begin{definition}\label{def:7.1}
\textit{An input/output map} $\mathbf{f}$ \textit{is a map} $\mathbf{f}: \mathbb{U} \rightarrow \mathbb{Y}$; $\mathbf{f}$ \textit{is called}
\begin{enumerate}
\item[(a)] (\textit{strictly}) \textit{causal iff, for all} $u, \hat{u}$ \textit{in} $\mathbb{U}$, \textit{and for all} $\tau$ in $\mathbb{Z}$, $u(t)=\hat{u}(t)$ \textit{for} $t<\tau$ \textit{implies} $\mathbf{f}(u)(t)=\mathbf{f}(\hat{u})(t)$ \textit{for} $t \leq \tau$;
\item[(b)] \textit{constant structure (or shift-invariant) iff} $\mathbf{f}(\sigma u)=\sigma \mathbf{f}(u)$ \textit{for all} $u$ \textit{in} $\mathbb{U}$.
\end{enumerate}
\textit{An input/output pair of} $\mathbf{f}$ \textit{is a pair} $(\mathbf{f}(u), u), u$ \textit{in} $\mathbb{U}$.
\end{definition}

An input/output map can be regarded as an abstract description of the external behavior of a "black box" (physical device, computer, etc.) operating at discrete instants $\ldots,-2,-1,0,1,2, \ldots$ of time. Constant structure means that this behavior is invariant under time shifts.

Rather than working with the input/output map, it is technically more convenient to work with the response map of the same "black box", i.e., with the description of outputs resulting immediately after the application of finite sequences of inputs.

\begin{definition}\label{def:7.2}
A \textit{response map} $f$ \textit{is a map}
\[
f: U[z] \rightarrow Y .
\]
Let $\eta: \mathbb{Y} \rightarrow Y: y \mapsto y(1)$; then \textit{the assignments}
\[
\mathbf{f} \mapsto f:=\eta \circ \mathbf{f} \mid U[z],
\]
\textit{and}
\[
f \mapsto \mathbf{f}(u)(t):=f\left(\sigma^{t} u\right),
\]
\textit{establish a bijection between causal constant-structure input/output maps and response maps}.
\end{definition}

In the above formulas we are implicitly identifying $U[z]$ with a subset of $\mathbb{U}$, via the rule
\label{op:56}%
\[
u_{\tau} z^{\tau-1}+\ldots+u_{1} \mapsto u,
\]
where $u(-t):=u_{t-1}$ if $0 \leq t<\tau$ and $u(t)=0$ otherwise.

In view of the above bijection, we shall state most of our definitions and results in terms of response maps. Only when dealing with input/output equations shall we refer again directly to input/ output maps.

The following abstract definitions and notations are mostly well-known. They belong to "general system theory"; no structure is imposed on systems or preserved by response maps. Later we shall refine these definitions by adding suitable algebraic structure.

\begin{definition}\label{def:7.3}
\textit{An abstract (constant structure) system} $\Sigma$ \textit{is an object} $\left(X, P, h, x^{\#}\right)$, \textit{where}
\begin{enumerate}
\item[(a)] $X$ \textit{is the state set};
\item[(b)] $P: X \times U \rightarrow X$ \textit{is the transition map};
\item[(c)] $h: X \rightarrow Y$ \textit{is the output map; and}
\item[(d)] $x^{\#}$ in $X$ \textit{is the initial state, and satisfies} $P\left(x^{\#}, 0\right)=x^{\#}$.
\end{enumerate}
\textit{The} $t$-\textit{th iterate of} $P$ \textit{is the map} $P^{(t)}: X \times U^{t} \rightarrow X$ \textit{defined recursively by}
\[
\begin{aligned}
& P^{(0)}:=1_{X}, P^{(1)}:=P \\
& P^{(t+1)}\left(x,\left(u_{t+1}, \ldots, u_{1}\right)\right):=P\left(P^{(t)}\left(x,\left(u_{t+1}, \ldots, u_{2}\right)\right), u_{1}\right) .
\end{aligned}
\]
\textit{The reachability map of} $\Sigma$ \textit{is}
\[
g: U[z] \rightarrow X: u_{t} z^{t-1}+\ldots+u_{1} \mapsto P^{(t)}\left(x^{\#},\left(u_{t}, \ldots, u_{1}\right)\right) ;
\]
\textit{the} $t$-\textit{step reachability map is} $g_{t}:=g \mid U^{t}$; \textit{the} $t$-\textit{step reachable set is} $X_{t}:=g_{t}\left(U^{t}\right)$.
\textit{For each} $w$ \textit{in} $U^{t}, t \geq 0$, \textit{the observable map induced by} $w$ \textit{is}
\[
h^{w}: X \rightarrow Y: x \mapsto h \circ P^{(t)}(x, w) .
\]
\label{op:57}%
\textit{The basic observables of} $\Sigma$ \textit{are the functions} $h_{j}^{w}:=\pi_{j} \circ h^{w}$, $j=1, \ldots, p$, $w$ \textit{in} $U^{*}$. \textit{The observability map of} $\Sigma$ \textit{is}
\[
h^{\Gamma}: X \rightarrow \Gamma: x \mapsto\left\{h^{w}(x), w \text { in } U^{*}\right\} .
\]
\begin{quote}
$\Sigma$ \textit{is called}
\textit{reachable iff} $g$ \textit{is onto};
\textit{observable iff} $h^{\Gamma}$ \textit{is one-to-one};
\textit{abstractly canonical iff both reachable and observable}.
\end{quote}
\textit{The response of} $\Sigma$ \textit{is} $f_{\Sigma}:=h \circ g$. \textit{Given a response} $f, \Sigma$ \textit{realizes} $f$ \textit{iff} $f=f_{\Sigma}$.
\end{definition}

\begin{definition}\label{def:7.4}
\textit{Consider abstract systems} $\Sigma, \hat{\Sigma}$. \textit{An abstract system morphism} $T: \Sigma \rightarrow \hat{\Sigma}$ \textit{is given by a map} $T: X \rightarrow \hat{X}$ \textit{such that}:
\begin{enumerate}
\item[(i)] $T\left(x^{\#}\right)=\hat{x}^{\#}$, and
\item[(ii)] \textit{the diagram}
\[
\begin{tikzpicture}[x=1cm,y=1cm]
\node (XU) at (0,0) {$X\times U$};
\node (X) at (4,0) {$X$};
\node (Y) at (8,0) {$Y$};
\node (XUh) at (0,-2.2) {$\hat X\times U$};
\node (Xh) at (4,-2.2) {$\hat X$};
\draw[arr] (XU) -- node[above,lab] {$P$} (X);
\draw[arr] (XU) -- node[left,lab] {$T\times 1$} (XUh);
\draw[arr] (XUh) -- node[below,lab] {$\hat P$} (Xh);
\draw[arr] (X) -- node[right,lab] {$T$} (Xh);
\draw[arr] (X) -- node[above right,lab] {$h$} (Y);
\draw[arr] (Xh) -- node[below right,lab] {$\hat h$} (Y);
\end{tikzpicture}
\]
\textit{commutes}.
\end{enumerate}
\end{definition}

Note that the existence of a system morphism $T: \Sigma \rightarrow \hat{\Sigma}$ implies that $f_{\Sigma}=f_{\hat{\Sigma}}$; this is proved by a trivial induction.

We shall call $\mathbf{Syst}_{\mathrm{abs}}$ the category consisting of abstract systems as objects and morphisms defined as in Definition~\ref{def:7.4}.

Analogously with concepts related to $\Sigma$, the same concepts can be defined directly in terms of the response map:

\begin{definition}\label{def:7.5}
\textit{Let} $f$ \textit{be a response map. The observable map of} $f$ \textit{induced by} $w$, \textit{where} $w$ \textit{is in} $U^{*}$, \textit{is}
\label{op:58}%
\[
f^{w}: U[z] \rightarrow Y: v \mapsto f(v w) .
\]
\textit{The basic observables of} $f$ \textit{are the functions} $f_{j}^{w}:=\pi_{j} \circ f^{w}$, $j=1, \ldots, p$, $w$ \textit{in} $U^{*}$. \textit{The observability map of} $f$ \textit{is}
\[
f^{\Gamma}: U[z] \rightarrow \Gamma: v \mapsto\left\{f^{w}(v), \quad w \text { in } U^{*}\right\} .
\]
Thus the \textit{value} of $f^{w}$ is the result of an input/output experiment in which the output is observed at the end of the application of the concatenated input $v w$, $v=$ given, $w=$ chosen. The collection of results of \textit{all} such experiments is $f^{\Gamma}$.
\end{definition}

Since $f^{(\emptyset)}=f$, clearly
\[
f=\hat{f} \text { iff } f^{\Gamma}=\hat{f}^{\Gamma} .
\]
In terms of the new notations, the definition of realization may be restated as follows. For any system $\Sigma$ and any $w$ in $U^{t}$, $v$ in $U^{s}$,
\[
\begin{aligned}
f_{\Sigma}^{v}(w)=f_{\Sigma}(w v)=h \circ g(w v) & =h \circ P^{(s+t)}\left(x^{\#}, w v\right)= \\
& =h \circ P^{(s)}(g(w), v)=h^{v} \circ g(w),
\end{aligned}
\]
so $f_{\Sigma}^{\Gamma}=h^{\Gamma} \circ g$. Thus,
\begin{equation}
\Sigma \text { realizes } f \text { iff } f^{\Gamma}=h^{\Gamma} \circ g . \label{eq:7.6}
\end{equation}
We now restate in our formalism largely well known facts (see, for instance, \citealp[Chapter XII]{eilenberg1974automata}).

\begin{lemma}\label{lem:7.7}
Let $\Sigma=\left(X, P, h, x^{\#}\right)$ and $\hat{\Sigma}=\left(\hat{X}, \hat{P}, \hat{h}, \hat{x}^{\#}\right)$ have the same response map. Assume that $\Sigma$ is reachable and that $\hat{\Sigma}$ is observable. Then there exists a unique map $T: X \rightarrow \hat{X}$ such that the following diagram commutes:

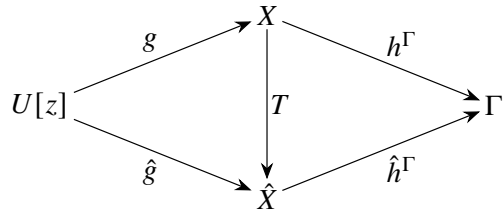
\[
\begin{tikzpicture}[x=1cm,y=1cm]
\node (U) at (0,0) {$U[z]$};
\node (X) at (3,1.2) {$X$};
\node (Xh) at (3,-1.2) {$\hat X$};
\node (G) at (6,0) {$\Gamma$};
\draw[arr] (U) -- node[above left,lab] {$g$} (X);
\draw[arr] (U) -- node[below left,lab] {$\hat g$} (Xh);
\draw[arr] (X) -- node[right,lab] {$T$} (Xh);
\draw[arr] (X) -- node[above right,lab] {$h^\Gamma$} (G);
\draw[arr] (Xh) -- node[below right,lab] {$\hat h^\Gamma$} (G);
\end{tikzpicture}
\]

\label{op:59}%
This unique $T$ induces an abstract system morphism $\Sigma \rightarrow \hat{\Sigma}$.
\end{lemma}

\begin{proof}
We first prove uniqueness. If such a $T$ exists, then $\hat{h}^{\Gamma}(T(g(w)))=\hat{h}^{\Gamma}(\hat{g}(w))$ for all $w$ in $U[z]$, so $T(g(w))=\hat{g}(w)$ by observability of $\hat{\Sigma}$. Since $\Sigma$ is reachable, every state in $X$ is of the form $g(w)$. Thus there is a unique choice of $T$,
\begin{equation*}
T: g(w) \mapsto \hat{g}(w) . \tag{*}
\end{equation*}
To prove existence, we need to see that $T$ as in (*) is a well-defined map. In other words, if $g(w)=g(v)$, it should follow that $\hat{g}(w)=\hat{g}(v)$. But $\hat{h}^{\Gamma}(\hat{g}(w))=h^{\Gamma}(g(w))=h^{\Gamma}(g(v))=\hat{h}^{\Gamma}(\hat{g}(v))$ implies $\hat{g}(w)=\hat{g}(v)$ by observability of $\hat{\Sigma}$.

We now prove that $T$ is an abstract system morphism. Since $T\left(x^{\#}\right)=T(g(0))=\hat{g}(0)=\hat{x}^{\#}$, we are only left to prove that $T(P(g(w), u))=\hat{P}(T(g(w)), u)$ for every $w$ in $U[z]$ and every $u$ in $U$. But $T(P(g(w), u))=T(g(w u))=\hat{g}(w u)=\hat{P}(\hat{g}(w), u)$, as required.
\end{proof}

\begin{theorem}\label{thm:7.8}
Let $f$ be a response map. Then $f$ has an abstractly canonical realization $\Sigma_{a c}$. If $\Sigma$ is any other abstractly canonical realization of $f$ then there exists a unique isomorphism $T: \Sigma_{a c} \rightarrow \Sigma$.
\end{theorem}

\begin{proof}
Abstractly canonical realizations $\Sigma$ correspond to factorizations $f^{\Gamma}=g_{\Sigma} \circ h_{\Sigma}^{\Gamma}$, where $g_{\Sigma}$ is onto and $h_{\Sigma}^{\Gamma}$ is one-to-one. We first prove the uniqueness part. Given any two abstractly canonical systems $\Sigma$ and $\hat{\Sigma}$, there exist by Lemma~\ref{lem:7.7} $T_{1}: \Sigma \rightarrow \hat{\Sigma}, T_{2}: \hat{\Sigma} \rightarrow \Sigma$. Since $T_{1} \circ T_{2}$ is an abstract system morphism, it follows by the uniqueness part of Lemma~\ref{lem:7.7} that $T_{1} \circ T_{2}=$ identity abstract system morphism $\hat{\Sigma} \rightarrow \hat{\Sigma}$. Similarly, $T_{2} \circ T_{1}$ is the identity $\Sigma \rightarrow \Sigma$. Thus $T_{1}, T_{2}$ are inverse abstract system isomorphisms.

For the existence part, simply take $X_{\mathrm{ac}}:=f^{\Gamma}(U[z]) \subseteq \Gamma$, $x_{a c}^{\#}:=f^{\Gamma}(0), h_{a c}:=$ projection onto ($\emptyset$)-th factor, and $P_{a c}\left(f^{\Gamma}(w), u\right):=f^{\Gamma}(w u)$. Note that $P_{a c}$ is well defined, since $f^{\Gamma}(w)=f^{\Gamma}(\hat{w})$ implies that for any $u, v$, $f^{v}(w u)=f(w u v)=f^{u v}(w)=f^{u v}(\hat{w})=f^{v}(\hat{w} u)$.
\end{proof}

\label{op:60}%
\subsection{Polynomial Response Maps and $k$-Systems}\label{sec:8}

We begin now to study response maps which are polynomial (Definition \ref{def:6.13}). The input/output map associated to a polynomial response map will be called a \textit{polynomial input/output map}.

The notion of polynomial response maps is new, as is the following

\begin{definition}\label{def:8.1}
\textit{The abstract constant structure system} $\Sigma=(X, P$, $h, x^{\#}$) \textit{is a} $k$-\textit{system iff} $X$ \textit{is a} $k$-\textit{space and both} $P$ \textit{and} $h$ \textit{are polynomial maps}.
\end{definition}

The terminology " $k$-system" is just a convenient short name for our systems. In no way is it intended to imply that Definition~\ref{def:8.1} constitutes the most general class of systems definable in terms of the field structure of $k$.

\begin{lemma}\label{lem:8.2}
The reachability map $g: U[z] \rightarrow X$ of a $k$-system $\Sigma$ is a polynomial map.
\end{lemma}

\begin{proof}
By definition~\ref{def:6.13}, we need to find a polynomial extension $g^{\Omega}: \Omega \rightarrow X$ of $g$. Dually, we construct a homomorphism $A\left(g^{\Omega}\right): A(X) \rightarrow \Psi$ such that $\epsilon_{t} \circ A\left(g^{\Omega}\right)=A\left(g_{t}\right)$ for each $t \geq 0$.

Pick any $a$ in $A(X)$. Then $A\left(g_{t}\right)(a)=a \circ g_{t}$ belongs to $k\left[\xi_{1}, \ldots, \xi_{t}\right]$ for each $t \geq 0$. We define a power series $\psi^{a}$ in the variables $\xi_{i j}$ by
\[
\psi_{\alpha}^{a}:=\text { coefficient of } \xi^{\alpha} \text { in } A\left(g_{|\alpha|}\right)(a) \text {, }
\]
for each $\alpha$ in $\Delta$.

Now we claim that $\operatorname{deg}_{i j} \psi^{a}<\infty$ for each $i, j$, i.e., $\psi^{a}$ is a Volterra series. Take any $i, j$. Let $t \geq j$. From the definition of $P^{(t)}$ it is clear that
\[
P^{(t)}=P^{(j)} \circ\left(P^{(t-j)} \times 1_{U^{j}}\right) .
\]
As $g_{t}=P^{(t)}\left(x^{\#}, \cdot\right)$,
\label{op:61}%
\[
a \circ g_{t}=a \circ P^{(j)} \circ\left(P^{(t-j)}\left(x^{\#}, \cdot\right) \times 1_{U^{j}}\right) .
\]
Thus the degree to which $\xi_{i j}$ appears in $a \circ g_{t}$, for any $t \geq j$, is at most the degree of the polynomial $a \circ P^{(j)}$ in the variable $\xi_{i j}$. So $\psi^{a}$ is a Volterra series.

Let $A\left(g^{\Omega}\right):=a \mapsto \psi^{a}$. By construction, $\epsilon_{t} \psi^{a}=A\left(g_{t}\right)(a)$ for each $a$ in $A(X)$ and $t \geq 0$. Since each $A\left(g_{t}\right)$ is a homomorphism, $A\left(g^{\Omega}\right)$ is also a homomorphism.
\end{proof}

The following result, the main of this section, suggests that $k$-\textit{systems are the natural realizations of polynomial response maps}. This will be confirmed later by the result on canonical realizations. We shall reserve the name "polynomial systems" for a special type of $k$-system in which a strong finiteness condition holds, which will allow $P$ and $h$ to be represented by actual polynomials.

\begin{theorem}\label{thm:8.3}
The response map $f$ is polynomial if and only if $f$ is realized by some $k$-system.
\end{theorem}

\begin{proof}
["only if"] The \textit{free realization} $\Sigma_{\text {free }}(f):= \left(\Omega, \delta_{1}, f^{\Omega},(0)\right)$ is a $k$-system realizing $f$.

["if"]. Let $f=f_{\Sigma}$ for some $k$-system $\Sigma$. Define
\[
f^{\Omega}:=h \circ g^{\Omega} .
\]
Then $f^{\Omega}\left|U[z]=h \circ g^{\Omega}\right| U[z]=h \circ g=f$. So $f$ is polynomial.
\end{proof}

Properties of $X$ serve to classify $k$-systems.

We shall say that $\Sigma$ is a \textit{polynomial system} [respectively \textit{almost polynomial}] iff $X_{\Sigma}$ is a variety [respectively an almost variety].

An \textit{irreducible} $\Sigma$ is one for which $X_{\Sigma}$ is irreducible.

Similarly, we define (recall Section~\ref{sec:4})
\[
\operatorname{dim} \Sigma:=\operatorname{dim} X_{\Sigma} .
\]
\label{op:62}%
The polynomial functions $X_{\Sigma} \rightarrow k$ are the \textit{costates} of $\Sigma$.

The notion of abstract system morphism is too weak to serve for comparing $k$-systems. A suitable category $\mathbf{Syst}_{k}$ of $k$-systems is obtained with morphisms as in the following

\begin{definition}\label{def:8.4}
\textit{An abstract system morphism} $T: \Sigma \rightarrow \hat{\Sigma}$ \textit{between two} $k$-\textit{systems is a} $k$-\textit{system morphism iff} $T: X \rightarrow \hat{X}$ \textit{is a polynomial map}. $T: \Sigma \rightarrow \hat{\Sigma}$ \textit{is dominating, a closed immersion, etc., iff} $T: X \rightarrow \hat{X}$ \textit{has the corresponding property}; $\Sigma$ \textit{dominates} $\hat{\Sigma}$ \textit{iff there exists a dominating} $T: \Sigma \rightarrow \hat{\Sigma} ; \Sigma$ \textit{is a closed subsystem of} $\hat{\Sigma}$ \textit{iff there exists a closed embedding} $T: \Sigma \rightarrow \hat{\Sigma}$.
\end{definition}

It is easy to see that $k$-systems form a category with the above notion of morphism. Note that a $k$-system isomorphism $T: \Sigma \simeq \hat{\Sigma}$ is the same as a polynomial change of coordinates in the state space.

\subsection{Quasi-Reachability}\label{sec:9}

\textit{From here until the end of Section~\ref{ch:4}}, $f$ \textit{is an arbitrary response map and} $\Sigma=\left(X, P, h, x^{\#}\right)$ \textit{is an arbitrary} $k$-\textit{system}.

\textit{The reachable set} of $\Sigma$ is
\[
X_{R}:=g(U[z])=\bigcup_{t \geq 0} g\left(U^{t}\right)=\bigcup_{t \geq 0} X_{t} .
\]

\begin{definition}\label{def:9.1}
$\Sigma$ \textit{is quasi-reachable iff} $\overline{X}_{R}=X$.
\end{definition}

By Lemma~\ref{lem:6.6} and Lemma~\ref{lem:3.11} we have the following

\begin{lemma}\label{lem:9.2}
The following statements are equivalent:
\begin{enumerate}
\item[(a)] $\Sigma$ is quasi-reachable.
\item[(b)] $g^{\Omega}$ is dominating.
\item[(c)] $A\left(g^{\Omega}\right)$ is one-to-one.
\end{enumerate}
\end{lemma}

\begin{lemma}[\citealp{sontag1975discrete}]\label{lem:9.3}
$\bar{X}_{t}=\bar{X}_{t+1}$ for some $t \geq 0$ implies $\bar{X}_{t}=\bar{X}_{R}$.
\end{lemma}

\label{op:63}%

\begin{proof}
Since $P$ is polynomial, it is continuous; thus
\[
X_{t+2}=P\left(X_{t+1} \times U\right) \subseteq P\left(\overline{X}_{t+1} \times U\right)=P\left(\overline{X}_{t} \times U\right) \subseteq \overline{P\left(X_{t} \times U\right)}=\overline{X}_{t+1},
\]
Since clearly $X_{t+1} \subseteq X_{t+2}$, it follows that $\overline{X}_{t+1} \subseteq \overline{X}_{t+2} \subseteq \overline{\overline{X}}_{t+1}=\overline{X}_{t+1}$, so $\bar{X}_{t+1}=\bar{X}_{t+2}$ and the result follows by induction.
\end{proof}

\begin{corollary}\label{cor:9.4}
If $\operatorname{dim} \Sigma=n<\infty$ then $\bar{X}_{n}=\bar{X}_{R}$.
\end{corollary}

\begin{proof}
Since $g^{\Omega}$ is continuous and $U^{t}$ is irreducible for each $t \geq 0, \bar{X}_{t}$ is irreducible. By Lemma~\ref{lem:4.5}, the chain $\left\{\bar{X}_{t}\right\}$ cannot have length greater than $n$. So Lemma~\ref{lem:9.3} gives the desired result.
\end{proof}

Thus \textit{in the finite-dimensional case a quasi-reachable system is quasi reachable in bounded time}. The analogous statement for reachability is false, as illustrated by the following example. Take $k:=\mathbb{R}, m=p:=1, X:=\mathbb{R}, x^{\#}:=0, P(x, u):=x+u^{2}-2 u$, and $h$ arbitrary. Then $X_{t}=\{x$ in $\mathbb{R} \mid x \geq-t\} \neq X_{R}=X$ for all $t \geq 0$.

\begin{lemma}\label{lem:9.5}
$\Sigma$ has a quasi-reachable closed subsystem $\Sigma_{Q}$.
\end{lemma}

\begin{proof}
Let $X_{Q}:=\bar{X}_{R}$. Since $P$ is continuous, $P\left(X_{Q} \times U\right) \subseteq X_{Q}$. We may therefore define $\Sigma_{Q}:=\left(X_{Q}, P\left|X_{Q} \times U, h\right| X_{Q}, x^{\#}\right)$. The inclusion of $X_{Q} \subseteq X$ exhibits $\Sigma_{Q}$ as a closed subsystem of $\Sigma$.
\end{proof}

\subsection{Algebraic Observability}\label{sec:10}

As discussed in intuitive terms in \citet[Chapter 10]{kalman1968lectures}, observability of $\Sigma$ means the existence of a procedure for determining the state $x$ of $\Sigma$ from data obtained by experiments of the type: "apply an input sequence to $\Sigma$ beginning in state $x$ and observe the corresponding output sequence". In terms of the basic observables $\left\{h_{j}^{w}\right\}$ introduced in Definition~\ref{def:7.3}, this informal description of observability can be made precise by requiring the existence of a set of arbitrary functions of experiments
\[
\eta_{\lambda}\left(h_{j_{1}}^{w_{1}, \lambda}, \ldots, h_{j_{r}(\lambda)}^{w_{r}(\lambda), \lambda}\right),
\]
\label{op:64}%
with $\Lambda=\{\lambda\}$ some arbitrary indexing set, such that each state $x$ is uniquely determined by the data $\left\{\eta_{\lambda}(x)\right\}_{\lambda \in \Lambda}$.

When this procedure is interpreted in the weakest possible, nonconstructive sense, the functions $\eta_{\lambda}$ are completely arbitrary and "observability" reduces to the abstract definition~\ref{def:7.3}. In the case of (finite-dimensional) \textit{linear} systems over a field this abstract definition turns out to be equivalent to the existence of linear combinations $\eta_{\lambda}$ which give every coordinate of the state; see \citet[Chapter 10]{kalman1968lectures}. For linear systems over a \textit{commutative ring}, however, the abstract notion of observability is no longer equivalent to the existence of a linear procedure; some of the resulting problems are studied in \citet{sontag1976linear,sontag1978split}. In general, observability should be formalized with reference to the particular category over which the system in question is defined. Thus, in the context of $k$-systems observability is defined by requiring that each coordinate of the state (i.e., every costate of the system) be a \textit{polynomial} in the basic observables. This is the definition given below, which is a direct generalization of that given in \citet{sontag1975discrete} for polynomial systems. A direct study of bilinear response maps, \citet{kalman1979realization} suggests the same conclusion.

\begin{definition}\label{def:10.1}
\textit{The observation space} $\mathcal{L}(\Sigma)$ [\textit{respectively the observation algebra} $\mathcal{A}(\Sigma)$] \textit{is the linear subspace [respectively the subalgebra] of} $A(X)$ \textit{generated by the basic observables. An observable is a costate in} $\mathcal{A}(\Sigma)$. $\Sigma$ \textit{is algebraically observable iff} $\mathcal{A}(\Sigma)=A(X)$. \textit{When} $X$ \textit{is irreducible, the observation field} $\mathcal{Q}(\Sigma)$ \textit{is the quotient field of} $\mathcal{A}(\Sigma)$.
\end{definition}

Consider the observability map $h^{\Gamma}: X \rightarrow \Gamma=Y^{U^{*}}$ introduced in Definition~\ref{def:7.3}. Since each $h^{w}$ is a polynomial map, $h^{\Gamma}$ is also a polynomial map. By Remark~\ref{rem:1.9}, $A(\Gamma)$ is generated by the algebras $A(Y)$ appearing in the coproduct. So the image $A\left(h^{\Gamma}\right)(A(\Gamma))$ coincides with the algebra generated by the $A\left(h^{W}\right)(A(Y))$, each of which is itself generated by $h_{1}^{W}, \ldots, h_{p}^{W}$. We conclude that $\mathcal{A}(\Sigma)$ is the image of $A(\Gamma)$. The dual of this fact is:

\label{op:65}%

\begin{lemma}\label{lem:10.2}
$\Sigma$ is algebraically observable iff $h^{\Gamma}$ is a closed immersion.
\end{lemma}

\begin{corollary}\label{cor:10.3}
If $\Sigma$ is algebraically observable then $\Sigma$ is (abstractly) observable.

Algebraic observability is a stronger requirement than abstract observability. This is clear from the counterexamples given in Discussion~\ref{disc:3.12}.
\end{corollary}

We remarked in Lemma~\ref{lem:9.5} that every $k$-system has a closed quasi-reachable subsystem. It is less trivial to prove the corresponding statement for algebraic observability:

\begin{proposition}\label{prop:10.4}
$\Sigma$ dominates an algebraically observable system $\Sigma^{\text {obs }}$.
\end{proposition}

\begin{proof}
Let $i: \mathcal{A}(\Sigma) \rightarrow A(X)$ be the inclusion map. Let $X^{\text {obs }}:=X(\mathcal{A}(\Sigma)), x^{\text {\#obs }}:=X(i)\left(x^{\#}\right)$. Since $A(h)(A(Y)) \subseteq A(X)$, we may factor $h: X \rightarrow Y$ as

\[
\begin{tikzpicture}[x=1cm,y=1cm]
\node (X) at (0,0) {$X$};
\node (Y) at (4,0) {$Y$};
\node (Xobs) at (2,-1.6) {$X^{\mathrm{obs}}$,};
\draw[arr] (X) -- node[above,lab] {$h$} (Y);
\draw[arr] (X) -- node[below left,lab] {$X(i)$} (Xobs);
\draw[arr] (Xobs) -- node[right,lab] {$h^{\mathrm{obs}}$} (Y);
\end{tikzpicture}
\]
for some $h^{\text {obs }}$.

Thus the proof will be complete if we can prove that $P$ induces through $X(i)$ a $k$-system morphism $P^{\text {obs }}: X^{\text {obs }} \times U \rightarrow X^{\text {obs }}$; then $X(i)$ is the required dominating $k$-system morphism $\Sigma \rightarrow \Sigma^{\text {obs }}$. Therefore we must show that
\begin{equation}\label{eq:10.5}
A(P)(\mathcal{A}(\Sigma)) \subseteq \mathcal{A}(\Sigma)\left[T_{1}, \ldots, T_{m}\right] .
\end{equation}
i.e., we must prove that when an element $q$ of $A(P)(\mathcal{A}(\Sigma))$ is expressed as a polynomial in the variables $T_{1}, \ldots, T_{m}$, the coefficients of such a polynomial $q$ are again in $\mathcal{A}(\Sigma)$. Since the algebra $\mathcal{A}(\Sigma)$ is generated by the space $\mathcal{L}(\Sigma)$, statement \eqref{eq:10.5} follows from the following

\label{op:66}%
\end{proof}

\begin{lemma}\label{lem:10.6}
Let $\hat{\mathcal{L}}_{t}$ be the subspace of $A(X)$ generated by all those $h_{j}^{w}$ with $w$ in $U^{t}$. Then
\[
A(P)\left(\hat{\mathcal{L}}_{t}\right) \subseteq \hat{\mathcal{L}}_{t+1}\left[T_{1}, \ldots, T_{m}\right] .
\]
\end{lemma}

\begin{proof}
We first observe that when $T_{1}, \ldots, T_{m}$ are specialized at an $u$ in $U=k^{m}$, a polynomial in $A(P) \hat{\mathcal{L}}_{t}$ becomes an element of $\hat{\mathcal{L}}_{t+1}$. Indeed, denote by $e(u): A(X)\left[T_{1}, \ldots, T_{m}\right] \rightarrow k$ the corresponding specialization. Then
\[
e(u) \circ A(P) \circ A\left(h^{w}\right)=A\left(h^{u w}\right) .
\]
Therefore $e(u)\left(A(P) \hat{\mathcal{L}}_{t}\right) \subseteq \hat{\mathcal{L}}_{t+1}$, as wanted. Our claim follows from the following more general result (with $A=A(X), F=k\left[T_{1}, \ldots, T_{m}\right]$ and the $\left\{c_{i}\right\}$ a finite set of monomials in $T_{1}, \ldots, T_{m}$):
\end{proof}

\begin{lemma}[Main Lemma, Part 1]\label{lem:10.7}
Let $A, F$ be vector spaces over $k$, with $F$ a space of functions $Z \rightarrow k$ for some set $Z$. Assume that $c_{1}, \ldots, c_{n}$ are linearly independent elements of $F$, and let $a_{1}, \ldots, a_{n}$ be in $A$. Then the linear subspace of $A$ generated by
\[
\left\{\sum_{i=1}^{n} c_{i}(z) a_{i}, \quad z \text { in } Z\right\},
\]
coincides with the subspace generated by $a_{1}, \ldots, a_{n}$.
\end{lemma}

\begin{proof}
Clearly $\sum c_{i}(z) a_{i}$ is in the space generated by $a_{1}, \ldots, a_{n}$. It is then enough to prove that each $a_{i}$, say $a_{1}$, can be written as
\[
a_{1}=\sum_{j=1}^{n} \lambda_{j}\left(\sum_{i=1}^{n} c_{i}\left(z_{j}\right) a_{i}\right),
\]
for some $z_{1}, \ldots, z_{n} \in Z$. Rewriting this expression as
\[
a_{1}=\sum_{i=1}^{n} a_{i}\left(\sum_{j=1}^{n} \lambda_{j} c_{i}\left(z_{j}\right)\right),
\]
\label{op:67}%
we see that it is enough to prove the existence of a $\lambda$ in $k^{n}$ such that $T \lambda=(1,0,0, \ldots, 0)^{\prime}$, where
\[
T=\left(\begin{array}{cccc}
c_{1}\left(z_{1}\right) & c_{1}\left(z_{2}\right) & \cdots & c_{1}\left(z_{n}\right) \\
\vdots & \vdots & & \vdots \\
c_{n}\left(z_{1}\right) & c_{n}\left(z_{2}\right) & \cdots & c_{n}\left(z_{n}\right)
\end{array}\right) .
\]
It is therefore enough to find $z_{1}, \ldots, z_{n}$ such that $T$ is nonsingular.

Form the $n \times Z$ matrix $\hat{T}$ whose $i$-th row is $c_{i}$ seen as an element of $k^{Z}$. Then existence of $T$ (as a submatrix of $\hat{T}$) follows from the fact that rank $T=n$.
\end{proof}

\subsection{Existence and Uniqueness of Canonical Realizations}\label{sec:11}

\begin{definition}\label{def:11.1}
$\Sigma$ \textit{is canonical iff} $\Sigma$ \textit{is quasi-reachable and algebraically observable}.
\end{definition}

We associate yet another map to $f$. The \textit{extended observability map} $f^{\Omega \Gamma}: \Omega \rightarrow \Gamma$ of $f$ is the observability map of the system $\Sigma_{\text {free }}(f)$ introduced in Theorem~\ref{thm:8.3}. The observability map $f^{\Gamma}: U[z] \rightarrow \Gamma$ introduced in Definition~\ref{def:7.5} is clearly the restriction of $f^{\Omega \Gamma}$ to $U[z]$. Since $U[z]$ is dense in $\Omega$, we may immediately generalize \eqref{eq:7.6}:
\begin{equation}\label{eq:11.2}
\Sigma \text{ realizes } f \text{ iff } f^{\Omega \Gamma}=h^{\Gamma} \circ g^{\Omega} .
\end{equation}

\begin{lemma}\label{lem:11.3}
Let $\Sigma=\left(X, P, h, x^{\#}\right)$ be a quasi-reachable and $\hat{\Sigma}=\left(\hat{X}, \hat{P}, \hat{h}, \hat{x}^{\#}\right)$ an algebraically observable $k$-system which realize $f$. Then there exists a unique $k$-system morphism $T: \Sigma \rightarrow \hat{\Sigma}$.
\end{lemma}

\begin{proof}
Consider the diagram

\begin{equation}
\begin{tikzpicture}[x=1cm,y=1cm]
\node (Om) at (0,0) {$\Omega$};
\node (X) at (3,1.2) {$X$};
\node (Xh) at (3,-1.2) {$\hat X$};
\node (G) at (6,0) {$\Gamma$};
\draw[arr] (Om) -- node[above left,lab] {$g^\Omega$} (X);
\draw[arr] (Om) -- node[below left,lab] {$\hat g^\Omega$} (Xh);
\draw[arr,densely dotted] (X) -- node[right,lab] {$T$} (Xh);
\draw[arr] (X) -- node[above right,lab] {$h^\Gamma$} (G);
\draw[arr] (Xh) -- node[below right,lab] {$\hat h^\Gamma$} (G);
\end{tikzpicture} \label{eq:11.4}
\end{equation}

\label{op:68}%
By hypothesis, $g^{\Omega}$ is dominating, $\hat{h}^{\Gamma}$ is a closed immersion, and $h^{\Gamma} \circ g^{\Omega}=f^{\Omega \Gamma}=\hat{h}^{\Gamma} \circ \hat{g}^{\Omega}$. Thus by the dual of Lemma~\ref{lem:1.5} there exists a polynomial map $T: X \rightarrow \hat{X}$ making \eqref{eq:11.4} commutative. Restricting to $U[z] \subseteq \Omega$, the diagram

\[
\begin{tikzpicture}[x=1cm,y=1cm]
\node (U) at (0,0) {$U[z]$};
\node (XR) at (3,1.2) {$X_R$};
\node (Xh) at (3,-1.2) {$\hat X$};
\node (G) at (7.2,0) {$\Gamma$};
\draw[arr] (U) -- node[above left,lab] {$g$} (XR);
\draw[arr] (U) -- node[below left,lab] {$\hat g$} (Xh);
\draw[arr] (XR) -- node[right,lab] {$T|_{X_R}$} (Xh);
\draw[arr] (XR) -- node[above right,lab] {$(h|X_R)^\Gamma$} (G);
\draw[arr] (Xh) -- node[below right,lab] {$\hat h^\Gamma$} (G);
\end{tikzpicture}
\]
commutes. Thus we may apply Lemma~\ref{lem:7.7} to the abstract systems $\left(X_{R}, P \mid X_{R} \times U\right.$, $\left.h \mid X_{R}, x^{\#}\right)$ and $\hat{\Sigma}$. We conclude that the continuous maps $T \circ P: X \times U \rightarrow \hat{X}$ and $\hat{P} \circ\left(T \times 1_{u}\right)$ coincide in the dense subset $X_{R} \times U$, so they are equal.
\end{proof}

The main result of this section is:

\begin{theorem}\label{thm:11.5}
Let $f$ be a polynomial response map. Then there is a canonical $k$-system $\Sigma_{f}$ realizing $f$. If $\hat{\Sigma}$ is any other canonical $k$-system realizing $f$, there is a unique $k$-system isomorphism $T: \Sigma_{f} \rightarrow \hat{\Sigma}$.
\end{theorem}

\begin{proof}
Uniqueness is clear by \eqref{eq:11.4}. To prove existence, take the system $\Sigma_{\text {free }}(f)$; this is quasi-reachable because $g^{\Omega}$ is the identity. Applying Proposition~\ref{prop:10.4} we obtain the observable system $\Sigma_{f}:=\left(\Sigma_{\text {free }}(f)\right)^{\text {obs }}$. Since $\Sigma_{\text {free }}(f)$ dominates $\Sigma_{f}$, the latter is also quasi-reachable.
\end{proof}

\label{op:69}%
\section{Finiteness Conditions}\label{ch:4}

We have shown in the previous section that any polynomial response map $f$ is realizable by a canonical $k$-system. We now turn to studying what conditions must $f$ satisfy in order that the canonical system $\Sigma_{f}$ has various finiteness properties.

The main tool in this study will be three structures obtained from the basic observables Definition~\ref{def:7.5} of $f$ by different algebraic operations: the observation space $\mathcal{L}_{f}$, algebra $\mathcal{A}_{f}$ and field $\mathcal{Q}_{f}$. We show that the conditions
\[
\begin{aligned}
& \mathcal{Q}_{f}=\text { finitely generated field over } k, \\
& \mathcal{A}_{f}=\text { finitely generated algebra over } k, \\
& \mathcal{L}_{f}=\text { finitely generated vector space over } k,
\end{aligned}
\]
each corresponds to an important characterization of $\Sigma_{f}$.

We then relate each of the above conditions to the existence of certain input/output equations for $f$.

We also show how to check the finiteness condition on $\mathcal{Q}_{f}$ via a Jacobian criterion. As an application we show that
\[
y(t)=u(t-1)+u(t-2)^{2}+\ldots+u(t-\ell)^{\ell}+\ldots,
\]
has no possible finite-dimensional realization.

In the final section we discuss examples and counterexamples associated to the results and constructions of the last two sections.

We continue to denote by $f$ an arbitrary polynomial response map.

\subsection{The Observables of $f$}\label{sec:12}

Since $\Sigma_{f}$ is quasi-reachable and $\Omega$ is irreducible, it follows that $\Sigma_{f}$ is irreducible, so $Q\left(A\left(X_{f}\right)\right)$ is well-defined:

\begin{definition}\label{def:12.1}
\textit{The observation space} $\mathcal{L}_{f}$ [\textit{respectively observation algebra} $\mathcal{A}_{f}$, \textit{respectively observation field} $\mathcal{Q}_{f}$] \textit{of} $f$ \textit{is}
\end{definition}

\label{op:70}%
\noindent $\mathcal{L}\left(\Sigma_{f}\right)$ [\textit{respectively} $\mathcal{A}\left(\Sigma_{f}\right)$, \textit{respectively} $\mathcal{Q}\left(\Sigma_{f}\right)$].

Thus $\mathcal{L}_{f}, A_{f}, Q_{f}$ are the space, algebra, and field generated by the basic observables $\left\{\left(f^{\Omega}\right)_{j}^{W}, w\right.$ in $\left.U^{*}, j=1, \ldots, p\right\}$ of $\Sigma_{\text {free }}(f)$. Each function $\left(f^{\Omega}\right)_{j}^{w}: \Omega \rightarrow k: \omega \mapsto \pi_{j}(f(\omega w))$ is already determined by its restriction $f_{j}^{W}$ to $U[z]$, which is dense in $\Omega$. The restriction $\left(f^{\Omega}\right)_{j}^{w} \mapsto f_{j}^{w}$ serves to establish the following identifications in terms of the basic observables of $f$ introduced in Definition~\ref{def:7.5}:
\begin{equation}\label{eq:12.2}
\mathcal{L}_{f} \text{ is the subspace of } k^{U[z]} \text{ generated by the basic observables } f_{j}^{w} , w \text{ in } U^{*} , j=1, \ldots, p .
\end{equation}
\begin{equation}\label{eq:12.3}
\mathcal{A}_{f} \text{ is the subalgebra of } k^{U[z]} \text{ generated by } \mathcal{L}_{f} .
\end{equation}

The $f_{j}^{w}$ can be also viewed as maps $U^{*} \rightarrow k$, so one can also identify $\mathcal{L}_{f}$ and $\mathcal{A}_{f}$ with the subspace and subalgebra of $k^{U^{*}}$ generated by the $f_{j}^{w}$.

Yet another representation of the observables is obtained via the (vector) Volterra series $\psi_{f}=\left(\psi_{f}^{(1)}, \ldots, \psi_{f}^{(p)}\right)$, of $f^{\Omega}$ (see \eqref{eq:6.15}). By the discussion in \eqref{eq:5.16}-\eqref{eq:5.17}, the Volterra series of $f^{u_{t}} \cdots u_{1}$ is precisely $\psi_{f}\left(u_{1}, \ldots, u_{t}\right)$; coordinatewise:
\begin{equation}\label{eq:12.4}
\text{ The Volterra series of } f_{j}^{u_{t} \cdots u_{1}} \text{ is } \psi_{f}^{(j)}\left(u_{1}, \ldots, u_{t}\right) .
\end{equation}

Thus $L_{f}$, $A_{f}$, and $Q_{f}$ can be interpreted as the space, algebra, and field generated by the series $\left\{\psi_{f}^{(j)}\left(u_{1}, \ldots, u_{t}\right), j=1, \ldots, p\right.$, $u_{t} \ldots u_{1}$ in $\left.U^{*}\right\}$. Using this interpretation, we may define
\begin{equation}\label{eq:12.5}
\operatorname{deg}_{i j} f:=\operatorname{deg}_{i j} \psi_{f}, \quad \operatorname{deg} f:=\operatorname{deg} \psi_{f} .
\end{equation}
By \eqref{eq:12.4} it follows that
\begin{equation}\label{eq:12.6}
\operatorname{deg}_{i j} f^{w} \leq \operatorname{deg}_{i, j+t} f, \quad \operatorname{deg} f^{w} \leq \operatorname{deg} f ,
\end{equation}
for $w$ in $U^{t}$.

\textit{The observables are the main system invariants in our approach}. The study of $\mathcal{L}_{f}, \mathcal{A}_{f}, \mathcal{Q}_{f}$ will be simplified by the consideration of various chains which approximate them:

\label{op:71}%

\begin{definition}\label{def:12.7}
\textit{The observability chains} $\left\{\mathcal{L}^{0, t}(\Sigma), t \geq 0\right\}$ and $\left\{\mathcal{A}^{0, t}(\Sigma), t \geq 0\right\}$ \textit{of} $\Sigma$ \textit{are, respectively, the subspaces and subalgebras of} $\mathcal{A}(\Sigma)$ \textit{generated for each} $t$ \textit{by the elementary observables} $h_{j}^{w}$, $j=1, \ldots, p$, $w$ in $U^{r}, r \leq t$. \textit{The reachability chains} $\left\{\mathcal{L}^{R, t}(\Sigma), t \geq 0\right\}$ and $\left\{\mathcal{A}^{R, t}(\Sigma), t \geq 0\right\}$ \textit{are, respectively, the subspaces and subalgebras generated for each} $t$ \textit{by the restrictions} $h_{j}^{w} \mid X_{t}, j=1, \ldots, p, w$ \textit{in} $U^{*}$. \textit{The diagonal chains} $\left\{\mathcal{L}^{t}(\Sigma), t \geq 0\right\}$ \textit{and} $\left\{\mathcal{A}^{t}(\Sigma), t \geq 0\right\}$ \textit{are, respectively, the subspaces and subalgebras generated for each} $t$ \textit{by the restrictions} $h_{j}^{w} \mid X_{t}, j=1, \ldots, p, w$ \textit{in} $U^{r}, r<t$. \textit{When} $\Sigma$ \textit{is irreducible there are corresponding chains of fields} $\mathcal{Q}^{0, t}(\Sigma)$, \textit{etc. The observability, reachability, and diagonal chains of} $f$ \textit{are} $\mathcal{L}^{0, t}:=\mathcal{L}^{0, t}\left(\Sigma_{f}\right)$, etc.
\end{definition}

We write $\mathcal{L}^{0, t}$, etc., instead of $\mathcal{L}^{0, t}(\Sigma)$, etc., when there is no danger of confusion. In terms of Volterra series, we have:
\begin{equation}\label{eq:12.8}
\mathcal{L}_{f}^{0, t}, \quad \mathcal{A}_{f}^{0, t} \subseteq \Psi \text{ are generated by all } \psi_{f}\left(u_{1}, \ldots, u_{r}\right), \quad r \leq t ;
\end{equation}
\begin{equation}\label{eq:12.9}
\mathcal{L}_{f}^{R, t}, \quad \mathcal{A}_{f}^{R, t} \subseteq k\left[\xi_{1}, \ldots, \xi_{t}\right] \text{ are generated by all } \epsilon_{t} \psi_{f} ;
\end{equation}
\begin{equation}\label{eq:12.10}
\mathcal{L}_{f}^{t}, \quad \mathcal{A}_{f}^{t} \subseteq k\left[\xi_{1}, \ldots, \xi_{t}\right] \text{ are generated by all } \epsilon_{t}\left(\psi_{f}\left(u_{1}, \ldots, u_{r}\right)\right) , r<t .
\end{equation}

Before proving some properties of the chains just introduced, we need the

\begin{lemma}[Main Lemma, Part 2]\label{lem:12.11}
Let $A$ be a $k$-algebra, and take a polynomial $Q$ in $A\left[T_{1}, \ldots, T_{r}\right]$, for some $r \geq 0$. Let $D$ be a dense subset of $k^{r}$. Then the linear subspace of $A$ spanned by $\{Q(u), u$ in $D\}$ is equal to the linear span of $\left\{Q(u), u\right.$ in $\left.k^{r}\right\}$.
\end{lemma}

\begin{proof}
We may assume that $Q \neq 0$. Write
\[
Q\left(T_{1}, \ldots, T_{r}\right)=\sum_{i=1}^{s} c_{i}\left(T_{1}, \ldots, T_{r}\right) a_{i},
\]
with all $c_{i}$ in $k\left[T_{1}, \ldots, T_{r}\right]$ and all $a_{i}$ in $A$, and with $s$ smallest
\label{op:72}possible. Minimality of $s$ implies that both $\left\{a_{1}, \ldots, a_{s}\right\}$ and $\left\{c_{1}, \ldots, c_{s}\right\}$ are linearly independent over $k$.

We claim that the restrictions $c_{i} \mid D$ are linearly independent as functions from $D$ into $k$. Otherwise there would exist $b_{1}, \ldots, b_{s}$ in $k$ such that $\sum b_{i} c_{i}(d)=0$ for all $d$ in $D$. By continuity of the polynomial function $\sum b_{i} c_{i}: k^{r} \rightarrow k$, it follows that $\sum b_{i} c_{i}=0$ on all of $k^{r}$, contradicting linear independence of the $c_{i}$.

Main Lemma~\ref{lem:10.7} can then be applied twice, first with $Z:=D$ and then with $Z:=k^{r}$. Thus $\{Q(u), u$ in $D\}$ and $\left\{Q(u), u\right.$ in $\left.k^{r}\right\}$ both span the same space as $\left\{a_{1}, \ldots, a_{s}\right\}$.
\end{proof}

We collect below some rather technical facts which will be needed in deriving the main results of this section.

\begin{proposition}\label{prop:12.12}
For any $\Sigma$ and $f$,
\begin{enumerate}
\item[(a)] If $\mathcal{L}^{0, t}(\Sigma)=\mathcal{L}^{0, t+1}(\Sigma)$ for some $t>0$, then $\mathcal{L}^{0, t}(\Sigma)=\mathcal{L}(\Sigma)$.
\item[(b)] If $\mathcal{A}^{0, t}(\Sigma)=\mathcal{A}^{0, t+1}(\Sigma)$ for some $t \geq 0$, then $\mathcal{A}^{0, t}(\Sigma)=\mathcal{A}(\Sigma)$.
\item[(c)] If $\mathcal{A}^{0, t+1}(\Sigma)$ is algebraic over $\mathcal{A}^{0, t}(\Sigma)$ for some $t \geq 0$ and if $\{P(x, u), x$ in $X, u$ in $U\}$ is dense in $X$ (e.g., if $\Sigma$ is quasi-reachable), then $\mathcal{A}(\Sigma)$ is algebraic over $\mathcal{A}^{0, t}(\Sigma)$.
\item[(d)] If $\mathcal{A}^{0, t+1}(\Sigma)$ is integral over $\mathcal{A}^{0, t}(\Sigma)$ for some $t \geq 0$, then $\mathcal{A}(\Sigma)$ is integral over $\mathcal{A}^{0, t}(\Sigma)$.
\item[(e)] $\mathcal{L}_{f}^{0, t}$ is finite dimensional and $\mathcal{A}_{f}^{0, t}$ is a finitely generated algebra for all $t \geq 0$.
\item[(f)] If $\operatorname{dim} \Sigma=n$ then $\operatorname{dim} \Sigma_{f_{\Sigma}} \leq n$.
\item[(g)] Let $\operatorname{dim} \Sigma_{f}=n$. Then $\epsilon_{n} \mid \mathcal{L}_{f}: \mathcal{L}_{f} \simeq \mathcal{L}_{f}^{R, n}$, and $\epsilon_{n} \mid \mathcal{A}_{f}: \mathcal{A}_{f} \simeq \mathcal{A}_{f}^{R, n}$.
\item[(h)] $\operatorname{dim} \Sigma_{f}=\sup _{t \geq 0}\left\{\operatorname{trdeg} \mathcal{A}_{f}^{t}\right\}$.
\end{enumerate}
\end{proposition}

\begin{proof}
(a) By Lemma~\ref{lem:10.6},

\label{op:73}%
\begin{align}
A(P)\left(\mathcal{L}^{0, r}\right) & =A(P)\left(\bigcup_{0 \leq s \leq r} \hat{\mathcal{L}}_{s}\right)=\bigcup_{0 \leq s \leq r} A(P)\left(\hat{\mathcal{L}}_{s}\right) \subseteq  \label{eq:12.13}\\
& \subseteq \bigcup_{0 \leq s \leq r} \hat{\mathcal{L}}_{s+1}\left[T_{1}, \ldots, T_{m}\right] \subseteq \notag\\
& \subseteq \mathcal{L}^{0, r+1}\left[T_{1}, \ldots, T_{m}\right], \notag\end{align}
for all $r \geq 0$. By the argument in Lemma~\ref{lem:10.6}, $\hat{\mathcal{L}}_{t+2}$ is spanned by the coefficients of the polynomials in
\[
A(P)\left(\hat{\mathcal{L}}_{t+1}\right) \subseteq A(P)\left(\mathcal{L}^{0, t+1}\right)=A(P)\left(\mathcal{L}^{0, t}\right) \subseteq \mathcal{L}^{0, t+1}\left[T_{1}, \ldots, T_{m}\right],
\]
(using Corollary~\ref{cor:10.3}). Thus $\hat{\mathcal{L}}_{t+2}$ is spanned by elements of $\mathcal{L}^{0, t+1}$, i.e $\hat{\mathcal{L}}_{t+2} \subseteq \mathcal{L}^{0, t+1}$. Then
\[
\mathcal{L}^{0, t+2}=\mathcal{L}^{0, t+1},
\]
and (a) follows by induction.

(b) Analogous to (a).

(c) Take any elementary observable $h_{j}^{w}$ with $w$ in $U^{t+1}$. Since, by hypothesis, $h_{j}^{w}$ is algebraic over $\mathcal{A}^{0, t}$, there exists an $s \geq 0$ and a polynomial $L\left(T_{0}, \ldots, T_{s}\right)$ in $k\left[T_{0}, \ldots, T_{s}\right]$ with the properties:

(i) $L\left(h_{j}^{w}, h_{j_{1}}^{v_{1}}, \ldots, h_{j_{s}}^{v_{s}}\right)=0$

\noindent for some elementary observables $h_{j_{1}}^{v_{1}}, \ldots, h_{j_{s}}^{v_{s}}, v_{i}$ in $U^{r_{i}}, r_{i} \leq t$, and (ii) if $\hat{L}\left(T_{1}, \ldots, T_{s}\right)$ is the leading coefficient of $L$ expressed as a polynomial in $T_{0}$, then
\begin{equation}
\hat{L}\left(h_{j_{1}}^{v_{1}}, \ldots, h_{j_{s}}^{v_{s}}\right) \neq 0 . \label{eq:12.14}
\end{equation}
If
\[
\hat{L}\left(h_{j_{1}}^{u v_{1}}(x), \ldots, h_{j_{s}}^{u v_{s}}(x)\right)=\hat{L}\left(h_{j_{1}}^{v_{1}}(P(x, u)), \ldots, h_{j_{s}}^{v_{s}}(P(x, u))\right)=0
\]
\label{op:74}%
for all $u$ in $U$ and $x$ in $X$, the density hypothesis leads to a contradiction of \eqref{eq:12.14}. So
\[
D:=\left\{u \text { in } U \mid \hat{L}\left(h_{j_{1}}^{u v_{1}}, \ldots, h_{j_{s}}^{u v_{s}}\right) \neq 0\right\} \neq \emptyset,
\]
and by Lemma~\ref{lem:2.4}, $D$ is open, hence dense. Take $u$ in $D$. Then
\[
L\left(h_{j}^{u w}, h_{j_{1}}^{u v_{1}}, \ldots, h_{j_{s}}^{u v_{s}}\right)=L\left(h_{j}^{w}(P(\cdot, u)), \ldots, h_{j_{s}}^{v_{s}}(P(\cdot, u))\right)=0,
\]
and the equation is nontrivial, i.e.
\[
\hat{L}\left(h_{j_{1}}^{u v_{1}}, \ldots, h_{j_{s}}^{u v_{s}}\right) \neq 0,
\]
by definition of $D$. Thus the elementary observables
\[
\left\{h_{j}^{u w}, u \text { in } D, j=1, \ldots, p\right\},
\]
are algebraic over $\mathcal{A}^{0, t+1}$. Let $Q\left(T_{1}, \ldots, T_{m}\right)$ be the polynomial $h_{j}^{w} \circ P$ considered as a polynomial in $T_{1}, \ldots, T_{m}$ (input variables) with coefficients in $A(X)$. We showed above that the elements $Q(u), u$ in $D$, are algebraic over $\mathcal{A}^{0, t+1}$. By Lemma~\ref{lem:12.11}, the span of $\left\{Q(u)\right.$, $u$ in $\left.k^{m}\right\}$ coincides with the span of $\{Q(u), u$ in $D\}$. Thus every generator $Q(u)=h_{j}^{u w}$ of $\mathcal{A}^{0, t+2}$ is algebraic over $\mathcal{A}^{0, t+1}$ for all $u w$ in $U^{t+2}$. The result follows by induction on $t$.

(d) This is analogous to (c).

(e) Since $\mathcal{A}_{f}^{0, t}$ is generated by $\mathcal{L}_{f}^{0, t}=$ span of $\hat{\mathcal{L}}_{f}^{0,0}, \ldots, \hat{\mathcal{L}}_{f}^{0, t}$, it is enough to prove that each $\hat{\mathcal{L}}_{f}^{0, t}$ is finite dimensional. But, by the Main Lemma~\ref{lem:10.7}, $\hat{\mathcal{L}}_{f}^{0, t}$ is generated by the coefficients of $\psi_{f}^{(j)}\left(S_{1}, \ldots, S_{t}\right)$.

(f) By Lemma~\ref{lem:9.5} and Proposition~\ref{prop:10.4}, $\Sigma$ has a subsystem $\Sigma_{Q}$ which dominates $\Sigma_{f}$ (i.e. $\Sigma_{f}$ is a "subquotient" of $\Sigma$). Thus (e) follows from Lemma~\ref{lem:4.2}.

(g) Apply Lemma~\ref{lem:9.3} to $\Sigma_{f}$; observe that $A\left(g_{n}\right)=\epsilon_{n}$.

\label{op:75}%
(h) By (c) above applied to $\Sigma_{f}$,
\[
\operatorname{dim} \Sigma_{f}=\sup _{t}\{\operatorname{trdeg} \mathcal{A}_{f}^{0, t}\} .
\]
The result is then clear by Lemma~\ref{lem:5.22} applied to each $\mathcal{A}_{f}^{0, t}$.
\end{proof}

\subsection{Finite Realizability and Minimality}\label{sec:13}

\begin{definition}\label{def:13.1}
$f$ \textit{is finitely realizable iff} $f$ \textit{has a finite dimensional realization}.
\end{definition}

We collect various characterizations of finite realizability in the following

\begin{theorem}\label{thm:13.2}
The following statements are equivalent:
\begin{enumerate}
\item[(a)] $f$ is finitely realizable.
\item[(b)] $\operatorname{dim} \Sigma_{f}<\infty$.
\item[(c)] $\Sigma_{f}$ is an almost-polynomial system.
\item[(d)] $\mathcal{Q}_{f}$ is a finitely generated field extension of $k$.
\item[(e)] $\epsilon_{t}: \mathcal{A}_{f} \simeq \mathcal{A}_{f}^{R, t}$ for some $t \geq 0$.
\item[(f)] $\mathcal{Q}_{f}=\mathcal{Q}_{f}^{0, t}$ for some $t \geq 0$.
\item[(g)] $\mathcal{A}_{f}$ is algebraic over $\mathcal{A}_{f}^{0, t}$ for some $t \geq 0$.
\end{enumerate}
\end{theorem}

\begin{proof}
We prove $(a) \Rightarrow(b) \Rightarrow(e) \Rightarrow(c) \Rightarrow(d) \Rightarrow(g) \Rightarrow(f) \Rightarrow(b) \Rightarrow$ ⇒ (a).

(a) ⇒ (b) This is Proposition~\ref{prop:12.12}(e).

$(b) \Rightarrow(e)$ This is Proposition~\ref{prop:12.12}(f).

$(e) \Rightarrow(c)$ By definition, $\mathcal{A}^{R, t}$ is a subalgebra of $A\left(U^{t}\right)=k\left[\xi_{1}, \ldots, \xi_{t}\right]$.

(c) ⇒ (d) Let $\mathcal{A}_{f}$ be included in the finitely generated algebra $k\left[b_{1}, \ldots, b_{s}\right]$. Then $\mathcal{Q}_{f}$ is included in the finitely generated field $k\left(b_{1}, \ldots, b_{s}\right)$. By \citet[Chapter X, Exercise 6]{lang1965algebra}, $\mathcal{Q}_{f}$ is itself finitely generated.

\label{op:76}%
(d) $\Rightarrow$ (g) Clear from $\mathcal{Q}_{f}=\bigcup_{t \geq 0} \mathcal{Q}_{f}^{0, t}$.

$(g) \Rightarrow(f)$ If $\mathcal{A}_{f}$ has the same quotient field as $\mathcal{A}_{f}^{0, t}$, then $\mathcal{A}_{f}$ is obviously algebraic over $\mathcal{A}_{f}^{0, t}$.

$(f) \Rightarrow(b)$ Clear from Proposition~\ref{prop:12.12}(e)

(b) ⇒ (a) Trivial.
\end{proof}

We consider briefly the topic of minimal-dimensional realizations:

\begin{definition}\label{def:13.3}
\textit{A system} $\Sigma$ \textit{is minimal iff}
\[
\operatorname{dim} \Sigma \leq \operatorname{dim} \hat{\Sigma}
\]
\textit{for any} $\hat{\Sigma}$ \textit{with} $f_{\hat{\Sigma}}=f_{\Sigma}$.
\end{definition}

By Proposition~\ref{prop:12.12}(f), $\quad \Sigma$ \textit{is minimal iff} $\operatorname{dim} \Sigma=\operatorname{dim} \Sigma_{f_{\Sigma}}$.

For any $x$ in $X$ we define the \textit{observability class} of $x$,
\begin{equation}\label{eq:13.4}
\text{ Obs } (x):=\left\{z\right. \text{ in } \left.X \mid h^{\Gamma}(z)=h^{\Gamma}(x)\right\} .
\end{equation}

We are interested in the case in which Obs ($x$) is generically finite:

\begin{definition}\label{def:13.5}
$\Sigma$ \textit{is weakly observable iff there exists an open dense subset} $U$ \textit{of} $X$ \textit{and an integer} $s \geq 0$ \textit{such that} Obs $(x) \cap U$ \textit{has cardinality} $\leq s$ \textit{for all} $x$ \textit{in} $U$. $\Sigma$ \textit{is weakly canonical iff} $\Sigma$ \textit{is quasi-reachable and weakly observable}.
\end{definition}

Let $\Sigma$ be a quasi-reachable realization of $f$. By Proposition~\ref{prop:10.4}, there is a dominating morphism $T: \Sigma \rightarrow \Sigma_{f}=\left(X_{f}, P_{f}, x_{f}^{\#}, h_{f}\right)$. Since $T$ is a system morphism, $h_{f}^{\Gamma} \circ T=h^{\Gamma}$. Thus
\[
h^{\Gamma}(x)=h^{\Gamma}(z) \text { iff } h_{f}^{\Gamma}(T(x))=h_{f}^{\Gamma}(T(z)),
\]
iff (observability of $\Sigma_{f}$!) $T(x)=T(z)$. So
\[
\text { Obs }(x)=T^{-1}(T(x))=\text { fiber of } T \text { through } x .
\]
\label{op:77}%
If either $k=\mathbb{R}$ or $k$ is algebraically closed, and if $\Sigma$ is a quasi-reachable almost-polynomial system, then Theorem~\ref{thm:4.6} implies that
\begin{equation}\label{eq:13.6}
\operatorname{dim} \operatorname{Obs}(x)=\operatorname{dim} \Sigma-\operatorname{dim} \Sigma_{f}
\end{equation}
for all $x$ in some open dense set $U$. In particular,

\begin{theorem}\label{thm:13.7}
Let either $k=\mathbb{R}$ or let $k$ be algebraically closed. The almost polynomial system $\Sigma$ is weakly canonical if and only if $\Sigma$ is both minimal and irreducible.
\end{theorem}

\begin{proof}
["if"] If $\Sigma$ is not quasi-reachable, then $\Sigma_{Q}$ has strictly lower dimension, contradicting minimality. So $\Sigma$ is quasi-reachable, and there exists $T: \Sigma \rightarrow \Sigma_{f_{\Sigma}}$ dominating. Apply Theorem~\ref{thm:4.6}, noting that $n=m$. Let $U:=j_{X}\left(X_{1}\right)$; then Theorem~\ref{thm:4.6}(c) shows that the sets Obs $(x) \cap U$ have a bounded cardinality when $x$ is in $U$.

["only if"] Irreducibility is a consequence of quasi-reachability. Assume that $\operatorname{dim} \Sigma_{f_{\Sigma}}<\operatorname{dim} \Sigma$. By \eqref{eq:13.6},
\[
\operatorname{dim}(\operatorname{Obs}(x) \cap U)=\operatorname{dim} \overline{(\operatorname{Obs}(x) \cap U)}=\operatorname{dim} \operatorname{Obs}(x)>0
\]
for almost all $x$, so Obs $(x) \cap U$ cannot be generically finite.
\end{proof}

\subsection{Polynomial Canonical Systems}\label{sec:14}

We have proved in the previous section that " $f$ finitely realizable" is equivalent to " $\Sigma_{f}$ is almost polynomial". We now ask under what (stronger) conditions $\Sigma_{f}$ is a \textit{polynomial} system.

The precise "input/output" condition for $\Sigma_{f}=$ polynomial is immediate at this stage. Since $\mathcal{A}\left(\Sigma_{f}\right)=\mathcal{A}_{f}$, we see that \textit{the canonical system} $\Sigma_{f}$ \textit{is polynomial iff the observation algebra} $\mathcal{A}_{f}$ \textit{is a finitely generated $k$-algebra}. Some useful characterizations of this latter condition are given in the next

\begin{theorem}\label{thm:14.1}
The following statements are equivalent:
\begin{enumerate}
\item[(a)] \label{op:78}$\Sigma_{f}$ is a polynomial system.
\item[(b)] $\mathcal{A}_{f}^{0, t}=\mathcal{A}_{f}$ for some $t \geq 0$.
\item[(c)] $\mathcal{A}_{f}$ is integral over $\mathcal{A}_{f}^{0, t}$ for some $t \geq 0$.
\end{enumerate}
\end{theorem}

\begin{proof}
(a) ⇒ (b) Since
\[
\mathcal{A}_{f}=\bigcup_{t \geq 0} \mathcal{A}_{f}^{0, t},
\]
there is a $t \geq 0$ such that $\mathcal{A}_{f}^{0, t}$ contains all the generators of $\mathcal{A}_{f}$.

$(b) \Rightarrow(c)$ Obvious.

(c) ⇒ (a) $\mathcal{A}_{f}$ integral over $\mathcal{A}_{f}^{0, t}$ implies $\mathcal{Q}_{f}=Q\left(\mathcal{A}_{f}\right)$ is algebraic over $Q\left(\mathcal{A}_{f}^{0, t}\right)$, so trdeg $\mathcal{Q}_{f}=\operatorname{trdeg} \mathcal{Q}_{f}^{0, t}=$ finite by Proposition~\ref{prop:12.12}(e). By (b) ⇔ (d) in Theorem~\ref{thm:13.2}, $\mathcal{Q}_{f}$ is finitely generated over $k$. By Lemma~\ref{lem:1.13}(a) we conclude that $\mathcal{A}_{f}$ is finitely generated.
\end{proof}

We shall now describe a procedure which enables one to obtain a polynomial canonical realization when a finite set of generators of $\mathcal{A}_{f}$ is known.

Assume now that $\mathcal{A}_{f}$ is generated by $\psi_{1}, \ldots, \psi_{n}$ as a subalgebra of $\Psi$. In particular trdeg $\mathcal{A}_{f} \leq n$, so $\epsilon_{n} \mid \mathcal{A}_{f}$ is an isomorphism. Thus any relation
\begin{equation}
\psi_{i}(u)=\sum_{\alpha} P_{\alpha, i}\left(\psi_{1}, \ldots, \psi_{n}\right) u^{\alpha}, \label{eq:14.2}
\end{equation}
where $u=\left(u_{1}, \ldots, u_{m}\right)^{\prime}$ is in $U$ and the $P_{\alpha, i}$ are finitely many polynomials, is equivalent to
\[
\epsilon_{n} \psi_{i}(u)=\sum_{\alpha} P_{\alpha, i}\left(\epsilon_{n} \psi_{1}, \ldots, \epsilon_{n} \psi_{n}\right) u^{\alpha} .
\]
The transition map of the canonical realization $\Sigma_{f}$ is given by $X\left(\theta_{1}\right)$, with $\theta_{1}$ defined as in \eqref{eq:6.7}, $\theta_{1}: \mathcal{A}_{f} \rightarrow \mathcal{A}_{f}\left[S_{1}\right]: \psi \mapsto \psi\left(S_{1}\right)$. Thus the canonical realization may be obtained from $\epsilon_{n+1} \psi_{1}, \ldots, \epsilon_{n+1} \psi_{n}$ as follows:

\label{op:79}%
Step 1. Substitute $\xi_{i 1} \mapsto T_{i}$ and $\xi_{i j} \mapsto \xi_{i, j-1}, i=1, \ldots, m$, $j \geq 2$, in each $\epsilon_{n+1} \psi$.

Step 2. Write $\epsilon_{n+1} \psi_{i}$ as a polynomial in the variables $T=T_{1}, \ldots, T_{m}:$
\begin{equation}
\epsilon_{n+1} \psi_{i}=\sum_{\alpha} q_{i \alpha}\left(\xi_{1}, \ldots, \xi_{n}\right) T^{\alpha} . \label{eq:14.3}
\end{equation}
Step 3. Express each $q_{i \alpha}$ as a polynomial combination of $\epsilon_{n} \psi_{1}, \ldots, \epsilon_{n} \psi_{n}$ :
\begin{equation}
q_{i \alpha}=P_{i \alpha}\left(\epsilon_{n} \psi_{1}, \ldots, \epsilon_{n} \psi_{n}\right) ; \label{eq:14.4}
\end{equation}
this is always possible by \eqref{eq:10.5}, since $\epsilon_{n} \psi_{1}, \ldots, \epsilon_{n} \psi_{n}$ generate $\mathcal{A}_{f}^{R, n}$.

Step 4. Since each $\psi_{f}^{(j)}$ is in $\mathcal{A}_{f}^{R, n}$, $j=1, \ldots, p$, there are polynomials $h_{j}$ with
\[
\epsilon_{n} \psi_{f}^{(j)}=h_{j}\left(\epsilon_{n} \psi_{1}, \ldots, \epsilon_{n} \psi_{n}\right), \quad j=1, \ldots, p .
\]
Then $\Sigma_{f}$ is the quasi-reachable subsystem $\Sigma_{Q}$ of $\Sigma=\left(k^{n}\right.$, $\left.\left(P_{i \alpha}\right),\left(h_{1}, \ldots, h_{p}\right), 0\right)$, i.e., $\Sigma_{f}$ \textit{is given by the equations}

\begin{equation}\label{eq:14.5}
\left\{\begin{aligned}
x_{i}(t+1)&=\sum_{\alpha} P_{i \alpha}\left(x_{1}(t), \ldots, x_{n}(t)\right) u(t)^{\alpha}, \quad x_{i}^{\#}:=0, \quad i=1, \ldots, n, \\
y_{j}(t)&=h_{j}\left(x_{1}(t), \ldots, x_{n}(t)\right), \quad j=1, \ldots, p,
\end{aligned}\right.
\end{equation}
\textit{subject to the constraints}
\[
Q\left(x_{1}(t), \ldots, x_{n}(t)\right)=0,
\]
\textit{for each polynomial} $Q$ \textit{for which} $Q\left(\epsilon_{n} \psi_{1}, \ldots, \epsilon_{n} \psi_{n}\right)=0$.

\begin{example}\label{ex:14.6}
A simple illustration of the above procedure is the following. Let $f=f_{\Sigma}$, where $\Sigma$ is given by $m=p=1, X:=k$ and
\[
x(t+1)=x^{2}(t)+u(t), \quad x^{\#}:=0, \quad y(t)=x^{2}(t) .
\]
\label{op:80}%
Thus, $f\left(u_{t}, \ldots, u_{1}\right)=\left(\left(\ldots\left(\left(u_{t}^{2}+u_{t-1}\right)^{2}+u_{t-2}\right)^{2}+\ldots+u_{2}\right)^{2}+u_{1}\right)^{2}$. Note that $f^{u}=(f+u)^{2}$, or equivalently, $\psi_{f}(u)=\left(\psi_{f}+u\right)^{2}$ for all $u$ in $U$. Thus $\psi_{f}$ generates $\mathcal{A}_{f}$, and $n=1$. Since $\epsilon_{2} \psi_{f}=\left(\xi_{2}^{2}+\xi_{1}\right)^{2}$, step 1 gives $\left(\xi_{1}^{2}+T\right)^{2}$. But $\epsilon_{1} \psi_{f}=\xi_{1}^{2}$, so $\left(\xi_{1}^{2}+T\right)^{2}=\left(\epsilon_{1} \psi_{f}+T\right)^{2}$. Thus \eqref{eq:12.5} becomes
\begin{equation}
x(t+1)=(x(t)+u)^{2}, \quad x^{\#}=0, \quad y(t)=x(t), \label{eq:14.7}
\end{equation}
with $X_{f}=k$. The procedure $\Sigma \rightarrow \Sigma_{f}$ has identified the pairs of unobservable states $\{x,-x\}$, giving the algebraically observable system $\Sigma_{f}$. On the other hand, (complete) reachability of $\Sigma$ became just quasi-reachability of $\Sigma_{f}$. A more direct way of obtaining $\Sigma_{f}$ in those cases (as with the present example) in which a quasi-reachable realization $\Sigma$ is already known is through Proposition~\ref{prop:10.4}. By Proposition~\ref{prop:12.12}(b) and Theorem~\ref{thm:14.1}(c) we must generate the algebras $\mathcal{A}^{0, t}(\Sigma)$ until $\mathcal{A}^{0, t}(\Sigma)=\mathcal{A}^{0, t+1}(\Sigma)$ for some $t$. Calculating, $\mathcal{A}^{0,0}(\Sigma)=$ algebra generated by $h=k\left[\eta^{2}\right]$; $\mathcal{A}^{0,1}(\Sigma)=$ algebra generated by $\eta^{2}$ and by all $\left(\eta^{2}+u\right)^{2}$, $u$ in $k=$ (by Lemma~\ref{lem:10.7}) algebra generated by $\eta^{2}$ and by the coefficients in $u$ of $\eta^{4}+2 \eta^{2} u+u^{2}=k\left[\eta^{2}, \eta^{4}\right]=k\left[\eta^{2}\right]=\mathcal{A}^{0,0}(\Sigma)$. So $\mathcal{A}(\Sigma)=\mathcal{A}^{0,0}(\Sigma)$. Restricting $A(P)$ to $\mathcal{A}(\Sigma)$ one obtains again \eqref{eq:14.7}.
\end{example}

We shall see in Section~\ref{sec:16} that $\Sigma_{f}$ is polynomial for a wide class of response maps, which includes all those maps defined by recursive polynomial equations.

\subsection{Bounded Maps}\label{sec:15}

An important rôle in studying the condition $\operatorname{dim} \mathcal{L}_{f}<\infty$ is played by the response maps introduced in the following

\begin{definition}\label{def:15.1}
$f$ \textit{is bounded iff} $\operatorname{deg} f<\infty$.
\end{definition}

(Recall Definitions \ref{def:5.4} and \eqref{eq:12.5} for $\operatorname{deg} f$.)

The system-theoretic meaning of a bounded map is that no input is raised to a power higher than a certain bound. Very simple systems can give rise to \textit{nonbounded} maps. For instance, let
\label{op:81}%
\[
\Sigma: x(t+1)=x^{2}(t)+u(t) ; \quad x^{\#}:=0, \quad y(t)=x(t) .
\]
Computing $f_{\Sigma}$ by iteration shows that it is clearly not bounded. However, the notion of bounded map is general enough to accommodate those classes of response maps for which a realization theory exists in the literature today. These cases correspond to the next three examples:

\begin{example}\label{ex:15.2}
$f$ is \textit{linear} iff $\epsilon_{0} \psi_{f}=0$ and each monomial in $\psi_{f}$ has degree $=1$. In particular, $\operatorname{deg} f=1<\infty$.
\end{example}

\begin{example}[\citealp{kalman1969pattern,kalman1979realization}]\label{ex:15.3}
$f$ is \textit{multilinear} of order $m$ iff all the nonzero monomials appearing in $\psi_{f}$ are of the form
\[
a \xi_{1 j_{1}} \cdots \xi_{m j_{m}}, \quad a \text { in } k,
\]
i.e. $f \mid U^{t}$ is $m$-linear as a map $U^{t}=\left(k^{m}\right)^{t}=\left(k^{t}\right)^{m} \rightarrow Y$, for all $t \geq 0$. Then, $\operatorname{deg} f=1<\infty$. (More generally one may consider s-linear maps $f: U_{1}^{t} \times \ldots \times U_{s}^{t} \rightarrow Y$, where $U_{i}:=k^{r_{i}}, i=1, \ldots, s$ and $r_{1}+\ldots+r_{s}=m$. Here we took all $r_{i}=1$ just for notational simplicity.)
\end{example}

\begin{example}[\citealp{brockett1972algebraic}, \citealp{isidori1973realization}, \citealp{isidori1973direct,isidori1974new}, \citealp{fliess1973sur}, \citealp{dalessandro1974realization}]\label{ex:15.4}
$f$ is \textit{internally-bilinear} iff all monomials appearing in $\psi_{f}$ are of the type
\[
a \xi_{i_{1} 1} \cdots \xi_{i_{t} t}, \quad t \geq 0, \quad a \text { in } k .
\]
Again $\operatorname{deg} f=1<\infty$.
\end{example}

Note the difference between the bilinearity in Example~\ref{ex:15.4} with that in Example~\ref{ex:15.3} (with $m=2$). Neither of these implies the other.

\begin{theorem}\label{thm:15.5}
Let $f$ be a response map. The following statements are equivalent:
\begin{enumerate}
\item[(a)] \label{op:82}$f$ is bounded and finitely realizable.
\item[(b)] $f$ is bounded and realizable by a polynomial system.
\item[(c)] $f$ is bounded and $\Sigma_{f}$ is a polynomial system (i.e., $\mathcal{A}_{f}$ is finitely generated).
\item[(d)] $f$ is bounded and $\mathcal{L}_{f}$ is finite dimensional.
\item[(e)] $\operatorname{dim} \mathcal{L}_{f}<\infty$.
\item[(f)] $f$ admits a realization with $X=k^{n}$ and transitions of the form
\[
x(t+1)=F(u(t)) x(t)+G(u(t)),
\]
where $F(\cdot)$ is a matrix, and $G(\cdot)$ a vector, of polynomial functions.
\item[(g)] $f$ admits a realization as above with, further, $h$ linear.
\item[(h)] $f$ admits an observable realization as in (g).
\item[(i)] $f$ admits a realization as in (h), with $G=0$.
\end{enumerate}
\end{theorem}

\begin{proof}
The only nontrivial implications in the diagram

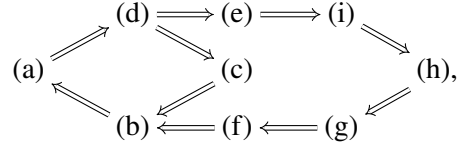
\[
\begin{tikzpicture}[x=1.15cm,y=0.75cm,baseline=(c.base),imp/.style={-{Implies},double,double distance=1.3pt,line width=0.4pt,shorten >=1pt,shorten <=1pt}]
\node (a) at (0,0) {(a)};
\node (d) at (1.2,1) {(d)};
\node (e) at (2.4,1) {(e)};
\node (i) at (3.6,1) {(i)};
\node (c) at (2.4,0) {(c)};
\node (h) at (4.7,0) {(h),};
\node (b) at (1.2,-1) {(b)};
\node (f) at (2.4,-1) {(f)};
\node (g) at (3.6,-1) {(g)};
\draw[imp] (a) -- (d); \draw[imp] (d) -- (e); \draw[imp] (e) -- (i);
\draw[imp] (i) -- (h.north west); \draw[imp] (h.south west) -- (g);
\draw[imp] (g) -- (f); \draw[imp] (f) -- (b); \draw[imp] (b) -- (a);
\draw[imp] (d) -- (c); \draw[imp] (c) -- (b);
\end{tikzpicture}
\]
are (a) $\Rightarrow$ (d), (e) $\Rightarrow$ (i), and $(f) \Rightarrow(a)$, which we now prove.

(a) ⇒ (d) Let $\operatorname{dim} \Sigma_{f}=n$ and $\operatorname{deg} f=d$. By Proposition~\ref{prop:12.12}(g), $\epsilon_{t}: \mathcal{L}_{f} \simeq \mathcal{L}_{f}^{R, n}$. By \eqref{eq:12.6}, $\operatorname{deg} f^{w} \leq d$ for all $w$ in $U^{*}$. So $\operatorname{deg} \epsilon_{t} \psi_{f, j}(w) \leq d$ for all $w$ in $U^{*}$ and $j=1, \ldots, p$. Thus $\mathcal{L}_{f}^{R, n}$ is a linear space generated by a set of polynomials in $n m$ variables of joint degree $\leq n m d$. Thus $\operatorname{dim} \mathcal{L}_{f}=\operatorname{dim} \mathcal{L}_{f}^{R, n} \leq(n m)^{n m d}<\infty$.

(e) ⇒ (i) Let $\operatorname{dim} \mathcal{L}_{f}=n$, and let $\left\{\psi_{1}, \ldots, \psi_{n}\right\}$ be a basis of $\mathcal{L}_{f}$ as a subspace of $\Psi$. By Theorem~\ref{thm:11.5} the polynomials $q_{i \alpha}$ introduced in \eqref{eq:14.3} are in $\mathcal{L}_{f}^{0, R}$. Thus each $P_{i \alpha}$ in \eqref{eq:14.4} can be chosen \textit{linear}, and (i) follows from the formulas \eqref{eq:14.5}.

\label{op:83}%
$(f) \Rightarrow(a)$ Let $F(u)=\left(F_{i j}(u)\right)$, a matrix of polynomials, and let $h(u)=\left(h_{1}, \ldots, h_{p}\right)$, a vector of polynomials. Denote $\operatorname{deg} F:=\sup _{i, j}\left\{\operatorname{deg} F_{i j}\right\}$ and $\operatorname{deg} h:=\sup _{j}\left\{\operatorname{deg} h_{j}\right\}$. Then
\[
\operatorname{deg} f \leq \operatorname{deg} F \cdot \operatorname{deg} h<\infty . \qedhere
\]
\end{proof}

We can see from either $(f)$ or $(g)$ above that an essential characteristic of bounded maps is the absence of any intrinsic nonlinear feedback.

The above theorem shows that the problem of deciding whether a bounded map is realizable is equivalent to checking whether a realization as in (g), (h), or (i) exists. We shall study in Section~\ref{ch:5} the connection between such special realizations and some questions of automata theory.

These special configurations are highly appealing because they can be studied by linear-algebraic methods. It should be noted carefully, however, that, for certain questions like synthesis, alternative (lower dimensional) representations may be more useful. This is illustrated by the following example. Consider the one-dimensional system $\Sigma$ :
\begin{equation}
x(t+1)=x(t)+u(t), \quad x^{\#}:=0, \quad y(t)=x^{s}(t), \label{eq:15.6}
\end{equation}
where $s \geq 1$ is an arbitrary integer. Then $\operatorname{dim} \mathcal{L}_{f_{\Sigma}}=s$ and therefore (see Section~\ref{ch:5}) representations as in (g), (h), or (i) have dimension at least $s$. The system $\Sigma$ may be, however, obtained back from any observable realization $\hat{\Sigma}$ as in (h) in the following way. Since $\Sigma$ is canonical (easy verification) by Lemma~\ref{lem:11.3} $\Sigma$ is a closed subsystem of $\hat{\Sigma}$. Indeed, the reachable set of $\hat{\Sigma}$ will be one-dimensional, and the restrictions of $\hat{P}$ and $\hat{h}$ to this reachable set will define a subsystem isomorphic to $\Sigma$.

\subsection{Input/Output Equations}\label{sec:16}

We shall now study the existence of input/output equations and relate them to various finiteness conditions. We treat only the case $p=1$, since when there is more than one output channel one may
\label{op:84}separately study each $\pi_{j} \circ f, j=1, \ldots, p$. We let $S_{i j}$ be indeterminates, $i=1, \ldots, m, j \geq 1$, and denote $S_{j}:=S_{1 j}, \ldots, S_{m j}$.

\begin{definition}\label{def:16.1}
\textit{Let} $r \geq 0$ \textit{be an integer, and let} $E$ \textit{be a polynomial in} $k\left[L_{0}, L_{1}, \ldots, L_{r}, S_{r}, \ldots, S_{1}\right]$, \textit{nontrivial in} $L_{0}$. \textit{A pair} $(y, u), y$ \textit{in} $\mathbb{Y}, u$ \textit{in} $\mathbb{U}$ \textit{satisfies} $E$ \textit{iff}
\[
E(y(t), y(t-1), \ldots, y(t-r), u(t-1), \ldots, u(t-r))=0
\]
\textit{for all} $t$ \textit{in} $\mathbb{Z}$. \textit{The response map} $f$ \textit{satisfies the algebraic difference equation} $E$ \textit{iff every input/output pair of} $f$ \textit{satisfies} $E$. \textit{The order of the equation} $E$ \textit{is} $r$. \textit{The equation} $E$ \textit{is}
\begin{enumerate}
\item[(a)] \textit{rational when} $\operatorname{deg}_{L_{0}} E=1$, \textit{i.e.}
\[
E=E_{1}\left(L_{1}, \ldots, L_{r}, S_{r}, \ldots, S_{1}\right) L_{0}+E_{2}\left(L_{1}, \ldots, L_{r}, S_{r}, \ldots, S_{1}\right) ;
\]
\item[(b)] \textit{integral when}
\[
E=E_{1}\left(S_{r}, \ldots, S_{1}\right) L_{0}^{s}+E_{2}\left(L_{0}, L_{1}, \ldots, L_{r}, S_{r}, \ldots, S_{1}\right),
\]
\textit{with} $\operatorname{deg}_{L_{0}} E_{2}<s$;
\item[(c)] \textit{recursive when rational and integral}, i.e.
\[
E=E_{1}\left(S_{r}, \ldots, S_{1}\right) L_{0}+E_{2}\left(L_{1}, \ldots, L_{r}, S_{r}, \ldots, S_{1}\right) ;
\]
\item[(d)] \textit{affine when}
\[
E=\sum_{\ell=0}^{r} E_{\ell}\left(S_{r}, \ldots, S_{1}\right) L_{\ell}+\hat{E}\left(S_{r}, \ldots, S_{1}\right) .
\]
\end{enumerate}
\end{definition}

The main result of this section is the following.

\begin{theorem}\label{thm:16.2}
Let $f$ be a polynomial response map.
\begin{enumerate}
\item[(a)] The following statements are equivalent:
\begin{enumerate}
\item[(i)] $f$ is finitely realizable;
\item[(ii)] $f$ satisfies an algebraic difference equation;\label{op:85}
\item[(iii)] $f$ satisfies a rational difference equation.
\end{enumerate}
\item[(b)] The following statements are equivalent:
\begin{enumerate}
\item[(i)] $f$ is bounded and finitely realizable;
\item[(ii)] $f$ satisfies an affine difference equation.
\end{enumerate}
\item[(c)] If $f$ satisfies an integral difference equation then $\Sigma_{f}$ is a polynomial system and $f$ satisfies a recursive equation.
\end{enumerate}
\end{theorem}

The proof of this theorem will be postponed until we establish some technical facts about algebraic difference equations.

The study of such equations will be algebraized through the introduction of the field
\[
K=k\left(\left\{S_{i j}, i=1, \ldots, m, j \geq 1\right\}\right)
\]
obtained by adjoining the indeterminates $S_{i j}$ to $k$. Let $\mathcal{L}_{f}^{K}$, $\mathcal{A}_{f}^{K}$ be the $K$-subspace and $K$-subalgebra of $\Psi_{K}$ generated by $\psi_{f}$ and (recall \eqref{eq:5.17}) all the $\psi_{f}\left(S_{t}, \ldots, S_{1}\right)$ (note order of arguments!), for all $t \geq 0$. Let $\mathcal{Q}_{f}^{K}:=Q\left(\mathcal{A}_{f}^{K}\right)$.

\begin{lemma}\label{lem:16.3}
Let $f$ be a polynomial response map. Then
\begin{enumerate}
\item[(a)] $f$ satisfies an algebraic difference equation if and only if $\operatorname{trdeg}_{K} \mathcal{Q}_{f}^{K}<\infty$.
\item[(b)] $f$ satisfies a rational difference equation if and only if $\mathcal{Q}_{f}^{K}$ is a finitely generated field extension of $K$.
\item[(c)] $f$ satisfies a recursive difference equation if and only if $\mathcal{A}_{f}^{K}$ is a finitely generated $K$-algebra.
\item[(d)] If $f$ satisfies an integral equation and a (possibly different) rational equation then $f$ satisfies a recursive equation.
\item[(e)] $f$ satisfies an affine difference equation if and only if $\operatorname{dim}_{K} \mathcal{L}_{f}^{K}<\infty$.
\end{enumerate}
\end{lemma}

\begin{proof}
We begin with some general remarks. Observe that an equation $E$ of order $r$ is equivalent to
\label{op:86}%
\begin{equation}\label{eq:16.4}
E\left(f^{u_{r} \cdots u_{1}}, f^{u_{r} \cdots u_{2}}, \ldots, f^{u_{r}}, f, u_{1}, \ldots, u_{r}\right)=0 ;
\end{equation}
for all ($u_{r}, \ldots, u_{1}$) in $U^{r}$. In terms of Volterra series, \eqref{eq:16.4} is equivalent to

\begin{equation}
\begin{array}{r}
E\left(\psi_{f}\left(S_{r}, \ldots, S_{1}\right), \psi_{f}\left(S_{r-1}, \ldots, S_{1}\right), \ldots, \psi_{f}\left(S_{1}\right), \psi_{f},\right.  \label{eq:16.5}\\
\left.S_{r}, \ldots, S_{1}\right)=0 .
\end{array}
\end{equation}
A change of variables $S_{\ell} \mapsto S_{\ell+1}, \ell=1, \ldots, r, \xi_{1} \mapsto S_{1}$, and $\xi_{j} \mapsto \xi_{j-1}$ if $j>1$ transforms \eqref{eq:16.5} into

\begin{equation}
\begin{array}{r}
E\left(\psi_{f}\left(S_{r+1}, \ldots, S_{1}\right), \psi_{f}\left(S_{r}, \ldots, S_{1}\right), \ldots, \psi_{f}\left(S_{2}, S_{1}\right), \psi_{f}\left(S_{1}\right),\right.  \label{eq:16.6}\\
\left.S_{r+1}, \ldots, S_{2}\right)=0 .
\end{array}
\end{equation}
Thus an equation of order $r$ gives rise to an equation of order $r+1$.

We now prove (a). Let $f$ satisfy $E$, with $E$ of order $r$ smallest possible. Let
\[
c\left(L_{1}, \ldots, L_{r}, S_{r}, \ldots, S_{1}\right) \neq 0
\]
be the leading coefficient of $E$ as a polynomial in $L_{0}$. Suppose that
\begin{equation*}
c\left(\psi_{f}\left(S_{r-1}, \ldots, S_{1}\right), \ldots, \psi_{f}, S_{r}, \ldots, S_{1}\right)=0 . \tag{*}
\end{equation*}
Dividing by $S_{r}$ if necessary, we may assume that $c\left(L_{1}, \ldots, L_{r}, 0, S_{r-1}, \ldots, S_{1}\right) \neq 0$. Thus
\[
c\left(\psi_{f}\left(S_{r-1}, \ldots, S_{1}\right), \ldots, \psi_{f}, 0, S_{r-1}, \ldots, S_{1}\right)=0 ;
\]
is an equation for $f$ of order $\leq r-1$, contradicting minimality of $r$. Therefore (*) is false. So \eqref{eq:16.4} is a nontrivial equation for $\psi_{f}\left(S_{r}, \ldots, S_{1}\right)$ with coefficients in
\[
\hat{K}_{r}:=K\left(\psi_{f}, \ldots, \psi_{f}\left(S_{r-1}, \ldots, S_{1}\right)\right) .
\]
Since the derivation of \eqref{eq:16.5} from \eqref{eq:16.4} is just a relabeling of variables, the same argument gives a dependency for $\psi_{f}\left(S_{r+1}, \ldots, S_{1}\right)$ over $K\left(\psi_{f}, \ldots, \psi_{f}\left(S_{r}, \ldots, S_{1}\right)\right)$. Since algebraic extensions of algebraic extensions are again algebraic (formula \eqref{eq:4.1}) we conclude that
\label{op:87}$\psi_{f}\left(S_{r+1}, \ldots, S_{1}\right)$ is also algebraic over $\hat{K}_{r}$, and by induction also $\mathcal{Q}_{f}^{K}$ is algebraic over $\hat{K}_{r}$. Thus
\[
\operatorname{trdeg}_{K} \mathcal{Q}_{f}^{K}=\operatorname{trdeg}_{K} \hat{K}_{r} \leq r<\infty .
\]
Conversely, assume that trdeg $_{K} \mathcal{Q}_{f}^{K}<\infty$. Since the $\psi_{f}\left(S_{t}, \ldots, S_{1}\right)$ generate $\mathcal{Q}_{f}^{K}$, there exists an $r \geq 1$ such that $\mathcal{Q}_{f}^{K}$ is algebraic over $\hat{K}_{r}$. In particular, $\psi_{f}\left(S_{r}, \ldots, S_{1}\right)$ is algebraic over $\hat{K}_{r}$, so there is an equation of algebraic dependence $Q\left(\psi_{f}\left(S_{r}, \ldots, S_{1}\right)\right)=0$ with coefficients in $\hat{K}_{r}$. Multiplying by a common denominator if necessary, we may assume that all coefficients of $Q$ are in $k\left[\left\{S_{i j}\right\} \cup\left\{\psi_{f}, \ldots, \psi_{f}\left(S_{r-1}, \ldots, S_{1}\right)\right\}\right]$. Dividing if necessary, we may further assume that no $S_{i j}$ divides $Q$. Therefore the evaluation $S_{i j} \mapsto 0$ for all $j>r$ gives an algebraic equation of order $r$.

(b) A rational equation E is equivalent to the statement " $\psi_{f}\left(S_{r}, \ldots, S_{1}\right)$ belongs to $\hat{K}_{r}$ ". Since $\mathcal{Q}_{f}^{K}$ is the union of the $\hat{K}_{r}$, (b) is clear.

(c) Existence of a recursive equation of order $r$ is equivalent to $\psi_{f}\left(S_{r}, \ldots, S_{1}\right)$ belonging to the algebra
\[
A_{r}=K\left[\psi_{f}, \ldots, \psi_{f}\left(S_{r-1}, \ldots, S_{1}\right)\right],
\]
which is in turn equivalent to $\mathcal{A}_{f}^{K}$ being finitely generated.

(d) An integral equation corresponds to $\psi_{f}\left(S_{r}, \ldots, S_{1}\right)$ being integral over $A_{r}$, so by induction as in passing from \eqref{eq:16.5} to \eqref{eq:16.6}, we conclude that $\mathcal{A}_{f}^{K}$ is integral over $A_{r}$. Since, by (b), $\mathcal{Q}_{f}^{K}=Q\left(\mathcal{A}_{f}^{K}\right)$ is finitely generated over $K$, we conclude from Lemma~\ref{lem:1.13}(a) that $\mathcal{A}_{f}^{K}$ is finitely generated.

(e) Similar to (c); just observe that an affine equation of order $r$ is equivalent to $\psi_{f}\left(S_{r}, \ldots, S_{1}\right)$ belonging to the span of $\psi_{f}, \ldots, \psi_{f}\left(S_{r-1}, \ldots, S_{1}\right)$ as a space over $K$.
\end{proof}

\begin{lemma}\label{lem:16.7}
Let $f$ be a polynomial response map. Then
\begin{enumerate}
\item[(a)] \label{op:88}$\operatorname{trdeg}_{K} \mathcal{Q}_{f}^{K}<\infty$ if and only if $\operatorname{trdeg} \mathcal{Q}_{f}<\infty$.
\item[(b)] $\operatorname{dim}_{K} \mathcal{L}_{f}^{K}<\infty$ if and only if $\operatorname{dim} \mathcal{L}_{f}<\infty$.
\item[(c)] If $\mathcal{A}_{f}^{K}$ is a finitely generated $K$-algebra, then $\mathcal{A}_{f}$ is finitely generated (over $k$).
\end{enumerate}
\end{lemma}

\begin{proof}
We first prove that the finiteness statements on $\mathcal{L}_{f}^{K}$, $\mathcal{A}_{f}^{K}$, and $\mathcal{Q}_{f}^{K}$ imply the corresponding conditions on $\mathcal{L}_{f}, \mathcal{A}_{f}$, and $\mathcal{Q}_{f}$. By Lemma~\ref{lem:16.3} these statements are equivalent to the existence, respectively, of an affine, recursive or algebraic equation $E$ for $f$. Let $c\left(L_{1}, \ldots, L_{r}, S_{r}, \ldots, S_{1}\right)$ be the leading coefficient of $E$ as a polynomial in $L_{0}$. Let
\[
\begin{array}{r}
D:=\left\{\left(u_{r}, \ldots, u_{1}\right) \text { in } U^{r} \mid c\left(\psi_{f}\left(u_{r-1}, \ldots, u_{1}\right), \ldots, \psi_{f}\left(u_{1}\right),\right.\right. \\
\left.\left.\ldots, \psi_{f}, u_{r}, \ldots, u_{1}\right)=0\right\} .
\end{array}
\]
By \eqref{eq:16.4}, $\psi_{f}\left(u_{r}, \ldots, u_{1}\right)$ is in the space $\mathcal{L}_{f}$ [respectively in the algebra $\mathcal{A}_{f}$, respectively algebraic over $\left.\mathcal{Q}_{f}\right]$ for all ($u_{r}, \ldots, u_{1}$) not in $D$. The conclusion is then clear from the Main Lemma, Part II \ref{lem:12.11} and Proposition~\ref{prop:12.12}(e).

Now we prove the "if" parts in (a) and (b). By the Main Lemma~\ref{lem:10.7}, the coefficients of $\psi_{f}\left(S_{r}, \ldots, S_{1}\right)$, considered as a polynomial in $S_{1}, \ldots, S_{r}$, are in $\mathcal{L}_{f}$. Thus $\mathcal{L}_{f}^{K}$ is in the linear space over $K$ generated by $\mathcal{L}_{f}$, which has a finite basis by hypothesis. Since $\mathcal{Q}_{f}^{K}$ is generated as a field by the elements of $\mathcal{L}_{f}^{K}$, the conclusion about $\mathcal{Q}_{f}^{K}$ is also immediate.
\end{proof}

We may now complete the

\begin{proof}[Proof of Theorem~\ref{thm:16.2}]
(a) Follows from Lemma~\ref{lem:16.3}(a,b), Lemma~\ref{lem:16.7}(a) and the equivalence of (a) and (d) in Theorem~\ref{thm:13.2}.

(b) Clear from Lemma~\ref{lem:16.3}(e) and Lemma~\ref{lem:16.7}(b).

(c) An integral equation is in particular an algebraic equation, so by (a) $f$ satisfies a rational equation. The conclusion is then clear from $Lemma~\ref{lem:16.3}(d),Lemma~\ref{lem:16.3}(c)$, and $Lemma~\ref{lem:16.7}(c)$.
\end{proof}

\label{op:89}%

\begin{remarks}[A]\label{rem:16.8}
For a finitely realizable $f$ there is up to scalar multiplication a \textit{unique} algebraic difference equation $E$ with the two properties: (i) $E$ is of minimal order $r$, and (ii) $E$ is irreducible as a polynomial in $k\left[L_{0}, L_{1}, \ldots, L_{r}, S_{r}, \ldots, S_{1}\right]$. Uniqueness follows easily from the discussion in \citet[page 110]{hodge1968methods}.
\begin{enumerate}
\item[(b)] It is easily verified that two polynomial input/output maps satisfying the same rational equation necessarily coincide. Moreover, by simple division of polynomials, $f$ (more precisely, $\psi_{f}$) can be reconstructed from the rational equation $E$.
\item[(c)] It can be seen by a counterexample that $\Sigma_{f}=$ polynomial system does not imply that $f$ satisfies an integral difference equation.
\end{enumerate}
\end{remarks}

\subsection{Jacobian Condition}\label{sec:17}

We now show how to check the condition " $f$ is finitely realizable" by examining increasing truncations of the Volterra series. We take $k$ to be a field of zero characteristic (when char $k=p \neq 0$ one may generalize the criterion via spaces of differentials).

As in the previous section, we assume $p=1$ without loss of generality.

Denote by $D_{r s}, 1 \leq r \leq m, s \geq 1$, the operator which takes partial derivatives of polynomials with respect to the indeterminate $\xi_{r s}$. Let $D_{s} P, s \geq 1$, be the row vector $\left(D_{1 s} P, \ldots, D_{m s} P\right)$.

\begin{definition}\label{def:17.1}
\textit{The} $n$-\textit{th Jacobian matrix} $J_{n}(f)$ \textit{of} $f$ \textit{is}
\[
J_{n}(f):=\left(\begin{array}{cccc}
D_{1} \epsilon_{n} \psi_{f} & D_{2} \epsilon_{n} \psi_{f} & \ldots & D_{n} \epsilon_{n} \psi_{f} \\
D_{1} \epsilon_{n} \psi_{f}\left(S_{1}\right) & . & \ldots & D_{n} \epsilon_{n} \psi_{f}\left(S_{1}\right) \\
\vdots & \vdots & \ldots & \vdots \\
D_{1} \epsilon_{n} \psi_{f}\left(S_{n-1}, \ldots, S_{1}\right) & . & \ldots & D_{n} \epsilon_{n} \psi_{f}\left(S_{n-1}, \ldots, S_{1}\right)
\end{array}\right) .
\]
\end{definition}

\label{op:90}%

\begin{example}\label{ex:17.2}
Let $f$ be a linear response map with "impulse response" $A_{1}, \ldots, A_{n}, \ldots$, i.e. $\psi_{f}=\sum a_{i j} \xi_{i j}$ and $A_{i}=\left(a_{1 j}, \ldots, a_{m j}\right), j \geq 1$. Then
\[
\epsilon_{n} \psi_{f}\left(S_{r}, \ldots, S_{1}\right)=\sum_{j=1}^{r} \sum_{i=1}^{m} a_{i j} S_{i, r-j}+\sum_{j=1}^{n} \sum_{i=1}^{m} a_{i, j+r} \xi_{i j},
\]
so
\[
D_{s} \epsilon_{n} \psi_{f}\left(S_{r}, \ldots, S_{1}\right)=A_{s+r} .
\]
Thus $J_{n}(f)$ is the $n$-th principal minor of the block Hankel matrix of $f$ (\citealp[Chapter 10]{kalman1969topics}), and it is well-known that realizability of $f$ is equivalent to the existence of an integer $s$ such that rank $J_{n}(f)<s$ for all $n$. The Jacobian of $f$ gives a new way of interpreting the classical Hankel matrix of $f$.
\end{example}

\begin{theorem}\label{thm:17.3}
$f$ is finitely realizable if and only if there exists an $s \geq 0$ such that
\[
\operatorname{rank} J_{t}(f) \leq s \text { for all } t \geq 0 \text {, }
\]
i.e., if and only if every $(s+1)$-minor of $J_{t}(f)$ is zero for all $t \geq 0$.
\end{theorem}

\begin{proof}
By \citet[III.7, Theorem III]{hodge1968methods},
\[
\operatorname{rank} J_{n}(f)=\operatorname{trdeg}_{K} K\left[\epsilon_{n} \psi_{f}, \ldots, \epsilon_{n} \psi_{f}\left(S_{n-1}, \ldots, S_{1}\right)\right] .
\]
($K$ is $k\left(\left\{S_{i j}\right\}\right)$ as in the previous section.)

["only if"] By Proposition~\ref{prop:12.12}(h), there is an $s$ such that $\operatorname{trdeg}_{k} \mathcal{A}_{f}^{n} \leq s$ for all $n \geq 0$. By the Main Lemma~\ref{lem:10.7} each $\epsilon_{n} \psi_{f}\left(S_{n-1}, \ldots, S_{1}\right), j=0, \ldots, n-1$ is a $K$-linear combination of elements of $\mathcal{A}_{f}^{n}$. So rank $J_{n}(f) \leq s$.

["if"] If rank $J_{t}(f) \leq s$ for all $t \geq 0$, in particular
\label{op:91}%
\begin{equation}\label{eq:17.4}
\operatorname{trdeg}_{K} K\left[\epsilon_{t} \psi_{f}, \ldots, \epsilon_{t} \psi_{f}\left(S_{s}, \ldots, S_{1}\right)\right] \leq s ,
\end{equation}
for all $t \geq s$. Let $B$ be the subalgebra of $\Psi_{K}$ generated by $\psi_{f}, \ldots, \psi_{f}\left(S_{s}, \ldots, S_{1}\right)$. By \eqref{eq:17.4}, $\operatorname{trdeg}_{K} \epsilon_{t}(B) \leq s$ for all $t \geq s$. It follows from Lemma~\ref{lem:5.22} that $\operatorname{trdeg}_{K} B \leq s$. So there is an $r \leq s$ such that $\psi_{f}\left(S_{r}, \ldots, S_{1}\right)$ is algebraically dependent over $\psi_{f}, \ldots, \psi_{f}\left(S_{r-1}, \ldots, S_{1}\right)$, and the conclusion follows from Lemma~\ref{lem:16.7}(a).
\end{proof}

\begin{example}\label{ex:17.5}
As an application of Theorem~\ref{thm:17.3}, we prove that the polynomial response map $f$, with $m=p=1$ and
\[
\psi_{f}:=\xi_{11}+\xi_{12}^{2}+\xi_{13}^{3}+\xi_{14}^{4}+\xi_{15}^{5}+\ldots,
\]
is not finitely realizable. Since
\[
\epsilon_{n} \psi_{f}\left(S_{r}, \ldots, S_{1}\right)=S_{r}+S_{r-1}^{2}+\ldots+S_{1}^{r}+\xi_{11}^{r+1}+\ldots+\xi_{1 n}^{r+n},
\]
for any $s \leq n, D_{s} \epsilon_{n} \psi_{f}\left(S_{r}, \ldots, S_{1}\right)=(r+s) \xi_{1 s}^{r+s-1}$; this latter expression is also $D_{s} \epsilon_{n-1} \psi_{f}\left(S_{r}, \ldots, S_{1}\right)$ if $s<n$. Therefore
\[
J_{n}(f)=\left(\begin{array}{c:c}
J_{n-1}(f) & * \\
\hdashline * & (2 n-1) \xi_{1 n}^{2 n-2}
\end{array}\right),
\]
and so
\[
\operatorname{det} J_{n}(f)=(2 n-1) \operatorname{det} J_{n-1}(f) \xi_{1 n}^{2 n-2}+t\left(\xi_{11}, \ldots, \xi_{1 n}\right),
\]
where $\operatorname{deg}_{1 n} t<2 n-2$. Since $\xi_{1 n}$ is algebraically independent over $k\left(\xi_{11}, \ldots, \xi_{1, n-1}\right)$, an induction on $n$ shows that $\operatorname{det} J_{n}(f) \neq 0$ for all $n$.
\end{example}

\subsection{Some Examples and Counterexamples}\label{sec:18}

We discuss now several examples related to canonical realizations and input/output equations.

\label{op:92}%

\begin{example}\label{ex:18.1}
We wish to illustrate the calculation of a nonpolynomial canonical system, via Proposition~\ref{prop:10.4}. Consider the system $\Sigma_{0}:=\left(k^{2}, P, h, 0\right)$, where $m=p=1$ and the equations of $\Sigma_{o}$ are
\[
\begin{aligned}
x_{1}(t+1) & =x_{1}(t)+u(t), \\
x_{2}(t+1) & =x_{1}(t) x_{2}(t)+x_{1}(t)+x_{2}(t), \\
y(t) & =x_{2}(t) .
\end{aligned}
\]
The 2-step reachability map of $\Sigma_{o}$ is
\begin{equation}
g_{2}: U^{2} \rightarrow k^{2}:\left(u_{2}, u_{1}\right) \mapsto\binom{u_{1}+u_{2}}{u_{2}}, \label{eq:18.2}
\end{equation}
which is obviously onto. So $\Sigma_{o}$ is reachable.
\end{example}

We want to calculate a canonical realization of $f_{0}:=f_{\Sigma_{0}}$. Since $\Sigma_{0}$ is reachable, $\Sigma_{f_{0}} \simeq \Sigma_{0}^{o b s}$; we apply the construction in Proposition~\ref{prop:10.4} to obtain $\Sigma_{f_{0}}$. The observation algebra $\mathcal{A}\left(\Sigma_{0}\right)$ will be determined as $\bigcup_{t \geq 0} \mathcal{A}^{0, t}\left(\Sigma_{0}\right)$ by induction on $t$, using the Main Lemma~\ref{lem:10.7}. Write $A(X)=k\left[\eta_{1}, \eta_{2}\right]$, so
\[
A(P): \eta_{1} \mapsto \eta_{1}+T, \quad \eta_{2} \mapsto \eta_{1} \eta_{2}+\eta_{1}+\eta_{2},
\]
and
\[
h=\eta_{2} .
\]
To simplify calculations, let $\hat{\eta}_{i}:=\eta_{i}+1, i=1,2$, so that $A(P): \hat{\eta}_{1} \mapsto \hat{\eta}_{1}+T, \quad \hat{\eta}_{2} \mapsto \hat{\eta}_{1} \hat{\eta}_{2}$ and $h=\hat{\eta}_{2}-1$. Then $\mathcal{A}^{0,0}=k\left[\hat{\eta}_{2}-1\right]=k\left[\hat{\eta}_{2}\right]$. By induction,
\[
\mathcal{A}^{0, t}=k\left[\hat{\eta}_{2}, \hat{\eta}_{2} \hat{\eta}_{1}, \ldots, \hat{\eta}_{2} \hat{\eta}_{1}^{t}\right],
\]
because
\[
\begin{aligned}
A(P)\left(\hat{\eta}_{2} \hat{\eta}_{1}^{t}\right) & =A(P)\left(\hat{\eta}_{2}\right) \cdot\left(A(P)\left(\hat{\eta}_{1}\right)\right)^{t} \\
& =\hat{\eta}_{1} \hat{\eta}_{2}\left(\hat{\eta}_{1}+T\right)^{t}
\end{aligned}
\]

\label{op:93}%
\[
=\sum_{j=0}^{t}\binom{t}{j} \hat{\eta}_{2} \hat{\eta}_{1}^{j+1} T^{t-j},
\]
so the new generator $\hat{\eta}_{2} \hat{\eta}_{1}^{t+1}$ appears as the coefficient of $T^{0}$. Thus
\begin{equation}
\mathcal{A}_{f_{0}}=\mathcal{A}\left(\Sigma_{0}\right)=k\left[\left\{\left(\eta_{2}+1\right)\left(\eta_{1}+1\right)^{t}, \quad t \geq 0\right\}\right] \label{eq:18.3}
\end{equation}
is not a finitely generated algebra. Therefore $\Sigma_{f_{o}}$ is not a polynomial system.

By Definition~\ref{def:3.17} the canonical state space $X_{f_{o}}=X\left(\mathcal{A}_{f_{o}}\right)=X\left(k\left[\hat{\eta}_{2} \hat{\eta}_{1}^{t}, t \geq 0\right]\right)$ is an almost-variety which can be represented by the principal open set $D:=\left\{\eta_{2} \neq-1\right\}$ in $k^{2}=X\left(k\left[\eta_{1}, \eta_{2}\right]\right)$ (or, equivalently, $\left\{\hat{\eta}_{2} \neq 0\right\}$ in $X\left(k\left[\hat{\eta}_{1}, \hat{\eta}_{2}\right]\right)$) \textit{plus} an extra point $\{*\}$; see Example~\ref{ex:3.19}. The morphism
\[
T: X \rightarrow X^{\text {obs }}=X_{f_{0}}
\]
in Proposition~\ref{prop:10.4} is the map $X(i)$, where $i: A_{f_{0}} \rightarrow k\left[\hat{\eta}_{1}, \hat{\eta}_{2}\right]$ is the inclusion map. Therefore
\[
\begin{aligned}
& T(x)=x \text { if } x \text { is in } D, \text { and } \\
& T(x)=* \text { if } x_{2}=-1 .
\end{aligned}
\]
Since $T$ is a $k$-system morphism, the transitions in $\Sigma_{f_{0}}$ are
\[
\begin{aligned}
& P_{f_{o}}(x, u)=T(P(x, u)) \text { if } x \text { is in } D \text {, and } \\
& P_{f_{o}}(*, u)=* \text { for all } u \text { in } U .
\end{aligned}
\]
The initial state in $\Sigma_{f_{o}}$ is 0 in $D$ and the output is given by
\[
\begin{aligned}
& h(x)=x_{2} \text { if } x \text { is in } D, \text { and } \\
& h(*)=-1 .
\end{aligned}
\]
Note that $T$ identifies precisely those states of $X_{\Sigma_{0}}$ which are indistinguishable.

\label{op:94}%

\begin{example}\label{ex:18.4}
Consider $\Sigma_{0}$ and $f_{0}$ as in Example~\ref{ex:18.1}. We now prove that \textit{there exists no polynomial realization of} $f_{0}$ \textit{in which every pair of reachable states is distinguishable}. In other words, it is impossible to find any polynomial system $\hat{\Sigma}$ realizing $f_{0}$ and a one-to-one abstract system morphism
\[
T_{\mathrm{ac}}: \Sigma_{\mathrm{ac}} \rightarrow \hat{\Sigma},
\]
where $\Sigma_{a c}$ is the abstractly canonical realization of $f_{0}$. Our claim depends on two facts. (i) Let $A:=A\left(g_{2}\right)\left(\mathcal{A}\left(\Sigma_{0}\right)\right) \subseteq A\left(U^{2}\right)$. Then $A$ \textit{is maximally separating Definition~\ref{def:1.11} with respect to} $A\left(U^{2}\right)$. Indeed, write $A\left(U^{2}\right)=k\left[T_{1}, T_{2}\right]$; by \eqref{eq:18.2} and \eqref{eq:18.3},
\[
A=k\left[\left\{\left(T_{2}+1\right)\left(T_{1}+T_{2}+1\right)^{t}, \quad t \geq 0\right\}\right] .
\]
With a change of variables $\hat{T}_{2}:=T_{2}+1, \hat{T}_{1}:=T_{1}+T_{2}+1$ in $A\left(U^{2}\right)$, $A$ becomes $k\left[\left\{\hat{T}_{2} \hat{T}_{1}^{t}, t \geq 0\right\}\right]$. By Example~\ref{ex:1.12}, $A$ is maximally separating. (ii) Now our claim follows from the following more general fact:
\end{example}

\begin{lemma}\label{lem:18.5}
Let $f$ be any polynomial response map. Assume that (i) the $n$-step reachability map $g_{f, n}$ of $\Sigma_{f}$ is onto and (ii) $A\left(g_{f, n}\right)\left(X_{f}\right)$ is maximally separating in $A\left(U^{n}\right)$. Then any quasi-reachable $\Sigma$ realizing $f$ whose reachable states are distinguishable is isomorphic to $\Sigma_{f}$.
\end{lemma}

\begin{proof}
Let $T: \Sigma \rightarrow \Sigma_{f}$ be the unique dominating morphism Proposition~\ref{prop:10.4}. We must prove that $T$ is an isomorphism. The restriction of $T$ to the reachable part $\Sigma_{R}$ of $\Sigma$ induces a morphism of \textit{abstract} systems $T_{R}: \Sigma_{R} \rightarrow \Sigma_{f}$. Since by hypothesis $\Sigma_{R}$ is abstractly canonical, $T_{R}$ is a bijection Theorem~\ref{thm:7.8}. Therefore $\Sigma_{R}$ is also reachable in $n$ steps. In particular, $\Sigma$ is quasi-reachable in $n$ steps and $A\left(g_{n}\right)$ is one-to-one. Let $B:=A\left(g_{n}\right)\left(X_{\Sigma}\right)$. Since $T \circ g_{n}=g_{f, n}, T$ is an isomorphism iff $B=A:=A\left(g_{f, n}\right)\left(X_{f}\right)$. By definition of "maximally separating" Definition~\ref{def:1.11} it will be enough to prove that $B$ separates no more points of $U^{n}$ than $A$. So take $v, w$ in $U^{n}$ and suppose that $A$ does not separate $v$ and $w$, i.e.

\label{op:95}%
\begin{equation}
a(v)=a(w) \quad \text { for all } a \text { in } A . \label{eq:18.6}
\end{equation}
By definition of $A\left(g_{f, 2}\right)$, Lemma~\ref{lem:18.5} is equivalent to
\[
c\left(g_{f, n}(v)\right)=c\left(g_{f, n}(w)\right) \text { for all } c \text { in } \mathcal{A}_{f},
\]
in other words, $g_{f, n}(v)=g_{f, n}(w)$. Since $g_{f, n}=T \circ g_{n}$ and $T$ is one-to-one on reachable states, we conclude that $g_{n}(v)=g_{n}(w)$. So $d\left(g_{n}(v)\right)=d\left(g_{n}(w)\right)$ for all $d$ in $A(X)$, which means that $b(v)=b(w)$ for all $b$ in $B$, i.e. neither does $B$ separate $v$ and $w$.
\end{proof}

When $k=\mathbb{R}$ or $k$ is algebraically closed the hypothesis of Lemma \ref{lem:18.5} can be weakened considerably, replacing (i) by just $\operatorname{dim} \Sigma_{f}=n$. The (easy) proof of this stronger lemma uses Theorem~\ref{thm:4.6} plus Lemma~\ref{lem:18.5}.

\begin{example}\label{ex:18.7}
We continue to investigate $\Sigma_{o}, f_{o}$, and determine the (unique) irreducible equation of minimal order satisfied by $f_{0}$.

Let $\hat{\eta}_{1}, \hat{\eta}_{2}$ be coordinates for $\Sigma_{0}$ defined as before. Working in $\mathcal{Q}\left(\Sigma_{0}\right)$,
\[
h+1=\hat{\eta}_{2}, \quad h^{u_{2}}+1=\hat{\eta}_{1} \hat{\eta}_{2} \quad \text { and } \quad h^{u_{2} u_{1}}+1=\hat{\eta}_{1} \hat{\eta}_{2}+\hat{\eta}_{1} \hat{\eta}_{2} u_{2},
\]
for all $\left(u_{2}, u_{1}\right)$ in $U^{2}$. Since
\[
\hat{\eta}_{1}^{2} \hat{\eta}_{2}=\left(\hat{\eta}_{1} \hat{\eta}_{2}\right)^{2} / \hat{\eta}_{2}=\left(h^{u_{2}}+1\right)^{2} /(h+1),
\]
we conclude that
\[
(h+1)\left(h^{u_{2} u_{1}}+1\right)-\left(h^{u_{2}}+1\right)^{2}-\left(h^{u_{2}}+1\right)(h+1) u_{2}=0,
\]
for all $\left(u_{2}, u_{1}\right)$ in $U^{2}$. Since $h(g(w))=f(w)$ for all $w$ in $U[z]$, the polynomial
\label{op:96}%
\[
E:=\left(L_{2}+1\right)\left(L_{0}+1\right)-\left(L_{1}+1\right)^{2}-\left(L_{1}+1\right)\left(L_{2}+1\right) S_{2},
\]
gives an algebraic (rational) difference equation for $f_{0}$, i.e.
\[
\begin{aligned}
{[y(t-2)} & +1][y(t)+1]-[y(t-1)+1]^{2}- \\
& -[y(t-1)+1][y(t-2)+1] u(t-2)=0,
\end{aligned}
\]
for all input/output pairs $(y, u)$ of $\mathbf{f}_{0}$. The polynomial $E$ is easily seen to be irreducible, and $E$ is of minimal order for $f_{0}$, because $h=\hat{\eta}_{2}-1$ and $h^{u}=\hat{\eta}_{1} \hat{\eta}_{2}-1$ are algebraically independent.
\end{example}

\begin{example}\label{ex:18.8}
Since always $\operatorname{trdeg}_{K} \mathcal{Q}_{f}^{K} \leq \operatorname{trdeg} \mathcal{Q}_{f}$, there is always an algebraic difference equation of degree $\operatorname{dim} \Sigma_{f}$. \textit{There may exist}, however, \textit{equations of order strictly lower than} $\operatorname{dim} \Sigma$. To illustrate this, consider the bilinear input/output map $f_{1}=f_{\Sigma_{1}}$, where $m=2, p=1$, and $\Sigma_{1}:=\left(k^{2 n+1}, P, h, 0\right), n \geq 1$ arbitrary, and $P, h$ are given by the equations
\[
\begin{array}{cc}
x_{1}(t+1)=u_{1}(t), & x_{n+1}(t+1)=u_{2}(t), \\
x_{2}(t+1)=x_{1}(t), & x_{n+2}(t+1)=x_{n+1}(t), \\
\vdots & \vdots \\
x_{n}(t+1)=x_{n-1}(t), & x_{2 n}(t+1)=x_{2 n-1}(t), \\
x_{2 n+1}(t)=x_{n}(t) u_{2}(t)+x_{2 n}(t) u_{1}(t), & y(t)=x_{2 n+1}(t) .
\end{array}
\]
It is easily shown that $\Sigma_{1}$ is canonical, so $\operatorname{dim} \mathcal{Q}_{f_{1}}=\operatorname{dim} \Sigma_{1}=2 n+1$. However, $f_{1}$ satisfies the (affine) difference equation
\[
y(t+1)=u_{1}(t-n) u_{2}(t)+u_{2}(t-n) u_{1}(t),
\]
of order $n+1<2 n+1$. Note that $f_{1}$ is also a counterexample to $\operatorname{dim} \mathcal{L}_{f}=\operatorname{dim}\mathcal{L}_{f}^{K}$. Counterexamples with $m=1$ also exist; e.g. $X_{\Sigma_{2}}:=k^{3}$, $x_{\Sigma_{2}}^{\#}:=0$ and $\Sigma_{2}$ given by
\label{op:97}%
\[
\begin{aligned}
x_{1}(t+1) & =u(t), \\
x_{2}(t+1) & =x_{3}(t), \\
x_{3}(t+1) & =x_{3}(t) x_{1}(t)+x_{1}(t)+x_{2}(t) u(t), \\
y(t) & =x_{3}(t),
\end{aligned}
\]
which satisfies an affine equation of order 2 :
\[
y(t)=y(t-1) u(t-2)+y(t-2) u(t-1)+u(t-2) .
\]
\end{example}

\begin{example}\label{ex:18.9}
When $\Sigma_{f}$ is a polynomial system, it is natural to represent it by a system of simultaneous polynomial difference equations. An important question in applications concerns the minimal possible number $r$ of equations in such representations of a given $\Sigma_{f}$. By (3.4ff), $r=$ smallest possible cardinality of a set of generators for $\Sigma_{f}$. In general, $r \leq \operatorname{trdeg} \mathcal{A}_{f}$, and equality holds if and only if $\mathcal{A}_{f}=k\left[T_{1}, \ldots, T_{r}\right]$, which is always the case for linear response maps $f$. A result of \citet{kalman1979realization} shows that $\mathcal{A}_{f}$ is a polynomial ring also for bilinear single-output ($p=1$) response maps. For more general response maps, even bounded ones, $X_{f}$ may be different from affine space. For instance, let $m=p=1$ and let $f_{3}$ be the response map satisfying
\[
y(t)=u^{2}(t-2) u(t-1)+u^{3}(t-2) .
\]
Then $f_{3}=f_{\Sigma_{3}}$, where $X_{\Sigma_{3}}:=k^{2}, x_{\Sigma_{3}}^{\#}:=0$ and $\Sigma_{3}$ has equations
\[
\begin{aligned}
x_{1}(t+1) & =u(t), \\
x_{2}(t+1) & =x_{1}^{2}(t) u(t)+x_{1}^{3}(t), \\
y(t) & =x_{2}(t) .
\end{aligned}
\]
Since $\Sigma_{3}$ is quasi-reachable, $\mathcal{A}_{f}=\mathcal{A}(\Sigma)=k\left[\eta_{1}^{2}, \eta_{1}^{3}, \eta_{2}\right]$. Thus to represent the canonical realization (polynomial by Theorem~\ref{thm:16.2}(c)) $\Sigma_{f_{3}}$ one needs at least 3 equations. (Note that the noncanonical realization $\Sigma_{3}$ of $f_{3}$ requires only 2 equations!) A representation of $\Sigma_{f_{3}}$ is, for instance,
\[
\begin{aligned}
x_{1}(t+1) & =u^{2}(t), \\
x_{2}(t+1) & =u^{3}(t), \\
x_{3}(t+1) & =x_{1}(t) u(t)+x_{2}(t), \quad x^{\#}=0, \\
y(t) & =x_{3}(t),
\end{aligned}
\]
\label{op:98}%
where $X_{f_{3}}=\left\{\left(x_{1}, x_{2}, x_{3}\right)\right.$ in $\left.k^{3} \mid x_{1}^{3}=x_{2}^{2}\right\}$, a surface with a singularity at $x=0$. Since $\operatorname{dim} X_{f}=2, \Sigma_{3}$ is minimal, a fact which agrees with $\Sigma_{3}$ being quasi-reachable and abstractly observable, so weakly canonical Theorem~\ref{thm:13.7}.
\end{example}

\begin{example}\label{ex:18.10}
Consider a \textit{homogeneous} polynomial response map $f$; i.e. all monomials in $\psi_{f}$ have the same degree, say $s$. Take, for simplicity, $m=p=1$. One way of applying the theory of \textit{multilinear} response maps Example~\ref{ex:15.3} for obtaining realizations of such $f$ is the following. Let
\[
\psi_{f}=\sum a_{i_{1} \ldots i_{s}} \xi_{1 i_{1}} \ldots \xi_{1 i_{s}},
\]
where the sum runs over all possible sequences (with repetitions) $i_{1}, \ldots, i_{s}$ of integers $\geq 1$. Define the $s$-linear response map $f^{o}$ (with $m=s, p=1$) by
\[
\psi_{f^{o}}:=\sum a_{i_{1} \ldots i_{s}} \xi_{1 i_{1}} \ldots \xi_{s i_{s}} .
\]
It is easily verified that $f$ is finitely realizable iff $f^{o}$ is finitely realizable. Given any realization $\Sigma$ of $f^{o}$, a realization $\hat{\Sigma}$ of $f$ is obtained by applying the same input to all the input channels of $\Sigma$. Such a method of realizing a homogeneous $f$ has been suggested by several authors (see, for instance, \citealp{bush1965kernels}).
\end{example}

The method is theoretically unsatisfactory, however, given the noncanonical nature of $\hat{\Sigma}$. The following example shows that this procedure may be also unsatisfactory from a practical (synthesis) point of view. Specifically, we shall give a homogeneous response map $f$ of
\label{op:99}(arbitrary) degree $s$ such that $\Sigma_{f}$ has $X_{f}=k$ but such that the lowest possible dimension for a realization $\hat{\Sigma}$ obtained by the above procedure is $s$. That is, we look for an $f$ with $X_{f}=k$ and $\operatorname{dim} \Sigma_{f^{o}}=s$.

Let $\Sigma_{f_{4}}$ have the equations (with $X=k$)
\[
\begin{aligned}
& x(t+1)=x(t)+u(t), \quad x^{\#}:=0, \\
& y(t)=x^{s}(t) .
\end{aligned}
\]
Then $f_{4}^{o}=f_{\Sigma_{4}}$, where $\Sigma_{4}$ has $m=s, p=1, X=k^{s}$ and equations
\[
\begin{aligned}
x_{i}(t+1) & =x_{i}(t)+u_{i}(t), \quad x_{i}^{\#}:=0, \quad i=1, \ldots, s . \\
y(t) & =x_{1}(t) \ldots x_{s}(t) .
\end{aligned}
\]
Clearly, $\Sigma_{4}$ is reachable. Moreover, $\Sigma_{4}$ is algebraically observable, since, writing $A\left(k^{s}\right)=k\left[\eta_{1}, \ldots, \eta_{s}\right]$, the coefficient of $T_{1} \ldots \hat{T}_{i} \ldots T_{s}$ ($T_{i}$ omitted) in
\[
A(P) h=\left(\eta_{1}+T_{1}\right) \ldots\left(\eta_{s}+T_{s}\right),
\]
is $\eta_{i}$. Thus $\Sigma_{4}$ is canonical and has dimension $s$.

\begin{example}\label{ex:18.11}
Consider \textit{strictly recursive} equations
\begin{equation*}
y(t)=R(y(t-1), \ldots, y(t-r), u(t-1), \ldots, u(t-r)), \tag{*}
\end{equation*}
for finitely realizable polynomial responses $f$, where $R$ is a polynomial. Such equations are well-known to exist for linear response maps. By the Cayley-Hamilton Theorem, the internally-bilinear response maps Example~\ref{ex:15.4} of systems of the form
\[
\begin{aligned}
x(t+1) & =F x(t) u(t), \\
y(t) & =H x(t),
\end{aligned}
\]
($m=1, F, H=$ linear) are easily seen to also satisfy equations (*).
\label{op:100}%
In general, however, a finitely realizable polynomial response map satisfies no equation such as (*). To construct counterexamples it is enough to exhibit systems $\Sigma$ such that for each $r \geq 1$ there exist pairs of inputs $w, \hat{w}$ for which $w(t)=\hat{w}(t)$ for $t \geq 0$, $\mathbf{f}(w)(t)=\mathbf{f}(\hat{w})(t)$ for $0 \leq t<r$ but $\mathbf{f}(w)(r) \neq f(\hat{w})(r)$. Such $w, \hat{w}$ clearly contradict (*). We give three such counterexamples.
\begin{enumerate}
\item[(a)] A one-dimensional system $\Sigma_{5}$ with $m=p=1$, where
\[
\begin{aligned}
x(t+1) & =x(t)+u(t), \quad x^{\#}:=0, \\
y(t) & =x^{2}(t) .
\end{aligned}
\]
Here take $w(t)=\hat{w}(t):=0$ if $t \neq-1, r-1, w(r-1)=\hat{w}(r-1):=1$, $w(-1):=-1$ and $\hat{w}(-1):=1$.
\item[(b)] Again $m=p=1$, and $\Sigma_{6}$ given by
\[
\begin{aligned}
& x_{1}(t+1)=x_{1}(t)+x_{1}(t) u(t), \quad x_{1}^{\#}:=1, \\
& x_{2}(t+1)=x_{2}(t)+x_{2}(t) u(t), \quad x_{2}^{\#}:=1, \\
& y(t)=x_{1}(t)+x_{2}(t) .
\end{aligned}
\]
Here take $w(t)=\hat{w}(t):=0$ if $t \neq-1, r-1, w(r-1)=\hat{w}(r-1):=\binom{1}{-1}$, $w(-1):=0$ and $\hat{w}(-1):=\binom{1}{-1}$. Note that $f_{\Sigma_{6}}$ is an internally-bilinear response map.
\item[(c)] Let $m=2, p=1$, and $\Sigma_{7}$ given by
\[
\begin{aligned}
x_{1}(t+1) & =x_{1}(t)+u_{1}(t), \quad x_{1}^{\#}:=0, \\
x_{2}(t+1) & =x_{2}(t)+u_{2}(t), \quad x_{2}^{\#}:=0, \\
y(t) & =x_{1}(t) x_{2}(t) .
\end{aligned}
\]
Here take $w(t)=\hat{w}(t):=0$ if $t \neq-1, r-1, w(r-1)=\hat{w}(r-1)=\hat{w}(-1):=\binom{1}{1}$ and $w(-1):=-\binom{1}{1}$. Note that $f_{\Sigma_{7}}$ is a bilinear response map.
\end{enumerate}
\end{example}

\label{op:101}%
\section{State-Affine Systems}\label{ch:5}

Some special types of system configurations arose in Theorem \ref{thm:15.5} as natural realizations for bounded input/output maps. The most useful of these configurations are state-affine systems; they are studied in this section.

As indicated in Section~\ref{ch:1}, we shall base our development on the notion of a representation of the Volterra (or, equivalently, the exponent) series of $f$.

\subsection{Recognizable Series}\label{sec:19}

Throughout this section, $\varphi$ is an exponent series \eqref{eq:5.6} whose support is contained in $\Delta_{J}$, where
\[
J=\left\{\delta_{0}=(0), \delta_{1}, \ldots, \delta_{s}\right\} .
\]

\begin{definition}\label{def:19.1}
\textit{A representation} (\textit{with support in} $\Delta_{J}$) \textit{is an object}
\[
R=\left(X,\left\{F_{\alpha}, \alpha \text { in } J\right\},\left\{g_{\alpha}, \alpha=\delta_{1}, \ldots, \delta_{s}\right\}, h\right),
\]
\textit{where}
\begin{enumerate}
\item[(a)] $X$ \textit{is a vector space over} $k$,
\item[(b)] \textit{each} $F_{\alpha}: X \rightarrow X$ \textit{is a linear map},
\item[(c)] \textit{each} $g_{\alpha}$ \textit{is in} $X$, \textit{and}
\item[(d)] $h: X \rightarrow Y=k^{p}$ \textit{is a linear map}.
\end{enumerate}
\textit{For each} $\alpha=\alpha_{1} \ldots \alpha_{t}$ \textit{in} $J^{*}$ \textit{let} $F_{\alpha}:=F_{\alpha_{1}} \ldots F_{\alpha_{t}}$, \textit{and write} $F_{\Lambda}:=1_{X}$. \textit{If} $\gamma=\alpha_{1} \ldots \alpha_{t}(0) \ldots(0)$ is in $J^{*}$, \textit{with} $\alpha_{t} \neq(0)$, \textit{then}
\[
g_{\gamma}:=F_{\alpha_{1} \ldots \alpha_{t-1}} g_{\alpha_{t}} ;
\]
\textit{if} $\gamma=(0) \ldots(0), g_{\gamma}:=0$. \textit{With these notations,} $R$ \textit{is accessible iff}
\[
\operatorname{span}\left\{g_{\alpha}, \alpha \text { in } \Delta_{J}\right\}=X ;
\]
$R$ \textit{is reduced iff}
\label{op:102}%
\[
\bigcap_{\alpha \text { in } J^{*}} \operatorname{ker} h F_{\alpha}=\{0\} ;
\]
$R$ \textit{is canonical iff} $R$ \textit{is both accessible and reduced. The dimension of} $R$ \textit{is}
\[
\operatorname{dim} R:=\operatorname{dim} X \leq \infty .
\]
\textit{The exponent series} $\varphi_{R}$ \textit{represented by} $R$ \textit{is given by}
\[
\varphi_{R}(\alpha):=h g_{\alpha} \text { for all } \alpha \text { in } \Delta_{J} .
\]
\textit{The representation} $R$ \textit{is minimal iff, for any} $\hat{R}$ \textit{for which} $\varphi_{R}=\varphi_{\hat{R}}$, \textit{necessarily} $\operatorname{dim} R \leq \operatorname{dim} \hat{R}$. \textit{When} $\operatorname{dim} R<\infty, \varphi_{R}$ \textit{is recognizable. The linear map} $T: X \rightarrow \hat{X}$ \textit{induces a morphism of representations}
\[
T: R=\left(X,\left\{F_{\alpha}\right\},\left\{g_{\alpha}\right\}, h\right) \rightarrow \hat{R}=\left(\hat{X},\left\{\hat{F}_{\alpha}\right\},\left\{\hat{g}_{\alpha}\right\}, \hat{h}\right),
\]
\textit{iff} $T g_{\alpha}=\hat{g}_{\alpha}$ \textit{and} $T \circ F_{\alpha}=\hat{F}_{\alpha} \circ T$ \textit{for all} $\alpha$ \textit{and} $h=\hat{h} \circ T$.
\end{definition}

The terminology "recognizable" is taken from automata theory, as explained in Section~\ref{ch:1}.

It is easy to see that representations form a category with the above notion of morphism.

\begin{definition}\label{def:19.2}
\textit{The behavior matrix} $\mathcal{B}(\varphi)$ \textit{of} $\varphi$ \textit{is an infinite block matrix, with rows indexed by} $J^{*}$ \textit{and columns indexed by} $\Delta_{J}$, \textit{whose} $(\alpha, \beta)$-\textit{th entry is} $\varphi(\alpha \beta)$, \textit{a column vector in} $Y=k^{p}$.
\end{definition}

We denote by $B_{\beta}$ the $\beta$-th column of $\mathcal{B}(\varphi)$.

\begin{theorem}\label{thm:19.3}
Any $\varphi$ has a canonical representation $R_{\varphi}$. If $R$ and $\hat{R}$ are two canonical representations of $\varphi$, there exists a unique representation isomorphism $T: R \rightarrow \hat{R}$. Further, $\varphi$ is recognizable if and only if rank $\mathcal{B}(\varphi)<\infty$; in this case a representation $R$ of $\varphi$ is canonical (i) if and only if $R$ is minimal and (ii) if and only if $\operatorname{dim} R=\operatorname{rank} \mathcal{B}(\varphi)$.
\end{theorem}

\label{op:103}%

\begin{proof}
We define $R_{\varphi}=\left(X,\left\{F_{\alpha}\right\},\left\{g_{\alpha}\right\}, h\right)$ as follows:

$X:=$ linear space spanned by $\left\{B_{\beta}, \beta\right.$ in $\left.\Delta_{J}\right\}$,

$g_{\delta_{i}}:=B_{\delta_{i}}, \quad i=1, \ldots, s$,

$h:=$ linear map induced by the projections in the $\Lambda$-th block row: $B_{\beta} \mapsto \varphi(\beta)$, and

$F_{\delta_{i}}:=$ linear map induced by the column shifts $B_{\beta} \mapsto B_{\delta_{i} \beta}$, $i=0, \ldots, s$.

Note that the $F_{\alpha}$ are well-defined, since any relation among columns
\begin{equation}
\sum_{\text {(finite) }} r_{\beta} B_{\beta}=0, \quad r_{\beta} \text { in } k, \label{eq:19.4}
\end{equation}
implies $\sum r_{\beta} B_{\delta_{i} \beta}=0, i=0, \ldots, s$. Indeed, the $\alpha$-th (block) row of $\sum r_{\beta} B_{\delta_{i} \beta}$ is $\sum r_{\beta} \varphi\left(\alpha \delta_{i} \beta\right)$, which is also the $\alpha \delta_{i}$-th row of \eqref{eq:19.4}, and hence zero.

Clearly,
\[
g_{\beta}=F_{\beta_{1} \ldots \beta_{t-1}} g_{\beta_{t}}=B_{\beta},
\]
for any $\beta$ in $\Delta_{J}$, so $R$ is accessible. By definition of $h$,
\[
\varphi_{R_{\varphi}}(\beta)=h\left(g_{\beta}\right)=h\left(B_{\beta}\right)=\varphi(\beta),
\]
so $R_{\varphi}$ represents $\varphi$. Now take $x=\sum r_{\beta} B_{\beta}$ in $\bigcap_{\alpha} \operatorname{ker} h F_{\alpha}$. Then
\[
0=h F_{\alpha} x=\sum r_{\beta} h F_{\alpha \beta_{1} \ldots \beta_{t-1}} g_{\beta_{t}}=\sum r_{\beta} \varphi(\alpha \beta)=\alpha \text {-th block row of } x \text {, }
\]
for all $\alpha$ in $J^{*}$. Thus $x=0$ and $R$ is reduced, hence canonical.

The proof of the theorem will be complete after we establish the following lemma (similar to Lemma~\ref{lem:7.7} and Lemma~\ref{lem:11.3}):
\end{proof}

\begin{lemma}\label{lem:19.5}
Let $R, \hat{R}$ be representations of $\varphi$, with $R$ accessible and $\hat{R}$ reduced. Then there exists a unique representation morphism
\label{op:104}$T: R \rightarrow \hat{R}$. When $\hat{R}$ is canonical, $T$ is onto.
\end{lemma}

\begin{proof}
Define $T: X \rightarrow \hat{X}$ as the linear extension of
\[
T\left(g_{\alpha}\right):=\hat{g}_{\alpha}, \alpha \text { in } \Delta_{J} .
\]
By definition of representation morphism, this is clearly the only possible choice for $T$. Thus the lemma will follow if $T$ is well-defined. We must show that if $x=\sum r_{\beta} g_{\beta}=0$ (finite sum) then also $\hat{x}:=\sum r_{\beta} \hat{g}_{\beta}=0$. Pick $\alpha$ in $J^{*}$. Then
\[
\hat{h} \circ \hat{F}_{\alpha}(\hat{x})=\sum r_{\beta} \hat{h}\left(\hat{g}_{\alpha \beta}\right)=\sum r_{\beta} \varphi(\alpha \beta)=\sum r_{\beta} h\left(g_{\alpha \beta}\right)=h F_{\alpha}(x)=0 .
\]
Since $\hat{R}$ is reduced, this means that $\hat{x}=0$.
\end{proof}

\begin{remark}\label{rem:19.6}
The above theorem gives rise to an algorithm for constructing representations of a recognizable series $\varphi$ from $\mathcal{B}(\varphi)$. It is only necessary for this purpose \textit{to find a submatrix} $\Phi$ \textit{of full rank of} $\mathcal{B}(\varphi)$ and to express the $g_{\alpha}, F_{\alpha}$ and $h$ with respect to the basis consisting of the set of columns used in the definition of $\Phi$. This algorithm is a minor variation of that given by \citet{fliess1972sur} and generalizes the one given for the linear case by \citet{rouchaleau1972linear}.
\end{remark}

\begin{example}\label{ex:19.7}
Let $m=p=1$ and take a linear Example~\ref{ex:15.2} response map $f$, with $\psi_{f}=\sum a_{i} \xi_{1 i}, \varphi_{f}=\left(\sum a_{i} 0^{i}\right) 1$. The only possible nonzero columns of $\mathcal{B}\left(\varphi_{f}\right)$ are those indexed by $1,01, \ldots$, and the only possible nonzero rows are those indexed by $\left\{0^{n}, n \geq 0\right\}$ :

\[
\begin{array}{c|ccccccc}
 & 1 & 01 & 0^{2} 1 & \cdots & 0^{n} 1 & \cdots \\
\hline
\Lambda & a_{1} & a_{2} & a_{3} & \cdots & a_{n+1} & \cdots \\
0 & a_{2} & a_{3} & a_{4} & \cdots & a_{n+2} & \cdots \\
0^{2} & a_{3} & a_{4} & a_{5} & \cdots & a_{n+3} & \cdots \\
\vdots & \vdots & \vdots & \vdots & & & \\
0^{n} & a_{n+1} & \cdot & \cdot & & & \\
\vdots & \vdots & \vdots & \vdots & & &
\end{array} \qquad ,
\]

\label{op:105}%
the classical \textit{Hankel matrix} of the linear response map $f$. When $m, p \neq 1$, the only possible nonzero part of $\mathcal{B}(\varphi)$ becomes, with a suitable ordering of row and column indexes, the block Hankel matrix of $f$; the realization procedure in Remark~\ref{rem:19.6} coincides with the well-known "Silverman's formulas" (\citealp{silverman1971realization})
\end{example}

\begin{example}\label{ex:19.8}
When $f$ is internally-bilinear Example~\ref{ex:15.4}, our $\mathcal{B}(\varphi)$ becomes the generalized Hankel matrix introduced by \citet{isidori1973direct} and \citet{fliess1973sur}. When $f$ is bilinear in the input/output sense Example~\ref{ex:15.3}, the nonzero part of $\mathcal{B}(\varphi)$ can be arranged so as to become the matrix introduced by \citet{kalman1979realization}.
\end{example}

\subsection{State-Affine Systems}\label{sec:20}

\begin{definition}\label{def:20.1}
\textit{The polynomial system} $\Sigma$ \textit{is state affine iff}
\begin{enumerate}
\item[(a)] $X=k^{n}$ \textit{for some integer} $n$,
\item[(b)] $h: X \rightarrow k^{p}$ \textit{is a linear map},
\item[(c)] $x^{\#}=0$, \textit{and}
\item[(d)] \textit{for each fixed input value} $u, P(x, u)$ \textit{is an affine} (= \textit{linear} + \textit{translation}) \textit{function of} $x$.
\end{enumerate}
\textit{In other words, there exist polynomials} $p_{i j}\left(T_{1}, \ldots, T_{m}\right)$ \textit{and} $q_{i}\left(T_{1}, \ldots, T_{m}\right)$ \textit{such that the transition equations for} $\Sigma$ \textit{are given by}
\[
x_{i}(t+1)=\sum_{j=0}^{n} p_{i j}(u(t)) x_{j}(t)+q_{i}(u(t)), \quad i=1, \ldots, n .
\]
$\Sigma$ \textit{is span-reachable iff the set of reachable states} $X_{R}$ \textit{spans} $k^{n}$; \textit{span-canonical iff both span-reachable and observable; minimal iff} $\operatorname{dim} \Sigma \leq \operatorname{dim} \hat{\Sigma}$ \textit{for any state-affine system} $\hat{\Sigma}$ \textit{for which} $f_{\Sigma}=f_{\hat{\Sigma}}$. \textit{A system morphism} $T: \Sigma \rightarrow \hat{\Sigma}$ \textit{between state-affine systems is a morphism of state-affine systems iff} $T: X \rightarrow \hat{X}$ \textit{is linear}.
\end{definition}

\begin{remark}\label{rem:20.2}
It follows from Theorem \ref{thm:15.5} that a polynomial response map $f$ is bounded and finitely realizable iff $f$ can be
\label{op:106}realized by a system which is state affine except for the condition $x^{\#}=0$. A coordinate translation may of course be used to make $x^{\#}=0$, without changing the affine form of $P$. However, this makes $h$ an affine, rather than a linear, map. In other words, the equilibrium-level output $h\left(x^{\#}\right)$ may become nonzero. \textit{We shall assume} for the rest of this section that a coordinate translation has been performed (if necessary) on $Y=k^{p}$ so that
\[
\varphi_{f}(\Lambda)=0 \text { for all response maps } f .
\]
With this convention, $h\left(x^{\#}\right)=0$ for every system. Thus bounded + finitely realizable $=$ realizable by a state-affine system.
\end{remark}

\begin{remark}\label{rem:20.3}
The basic observables of a state-affine system are affine functions of the state; thus

\[
\text {(abstract) observability }=\text { algebraic observability, }
\]
for such systems.
\end{remark}

Let $\Sigma$ be a state affine system. Since $P$ is affine in $x$ and polynomial in $u$, there exists a subset $J=\left\{\delta_{0}=(0), \delta_{1}, \ldots, \delta_{s}\right\}$ of $\mathbb{N}^{m}$ and matrices $\left\{F_{\delta_{i}}, \delta_{i}\right.$ in $\left.J\right\}$ and $\left\{g_{\delta_{i}}, i=1, \ldots, s\right\}$ such that
\begin{equation}
P(x, u)=\sum_{i=0}^{s} F_{\delta_{i}} x u^{\delta_{i}}+\sum_{i=1}^{s} g_{\delta_{i}} u^{\delta_{i}} \label{eq:20.4}
\end{equation}
We shall call $R=R_{\Sigma}:=\left(k^{n},\left\{F_{\delta_{i}}, i=0, \ldots, s\right\},\left\{g_{\delta_{i}}, i=1, \ldots, s\right\}\right.$, h) the \textit{representation associated to} $\Sigma$. Take $w=w_{t} \cdots w_{1}$ in $U^{t}$. An easy calculation shows that
\begin{equation}
P^{(t)}(x, w)=\sum_{\alpha \text{ in } J^{t}}\left(F_{\alpha} x+g_{\alpha}\right) w^{\alpha} . \label{eq:20.5}
\end{equation}
Thus the $t$-step reachability map $g_{t}: U^{t} \rightarrow k^{n}$ becomes

\label{op:107}%
\begin{equation}
g_{t}(w)=P^{(t)}(0, w)=\sum_{\alpha \text{ in } J^{t}} g_{\alpha} w^{\alpha}, \label{eq:20.6}
\end{equation}
and therefore
\[
f_{\Sigma}(w)=h\left(g_{t}(w)\right)=\sum h\left(g_{\alpha}\right) w^{\alpha}=\sum \varphi_{R}(\alpha) w^{\alpha},
\]
so that $\varphi_{f}=\varphi_{R}$.

We now have proved the following

\begin{lemma}\label{lem:20.7}
The assignment $\Sigma \rightarrow R_{\Sigma}$ is a bijection between state affine realizations of a bounded polynomial response map $f$ and finite-dimensional representations (with fixed basis for $X$) of $\varphi_{f}$.
\end{lemma}

The above assignment preserves reachability and observability properties:

\begin{lemma}\label{lem:20.8}
$\Sigma$ is span-reachable [respectively observable] if and only if $R_{\Sigma}$ is accessible [respectively reduced].
\end{lemma}

\begin{proof}
Let $b: k^{n} \rightarrow k$ be a linear function, and take $w$ in $U^{t}$. Then \eqref{eq:20.6} implies that
\begin{equation}\label{eq:20.9}
b\left(g_{t}(w)\right)=\sum_{\alpha \text{ in } J^{t}} b\left(g_{\alpha}\right) w^{\alpha} .
\end{equation}

Since $k$ is an infinite field, the right side of \eqref{eq:20.9} is zero for all $w$ in $U^{t}$ iff $b\left(g_{\alpha}\right)=0$ for all $\alpha$ in $J^{t}$. Thus there exists a $b \neq 0$ with $b\left(g_{t}(w)\right)=0$ for all $w$ in $U[z]$ (i.e., $\Sigma$ is not span-reachable) iff there is a $b \neq 0$ with $b\left(g_{\alpha}\right)=0$ for all $\alpha$ (i.e., $R_{\Sigma}$ is not accessible).

We now prove the observability part. Take $w$ in $U^{t}$. Then
\begin{equation}
h^{w}(x)=h \circ P^{(t)}(x, w)=\sum_{\alpha \text { in } J^{t}} h\left(F_{\alpha} x+g_{\alpha}\right) w^{\alpha} . \label{eq:20.10}
\end{equation}

\label{op:108}%
Two states $x$ and $\hat{x}$ are indistinguishable iff $h^{w}(x)=h^{w}(\hat{x})$ for all $w$ in $U^{*}$, which by \eqref{eq:20.10} means that
\begin{equation}
\sum_{\alpha \text{ in } J^{t}} h F_{\alpha}(x-\hat{x}) w^{\alpha}=0 \quad \text { for all } w \text { in } U^{t} \text { and all } t \geq 0 . \label{eq:20.11}
\end{equation}
Using again the fact that $k$ is infinite, \eqref{eq:20.11} is equivalent to $x-\hat{x}$ being in $\operatorname{ker} h F_{\alpha}$ for all $\alpha$ in $J^{*}$. So $h^{\Gamma}$ is one-to-one iff $\bigcap \operatorname{ker} h F_{\alpha}=\{0\}$, as required.
\end{proof}

\begin{remark}\label{rem:20.12}
The hypothesis $k=$ infinite is essential for the above result in its present form. The correct generalization to arbitrary fields is that only monomials $u^{\delta_{i}}$ which are linearly independent as functions should be used in the definition of a state affine system.
\end{remark}

From Lemma~\ref{lem:20.7}, \eqref{eq:20.9}, and Theorem~\ref{thm:19.3} we obtain the main result of this section:

\begin{theorem}\label{thm:20.13}
Any bounded and finitely realizable response map $f$ has a span-canonical state-affine realization, unique up to isomorphism (= change of basis in the state-space). A realization $\Sigma$ of $f$ is span-canonical (i) if and only if $\Sigma$ is minimal and (ii) if and only if $\operatorname{dim} \Sigma=\operatorname{rank} \mathcal{B}\left(\varphi_{f}\right)$.
\end{theorem}

Of course, realizations can be obtained from $\mathcal{B}\left(\varphi_{f}\right)$ using the algorithm in \eqref{eq:19.4}. There is an interesting interpretation of the row-space $\mathcal{R}_{f}$ of $\mathcal{B}\left(\varphi_{f}\right)$, which implies a natural duality between the observation space $\mathcal{L}_{f}$ and the state space of the span-canonical state-affine realization of $f$ (i.e., the column space of $\mathcal{B}\left(\varphi_{f}\right)$). Assume for simplicity that $p=1$, and define
\[
\eta: \mathcal{L}_{f} \rightarrow \mathcal{R}_{f},
\]
on generators by
\begin{equation}\label{eq:20.14}
\beta\text{-th column of } \eta\left(f^{w}\right):=\sum_{\alpha \text { in } J^{t}} \varphi(\alpha \beta) w^{\alpha} \text{ for } w \text{ in } U^{t}, t \geq 0 .
\end{equation}

\label{op:109}%
It is not difficult to verify the following

\begin{proposition}\label{prop:20.15}
$\eta: \mathcal{L}_{f} \simeq \mathcal{R}_{f}$. In particular, $\operatorname{dim} \mathcal{L}_{f}=\operatorname{rank} B\left(\varphi_{f}\right)=$ dimension of span-canonical state-affine realization of $f$.
\end{proposition}

\subsection{Finite Response Maps and Cascades of Linear Systems}\label{sec:21}

\begin{definition}\label{def:21.1}
\textit{The total degree of the polynomial response map} $f$ (\textit{or of its Volterra series} $\psi_{f}$) \textit{is}

\[
\operatorname{tdeg} f=\operatorname{tdeg} \psi_{f}:=\sup \left\{\|\alpha\| \mid \psi_{f}(\alpha) \neq 0\right\} \leq \infty ;
\]
$f$ \textit{is finite iff} $\operatorname{tdeg} f<\infty$.
\end{definition}

Since obviously $\operatorname{deg} \psi_{f} \leq$ tdeg $\psi_{f}$, a finite response map is in particular bounded.

\begin{examples}\label{ex:21.2}
The nonzero polynomial response map $f$ is \textit{linear} if and only if tdeg $f=1$; when $f$ is bilinear, tdeg $f=2$. On the other hand, an internally-bilinear response map $f$ is bounded but not necessarily finite.
\end{examples}

A natural class of realizations for finite maps will consist of feedback-free interconnections of linear systems, defined below:

\begin{definition}\label{def:21.3}
\textit{The polynomial system} $\Sigma$ \textit{is a cascade of linear systems iff} $X=k^{n}, x^{\#}=0$ \textit{and there exists a direct sum decomposition}.
\[
X=X_{1} \oplus \ldots \oplus X_{d}, \quad d \geq 1,
\]
\textit{and linear maps} $A_{i}: X_{i} \rightarrow X_{i}, i=1, \ldots, d$ \textit{such that the transition equations of} $\Sigma$ \textit{become}
\[
\begin{aligned}
& x_{i}(t+1)=A_{i} x_{i}(t)+B_{i}\left(x_{1}(t), \ldots, x_{i-1}(t), u(t)\right), \\
& x_{i}(t) \text { in } X_{i} \text { for all } t, i=1, \ldots, d .
\end{aligned}
\]
Given a decomposition $X=X_{1} \oplus \ldots \oplus X_{d}$ as above, let
\label{op:110}%
\begin{equation}\label{eq:21.4}
Z_{j}:=X_{j} \oplus \ldots \oplus X_{d}, \quad j=1, \ldots, d .
\end{equation}
Conversely, for any chain $Z_{d} \subseteq Z_{d-1} \subseteq \ldots \subseteq Z_{1}=X$, there exists a decomposition $\left\{X_{i}\right\}$ such that \eqref{eq:21.4} holds. Thus the following result is easy to prove.
\end{definition}

\begin{lemma}\label{lem:21.5}
The state-affine system $\Sigma$ is a cascade of linear systems if and only if there exists a chain of subspaces
\[
\{0\}=Z_{d+1} \subseteq Z_{d} \subseteq \ldots \subseteq Z_{1}=X,
\]
such that, with the notations of \eqref{eq:20.4},
\[
\begin{aligned}
& F_{\delta_{0}} Z_{j} \subseteq Z_{j}, \quad j=1, \ldots, d, \text { and } \\
& F_{\delta_{i}} Z_{j} \subseteq Z_{j+1}, \quad i=1, \ldots, s, \quad j=1, \ldots, d .
\end{aligned}
\]
\end{lemma}

For any $\alpha=\alpha_{1} \ldots \alpha_{t}$ in supp $\psi$, let
\[
\#(\alpha):=\text { number of nonzero vectors among } \alpha_{1}, \ldots, \alpha_{t} \text {. }
\]
Clearly,
\[
\#(\alpha) \leq\|\alpha\| \leq d .
\]
The main result of this section is the following

\begin{theorem}\label{thm:21.6}
The following statements are equivalent for any polynomial response map $f$ :
\begin{enumerate}
\item[(a)] $f$ is bounded and the span canonical state-affine realization of $f$ is a cascade of linear systems.
\item[(b)] $f$ is realizable by a cascade of linear systems.
\item[(c)] $f$ is a finitely realizable finite map.
\end{enumerate}
\end{theorem}

\begin{proof}
(a) ⇒ (b) Trivial.

\label{op:111}%
$(b) \Rightarrow(c)$ Let $f=f_{\Sigma}$, with $\Sigma$ as in Definition~\ref{def:21.3}. Let each $G_{i}$ have total degree $s_{i}$ and let $h$ have total degree $s$. It follows by induction on $i$ that
\[
\operatorname{tdeg} \psi_{f_{\Sigma}} \leq s s_{1} \cdots s_{n} .
\]

(c) ⇒ (a) Let $\operatorname{tdeg} \psi_{f}=d$. Referring to the realization constructed via Theorem~\ref{thm:19.3}, let
\[
Z_{j}:=\operatorname{span}\left\{B_{\beta}, \quad \beta \text { in } \Delta_{J}, \quad j \leq \#(\beta)\right\}, \quad j=1, \ldots, d+1 .
\]
The result is then clear by Lemma~\ref{lem:21.5}.
\end{proof}

\subsection{Rationality}\label{sec:22}

We now indicate how the automata-theoretic notion of rationality generalizes to the present context. The convention Remark~\ref{rem:20.2} that $\varphi(\Lambda)=0$ for every exponent series is still in force.

\begin{definition}\label{def:22.1}
\textit{A scalar} $(p=1)$ \textit{exponent series} $\varphi$ \textit{is rational iff} $\varphi$ \textit{can be expressed in terms of finitely many vectors in} $\mathbb{N}^{m}$ \textit{by means of a finite number of any of the following three types of operations}:
\begin{enumerate}
\item[(i)] $(\varphi, \hat{\varphi}) \mapsto r \varphi+s \hat{\varphi}$, for any $r, s$ in $k$,
\item[(ii)] $(\varphi, \hat{\varphi}) \mapsto \varphi \hat{\varphi}$,
\item[(iii)] $\varphi \mapsto \varphi^{*}:=\sum_{i \geq 0} \varphi^{i}$,
\end{enumerate}
\textit{where} $\varphi^{*}$ \textit{is interpreted as}
\[
\varphi^{*}(\alpha):=\sum_{i \leq|\alpha|} \varphi^{i}(\alpha), \text { for all } \alpha \text { in } \Delta_{J} .
\]
\textit{A vector exponent series} $\left(\varphi_{1}, \ldots, \varphi_{p}\right)^{\prime}$ \textit{is rational iff each} $\varphi_{j}$ \textit{is rational}.
\end{definition}

\label{op:112}%

\begin{theorem}[Kleene-Schützenberger]\label{thm:22.2}
The exponent series $\varphi$ is rational if and only if $\varphi$ is recognizable.
\end{theorem}

\begin{proof}
See \citet[Theorem VII.5.1]{eilenberg1974automata}.
\end{proof}

The notion of rationality can be interpreted in terms of a calculus of interconnections of state-affine systems. This calculus permits, via exponent series, the determination of the Volterra series of a given system and, conversely, the construction of a realization given (a rational expression for) a Volterra series. The calculus is obtained as a straightforward translation of the well-known manipulations with automata, as given for instance in \citet{eilenberg1974automata}.

Consider the \textit{truncation} $f^{(d)}$ \textit{of} $f$ \textit{to kernels of degree at most} $d$, i.e., the response whose formal Volterra series has all those terms in $\psi_{f}$ corresponding to $\|\alpha\| \leq d$ and all other coefficients zero. As an illustration of the use of rationality, we shall give a short proof of the following result (variations of it, e.g. that each homogeneous part is realizable, are proved similarly):

\begin{proposition}\label{prop:22.3}
Let $f$ be a bounded finitely realizable response map. Then $f^{(d)}$ is also finitely realizable, for all $d \geq 1$.
\end{proposition}

\begin{proof}
In terms of the corresponding exponent series, $\varphi^{(d)}$ is the Hadamard or coefficientwise product of $\varphi$ and of $S_{d}$, where $S_{d}$ is the series with $a_{\alpha}=1$ iff $\|\alpha\| \leq d$ and zero otherwise. For any $d, S_{d}$ is rational (e.g., for $m=1, d=2, S_{d}=0^{*} 1 0^{*} 1+0^{*} 2+0^{*} 1$). The result is then an immediate consequence of the known fact that the Hadamard product of rational series is again rational (see e.g., \citealp{fliess1972sur}).
\end{proof}

\label{op:113}%
\section{Classes of Quasi-Reachable Realizations}\label{ch:6}
We study in this section the structure of various classes of realizations of a fixed polynomial response map $f$. The goal is to understand better the systems in each such class, as well as their interrelationships. As a corollary, we shall give a stronger version of the isomorphism theorem of canonical realizations Theorem~\ref{thm:11.5}.

Although part of the discussion could proceed in general, we shall restrict attention to quasi-reachable systems. With this restriction, the classes of interest become naturally endowed with a lattice structure, and the treatment is considerably simplified; results for more general realizations can be often obtained by restricting to the closure of the reachable set.

Unless otherwise stated, all systems in this section are quasi-reachable $k$-systems realizing an arbitrary but fixed polynomial response $f$.

\subsection{The Lattice $\mathrm{QR}(f)$}\label{sec:23}

In this section we do not impose any finiteness restrictions on $f$.

\begin{definition}\label{def:23.1}
\textit{We say that} $\Sigma_{1}$ \textit{dominates} $\Sigma_{2}$, \textit{and denote} $\Sigma_{2} \leq \Sigma_{1}$, \textit{if there exists a $k$-system morphism} $T: \Sigma_{1} \rightarrow \Sigma_{2}$.
\end{definition}

The above defines a pre-order among systems, which will become a partial order when isomorphic systems are identified.

\begin{lemma}\label{lem:23.2}
If $T_{i}: \Sigma_{1} \rightarrow \Sigma_{2}, i=1,2$, are morphisms, then $T_{1}=T_{2}$. Furthermore, the $T_{i}$ are dominating.
\end{lemma}

\begin{proof}
Since $\Sigma_{1}$ is quasi-reachable, the abstractly canonical state-space $X_{\mathrm{ac}}$ is dense in $X_{1}$. Thus by argument as in Lemma~\ref{lem:7.7}, $T_{1}=T_{2}$ on $X_{a c}$. The equality follows by continuity. A similar argument proves the last statement.
\end{proof}

\label{op:114}%
By a slight abuse of notation, the same letter will be used for a system and for its isomorphism class. Let $\mathrm{QR}(f)$ denote the set of all isomorphism classes of quasi-reachable realizations of $f$; then $\mathrm{QR}(f)$ inherits the preorder $\leq$; in fact:

\begin{corollary}\label{cor:23.3}
$\mathrm{QR}(f)$ is partially ordered by $\leq$.
\end{corollary}

\begin{proof}
If $T: \Sigma_{1} \rightarrow \Sigma_{2}$ and $S: \Sigma_{2} \rightarrow \Sigma_{1}$ are morphisms, then $T S: \Sigma_{2} \rightarrow \Sigma_{2}$ must be equal to the identity morphism, by Lemma~\ref{lem:23.2}. Similarly, $S T$ is the identity. So $T$ is an isomorphism.
\end{proof}

Recall that $\Sigma_{\text {free }}(f)$ is the system having the input space $\Omega$ as its state-space, and with transitions extending the concatenation operation on input sequences; see Lemma~\ref{lem:6.10}. By Lemma~\ref{lem:8.2}, the reachability map $g: U[z] \rightarrow X_{\Sigma}$ extends to a polynomial map $g^{\Omega}$ from $\Omega$ to $X_{\Sigma}$, for any $k$-system $\Sigma$. If $\Sigma$ realizes $f$, then $g$ induces an abstract-system morphism from the system with $X=U[z], P:=$ concatenation, and $h, x^{\#}$ as in $\Sigma_{\text {free }}(f)$, into $\Sigma$. Thus $g^{\Omega}$ induces a $k$-system morphism from $\Sigma_{\text {free }}(f)$ into $\Sigma$. Since on the other hand, by Lemma~\ref{lem:11.3}, the canonical realization $\Sigma_{f}$ is terminal among quasi-reachable ones, it follows that

\begin{proposition}\label{prop:23.4}
$\Sigma_{\text {free }}(f)$ is the (unique) largest, and $\Sigma_{f}$ the (unique) smallest, element of $\mathrm{QR}(f)$.
\end{proposition}

If $T: \Sigma_{1} \rightarrow \Sigma_{2}$ is a dominating $k$-system morphism, $A(T)$ gives $A\left(X_{2}\right)$ as a subalgebra of $A\left(X_{1}\right)$, with "co-transitions" $A\left(P_{1}\right)$ and "co-output map" $A\left(h_{1}\right)$ extending $A\left(P_{2}\right), A\left(h_{2}\right)$. Conversely, given any $k$-subalgebra $A$ of $A\left(X_{1}\right)$ such that
\begin{equation}\label{eq:23.5}
A \text{ includes } \mathcal{A}_{f}
\end{equation}
(note that $\mathcal{A}_{f}$ is a subalgebra of $A\left(X_{1}\right)$, by the quasi-reachability assumption) and
\begin{equation}\label{eq:23.6}
A\left(P_{1}\right)(A) \text{ is included in } A\left[T_{1}, \ldots, T_{m}\right] ,
\end{equation}
\label{op:115}then the restriction of $A\left(P_{1}\right)$ to $A$, together with the restriction of $A\left(x_{1}^{\#}\right)$ to $A$ and $A\left(h_{1}\right)$ (seen as a homomorphism into $A$), define a system $\Sigma_{2}$ with $A\left(X_{2}\right)=A$ and $\Sigma_{2} \leq \Sigma_{1}$. Furthermore, $A$ determines a \textit{unique} such $\Sigma_{2}$ (up to isomorphism), since $A\left(P_{2}\right), A\left(h_{2}\right), A\left(x_{2}^{\#}\right)$ are given necessarily by the above procedure.

Thus, the (isomorphism classes of) systems less or equal than $\Sigma_{1}$ are in a one-to-one correspondence with the algebras satisfying \eqref{eq:23.5} and \eqref{eq:23.6}. Furthermore, this correspondence preserves orderings, when the subalgebras $A$ are ordered by inclusion. But \eqref{eq:23.5} and \eqref{eq:23.6} are preserved under intersections; similarly, if a family $A_{i}$ satisfies \eqref{eq:23.5} and \eqref{eq:23.6} then the algebra generated by the union of the $A_{i}$ again satisfies these properties. Translating these facts into the partial order for systems, and applying them for $\Sigma_{1}=\Sigma_{\text {free }}(f)$ :

\begin{theorem}\label{thm:23.7}
$\mathrm{QR}(f)$ is a complete lattice.
\end{theorem}

Although the technicalities are very different, the above is formally very similar to the result for linear responses over rings presented in \citet{sontag1978split}.

\subsection{Examples Using the Lattice Construction}\label{sec:24}

The join $\Sigma_{1} \vee \Sigma_{2}$, (corresponding to the algebra generated by $A_{1}$ and $A_{2}$) can be described somewhat more explicitly than above. In fact, $\Sigma_{1} \vee \Sigma_{2}$ is the "fibre product" of $\Sigma_{1}$ and $\Sigma_{2}$, when each system $\Sigma_{i}$ is interpreted as an ordered pair ($\Sigma_{i}, T_{i}$), with $T_{i}: \Sigma_{i} \rightarrow \Sigma_{f}$ the (unique) $k$-system morphism into the canonical realization. In other words, $\Sigma_{1} \vee \Sigma_{2}$ is such that for every pair of $k$-system morphisms $\Sigma \rightarrow \Sigma_{i}$ there is a unique $k$-system morphism $\Sigma \rightarrow \Sigma_{1} \vee \Sigma_{2}$ such that the compositions $\Sigma \rightarrow \Sigma_{1} \vee \Sigma_{2} \rightarrow \Sigma_{i}$ are the original morphisms (the second morphism being the natural one giving the dominance $\Sigma_{1} \vee \Sigma_{2} \geq \Sigma_{i}$). Thus $\Sigma_{1} \vee \Sigma_{2}$ can be explicitly defined as follows: $\Sigma$ has as its state-space a closed subset of $X_{1} \times X_{2}$ :
\label{op:116}%
\begin{equation}\label{eq:24.1}
X=\overline{\left\{\left(g_{1}(\omega), g_{2}(\omega)\right), \omega \text { in } U^{*}\right\}} ,
\end{equation}
and $P\left(\left(x_{1}, x_{2}\right), \omega\right):=\left(P_{1}\left(x_{1}, \omega\right), P_{2}\left(x_{2}, \omega\right)\right)$, initial state $\left(x_{1}^{\#}, x_{2}^{\#}\right)$, and $H\left(x_{1}, x_{2}\right):=h_{1}\left(x_{1}\right)$ (or $\left.h_{2}\left(x_{2}\right)\right)$.

In the following examples initial states are zero, and $U=Y=k$, unless otherwise stated:

\begin{example}\label{ex:24.2}
Let $\Sigma_{1}$ be (with $X=k$)
\[
\begin{aligned}
& x(t+1)=u^{2}(t) \\
& y(t)=x^{3}(t)
\end{aligned}
\]
and $\Sigma_{2}$ be (with $X=k$)
\[
\begin{aligned}
& x(t+1)=u^{3}(t) \\
& y(t)=x^{2}(t)
\end{aligned}
\]
Both realize the same response map $f$ with canonical realization (which is also their meet):
\[
\begin{aligned}
& x(t+1)=u^{6}(t) \\
& y(t)=x(t)
\end{aligned}
\]
Again, $X=k$. Their join is the system whose state-space is the "cusp" $\left\{\left(x_{1}, x_{2}\right)\right.$ in $k^{2}$ with $\left.x_{1}^{3}=x_{2}^{2}\right\}$ and
\[
\begin{aligned}
& x_{1}(t+1)=u^{2}(t) \\
& x_{2}(t+1)=u^{3}(t) \\
& y(t)=x_{1}^{3}(t),
\end{aligned}
\]
\label{op:117}%
which is more complex than the original systems. The form of the join follows from the above remarks. We now prove that $\Sigma_{1} \wedge \Sigma_{2}$ is the same as $\Sigma_{f}$. Since $\Sigma_{1}$ and $\Sigma_{2}$ both have dimension one, the one-step reachability maps $u \mapsto u^{2}$ and $u \mapsto u^{3}$ permit identifying the algebra $A_{1}$ of $\Sigma_{1}$ with $k\left[T^{2}\right]$ and the algebra $A_{2}$ of $\Sigma_{2}$ with $k\left[T^{3}\right]$. Under these identifications the output map $x \mapsto x^{3}$ of $\Sigma_{1}$ (or equivalently, that of $\Sigma_{2}$) dualizes to
\[
A(h): A(Y)=k[L] \rightarrow k[T]: L \mapsto T^{6}
\]
Thus $\mathcal{A}_{f}=k\left[T^{6}\right]$, which is also the intersection of $A_{1}$ and $A_{2}$; thus (translating in terms of systems), $\Sigma_{f}$ is indeed the meet of the $\Sigma_{i}$.
\end{example}

\begin{example}\label{ex:24.3}
Here $\Sigma_{1}$ and $\Sigma_{2}$ have as state-space the closed set consisting of those vectors $\left(x_{1}, x_{2}, x_{3}, x_{4}\right)$ in $k^{4}$ with $x_{1} x_{3}=x_{2}^{2}$, and input set $U=k^{2}$. The equations are, for $\Sigma_{1}$ :
\[
\begin{aligned}
& x_{1}(t+1)=u(t) \\
& x_{2}(t+1)=u(t) v(t) \\
& x_{3}(t+1)=u(t) v(t)^{2} \\
& x_{4}(t+1)=x_{2}(t)+x_{1}(t) x_{2}(t) u(t)+x_{3}(t) v(t) \\
& y(t)=x_{4}(t)
\end{aligned}
\]
and for $\Sigma_{2}$ :
\[
\begin{aligned}
& x_{1}(t+1)=v(t) \\
& x_{2}(t+1)=u(t) v(t) \\
& x_{3}(t+1)=u(t)^{2} v(t) \\
& x_{4}(t+1)=x_{2}(t)+x_{3}(t) u(t)+x_{1}(t) x_{2}(t) v(t) \\
& y(t)=x_{4}(t) .
\end{aligned}
\]
\label{op:118}%
Their meet is the canonical realization $\Sigma_{f}$ with $X_{f}=$ all 4-vectors with $x_{1}^{3}=x_{2} x_{3}$, and:
\[
\begin{aligned}
& x_{1}(t+1)=u(t) v(t) \\
& x_{2}(t+1)=u^{2}(t) v(t) \\
& x_{3}(t+1)=u(t) v(t)^{2} \\
& x_{4}(t+1)=x_{1}(t)+x_{2}(t) u(t)+x_{3}(t) v(t) \\
& y(t)=x_{4}(t) .
\end{aligned}
\]
Here the join turns out to be \textit{simpler} than all of the above: it is the system with $X=k^{3}$ and
\begin{equation}\label{eq:24.4}
\begin{aligned}
& x_{1}(t+1)=u(t) \\
& x_{2}(t+1)=v(t) \\
& x_{3}(t+1)=x_{1}(t) x_{2}(t)+x_{1}(t)^{2} x_{2}(t) u(t)+x_{1}(t) x_{2}(t)^{2} v(t) \\
& y(t)=x_{3}(t) .
\end{aligned}
\end{equation}
Since both $\Sigma_{i}$ are quasi-reachable in two steps, the calculations of $\Sigma_{1} \vee \Sigma_{2}$ and $\Sigma_{1} \wedge \Sigma_{2}$ are straightforward when carried out in $A\left(U^{2}\right)=$ polynomial ring in 4 variables. Alternatively, we may use the above fibre product construction for $\Sigma_{1} \vee \Sigma_{2}$. Its state-space becomes then the variety in $k^{8}$ given by the set of those ($x_{i}$) satisfying $x_{2}=x_{1} x_{5}, x_{3}=x_{8} x_{5}^{2}, x_{4}=x_{8}, x_{6}=x_{1} x_{5}$, and $x_{7}=x_{1}^{2} x_{5}$. Thus the projection $k^{8} \rightarrow k^{3}$ which sends $\left(x_{i}\right)$ onto $\left(x_{1}, x_{5}, x_{4}\right)$ gives an isomorphism between this variety and $k^{3}$, and equations \eqref{eq:24.4} result.
\end{example}

\begin{example}\label{ex:24.5}
We shall consider the following subalgebras of $k\left[\eta_{1}, \eta_{2}\right]$ :
\[
A_{i}:=k\left[\eta_{1}, \eta_{1} \eta_{2}, \eta_{1} \eta_{2}^{2}, \ldots, \eta_{1} \eta_{2}^{i-1}, \eta_{2}^{i}\right] .
\]
\label{op:119}%
Note that $A_{r}$ includes $A_{s}$ whenever $r$ divides $s$. Thus $A_{1}=k\left[\eta_{1}, \eta_{2}\right]$ is the algebra of the quasi-reachable system having $X=k^{2}$, initial state $(1,1)^{\prime}, y(t)=x_{1}(t)-1$, and $P_{1}$ given by
\[
\begin{aligned}
A\left(P_{1}\right) & : k\left[\eta_{1}, \eta_{2}\right] \mapsto k\left[\eta_{1}, \eta_{2}\right][T] \\
& : \eta_{1} \mapsto \eta_{1} \eta_{2}, \quad \eta_{2} \mapsto \eta_{2}(T+1) .
\end{aligned}
\]
Since $\mathcal{A}_{f}=k\left[\eta_{1} \eta_{2}^{t}, t \geq 0\right]$, the algebras $A_{i}$ satisfy the conditions in \eqref{eq:23.5} - \eqref{eq:23.6}. We call $\Sigma_{i}$ the system corresponding to $A_{i}$, and $f$ the common response of all these systems. Note that $\mathcal{A}_{f}$ is the intersection of all the $A_{i}$, or just of $A_{1}, A_{2}, A_{4}, A_{8}, \ldots$. In terms of systems, we have a \textit{chain of polynomial systems}
\begin{equation}\label{eq:24.6}
\Sigma_{1}>\Sigma_{2}>\Sigma_{4}>\Sigma_{8}>\ldots
\end{equation}
whose meet is $\Sigma_{f}$, a nonpolynomial system. Moreover, the systems $\Sigma_{i}$ have an interesting cofinality property in $\mathrm{QR}(f)$ : for any noncanonical $\Sigma$ in $\mathrm{QR}(f)$, either (1) $\Sigma \wedge \Sigma_{1}=\Sigma_{f}$, or (2) there is some $i$ with $\Sigma \wedge \Sigma_{1} \geq \Sigma_{i}$.

Indeed, let $\Sigma$ be noncanonical, $A=A\left(X_{\Sigma}\right), P$ the transition map of $\Sigma$. Since $\mathcal{A}_{f}$ is naturally included in $A$, we may think of $\mathcal{Q}_{f}=k\left(\eta_{1}, \eta_{2}\right)$ as a subfield of the quotient field $Q(A)$. Working in the latter, we have that $A\left(P_{f}\right)$ extends both to the algebra $A_{1}=k\left[\eta_{1}, \eta_{2}\right]$ (as $A\left(P_{1}\right)$) and to the algebra $A$ (as $A(P)$). Since $Q\left(\mathcal{A}_{f}\right)=Q\left(A_{1}\right)$, both $A\left(P_{1}\right)$ and $A(P)$ coincide on $A_{1} \cap A$. If the latter is $\mathcal{A}_{f}$, then (1) holds. Else, since $A_{1} \cap A$ satisfies \eqref{eq:23.6}, we may replace $A$ by its subalgebra $A_{1} \cap A$, and so assume that $A$ is included in $A_{1}= k\left[\eta_{1}, \eta_{2}\right]$. Since $\Sigma$ is noncanonical, $A \neq \mathcal{A}_{f}$. Thus there is some $a$ in $A$ with
\[
a=c\left(\eta_{2}\right)+b
\]
where $c\left(\eta_{2}\right)$ is a polynomial in $\eta_{2}$ alone, of positive degree, and where
\label{op:120}$b$ is in $\mathcal{A}_{f}$. So $c\left(\eta_{2}\right)=a-b$ is also in $A$. Since $A(P)(A)$ is included in $A[T]$, applying $A(P)$ to $c\left(\eta_{2}\right)=\sum d_{i} \eta_{2}^{i}$ results in $\sum d_{i} \eta_{2}^{i} T^{i}$. By the main lemma~\ref{lem:10.7}, all $d_{i} \eta_{2}^{i}$ are in $A$. Thus $A$ contains both $\mathcal{A}_{f}$ and some $\eta_{2}^{i}$, i.e., it contains $\mathcal{A}_{f}\left[\eta_{2}^{i}\right]=A_{i}$, and $\Sigma$ dominates $\Sigma_{i}$, as wanted.
\end{example}

\begin{remark}\label{rem:24.7}
There is an interesting consequence of the above cofinality property. The "categorical" approach frequently suggested for a canonical realization theory consists in defining "canonical" as "final in the category of all 'reachable' realizations" (for a suitable notion of reachability), instead of via a direct definition of observability; see for example \citet{arbib1974machines}. Using "quasi-reachable" as our notion of reachability, this means that the canonical system should be the smallest element in the lattice $\mathrm{QR}(f)$, i.e., this alternative definition would result in the \textit{same} realization $\Sigma_{f}$ constructed in Section~\ref{ch:3}. One could ask, however, whether it is possible to obtain a "canonical realization theory" (in the present sense) in the context of \textit{polynomial} systems, i.e.: is there a polynomial realization which is smallest possible among all polynomial realizations? The above example provides a negative answer to this question: any such realization $\Sigma$ for the above $f$ would be either incomparable to $\Sigma_{1}$ (hence, not less than it) or it would be greater than one of the polynomial systems $\Sigma_{i}$, and hence not minimal either.
\end{remark}

\subsection{Some Relevant Sublattices}\label{sec:25}

The lattice $\mathrm{QR}(f)$ is too "large," in that it contains realizations of arbitrary dimensions. Certain sublattices described below are much more interesting; it is a remarkable fact that there seems to be no way to study any of these lattices without in some way first introducing $\mathrm{QR}(f)$. In this section, $f$ will be assumed to be finitely realizable.

\begin{definition}\label{def:25.1}
$\mathrm{MD}(f)$ \textit{denotes the (isomorphism classes of) minimal-dimensional realizations of} $f$, \textit{viewed as a partially-ordered subset of} $\mathrm{QR}(f)$.
\end{definition}

\label{op:121}%

\begin{theorem}\label{thm:25.2}
$\mathrm{MD}(f)$ is a complete sublattice of $\mathrm{QR}(f)$.
\end{theorem}

\begin{proof}
Minimal realizations correspond to those subalgebras $A$ of the algebra of Volterra series which satisfy \eqref{eq:23.5} and \eqref{eq:23.6} together with the additional condition that $A$ is algebraic over $\mathcal{A}_{f}$. This is again a complete lattice.
\end{proof}

\begin{remark}\label{rem:25.3}
By Lemma~\ref{lem:9.3}, if $\Sigma$ is a realization of dimension $n$ then the $n$-step reachability map is dominating. By the arguments in Proposition~\ref{prop:12.12}, it follows that two minimal realizations are isomorphic if and only if $A\left(g_{n}\right)(A)$ is the same subalgebra of $A\left(U^{n}\right)$ for both of them, where $n=\operatorname{dim} \mathcal{A}_{f}$. This permits calculations to be carried out explicitly, in $A\left(U^{n}\right)$.
\end{remark}

Realizations in $\mathrm{MD}(f)$ are characterized by the fact that their observation fields are algebraic over the canonical observation field $\mathcal{Q}_{f}$. (Note that the natural inclusion of the observation algebra $\mathcal{A}_{f}=A\left(X_{f}\right)$ in $\mathcal{A}(\Sigma)$ extends to an inclusion of $\mathcal{Q}_{f}$ in $\mathcal{Q}(\Sigma)$, for any quasi-reachable realization $\Sigma$). Another important subclass of realizations is:

\begin{definition}\label{def:25.4}
\textit{A realization} $\Sigma$ \textit{of} $f$ \textit{is quasi-canonical iff} $\mathcal{Q}(\Sigma)$ \textit{is equal to} $\mathcal{Q}_{f}$. \textit{The poset of quasi-canonical realizations is} $\mathrm{QC}(f)$.
\end{definition}

A dominating $k$-space morphism $T: X \rightarrow Z$ is birational when $A(Z)$ has the same quotient field as $A(X)$, (identifying via $A(T)$). The meaning of Definition~\ref{def:25.4} will be clarified by the algebraic:

\begin{lemma}\label{lem:25.5}
Let $X, Z$ be almost-varieties, $T: X \rightarrow Z$ dominating. Assume that the field $k$ is algebraically closed and has characteristic zero. Then $T$ is birational if and only if there is a (Zariski) open set $Z_{1}$ in $Z$ such that the fibre $T^{-1}(z)$ has precisely one element, for each $z$ in $Z_{1}$.
\end{lemma}

\begin{proof}
The argument is essentially that in Theorem~\ref{thm:4.6}. By \citet[Section 5.3]{dieudonne1974cours},\label{op:122} the varieties $X_{1}, Z_{1}$ can be chosen to be normal (i.e., $A\left(X_{1}\right), A\left(Z_{1}\right)$ are integrally closed). If $T$ is birational, $n=m$ in Theorem~\ref{thm:4.6} and the restriction map $X_{1} \rightarrow Z_{1}$ is finite and onto; furthermore, $s=$ cardinality of fibres $=1$, by \citet[Prop 5.3.2]{dieudonne1974cours}. Conversely, if fibres have generically a single point then the argument in Theorem~\ref{thm:4.6} proves that $n=m$, so $Q(X)$ is algebraic over $Q(Z)$, with separable degree one; since char $k=0$, they are equal.
\end{proof}

The above is a straightforward generalization of a result well-known for varieties. Since $\Sigma_{f}$ may be nonpolynomial, however, the almost-variety case is needed in order to conclude:

\begin{proposition}\label{prop:25.6}
Let $k$ be as in Lemma~\ref{lem:25.5}. The (quasi-reachable) almost-polynomial system $\Sigma$ is in $\mathrm{QC}(f)$ if and only if there exists an open (hence dense) subset $X_{1}$ of its state-space $X$ such that no two states in $X_{1}$ are indistinguishable.
\end{proposition}

\begin{proof}
Immediate from Lemma~\ref{lem:25.5}, by considering the canonical morphism $T: \Sigma \rightarrow \Sigma_{f}$.
\end{proof}

This justifies the terminology "quasi-canonical" = quasi-reachable plus "quasi-observable" in the above sense. Such systems have been also suggested in the context of minimality of discrete-time nonlinear systems by \citet{pearlman1979canonical} (for bilinear response maps). The "if" in Proposition~\ref{prop:25.6} is not true in general over the reals, but it is valid for restricted kinds of systems (e.g., state-affine).

Reasoning as in previous cases, we can conclude the

\begin{theorem}\label{thm:25.7}
$\mathrm{QC}(f)$ is a complete sublattice of $\mathrm{QR}(f)$.
\end{theorem}

In particular, there exists a \textit{largest} quasi-canonical realization $\Sigma^{f}$. Explicitly, $\Sigma^{f}$ can be obtained by intersecting $\mathcal{Q}_{f}$ with the algebra of Volterra series $\Psi$ (this gives $A^{f}$, the algebra of functions on the state-space $X^{f}$), and restricting the maps defining $\Sigma_{\text {free }}(f)$. That $A^{f}$ indeed satisfies \eqref{eq:23.6} follows from the more general result:

\label{op:123}%

\begin{lemma}\label{lem:25.8}
If $\Sigma_{2} \leq \Sigma_{1}$, then (with the notations in \eqref{eq:23.5} and \eqref{eq:23.6}), $A:=Q\left(A_{2}\right) \cap A_{1}$ satisfies \eqref{eq:23.6}.
\end{lemma}

\begin{proof}
Since $\Sigma_{1}$ is quasi-reachable, its transition map $P$ is dominating; thus $A(P)$ is one-to-one. So $A(P)$ extends to a homomorphism from $Q\left(A_{1}\right)$ into $Q\left(A_{1}\left[T_{1}, \ldots, T_{m}\right]\right)$, which itself restricts to a homomorphism from $Q\left(A_{2}\right)$ into $Q\left(A_{2}\left[T_{1}, \ldots, T_{m}\right]\right)$. Since $A_{2}$ satisfies \eqref{eq:23.6}, the result will follow from
\begin{equation}\label{eq:25.9}
Q\left(A_{2}\left[T_{1}, \ldots, T_{m}\right]\right) \cap\left(A_{1}\left[T_{1}, \ldots, T_{m}\right]\right) \subseteq Q\left(A_{2}\right)\left[T_{1}, \ldots, T_{m}\right] ,
\end{equation}
which is clear.
\end{proof}

The largest quasi-canonical realization $\Sigma^{f}$ is thus obtained using $\Sigma_{1}=\Sigma_{\text {free }}(f)$ and $\Sigma_{2}=\Sigma_{f}$ above.

Restricting even more the observability properties leads to two other classes of realizations:

\begin{definition}\label{def:25.10}
\textit{The sub-poset} $\mathrm{AO}(f)$ [\textit{respectively}, $\mathrm{RD}(f)$] \textit{consists of all realizations which are abstractly observable [respectively, whose reachable states are pairwise distinguishable}].
\end{definition}

Thus, the $\Sigma$ in $\mathrm{RD}(f)$ are those admitting a (necessarily one-to-one) abstract system morphism $\Sigma_{a c}(f) \rightarrow \Sigma$.

\begin{theorem}\label{thm:25.11}
$\mathrm{RD}(f)$ is a complete sublattice and $\mathrm{AO}(f)$ is a complete join-semilattice of $\mathrm{QR}(f)$.
\end{theorem}

\begin{proof}
We shall use the characterization in \eqref{eq:23.5} - \eqref{eq:23.6}. Two states of a system $\Sigma$ are indistinguishable if and only if they are mapped into the same state under the canonical $k$-system morphism $\Sigma \rightarrow \Sigma_{f}$. In terms of the algebra $A$ of $\Sigma$ (seen as a subalgebra of $\Psi$) states are homomorphisms $x: A \rightarrow k$; thus $x_{1}$ is indistinguishable from $x_{2}$ iff $x_{1}$ and $x_{2}$ restrict to the same homomorphism on the subalgebra $\mathcal{A}_{f}$. Thus $\Sigma$ is abstractly observable if and only if equality of $x_{1}$ and $x_{2}$ on $\mathcal{A}_{f}$ implies equality on all of $A$, and $\Sigma$ has its reachable
\label{op:124}states distinguishable if and only if this implication is true for all reachable $x_{1}$ and $x_{2}$.

If $A$ is the algebra generated by subalgebras $A_{i}$, and if $x_{1}$, $x_{2}$ are homomorphisms from $A$ into $k$, then $x_{1}$ and $x_{2}$ are equal on $A$ if and only if their restrictions to each $A_{i}$ are equal. Thus if the $A_{i}$ correspond to abstractly observable systems $\Sigma_{i}$, or to systems in $\mathrm{RD}(f)$, the same is true of $A$ (which corresponds to the join of the $\Sigma_{i}$).

For the closure of $\mathrm{RD}(f)$ under meets, it is enough to remark that, if $\Sigma_{1} \leq \Sigma_{2}$ and $\Sigma_{2}$ is in $\mathrm{RD}(f)$, then $\Sigma_{1}$ is also in $\mathrm{RD}(f)$. Indeed, for $x_{1}$ and $x_{2}$ indistinguishable states in $\Sigma_{1}, T_{1}\left(x_{1}\right)= T_{1}\left(x_{2}\right)$ ($T_{i}$ is here the canonical map $\Sigma_{i} \rightarrow \Sigma_{f}$). If the $x_{i}$ are reachable, $x_{i}=g_{1}\left(\omega_{i}\right)$ for some input sequences $\omega_{i}$. Then $z_{i}=g_{2}\left(\omega_{i}\right)$ are states of $\Sigma_{2}$ mapping onto the $x_{i}$, so $T_{2}\left(z_{1}\right)=T_{2}\left(z_{2}\right)$. Since $\Sigma_{2}$ is in $\mathrm{RD}(f), z_{1}=z_{2}$, so also $x_{1}=x_{2}$ as wanted.
\end{proof}

\begin{remark}\label{rem:25.12}
Closure under joins in $\mathrm{AO}(f)$ proves in particular that there exists (in the ordering of $\mathrm{QR}(f)$) a largest \textit{abstractly observable realization} $\Sigma_{a o}(f)$ of $f$. A "dual" approach to realization theory is that of finding such initial observable realizations, instead of characterizing $\Sigma_{f}$ as a final quasi-reachable realization. The above construction, together with the other results in this work, permit developing such a "dual" realization theory for polynomial response maps. Sometimes $\Sigma_{a o}(f)=\Sigma_{f}$, like for the $f_{o}$ in example~\ref{ex:18.1} (see Lemma~\ref{lem:18.5}), but in general they are different (see example~\ref{ex:25.13} below). Note that $\Sigma_{a o}(f)$ is initial for \textit{all}, not just quasi-reachable, realizations: for any $\Sigma$, one has a composition morphism $\Sigma_{a o}(f) \rightarrow \Sigma_{Q} \rightarrow \Sigma$, where $\Sigma_{Q}$ is the quasi-reachable subsystem of $\Sigma$.
\end{remark}

\begin{example}\label{ex:25.13}
The above proof cannot be used to conclude that $\mathrm{AO}(f)$ is closed under \textit{joins}, since existence of a morphism $\Sigma_{1} \rightarrow \Sigma_{2}$, with $\Sigma_{1}$ in $\mathrm{AO}(f)$, \textit{does not} imply abstract observability of $\Sigma_{2}$. In fact, we now give a family of realizations $\Sigma_{t}$ in $\mathrm{AO}(f)$ whose meet is not in $\mathrm{AO}(f)$. To construct the $\Sigma_{t}$, we begin with the system $\Sigma$ which has
\label{op:125}%
$U=k^{2}, \quad X=k^{2}, \quad Y=k^{2}, \quad x^{\#}=0$, and equations
\[
\begin{aligned}
& x_{1}(t+1)=u_{1}(t), \quad x_{2}(t+1)=u_{2}(t), \\
& y_{1}(t)=x_{1}(t), \quad y_{2}(t)=x_{1}(t)^{2} x_{2}(t) .
\end{aligned}
\]
Let $f$ be the response of this system. Then $\Sigma$ is in $\mathrm{QR}(f)$. We shall work in the algebra $A(X)=k\left[\eta_{1}, \eta_{2}\right]$, and use the characterizations \eqref{eq:23.5}, \eqref{eq:23.6} in order to define the systems $\Sigma_{i}$. As a subalgebra of $A(X)$,
\[
\mathcal{A}_{f}=k\left[\eta_{1}, \eta_{1}^{2} \eta_{2}\right] .
\]
Any subalgebra of $A(X)$ containing $\mathcal{A}_{f}$ satisfies also \eqref{eq:23.6}, since $A(f)(A(X))$ is in fact included in $A(X)\left[T_{1}, T_{2}\right]$. The following are all subalgebras containing $\mathcal{A}_{f}$, for $t=1,2,3, \ldots$ :
\[
\begin{aligned}
& A_{t}:=k\left[\eta_{1}, \eta_{1} \eta_{2}, \eta_{1} \eta_{2}^{t}, \eta_{1} \eta_{2}^{t+1}, \eta_{1} \eta_{2}^{t+2}, \ldots\right], \\
& B:=\bigcap_{t \geq 1} A_{t}=k\left[\eta_{1}, \eta_{1} \eta_{2}\right] .
\end{aligned}
\]
In particular, $B$ corresponds to the system $\Sigma^{\prime}$ with $X=k^{2}, x^{\#}=0$ and
\[
\begin{aligned}
& x_{1}(t+1)=u_{1}(t), \quad x_{2}(t+1)=u_{1}(t) u_{2}(t) \\
& y_{1}(t)=x_{1}(t), \quad y_{2}(t)=x_{1}(t) x_{2}(t) .
\end{aligned}
\]
Thus $\Sigma^{\prime}$ has all pairs of states distinguishable \textit{except} those with $x_{1}=0$; the states in the line ($x_{1}=0$) are in one indistinguishability class. Thus, $\Sigma^{\prime}$ is not observable, and it is the meet of the systems $\Sigma_{t}$ corresponding to the algebras $A_{t}$. We now claim that each $\Sigma_{t}$ is observable. In order to prove this claim, it is enough to prove that the morphisms $\Sigma_{t} \rightarrow \Sigma^{\prime}$ given by the corresponding inclusions are one-to-one and have images which intersect with the unobservable states ($x_{1}=0$) at just one point: $x_{1}=0, x_{2}=0$.
\end{example}

\label{op:126}%
The statements about the morphisms in turn follow from the following fact (when translated into the corresponding algebras): If $x: A_{t} \rightarrow k$ is a homomorphism with $x\left(\eta_{1}\right)=0$, then $x\left(\eta_{1} \eta_{2}^{i}\right)$ is also zero, for $i=1, t, t+1, \ldots$. Indeed, for $t=1$ this statement has been proven in Example~\ref{ex:3.19}; for $t>1,\left(\eta_{1} \eta_{2}\right)^{t}=\left(\eta_{1} \eta_{2}^{t}\right) \eta_{1}^{t-1}$ forces $x\left(\eta_{1} \eta_{2}\right)^{t}=0$, so $x\left(\eta_{1} \eta_{2}\right)=0$, and $\left(\eta_{1} \eta_{2}^{i}\right)^{2}=\left(\eta_{1} \eta_{2}^{2 i}\right) \eta_{1}$ forces $x\left(\eta_{1} \eta_{2}^{i}\right)^{2}=0$, and hence $x\left(\eta_{1} \eta_{2}^{i}\right)=0$, for all $i \geq t$, as wanted.

\subsection{Normal realizations}\label{sec:26}

Recall Lemma~\ref{lem:1.13} that the algebra $A$ is \textit{integral} over the subalgebra $B$ if every element $a$ of $A$ satisfies a monic equation with coefficients in $B$, i.e., $a$ is integral over $B$. The \textit{integral closure} $\bar{B}$ \textit{of} $B$ \textit{in} $A$ is the set of all $a$ in $A$ integral over $B$; $B$ is \textit{integrally closed in} $A$ when $\bar{B}=B$. When $A=Q(B)$, the quotient field of the integral domain $B$, one refers simply to the "integral closure" of $B$, and to "integrally closed" $B$. For example, a unique factorization domain (e.g., $k\left[\eta_{1}, \ldots, \eta_{n}\right]$,) is always integrally closed. The following definition is sufficient for our purposes, but it may be extended to nonirreducible spaces:

\begin{definition}\label{def:26.1}
\textit{An irreducible} $k$-\textit{space} $X$ \textit{is normal iff} $A(X)$ \textit{is integrally closed}.
\end{definition}

In algebraic geometry the notion of normality is closely related to the study of singularities. In fact, for varieties $X$ of dimension one, normality is \textit{equivalent} to the nonexistence of singular points (so, for $k=\mathbb{C}$, to $X$ being a Riemann surface); in general, nonsingularity implies normality, but the converse is only partially true. Since $k\left[\eta_{1}, \eta_{2}, \ldots, \eta_{n}\right]$ is integrally closed, $k^{n}$ is always normal; for an almost-variety we can take the canonical state-space of $f_{0}$ (cf. Example~\ref{ex:18.1}) as an

\begin{example}\label{ex:26.2}
$A=k\left[\eta_{1} \eta_{2}^{t}, t \geq 0\right]$ is integrally closed. Indeed, let $b$ be in the quotient field of $A, Q(A)=k\left(\eta_{1}, \eta_{2}\right)$, and assume that $b$ is integral over $A$. In particular, $b$ is integral over $k\left[\eta_{1}, \eta_{2}\right]$, which
\label{op:127}is integrally closed, so $b$ must belong to the latter. If $b$ is not in $A$, but is a polynomial, it has a term $c \eta_{2}^{r}$, $c$ in $k$, $r>0$. Since $b$ is integral over $A$, there is an equation
\begin{equation}\label{eq:26.3}
b^{n}+a_{n-1} b^{n-1}+\ldots+a_{0}=0 ,
\end{equation}
with the $a_{i}$ in $A$. Specializing $\eta_{1}$ into zero, there results an equation as in \eqref{eq:26.3} with the $a_{i}$ \textit{scalars} and $b\left(0, \eta_{2}\right)$ a polynomial in $\eta_{2}$ of positive degree $r$, which is impossible. Thus $b$ must be in $A$.
\end{example}

\begin{definition}\label{def:26.4}
$\mathrm{NOR}(f)$ \textit{is the subposet of} $\mathrm{QR}(f)$ \textit{consisting of all (quasi-reachable) normal realizations}.
\end{definition}

\begin{lemma}\label{lem:26.5}
$\Sigma_{\text {free }}(f)$ is in $\mathrm{NOR}(f)$
\end{lemma}

\begin{proof}
We must prove that
\[
\Psi=\bigcap_{n \geq 1} k\left[\left[\xi_{n}, \xi_{n+1}, \ldots\right]\right]\left[\xi_{1}, \ldots, \xi_{n-1}\right]
\]
is integrally closed. Since intersections of, and polynomial rings over, integrally closed domains are again so, (see e.g. \citealp[V.I.3, Corollary 2]{bourbaki1972commutative}) the problem reduces to proving that a power series domain in infinitely many variables, with coefficients in a field, is integrally closed. But this latter statement was proved by \citet{cashwell1963formal}.
\end{proof}

\begin{remark}\label{rem:26.6}
In contrast to a full power series ring, $\Psi$ is \textit{not} a unique factorization domain (and is not local, either). Indeed, taking $m=1$ for simplicity, let $\psi$ be the Volterra series whose terms are all those monomials $\xi_{\alpha_{1}} \ldots \xi_{\alpha_{n}}$ having $\alpha_{1}, \ldots, \alpha_{n}$ all distinct and all $\alpha_{j} \geq i$. Then, $\psi_{i}=\left(1+\xi_{i}\right) \psi_{i+1}$. Since $\left(1+\xi_{i}\right)$ is \textit{not} invertible in $\Psi$ (because $1+\xi_{i}+\xi_{i}^{2}+\xi_{i}^{3}+\ldots$ is not a Volterra series), there results a strictly increasing chain
\[
\left(\psi_{1}\right) \subset\left(\psi_{2}\right) \subset\left(\psi_{3}\right) \subset \ldots
\]
\label{op:128}%
of principal ideals; by the criterion in \citet[VII.3.2, Theorem 2]{bourbaki1972commutative}, $\Psi$ is not a unique factorization domain.
\end{remark}

\begin{proposition}\label{prop:26.7}
Let $\Sigma_{1}$ be in $\mathrm{QR}(f)$, with $\Sigma_{1} \leq \Sigma_{2}$ and $\Sigma_{2}$ in $\mathrm{NOR}(f)$. Identify $A_{1}=A\left(X_{1}\right)$ with a subalgebra of $A_{2}=A\left(X_{2}\right)$, and consider the two subalgebras: $A:=$ integral closure of $A_{1}$, and $B:=$ intersection of all those integrally closed subalgebras of $A_{2}$ which satisfy \eqref{eq:23.6} and include $A_{1}$. Then $A=B$.
\end{proposition}

\begin{proof}
Since the elements of $Q\left(A_{1}\right)$ integral over $A_{1}$ must belong to any integrally closed algebra containing $A_{1}$, $A$ is included in $B$. To prove the other inclusion, it will be necessary to establish that $A$ satisfies \eqref{eq:23.6} and it is integrally closed. The latter statement follows from the fact that $A_{2}$ is integrally closed. Consider now the algebra $A(P)(A)$, where $P$ is the transition map of $\Sigma_{2}$. Since $A$ is integral over $A_{1}$ and $A(P)$ is a homomorphism, $A(P)(A)$ has the same quotient field, and is integral over, $A(P)\left(A_{1}\right)$, which is in turn included in $A_{1}\left[T_{1}, \ldots, T_{m}\right]$, and hence in $A\left[T_{1}, \ldots, T_{m}\right]$. Since $A$ is integrally closed, $A\left[T_{1}, \ldots, T_{m}\right]$ also is, so $A(P)(A)$ must be included in $A\left[T_{1}, \ldots, T_{m}\right]$ as wanted.
\end{proof}

\begin{definition}\label{def:26.8}
\textit{In the situation of} Proposition~\ref{prop:26.7}, \textit{the realization corresponding to} $A$ \textit{and} $B$ \textit{is the integral closure of} $\Sigma_{1}$, \textit{denoted} $\bar{\Sigma}_{1}$. \textit{The canonical normal realization is} $\bar{\Sigma}_{f}$.
\end{definition}

\begin{remark}\label{rem:26.9}
The integral closure of any system is well-defined: given any $\Sigma_{1}$, by Lemma~\ref{lem:26.5} the pair ($\left.\Sigma_{1}, \Sigma_{2}=\Sigma_{\text {free }}(f)\right)$ satisfies the hypothesis of Proposition~\ref{prop:26.7}. Further, it is clear from the form of $A$ that the definition of $\bar{\Sigma}_{1}$ is independent of the $\Sigma_{2}$. Note also that from the definition of $B$ it follows that if $\Sigma_{1} \leq \Sigma_{2}$ then $\bar{\Sigma}_{1} \leq \bar{\Sigma}_{2}$ (integral closure is therefore an algebraic closure operator).
\end{remark}

\begin{examples}\label{ex:26.10}
For the response $f_{o}$ in Example~\ref{ex:18.1}, it follows from Example~\ref{ex:26.2} that $\Sigma_{f_{0}}$ is also the canonical normal realization of $f_{0}$. Consider instead the system $\Sigma$ with $U=k, Y=k^{2}$,
\label{op:129}%
\[
X=\left\{\left(x_{1}, x_{2}\right) \text { in } k^{2} \mid x_{1}^{2}=x_{2}^{3}\right\},
\]
initial state zero, and equations
\[
\begin{aligned}
& x_{1}(t+1)=u(t)^{3}, \quad x_{2}(t+1)=u(t)^{2}, \\
& y_{1}(t)=x_{1}(t), \quad y_{2}(t)=x_{2}(t) .
\end{aligned}
\]
Then, $\Sigma$ is not in $\mathrm{NOR}(f)$, because $\eta$ is in the quotient field $Q(\eta)$ of $A(X)=k\left[\eta^{2}, \eta^{3}\right]$ but $\eta$ satisfies the monic equation $z^{3}-\eta^{2}=0$, and is hence integral over $A(X)$. Its normalization $\bar{\Sigma}$ is the system with $X=k$ and
\[
\begin{aligned}
& x(t+1)=u(t), \\
& y_{1}(t)=x(t)^{3}, \quad y_{2}(t)=x(t)^{2},
\end{aligned}
\]
since the algebra $k[\eta]$ of $\bar{\Sigma}$ is the integral closure of $A(X)$. Since $\Sigma$ is canonical, $\bar{\Sigma}$ is $\bar{\Sigma}_{f}$ ($f=$ response of $\Sigma$), and is different from $\Sigma_{f}$. Note that $\Sigma$ had a singularity at the origin, while $\bar{\Sigma}$ has a nonsingular state space.
\end{examples}

\begin{theorem}\label{thm:26.11}
$\mathrm{NOR}(f)$ is a complete lattice.
\end{theorem}

\begin{proof}
It is easy to verify, either directly or using properties of algebraic closure operators, that the meet in $\mathrm{NOR}(f)$ of a family $\left\{\Sigma_{i}\right\}$ is their meet in $\mathrm{QR}(f)$, while their join in $\mathrm{NOR}(f)$ is the integral closure of their join in $\mathrm{QR}(f)$.
\end{proof}

We now turn to proving some variants of the isomorphism theorem~\ref{thm:11.5} and of Lemma~\ref{lem:11.3}. To simplify (but: see Remarks~\ref{rem:26.20} below) we shall assume for the rest of this section that

\begin{quote}
$k$ \textit{is algebraically closed, of characteristic zero}.
\end{quote}

Before proving any results, we need to recall (with some changes in terminology) some well-known definitions and results from algebra.

\label{op:130}%

\begin{definition}\label{def:26.12}
A \textit{polynomial map} $T: X_{1} \rightarrow X_{2}$ \textit{between} $k$-\textit{spaces is one-to-one as schemes iff the following property holds: If} $P_{1}, P_{2}$ \textit{are prime ideals in} $A\left(X_{1}\right)$ \textit{and} $A(T)^{-1}\left(P_{1}\right)=A(T)^{-1}\left(P_{2}\right)$ \textit{then} $P_{1}=P_{2}$.
\end{definition}

Note that when $T$ is dominating and $A\left(X_{2}\right)$ is identified through $A(T)$ with a subalgebra of $A\left(X_{1}\right)$, the property becomes: If $P_{1} \cap A\left(X_{2}\right)= P_{2} \cap A\left(X_{2}\right)$ \textit{then} $P_{1}=P_{2}$.

\begin{remark}\label{rem:26.13}
Since $k$-points correspond to $k$-ideals (which are maximal, hence prime), a $T$ as in Definition~\ref{def:26.12} is necessarily one-to-one in the usual sense. The converse, however, need not hold. For example, let $X_{1}=X\left(k\left[\eta_{1} \eta_{2}^{t}, t \geq 0\right]\right)$ and $X_{2}=X\left(k\left[\eta_{1}, \eta_{1} \eta_{2}\right]\right)$, with $T: X_{1} \rightarrow X_{2}$ the map dual to the inclusion. Then $T$ is one-to-one, as shown in Example~\ref{ex:25.13}. Take now $P_{1}:=$ the ideal of $A\left(X_{1}\right)$ generated by all the monomials $\eta_{1} \eta_{2}^{t}, t \geq 0$. Since $A\left(X_{1}\right) / P_{1}$ is isomorphic to $k, P_{1}$ is prime (and in fact, a $k$-ideal). Let $P_{2}$ be the ideal of $A\left(X_{1}\right)$ generated by $\eta_{1}$ and $\eta_{1} \eta_{2}$; then $A\left(X_{1}\right) / P_{2}$ is isomorphic to $k\left[\eta_{1} \eta_{2}^{t}, t \geq 2\right]$, an integral domain; thus, $P_{2}$ is also a prime ideal, different from $P_{1}$. But $P_{1} \cap A\left(X_{2}\right)=P_{2} \cap A\left(X_{2}\right)$ : this is the $k$-ideal of $A\left(X_{2}\right)$ generated by $\eta_{1}$ and $\eta_{1} \eta_{2}$. Thus $T$ is \textit{not} one-to-one as schemes.
\end{remark}

However, one has the following

\begin{lemma}\label{lem:26.14}
If $T: X_{1} \rightarrow X_{2}$ is one-to-one and $X_{1}, X_{2}$ are varieties, then $T$ is also one-to-one as schemes.
\end{lemma}

\begin{proof}
Prime ideals of $A(X)$ correspond to closed irreducible subsets of $X$ (cf. Lemma~\ref{lem:2.12}). Using Lemma~\ref{lem:3.11}(c), $T$ being one-to-one as schemes becomes: "If $V_{1}, V_{2}$ are irreducible closed subsets of $X$ such that $\overline{T\left(V_{1}\right)}=\overline{T\left(V_{2}\right)}$, then $V_{1}=V_{2}$ ". So assume that $\overline{T\left(V_{1}\right)}=\overline{T\left(V_{2}\right)}=W$. Let $T_{0}$ be the restriction of $T$ to $T^{-1}(W)$. Thus $T_{0}$ is dominating. By Theorem~\ref{thm:3.14}, $T_{0}\left(V_{i}\right)$ contains an open set $W_{i}$. Let $W_{3}:=W_{1} \cap W_{2}$, again open in $X_{2}$. Then $T_{0}^{-1}\left(W_{3}\right)$ is included in $V_{1} \cap V_{2}$, because $T$ is one-to-one. Since $T_{0}^{-1}\left(W_{3}\right)$ is open, hence dense, both $V_{i}=\bar{V}_{i}=T_{0}^{-1}(W)$, so $V_{1}=V_{2}$.
\end{proof}

\label{op:131}%
We shall need a further concept, that of an \textit{open immersion} $T: X_{1} \rightarrow X_{2}$. Its definition cannot be given without introducing the concept of \textit{nonaffine} schemes, which would complicate the exposition at this point; a discussion of immersions can be found in "EGA": \citet[Part 4]{grothendieck1971elements}. For our purposes it will be sufficient, however, to have the following consequence of the definition:

\begin{lemma}\label{lem:26.15}
If $T_{i}: X_{i} \rightarrow X$, $i=1,2$, $T_{2}$ is an open immersion, and $T_{1}\left(X_{1}\right) \subseteq T_{2}\left(X_{2}\right)$, then there exists a (unique) $T: X_{1} \rightarrow X_{2}$ such that $T_{2} \circ T=T_{1}$.
\end{lemma}

As before, we shall say that a $k$-system morphism $T: \Sigma_{1} \rightarrow \Sigma_{2}$ is an open immersion, or one-to-one as schemes, iff the corresponding property holds for the underlying $T: X_{1} \rightarrow X_{2}$.

The following technical result, based on Zariski's Main Theorem, is the key to the isomorphism theory for normal realizations.

\begin{lemma}\label{lem:26.16}
Let $\Sigma_{1}, \Sigma_{2}$ be in $\mathrm{QR}(f)$ and let $T: \Sigma_{1} \rightarrow \Sigma_{2}$ be one-to-one as schemes, with $\Sigma_{1}$ finite-dimensional and $\Sigma_{2}$ normal. Then $T$ is an open immersion.
\end{lemma}

\begin{proof}
By Zariski's Main Theorem (see \citealp[Corollary 18.12.13]{grothendieck1967elements}), $T: X_{1} \rightarrow X_{2}$ factors as $X_{1} \rightarrow Z \rightarrow X_{2}$, with $T^{\prime}: X_{1} \rightarrow Z$ an open immersion, and $T^{\prime \prime}: Z \rightarrow X_{2}$ a finite morphism. Since $\Sigma_{1}$ is finite-dimensional and quasi-reachable, it is almost polynomial; $T$ being dominating, $\Sigma_{2}$ is also almost polynomial. Hence, by Lemma~\ref{lem:25.5}, $T$ is birational. Thus $T^{\prime \prime}$ is also birational, and so (since $X_{2}$ is normal), it is an isomorphism. Thus $T=T^{\prime}$ is an open immersion, as wanted.
\end{proof}

\begin{corollary}\label{cor:26.17}
If $\Sigma$ is a polynomial system in $\mathrm{AO}(f) \cap \mathrm{NOR}(f)$ then the natural morphism $T: \Sigma \rightarrow \bar{\Sigma}_{f}$ is an open immersion.
\end{corollary}

\begin{proof}
Let $\omega_{1}, \ldots, \omega_{r}$ be input sequences such that
\label{op:132}%
\[
\begin{aligned}
H: & X \rightarrow Y \times Y \times \ldots \times Y \quad(r \text { times }) \\
& x \mapsto\left(h^{\omega_{1}}(x), \ldots, h^{\omega_{r}}(x)\right)^{\prime}
\end{aligned}
\]
is one-to-one (see \citealp[Prop. 7.2]{sontag1975discrete}). By Lemma~\ref{lem:26.14}, $H$ is also one-to-one as schemes. If $H_{f}$ is the analogous map for $\bar{\Sigma}_{f}$ (for the same $\omega_{i}$), $H_{f} \circ T=H$. Thus $T$ is also one-to-one as schemes. So Lemma~\ref{lem:26.16} can be applied.
\end{proof}

\begin{corollary}\label{cor:26.18}
Let $\Sigma_{1}, \Sigma_{2}$ be as in Corollary~\ref{cor:26.17}. Assume that $\Sigma_{1}$ is reachable. Then $\Sigma_{2} \leq \Sigma_{1}$.
\end{corollary}

\begin{proof}
Immediate from Lemma~\ref{lem:26.15} and Corollary~\ref{cor:26.17}.
\end{proof}

We can then conclude one of the main results of this section:

\begin{theorem}\label{thm:26.19}
Any two abstractly canonical normal polynomial realizations are isomorphic as $k$-systems.
\end{theorem}

\begin{remarks}[A]\label{rem:26.20}
Analogous results can be derived for an arbitrary field $k$, provided that "abstract observability" be \textit{re-defined}, taking into account points in the "extended" state-space which includes points in the algebraic closure of $k$. For example, the system over the reals $x(t+1)=u(t), y(t)=x^{3}(t)$ is \textit{not} abstractly observable in this restricted sense, because the map $x \rightarrow x^{3}$ is not one-to-one over the complex numbers.
\begin{enumerate}
\item[(b)] In Section~\ref{ch:1}, the first definition proposed for "polynomial systems" was that of a system of simultaneous first-order difference equations, i.e., $X=k^{n}$, thus a polynomial normal system. So Theorem~\ref{thm:26.19} insures that two systems of this type, realizing the same response and both abstractly observable and reachable, are isomorphic via a polynomial coordinate change.
\item[(c)] Restricting to systems with $X=k^{n}$, a rather strong result in fact holds: If $\Sigma_{1} \leq \Sigma_{2}$ and $\Sigma_{2}$ is abstractly observable, then $\Sigma_{1}$ is isomorphic to $\Sigma_{2}$. Indeed, the $T: X_{2} \rightarrow X_{1}$ must be one-to-one, by observability of $\Sigma_{2}$. But a one-to-one polynomial map from $k^{n}$ into
\label{op:133}$k^{n}$ must be onto (see e.g. \citealp[Chapter I]{cherlin1976model}). So $T$ is an isomorphism, by Lemma~\ref{lem:26.16}.
\end{enumerate}
\end{remarks}

\label{op:134}%
\section{Other Topics}\label{ch:7}

We have already seen that the response $f$ of a polynomial system $\Sigma$ does not in general admit a polynomial canonical realization, unless certain restrictions (boundedness, existence of a recursive equation, etc.) are imposed on $f$ (or on $\Sigma$). For the general case, the results in section~\ref{sec:27} will exhibit the canonical realization in terms of locally rational transition and output maps. Section~\ref{sec:28} deals with the nonexistence in general of sets of polynomial representations of "low" dimensions. Generalizations of the present work to the case of nonequilibrium initial states and more general input, state, and output spaces are discussed briefly in Section~\ref{sec:29}, while the last section includes a short discussion of the problem of checking polynomial realizability, as well as other extensions and suggestions for further research.

\subsection{The Canonical State-Space}\label{sec:27}

Before stating the main result of this section, we shall motivate our approach. Unless otherwise stated, $f$ will denote the response of a fixed but arbitrary polynomial system $\Sigma$.

Obtaining rational transitions for $\Sigma_{f}$ is in a sense trivial. Since the observation field $\mathcal{Q}_{f}$ is finitely generated (as a field), and since the algebra homomorphism
\[
A\left(P_{f}\right): \mathcal{A}_{f} \rightarrow \mathcal{A}_{f}\left[T_{1}, \ldots, T_{m}\right]
\]
is one-to-one (because of quasi-reachability), $A\left(P_{f}\right)$ can be uniquely extended to $\mathcal{Q}_{f}$ and is thus completely determined by its action on a set of generators $q_{1}, \ldots, q_{r}$ of $\mathcal{Q}_{f}$. Similarly, $A(h)\left(L_{i}\right)$ is rational in the $q_{i}$ for each generator $L_{i}$ of $A(Y)$. This gives a realization with state-space $k^{r}$ and transition and output maps rational (explicitly, $A\left(P_{f}\right)\left(q_{i}\right)$ gives the $i$-th coordinate of the next state as a rational function of previous state and input). When the field $k$ has characteristic zero, $r$ can be taken as low as $n+1$, $n=$ dimension of $\Sigma_{f}$. The drawback of this simple-minded approach is of course that there is no way
\label{op:135}to guarantee that a state and input configuration will not appear, which is a pole of the corresponding rational functions. Still, it is interesting to note that outputs can be calculated except for a "generic" input sequence (those not in a certain proper algebraic subset), so the response $f$ is \textit{completely determined} from this rational realization. A similar situation occurs with rational difference equations (Theorem~\ref{thm:16.2}) for $f$ : a rather low-order equation expresses future outputs as a rational function of past inputs and outputs; this permits a very efficient calculation for "generic" inputs, and the complete formal Volterra Series for $f$ can still be recovered from the equation (Remark~\ref{rem:16.8}(b), Example~\ref{ex:18.8}).

The problem is much less trivial if one is to explicitly define transitions for every possible state and input. One way to do this is to first define enough rational functions so that their domains of definition cover $X_{R} \times U$ ($X_{R}=$ reachable set), each rational function defined on a variety, and to implement transition and output maps via a series of
\begin{equation}\label{eq:27.1}
\text{ "if } Q_{i}(x, u) \text{ then } R_{i}(x, u) \text{ else" }
\end{equation}
statements, each $Q_{i}$ being a predicate consisting of polynomial equalities and inequalities and each $R_{j}$ a rational function defined at those ($x, u$) for which $Q_{i}(x, u)$ holds. We shall prove in the rest of this section that such a representation indeed exists. The proof rests upon a decomposition of (a large enough subset of) the state-space into (quasi-affine) varieties. An example of such a decomposition is provided by the response $f_{o}$ considered in example~\ref{ex:18.1}. Its canonical state-space can be decomposed into the variety
\[
X_{1}:=\left\{\left(x_{1}, x_{2}, x_{3}\right) \text { in } k^{3} \mid\left(x_{2}+1\right) x_{3}=1\right\}
\]
(this corresponds, via the natural projection $\left(x_{1}, x_{2}, x_{3}\right) \mapsto\left(x_{1}, x_{2}\right)$, to the set $D$ in p. 93) and an extra point (thought of as a variety $X_{2}$ of dimension zero). Thus a state $x$ can be either in $X_{1}$ or in $X_{2}$; if in the latter, $P(x, u)=x$ and $h(x)=-1$; if in $X_{1}$, then $h(x)=x_{2}$ and for transitions $P$ : if $x_{1} x_{2}+x_{1}+x_{2} \neq-1$ then
\label{op:136}%
\[
P(x, u):=\left(x_{1}+u, x_{1} x_{2}+x_{1}+x_{2},\left(x_{1} x_{2}+x_{1}+x_{2}+1\right)^{-1}\right),
\]
else, $P(x, u):=$ the only state in $X_{2}$ (a constant function).

The proofs and statements of the above facts involve algebraic-geometric notions somewhat less elementary than those used in previous sections. We shall not explain these notions in detail, but will give references to the relevant literature. We begin with a

\begin{definition}\label{def:27.2}
\textit{A decomposition of a} $k$-\textit{system} $\Sigma$ \textit{into quasi-affine varieties consists of a set} $Z_{1}, \ldots, Z_{r}$ \textit{of quasi-affine varieties and morphisms} $\varphi_{i}: Z_{i} \rightarrow X$ \textit{such that, denoting} $X_{i}:=\varphi_{i}\left(Z_{i}\right)$ \textit{and} $X_{0}:=$ \textit{union of the} $X_{i}$ :
\begin{enumerate}
\item[(a)] \textit{each} $\varphi_{i}$ \textit{is an immersion},
\item[(b)] $X_{i} \cap X_{j}$ \textit{is empty for} $i \neq j$,
\item[(c)] $P\left(X_{0} \times U\right) \subseteq X_{0}$, \textit{and}
\item[(d)] $x^{\#}$ \textit{is in} $X_{0}$.
\end{enumerate}
\end{definition}

A good reference for the algebraic-geometric concepts used above is \citet{hartshorne1977algebraic}: "morphism" means morphism of schemes, "immersion" means an isomorphism with an open subscheme of a closed subscheme of $X$ (Hartshorne, p. 120), and "quasi-affine variety" means an open subset of an affine variety (Hartshorne, p. 3). Since nonaffine varieties also appear, for the rest of this section the varieties introduced in Section~\ref{ch:2} will be called \textit{affine} varieties.

\begin{remark}\label{rem:27.3}
Given a decomposition as in Definition~\ref{def:27.2}, the (restriction to $X_{0}$ of the) transition and output maps of $\Sigma$ can be defined separately in each $X_{i}$, which is up to isomorphism a quasi-affine variety. For example, $h$ gives rise to $r$ maps $h_{i}=h \mid X_{i}$. Let $\bar{X}_{i}$ be an affine variety of which $X_{i}$ is an open subset. Since each $h_{i}$ is a morphism, it can be represented by a rational function on $\bar{X}_{i}$ which has no poles on $X_{i}$. To define $P$ explicitly, we may proceed as follows, for each $i$.

\label{op:137}%
Since, by Definition~\ref{def:27.2}(c), $P\left(X_{i} \times U\right) \subseteq X_{o}$, there is a covering of $X_{i} \times U$ by subsets $V_{i 1}, \ldots, V_{i r}$ such that $P\left(V_{i j}\right) \subseteq X_{j}$. In fact, letting
\[
V_{i j}:=P^{-1}\left(X_{j}\right) \cap\left(X_{i} \times U\right)
\]
shows that each $V_{i j}$ can be taken to be an open subscheme of a closed subscheme of $X_{i} \times U$. In terms of $X_{i} \times U$, each $V_{i j}$ can be therefore determined by a set of polynomial equalities and inequalities (the $Q_{i}(x, u)$ in \eqref{eq:27.1}), and $P$ restricted to $V_{i j}$ is given by a rational function with no poles in $V_{i j}$. Thus $h$ and $P$ can be indeed defined on $X_{o}$ by programs of the type in \eqref{eq:27.1}, and since by (27.2c, $d$) $X_{o}$ contains all reachable states, this is clearly sufficient in order to simulate $\Sigma_{f}$.
\end{remark}

The following theorem shows that we can always obtain a "stratification" as in Definition~\ref{def:27.2}. It proves a weaker version of a (still open) conjecture of M. Hazewinkel (personal communication) that decompositions always exist with $X_{0}=X_{f}$:

\begin{theorem}\label{thm:27.4}
$\Sigma_{f}$ admits a decomposition into quasi-affine varieties. Moreover, $Z_{1}$ can be taken to be a variety and $X_{1}$ a principal open subset of $X_{f}$. Further, if $T: \Sigma \rightarrow \Sigma_{f}$ is any $k$-system morphism with $\Sigma$ polynomial and if $k$ is algebraically closed, $X_{0}$ can be taken to be the image $T(X)$.
\end{theorem}

We shall first prove a technical

\begin{lemma}\label{lem:27.5}
If $T: X_{1} \rightarrow X_{2}$ is a dominating polynomial map, with $X_{1}$ an irreducible affine variety, then there are closed sets $X_{2}=F_{1}, \ldots, F_{r}$ and principal open sets $D_{1}, \ldots, D_{r}$ such that (i)' $T\left(X_{1}\right)$ is included in the union of the $F_{i} \cap D_{i}$, and (ii) each $F_{i} \cap D_{i}$ is an affine variety, i.e., if $A\left(D_{i}\right)=A_{2}\left[s_{i}^{-1}\right]$ and $F_{i}=V\left(I_{i}\right)$, then $A\left(F_{i} \cap D_{i}\right)= \left(A_{2} / I_{i}\right)\left[\bar{s}_{i}^{-1}\right]$ is finitely generated. (Here $\bar{s}_{i}$ is the coset of $s_{i}$ in $A_{2} / I_{i}$. )
\end{lemma}

\label{op:138}%

\begin{proof}
Using Lemma~\ref{lem:3.18}, there is an $s_{1}:=s$ in $A_{2}$ with $A_{2}\left[s^{-1}\right]$ finitely generated. We let $D_{1}:=X\left(A_{2}\left[s_{1}^{-1}\right]\right)$, and identify $A_{2}$ with a subalgebra of $A_{1}$. Let $s A_{1}$ be the ideal generated by $s$ in $A_{1}$, and let $J_{1}, \ldots, J_{s}$ be the set of prime ideals of $A_{1}$ which are minimal over $s A_{1}$ (finitely many, because $A_{1}$ is Noetherian). Let $x: A_{1} \rightarrow k$ be a homomorphism. If $x(s) \neq 0$, then the restriction of $x$ to $A_{2}$ (i.e., $T(x)$) is in $D_{1}$. If $x(s)=0$ then the kernel of $x$ is a prime ideal containing $s A_{1}$, so it contains some $J_{i}$. Thus $x$ factors through $A_{1} / J_{i}$, and $x \mid A_{2}$ factors through $A_{2} /\left(A_{2} \cap J_{i}\right)$, i.e. $T(x)$ is in the closed subset $V\left(A_{2} \cap J_{i}\right)$ of $X_{2}$. Since $A_{1} / J_{i}$ is again finitely generated and $A_{2} /\left(A_{2} \cap J_{i}\right)$ has less dimension than $A_{2}$, we may assume by induction on $\operatorname{dim} A_{2}$ that the lemma is true for each dominating polynomial map $X\left(A_{1} / J_{i}\right) \rightarrow X\left(A_{2} /\left(A_{2} \cap J_{i}\right)\right)$. Thus for each $V\left(A_{2} \cap J_{i}\right)$ there are open and closed sets as wanted. These give rise in turn to open sets $D_{2}, \ldots, D_{r}$ and closed sets $F_{2}, \ldots, F_{r}$ of $X_{2}$, and properties (i) and (ii) are satisfied. (In fact, property (ii) is true in the sense of schemes, i.e. for the map Spec $A_{1} \rightarrow$ Spec $A_{2}$ corresponding to $T$.)
\end{proof}

\begin{proof}
of Theorem~\ref{thm:27.4}. To apply Lemma~\ref{lem:27.5}, let $T: \Sigma \rightarrow \Sigma_{f}$ be a dominating $k$-system morphism, with $\Sigma$ polynomial. Let the $D_{i}, F_{i}$ be as in Lemma~\ref{lem:27.5}. Defining if necessary new closed sets $F_{1}':=F_{1}, F_{i}':=$ intersection of $F_{i}$ with the complements of $D_{1}, \ldots, D_{i-1}$, the $F_{i} \cap D_{i}$ can be assumed disjoint (the new algebras are quotients of the former ones, so they are still finitely generated). Consider $R_{i}:=T^{-1}\left(F_{i} \cap D_{i}\right)$. These are affine subvarieties of $X$, whose union covers $X$. Thus $T(X)$ is the union of the $T\left(R_{i}\right)$. Assume now that $k$ is algebraically closed. Each $T \mid R_{i}$ maps a variety into a variety, so by Chevalley's theorem~\ref{thm:3.14}(a), each $T\left(R_{i}\right)$ is a finite disjoint union of locally closed sets $V_{i j} \cap S_{i j}$, i.e. sets obtained as intersections of an open set $V_{i j}$ and a closed set $S_{i j}$. Each of the $V_{i j} \cap S_{i j}$ is itself locally closed as a subset of $X_{f}$, since each $F_{i} \cap D_{i}$ is locally closed. Thus each defines a scheme under the induced sheaf, giving rise to the $Z_{i}$ in Definition~\ref{def:27.2} (more precisely, we are restricting to the $k$-points of the corresponding schemes). Since the $F_{i} \cap D_{i}$ are (isomorphic to) affine varieties, each $V_{i j} \cap S_{i j}$ is an
\label{op:139}open subset of the variety $S_{i j} \cap\left(F_{i} \cap D_{i}\right)$, so the $Z_{i}$ are indeed quasi-affine. Since $T$ is a $k$-system morphism, $T\left(X_{1}\right)$ satisfies (c) and (d) of Definition~\ref{def:27.2}, and (a), (b) are valid by construction. The case of non-algebraically closed $k$ follows from the algebraically closed case by consideration of the system $\Sigma$ as a system over the algebraic closure $K$ of $k$, and operating in $K$; the sets $Z_{i}$, $X_{i}$ will then consist of the restriction to the $k$-points of the corresponding sets over $K$.
\end{proof}

\subsection{Unconstrained Realizations}\label{sec:28}

When the canonical realization is polynomial, it admits by definition a representation in terms of polynomial (rather than just rational) difference equations in finitely many variables. It becomes then of interest to ask \textit{how many} equations are needed, i.e., what is the smallest possible cardinality $r=r\left(\mathcal{A}_{f}\right)$ of a set of generators for $\mathcal{A}_{f}$. A lower bound for $r$ is $\operatorname{dim} \Sigma_{f}$, which is attained precisely when $\mathcal{A}_{f}$ is a polynomial ring, i.e. when $X_{f}=k^{r}$. In general we shall call a realization with $X$ an affine space $k^{n}$ an \textit{unconstrained} realization, since no algebraic relations exist between its state variables. A result of \citet{kalman1979realization} (and independently by \citealp{pearlman1979canonical}) states that $\Sigma_{f}$ is unconstrained in the very special case of a bilinear \textit{single}-output response map $f$. We saw in Example~\ref{ex:18.9} that a rather simple $f$, however, may have $r(\mathcal{A}_{f})>\operatorname{dim} \Sigma_{f}$. Counterexamples can also be given with $f$ bilinear with \textit{two} outputs or \textit{tri}linear single-output showing that the above result cannot be extended:

\begin{example}\label{ex:28.1}
Let $f_{1}$ be the response map of the system having $m=p=2, X=k^{3}$, initial state zero and:
\begin{align}
& x_{1}(t+1)=u_{1}(t)+x_{2}(t), \quad x_{2}(t+1)=x_{1}(t), \quad x_{3}(t+1)=u_{2}(t),  \label{eq:28.2}\\
& y_{1}(t)=x_{1}(t) x_{3}(t), \quad y_{2}(t)=x_{2}(t) x_{3}(t) . \notag\end{align}
Let $f_{2}$ be the response of the system having $m=3, p=1, X=k^{4}$,
\label{op:140}initial state zero, and:
\begin{align}
x_{1}(t+1) & =u_{1}(t), \quad x_{2}(t+1)=u_{2}(t), \quad x_{3}(t+1)=u_{3}(t), \notag\\
x_{4}(t+1) & =x_{1}(t) u_{2}(t) u_{3}(t)+x_{2}(t) u_{1}(t) u_{3}(t)+x_{1}(t) x_{3}(t) u_{2}(t)  \label{eq:28.3}\\
& +x_{2}(t) x_{3}(t) u_{1}(t) \notag\\
y(t) & =x_{4}(t) \notag\end{align}
Then $f_{1}$ is bilinear and $f_{2}$ is trilinear. Both systems \eqref{eq:28.2} and \eqref{eq:28.3} are quasi-reachable, so the observation algebras can be calculated directly. They are $k\left[\eta_{1}, \eta_{2}, \eta_{1} \eta_{3}, \eta_{2} \eta_{3}\right]$ and $k\left[\eta_{1}, \eta_{2}, \eta_{1} \eta_{3}, \eta_{2} \eta_{3}, \eta_{4}\right]$ respectively. Neither of these is isomorphic to a polynomial ring. In fact, neither of them is even a UFD (unique factorization domain). Indeed, the equation $\eta_{1}\left(\eta_{2} \eta_{3}\right)=\eta_{2}\left(\eta_{1} \eta_{3}\right)$ shows that $\eta_{1} \eta_{2} \eta_{3}$ can be decomposed in two different ways into irreducibles (note that both $\eta_{1} \eta_{3}$ and $\eta_{2} \eta_{3}$ are indeed irreducible in the corresponding algebras, since $\eta_{3}$ is not there).
\end{example}

A result parallel to the one for bilinear responses was obtained by \citet{gilbert1977minimal}, who proved that in the case of $m=p=1$ and $f$ \textit{homogeneous of degree two}, there is always an unconstrained realization of dimension equal to that of $\Sigma_{f}$. This result is different from the one on bilinear maps: the following example shows that in this case $\mathcal{A}_{f}$ may not be a polynomial ring:

\begin{example}\label{ex:28.4}
Let $f_{3}$ be the response map of the system having $m=p=1, X=k^{4}$, initial state zero, and:
\begin{align}
& x_{1}(t+1)=u(t), \quad x_{2}(t+1)=x_{1}(t), \quad x_{3}(t+1)=x_{2}(t)+x_{3}(t), \notag\\
& x_{4}(t+1)=x_{2}(t) x_{3}(t)+u(t) x_{2}(t)  \label{eq:28.5}\\
& y(t)=x_{4}(t) \notag\end{align}

\label{op:141}%
In other words, $f_{3}$ corresponds to the input-output map
\[
y(t)=u(t-3)(u(t-1)+u(t-4)+u(t-5)+u(t-6)+\ldots) .
\]
(In particular, it is easy to realize $f_{3}$ as a parallel connection of two linear systems whose outputs are multiplied.) The system in \eqref{eq:28.5} is quasi-reachable, since the 4-step reachability map
\[
\left(u_{1}, u_{2}, u_{3}, u_{4}\right) \mapsto\left(u_{4}, u_{3}, u_{1}+u_{2}, u_{2}\left(u_{1}+u_{4}\right)\right)
\]
is dominating (for example, because its Jacobian has full rank at $(1,0,0,0))$. Thus the observation algebra is $k\left[\eta_{1}, \eta_{2}, \eta_{2} \eta_{3}, \eta_{1} \eta_{3}, \eta_{4}\right]$, which is not even a UFD.
\end{example}

Not only does $\mathcal{A}_{f}$ not admit in general a system of $n$ generators, $n=\operatorname{dim} \Sigma_{f}$, but $r\left(\mathcal{A}_{f}\right)$ may in fact be arbitrarily large. Constructing examples of this serves also to illustrate some technical tools of rather general interest, which we shall discuss first.

\begin{lemma}\label{lem:28.6}
Let $\Sigma$ be an algebraically observable realization of $f$, with $X_{\Sigma}=k^{n}$ and initial state zero. Assume that the quasi-reachable set can be defined by equations $Q_{i}(x)=0$ where the $Q_{i}$ have no linear term. Then $r\left(\mathcal{A}_{f}\right)=r$.
\end{lemma}

\begin{proof}
By algebraic observability, $\mathcal{A}_{f}$ is the algebra of the quasi-reachable set $V$. Consider the tangent space $T_{0}(V)$ of $V$ at the origin (note $x^{\#}=0$ is in $V$). This has equations $J_{0} x=0$, where $\left(J_{0}\right)_{i j}=\left(\partial Q_{i} / \partial x_{j}\right)(0)$ is the Jacobian of the $Q_{i}$ at zero (see e.g. \citealp[Chapter VI]{dieudonne1974cours}, or \citealp[Chapter 3]{shafarevich1975basic}). By hypothesis, $J_{0}=0$, so $T_{0}(V)$ has dimension $r$. If $r\left(\mathcal{A}_{f}\right)$ would be less than $r$, there would exist an immersion of $V$ into a space $k^{d}, d<r$. This would imply that all points of $V$ would have tangent spaces of dimension less than $r$, a contradiction. (Note that this uses, implicitly, the invariance of tangent spaces under isomorphism).
\end{proof}

\begin{remark}\label{rem:28.7}
The utility of the above lemma depends on having a fairly
\label{op:142}simple method to find the quasi-reachable set of a polynomial system $\Sigma$. By Corollary~\ref{cor:9.4}, this is equivalent to finding the closure $X_{n}$ of the image of the $n$-step reachability map $g_{n}, n=\operatorname{dim} \Sigma$. When $X=k^{n}$, $g_{n}$ is dominating if and only if its Jacobian is nonzero at some point, as used in the previous example. In general, with $X \subseteq k^{n}, \overline{X}_{n}=V(I)$, where $I$ is the kernel of $A\left(g_{n}\right)$; see Lemma~\ref{lem:3.11}. Finding $I$ involves a classical syzygy problem. The effective decidability of this type of question has been studied; see for instance \citet{seidenberg1971length}, but no \textit{simple} method exists. A heuristic method for obtaining generators for an ideal $J$ with $V(J)=\bar{X}_{n}$ (not necessarily $J=I(X)$, but enough for finding $\overline{X}_{n}$ !), illustrated in Example~\ref{ex:28.8} below, is to find enough elements $R_{1}, \ldots, R_{t}$ in $I$ such that one will be able to prove that every point in some open dense subset of $V\left(\left\{R_{1}, \ldots, R_{t}\right\}\right)$ is in the image of $g_{n}$. This will imply that $\bar{X}_{n}=V\left(\left\{R_{1}, \ldots, R_{t}\right\}\right)$. In fact, in finding input sequences $w$ such that $g_{n}(w)$ equals a given state, it is allowable for this purpose to find inputs with values in the \textit{algebraic closure} $K$ of $k$. Indeed, if a polynomial map $T: k^{r} \rightarrow k^{s}$ has $\overline{T\left(K^{r}\right)}=V$, then $T\left(K^{r}\right) \cap k^{s}$ also has closure $V$. Otherwise, there would exist a polynomial function $Q: k^{s} \rightarrow k$ such that $Q \circ T=0$ on $k^{r}$ but not on $K^{r}$. But the field $k$ being infinite means that $Q \circ T$ can only be zero if it has as a polynomial every coefficient equal to zero, so it cannot be nonzero on $K^{r}$.
\end{remark}

\begin{example}\label{ex:28.8}
Fix $r \geq 3$ and let $f_{4}$ be the response of the system having $m=p=1, X=k^{r}$, initial state zero, and equations
\[
\begin{aligned}
& x_{i}(t+1)=x_{1}^{i-1}(t) u(t), \quad i=1, \ldots, r-1, \\
& x_{r}(t+1)=x_{1}(t) u(t)^{r-1}+x_{2}(t) u(t)^{r-2}+\ldots+x_{r-1}(t) u(t), \\
& y(t)=x_{r}(t) .
\end{aligned}
\]
This system is algebraically observable, because $x_{r}$ is in $\hat{\mathcal{L}}_{0}$ and $x_{1}, \ldots, x_{r-1}$ are in $\hat{\mathcal{L}}_{1}$, using Lemma~\ref{lem:10.6}. Its quasi-reachable set is defined by the equations
\label{op:143}\[
x_{1} x_{3}=x_{2}^{2}, \quad x_{2} x_{4}=x_{3}^{2}, \quad \ldots, \quad x_{r-3} x_{r-1}=x_{r-2}^{2}, \quad \text{and} \quad x_{1} x_{r-1}=x_{2} x_{r-2},
\]
as we shall prove below. By Lemma~\ref{lem:28.6}, $r\left(\mathcal{A}_{f_{4}}\right)=r$, but $\Sigma_{f_{4}}$ has dimension $n=3$ (see below). Thus $r\left(\mathcal{A}_{f}\right)$ \textit{may be arbitrarily larger} than $n$. Note also that $f_{4}$ is homogeneous (of degree $r$).
\end{example}

We now fill in the missing technical facts, using the method in Remark~\ref{rem:28.7}. (This rather easy example could, of course, be solved in many other ways; we shall use it to illustrate the above method, which constructs inputs explicitly.) The t-step reachability map is
\[
g_{t}\left(u_{1}, \ldots, u_{t}\right)=\left(u_{t}, u_{t-1} u_{t}, \ldots, u_{t-1}^{r-2} u_{t}, Z\right)^{\prime}
\]
where
\[
Z=u_{t-1} u_{t}^{r-1}+u_{t-2} u_{t-1} u_{t}^{r-2}+\ldots+u_{t-2}^{r-2} u_{t-1} u_{t},
\]
whenever $t \geq 3$. Thus $X_{R}=X_{3}$, so we shall work with $g_{3}$. The relations $x_{1} x_{3}=x_{2}^{2}$, etc., are easily found. Call $V$ the set of solutions of these equations. Now, given any $\left(x_{1}, \ldots, x_{r}\right)$ in $V$, \textit{if} $x_{2} \neq 0$, then also $x_{1} \neq 0$ and we may define $w:=\left(u_{1}, u_{2}, u_{3}\right)$ with $g(w)=x$ as follows: $u_{2}:=x_{2} x_{1}^{-1}, u_{3}:=x_{1}$, and $u_{1}:=$ any root $u$ of
\[
u_{2} u_{3}^{r-1}+u u_{2} u_{3}+\ldots+u^{r-2} u_{2} u_{3}=0 .
\]
(Since $u_{2} u_{3} \neq 0$, there is always a solution $u$ in the algebraic closure of $k$.) Thus $g(w)=x$, as is easily verified (e.g., $x_{1}=u_{3}$ and $x_{2}=u_{2} x_{1}=u_{2} u_{3}$ by definition of $u_{1}, u_{2}$, and $x_{1} x_{3}=x_{2}^{2}$ implies $u_{3} x_{3}= \left(u_{2} u_{3}\right)^{2}$ so $x_{3}=u_{2}^{2} u_{3}$, etc.) The case $x_{2} \neq 0$ is however generic in $V$ : we shall prove that if $Q$ is a polynomial which is zero on $X_{3}$ then $Q$ is zero on $V$. Indeed, let $Q$ be such a polynomial, and let $x$ be in $V$. If $x_{2} \neq 0, x$ is in $X_{3}$, and there is nothing to prove. If $x_{2}=0$, the above equations imply that $x_{3}=\ldots=x_{r-2}=0$ and either $x_{1}=0$ or $x_{r-1}=0$. We consider first the case $x_{1}=0$. Then $Q\left(x_{1}, \ldots, x_{r}\right)= Q_{1}\left(x_{r-1}, x_{r}\right)$, where $Q_{1}\left(T_{1}, T_{2}\right):=Q\left(0,0, \ldots, 0, T_{1}, T_{2}\right)$. But $Q_{1}$ is identically zero: it is enough to see for this that $Q_{1}$ is constant,
\label{op:144}since $Q(0, \ldots, 0)=0$. The degree of $Q_{1}$ in $T_{2}$ is zero, since $Q$ is zero on $V$ and $u_{1}$ is independent over $u_{2}, u_{3}$. The degree of $Q_{1}$ in $T_{1}$ is also zero, because otherwise $Q=0$ on $V$ would give rise to an equation $\left(u_{2}^{r-2} u_{3}\right)^{s}=$ polynomial in $u_{2}^{r-2} u_{3}$ of degree less than $s$, with coefficients which are themselves polynomials in $u_{2}^{i} u_{3}, i<r-2$, a contradiction (compare terms). For the case $x_{r-1}=0$, just note that $Q\left(u_{3}, u_{2} u_{3}, \ldots\right)=0$ implies (taking $u_{2}=0$) that $Q\left(x_{1}, 0, \ldots, 0\right)=0$. We are only left to prove that $V$ has dimension 3. This follows from the fact that $g_{3}: U^{3} \rightarrow V$ (but not $g_{2}$) is dominating. 

It is natural to ask in general if it is possible to find unconstrained \textit{minimal} realizations, i.e. realizations with $X=k^{n}, n=$ the dimension of the canonical realization. For $f$ homogeneous of degree 2, the above-mentioned result of Gilbert answers this question in a positive way. We show below that this is false in general. Construction of counterexamples is rather easy using variants of the following type of algebraic

\begin{lemma}\label{lem:28.9}
Let $A$ be a subalgebra of a polynomial ring $B=k\left[T_{1}, T_{2}, T_{3}, T_{4}, L_{1}, \ldots, L_{s}\right]$ such that (i) $A$ contains $k\left[T_{1} T_{2}, T_{1} T_{3}, T_{3} T_{4}, T_{2} T_{4}\right]$, and (ii) $A$ is a unique factorization domain. Then $A$ contains $k\left[T_{1}, \ldots, T_{4}\right]$.
\end{lemma}

\begin{proof}
Consider the elements $Q_{1}:=T_{1} T_{2}, Q_{2}:=T_{1} T_{3}, Q_{3}:=T_{3} T_{4}$, and $Q_{4}:=T_{2} T_{4}$, of $A$. Then $Q_{1} Q_{3}=Q_{2} Q_{4}$. Since $A$ is a UFD, there exist elements $a, b, c, d$ of $A$ such that $Q_{1}=a b, Q_{3}=c d, Q_{2}=a c$, and $Q_{4}=b d$. In particular, there is an equation $T_{1} T_{2}=a b$ in $B$. Using that $B$ is a UFD, it is clear that both $a$ and $b$ have zero degree in the $L_{i}$, and in fact one can assume that $a$ and $b$ are both monic monomials in $T_{1}$ and $T_{2}$. Thus there are four possibilities: $a=T_{2}$ and $b=T_{1}$, or $a=T_{1} T_{2}$ and $b=1$, or $a=1$ and $b=T_{1} T_{2}$, or $a=T_{1}$ and $b=T_{2}$. If the first or the second possibility hold, then $a c$ is divisible (in $B$) by $T_{2}$, contradicting the fact that $a c= Q_{2}=T_{1} T_{3}$. Analogously, the third contradicts $b d=Q_{4}=T_{2} T_{4}$. Thus $a=T_{1}$ and $b=T_{2}$ are both in $A$. From $T_{1} c=a c=T_{1} T_{3}$ it follows that $c=T_{3}$, and from $T_{2} d=b d=T_{2} T_{4}$ it follows that $d=T_{4}$. So

\label{op:145}%
$A$ contains all the $T_{i}$, as wanted.
\end{proof}

\begin{example}\label{ex:28.10}
Let $f_{5}$ be the response of the system having $m=4$, $p=1, \quad X=k^{5}$, initial state zero, and equations
\begin{align}
x_{1}(t+1)= & u_{1}(t) u_{2}(t), \quad x_{2}(t+1)=u_{1}(t) u_{3}(t), \notag\\
x_{3}(t+1)= & u_{3}(t) u_{4}(t), \quad x_{4}(t+1)=u_{2}(t) u_{4}(t)  \label{eq:28.11}\\
x_{5}(t+1)= & u_{3}(t) u_{4}(t) x_{1}(t)+u_{2}(t) u_{4}(t) x_{2}(t)+u_{1}(t) u_{2}(t) x_{3}(t)+ \notag\\
& u_{1}(t) u_{3}(t) x_{4}(t) \notag\\
y(t)= & x_{5}(t) . \notag\end{align}
This system is algebraically observable, and $X_{Q}$ is quasi-reachable in two steps (the quasi-reachable set $X_{Q}$ is 4-dimensional, and has equations $x_{1} x_{3}=x_{2} x_{4}$). The dual of the 2-step reachability map identifies the observation algebra with the subalgebra
\[
k\left[T_{1} T_{2}, T_{1} T_{3}, T_{3} T_{4}, T_{2} T_{4}, L\right]
\]
of
\[
B=k\left[T_{1}, T_{2}, T_{3}, T_{4}, L_{1}, L_{2}, L_{3}, L_{4}\right]
\]
where
\[
L=T_{3} T_{4} L_{1}+T_{2} T_{4} L_{2}+T_{1} T_{2} L_{3}+T_{1} T_{3} L_{4} .
\]
(Here $A\left(U^{2}\right)$ is a polynomial ring in 8 variables, identified with $B$.) Assume that there \textit{would} exist an unconstrained realization $\Sigma$ of $f_{5}$ of dimension 4. By Lemma~\ref{lem:11.3}, there is a dominating $k$-system morphism $T: \Sigma \rightarrow \Sigma_{f}$. Let $g_{2}$ be the 2-step reachability map of $\Sigma$. Since $T \circ g_{2}$ is the (dominating) 2-step reachability map of the 4-dimensional system $\Sigma_{f}$, it follows that $g_{2}$ is dominating; (otherwise, $\operatorname{dim} \overline{g_{2}\left(U^{2}\right)}$ is

\label{op:146}%
3 or less, contradicting $\left.\operatorname{dim} X_{Q}=4\right)$. Thus $A\left(g_{2}\right)$ identifies $A\left(X_{\Sigma}\right)$ with a subalgebra $A$ of $B$, such that $A$ satisfies the conditions in Lemma~\ref{lem:28.9}. It follows that $A$ contains $T_{1}, T_{2}, T_{3}, T_{4}$, and $L$. Since the latter is algebraically independent over the $T_{i}, A$ would have transcendence degree at least 5. But $A$ is isomorphic to $A\left(X_{\Sigma}\right)$, (a polynomial ring in 4 variables,) contradicting this latter fact.
\end{example}

\subsection{Generalizations}\label{sec:29}

The material in previous sections can be easily generalized in various directions. In particular, we shall lift here the restriction to shift-invariant input/output maps (and the corresponding equilibrium initial-state assumption for systems), without changing the nature of the results. Similarly, the input and output-value spaces $U$ and $Y$ will be allowed to be arbitrary $k$-spaces rather than $k^{m}$ and $k^{p}$; as explained in the introduction, this permits the incorporation of various constraints into the model. We shall only sketch proofs, since these are analogous to those for the particular case already treated.

The definition of polynomial response map can be given either in terms of formal Volterra series, or simply considering the polynomial maps $\Omega \rightarrow Y$. We shall use here the latter style of definition (but: see example~\ref{ex:29.11}), for which we must first introduce a suitable input space. The motivation for the construction of $\Omega$ was the need for a "completion" of $U[z]$, the latter being obtained from the set of all sequences $U^{*}$ by identifying $\left(u_{t}, \ldots u_{1}\right)$ with $\left(0, u_{t}, \ldots, u_{1}\right)$. An arbitrary polynomial response $U^{*} \rightarrow Y$ does not necessarily factor through $U[z]$, since no shift-invariance property insures that $f_{t}\left(u_{t}, \ldots, u_{1}\right)=f_{t+1}\left(0, u_{t}, \ldots, u_{1}\right)$. We shall define now a $k$-space $\Omega^{\prime}$ as a completion of $U^{*}$ itself.

For the rest of this section, $U=X(C)$ and $Y$ will denote arbitrary but fixed $k$-spaces. We also use the notations $C_{n}:=A\left(U^{n}\right)=C \otimes \ldots \otimes C$ ($n$ times), and $C_{*}:=$ product of the $C_{n}$, for $n \geq 0$ (note that $C_{0}= X(k)$ is just a point).

\label{op:147}%
There are canonical projections $C_{*} \rightarrow C_{n}$, which give rise to homomorphisms $C_{*} \otimes C \rightarrow C_{n} \otimes C$. These induce therefore a homomorphism
\begin{equation}\label{eq:29.1}
\alpha: C_{*} \otimes C \rightarrow \prod_{n \geq 0}\left(C_{n} \otimes C\right)=\prod_{n \geq 1} C_{n} .
\end{equation}

We introduce also a sequence of subalgebras $C_{(i)}$ of $C_{*}$ defined recursively by: $C_{(0)}:=C_{*}$, and
\begin{equation}\label{eq:29.2}
C_{(n)}:=(1 \times \alpha)\left(k \times\left(C_{(n-1)} \otimes C\right)\right) ,
\end{equation}
and denote by $C_{\infty}$ the intersection of all the $C_{(n)}$. It is easy to prove then that
\begin{equation}\label{eq:29.3}
C_{\infty}=(1 \times \alpha)\left(k \times\left(C_{\infty} \otimes C\right)\right) .
\end{equation}

Thus, restricting the homomorphism $(1 \times \alpha)$ to the subalgebra $k \times\left(C_{\infty} \otimes C\right)$, we can define
\begin{equation}\label{eq:29.4}
\beta:=\mathrm{pr}_{2} \circ(1 \times \alpha)^{-1}: C_{\infty} \rightarrow C_{\infty} \otimes C .
\end{equation}
We denote
\begin{equation}\label{eq:29.5}
\delta^{\prime}:=X(\beta)
\end{equation}
and
\begin{equation}\label{eq:29.6}
\Omega^{\prime}:=X\left(C_{\infty}\right) .
\end{equation}

The projections $C_{*} \rightarrow C_{n}$ restrict to (onto) homomorphisms $\gamma_{n}: C_{\infty} \rightarrow C_{n}$, which dualize to closed immersions $U^{n} \rightarrow \Omega^{\prime}$. Identifying through these inclusions the sets of input sequences $U^{n}$ with subspaces of $\Omega^{\prime}$, it can be proved as in Lemma~\ref{lem:6.10} that
\begin{equation}\label{eq:29.7}
\delta^{\prime}: \Omega^{\prime} \times U \rightarrow \Omega^{\prime}
\end{equation}
indeed extends the concatenation maps. Further,
\label{op:148}%
\begin{equation}\label{eq:29.8}
\operatorname{ker} \gamma_{i}+\operatorname{ker} \gamma_{j}=C_{\infty}
\end{equation}
whenever $i \neq j$, since the identity ($1,1,1, \ldots$) of $C_{\infty}$ can be written as $(1,1, \ldots, 1,0,1, \ldots)$ (a zero in the i-th position, thus in $\operatorname{ker} \gamma_{i}$) added to $(0, \ldots, 0,1,0, \ldots)$ ($a$ one only in the $i$-th position, thus in ker $\gamma_{j}$). The image of $U^{i}$ in $\Omega^{\prime}$ is thus disjoint with the image of $U^{j}$, and there results a canonical inclusion
\begin{equation}\label{eq:29.9}
U^{*} \rightarrow \Omega^{\prime} .
\end{equation}

The image of this map is dense, since the intersection of all the ker $\gamma_{n}$ is zero. Thus a polynomial map with domain $\Omega^{\prime}$ is completely determined by its restriction to $U^{*}$. This motivates the

\begin{definition}\label{def:29.10}
\textit{A generalized polynomial response map is a polynomial map} $f: \Omega^{\prime} \rightarrow Y$.
\end{definition}

Thus a generalized polynomial response map is a map $f: U^{*} \rightarrow Y$ which satisfies certain additional properties (namely, those that imply the existence of an extension to $\Omega^{\prime}$). The most important of these properties (obviously implied by Definition~\ref{def:29.10}) is that the restriction $f_{t}$ to each $U^{t}$ be a polynomial map. To obtain a useful characterization, one needs to make further assumptions on the input-value set $U$. When $U$ is an affine space $k^{m}$, the only further property needed is that the degree of each $f_{t}$ in the last $r$ inputs be bounded independently of $t$, for any $r$; this is shown below for $m=1$, but basically the same proof is valid in general. When $U$ is a variety, the statement of the characterizations is somewhat more complicated (a representation in terms of actual polynomials must be chosen for each polynomial map $f_{t}$), but again it is essentially the same as in the

\begin{example}\label{ex:29.11}
Let $U=k$. Then each $C_{n}$ is a polynomial ring $k\left[\xi_{1}, \ldots, \xi_{n}\right]$, and $\alpha$ is the linear extension of
\begin{equation}\label{eq:29.12}
\alpha\left(\left\{Q_{i}\right\} \otimes Q(\xi)\right):=\left\{Q_{i}^{\prime}\right\} ,
\end{equation}
\label{op:149}where
\begin{equation}
Q_{n+1}^{\prime}\left(\xi_{1}, \ldots, \xi_{n+1}\right)=Q_{n}\left(\xi_{2}, \ldots, \xi_{n+1}\right) Q\left(\xi_{1}\right) . \label{eq:29.13}
\end{equation}
Thus $C_{(1)}$ is the set of all sequences of polynomials in $0,1,2, \ldots$ variables such that the degree in $\xi_{1}$ is bounded. Iterating, $C_{(r)}$ is the set of all sequences with the degree in $\xi_{1}, \ldots, \xi_{r}$ bounded. Thus a polynomial function $\Omega^{\prime} \rightarrow k$, i.e. an element of $A\left(\Omega^{\prime}\right)=C_{\infty}$, is a sequence of polynomial functions $f_{t}: U^{t} \rightarrow k$ such that the degree of $f_{t}$ in $\xi_{1}, \ldots, \xi_{r}$ is bounded for each $r$ (independently of $t$). So when $Y=k^{p}$, a generalized polynomial response map $f$ corresponds to a set of $p$ polynomial functions on $\Omega^{\prime}$ subject to the above restriction. It is trivial to verify that when $f$ is shift-invariant, i.e. $f(0, w)= f(w)$, this definition coincides with the one in Section~\ref{ch:3}.
\end{example}

The definition of a \textit{generalized} $k$-\textit{system} only differs from Definition~\ref{def:8.1} in that the initial state is not required to satisfy $P\left(x^{\#}, 0\right)=x^{\#}$. As before, there are $t$-step reachability maps $g_{t}: U^{t} \rightarrow X$ and a reachability map $g: U^{*} \rightarrow X$ which extends to a polynomial map $g^{\prime}: \Omega^{\prime} \rightarrow X$. The \textit{response map of} $\Sigma$ is $f_{\Sigma}:=h \circ g^{\prime}$. Defining now $k$-system morphisms as before, and canonical $:=g^{\prime}$ dominating + algebraic observability, one concludes in analogy with previous results:

\begin{theorem}\label{thm:29.14}
Any generalized polynomial response has a canonical generalized $k$-system realization, unique up to isomorphism.
\end{theorem}

\begin{example}\label{ex:29.15}
Consider the generalized polynomial system $\Sigma$ with $X=k^{3}, \quad U=Y=k$, initial state $(1,0,0)^{\prime}$ and equations
\begin{align}
& x_{1}(t+1)=x_{2}(t) u(t) \notag\\
& x_{2}(t+1)=x_{1}(t) u(t) \notag\\
& x_{3}(t+1)=x_{3}(t)+u(t)  \label{eq:29.16}\\
& y(t)=x_{1}(t) . \notag\end{align}

\label{op:150}%
Calculating first $\Sigma^{\text {obs }}$ results in dropping the third coordinate (since only $x_{1}, x_{2}$ are observable); the corresponding canonical realization is thus obtained by restricting to the quasi-reachable set, which consists of the union of the line $x_{1}=0$ and the line $x_{2}=0$. Note that this is a polynomial system, but the canonical state-space is \textit{not} irreducible as in the equilibrium initial state case.
\end{example}

\begin{remark}\label{rem:29.17}
The above results properly generalize those in Section~\ref{ch:3}: it is not hard to prove that if $f$ is a (nongeneralized) polynomial response map, then any abstractly observable realization (and hence in particular its generalized canonical realization) has equilibrium initial state; thus the latter coincides with $\Sigma_{f}$. We shall not pursue here extensions of the finiteness results or of those on input/output equations. It is clear that further restrictions must be placed on $U$ and $Y$ in order to render these problems meaningful. Under reasonable hypothesis (e.g., $U$, $Y$ varieties), generalizations do exist and are rather straightforward.
\end{remark}

\begin{remarks}[A]\label{rem:29.18}
As in section~\ref{ch:5}, an algorithmic, matrix-theoretic realization theory for (generalized) bounded response maps, via (generalized) state-affine systems, is easy to give. This is done in detail in \citet{sontag1979realization}. In fact, even more general (nonpolynomial, e.g. piecewise linear) response maps are treated there using essentially the same methods. (The only result that fails to generalize to non shift-invariant maps is the implication "finite realizability implies state-affine realizability". Counterexamples are given in the above reference.)
\begin{enumerate}
\item[(b)] Non \textit{strictly} causal responses (present output may depend on present input) and corresponding "Mealy-machine" realizations (output $y(t)$ is function of present state and input) can be also treated in a totally analogous way.
\item[(c)] A much less trivial extension of the present setup consists in allowing for \textit{nonaffine} schemes as input, state, and output spaces. While the practical significance of rather abstract schemes is at best
\label{op:151}doubtful, it is of interest to consider quasi-affine varieties, allowing for locally rational transitions and outputs. (In fact, quasi-affine varieties appear naturally among all varieties when abstract observability is considered: as a consequence of Zariski's Main Theorem, the state-space of an observable system, with $Y=k^{p}$, is necessarily quasi-affine.) It is interesting to remark, however, that no \textit{new} generalization of polynomial response maps appears if nonaffine state-spaces are allowed, \textit{provided that} $U$ and $Y$ remain \textit{affine}. Indeed, the "affinization" functor $X \rightarrow X^{0}$ (\citealp[9.1.21]{grothendieck1967elements}) maps any such more general realization into another one with an affine state-space, so the response of both systems must be polynomial. (The existence and uniqueness theorem for canonical realizations appears to extend with no difficulty to the case of nonaffine $U, Y$, but other results are not so straightforward.)
\end{enumerate}
\end{remarks}

\subsection{Suggestions for Further Research}\label{sec:30}

Research in a new field is bound to suggest a wealth of open questions and new directions of investigation. In attacking the realization theory of nonlinear systems, the present work is no exception to that hope.

One of the byproducts of an \textit{algebraic} study of systems is of course the development of \textit{algorithms} for system analysis and design. In the case of bounded maps and state affine systems, we use linear-algebraic techniques in constructing canonical realizations; these methods are a rather simple generalization of the classical Hankel matrix technique used so successfully in linear system theory. Finite dimensionality of the observation space is responsible for the linear-algebraic character of the study of bounded maps. This means that a nonlinear computational technique is indispensable as soon as nonlinear feedback is present in a system. An important question is, then: How effective are calculations with fields, algebras and polynomials?

From its origins until (historically) not long ago, algebra remained to a great extent a computational discipline. The development of "modern"
\label{op:152}algebra [or the modern development of algebra] has shifted the emphasis towards generality and abstraction, permitting both the solution of heretofore unsolvable problems and the understanding of deep questions which can only now be even formulated in a rigorous way. Many questions of effective calculation have thus been left aside of the mainstream of algebra; a development which is particularly unfortunate in view of the advent of the digital computer. However, there are now signs of a trend toward the effectivization of various basic algebraic constructions. Some of these constructions can be used to solve system-theoretic questions. For instance, \citet{seidenberg1971length} has worked on effective versions of Hilbert's Basis Theorem, and his results find an immediate application to questions of observability (\citealp{sontag1975discrete}). The posthumous work of \citet{robinson1975algorithms} (see also \citealp{crossley1976effective}) represents a promising approach to questions of computability in algebra, attacking such questions from the point of view of mathematical logic (model theory), but most of the detailed work remains to be done.

Of course, there is a large number of classical results, dealing with resultants and derivatives, and sometimes referred to by the label \textit{elimination theory}, which permit the effective verification of certain conditions; our Jacobian criterion for finite-dimensional realizability involves a simple application of such results. It would certainly be of interest to explicitly compute the form of similar criteria for other problems.

Many theoretical algebraic problems are also suggested by the present work. For instance, as a rule system-theoretic questions depend for their clarification upon the development of a \textit{real} (as opposed to complex) algebraic geometry. The papers of \citet{whitney1957elementary} and of \citet{dubois1970algebraic} are among the few works in this area. In fact, the study of points in more arbitrary fields (e.g., the rational numbers) is needed from our viewpoint. For instance, the question of the validity over $k \neq$ reals or $k$ not algebraically closed, of the theorem: "A realization is minimal if and only if it is weakly canonical", leads to (unsolved) problems of an arithmetic, rather than geometric, nature.

\label{op:153}%
There are many problems which probably do not require finding new algebraic results. A typical open question of this type is: If $k$ is algebraically closed, is the reachable part of every polynomial system (with equilibrium initial state) actually reachable in \textit{bounded} time?

Essentially nothing has been said about the application of linear methods to the \textit{local} study (when $k=\mathbb{R}$ or $\mathbb{C}$) of nonlinear systems. The connection with realization theory is given by the (easily proved) fact that an unconstrained realization $\Sigma$ is in $\mathrm{MD}\left(f_{\Sigma}\right)$ iff it is locally canonical at some state (i.e., a strong neighborhood of some $x$ is reachable and observable), which in turn follows from the linearized system being canonical. It is as yet unclear whether this fact can be used in the construction of minimal realizations.

Another topic we have not treated is that of giving methods for deciding if a finitely realizable $f$ with $\Sigma_{f}$ nonpolynomial admits \textit{some} polynomial realization. In a sense this problem has an easy solution, say for $k=\mathbb{R}$, as follows. For any integers $n, d$, and $t$, there is a predicate $I_{n, d, t}$ consisting of polynomial equalities and inequalities such that, given any polynomial map $F: U^{t} \rightarrow \mathbb{R}, F$ is equal to $\left(f_{\Sigma}\right)_{t}$ for some unconstrained polynomial system $\Sigma$ of dimension $n$ and \textit{degree} at most $d$ (i.e., all polynomials appearing in the definition of $\Sigma$ have degree $\leq d)$ if and only if the condition $I_{n, d, t}$ is satisfied by the coefficients of $F$. This is an easy consequence of the Tarski-Seidenberg decision method for elementary geometry (cf. Theorem~\ref{thm:3.14}), seeing the coefficients of $\Sigma$ as indeterminates. (In the linear case $d=1$, for example, the predicates correspond to the requirement that all the $n$-minors of the $t$-th Hankel matrix be zero.) On the other hand, if $\Sigma_{f}$ is determined (e.g., via standard Jacobian arguments) to have dimension $r$, then $I_{n, d, 2 n}$ (with $n \geq r$) being satisfied for $f_{2 n}$ is \textit{equivalent} to $f$ itself having a polynomial realization of dimension $n$ and degree $\leq d$. Indeed, if $\Sigma$ partially realizes $f_{2 n}$, it follows from \citet[Theorem 6.1]{sontag1975discrete} that $\Sigma$ and $\Sigma_{f}$ have the same response, i.e. $f_{\Sigma}=f$. One should note that, although the Tarski-Seidenberg methods are of impractically high computational complexity, the
\label{op:154}implementation of the above procedure really relies on an \textit{a priori} calculation of the sequence of predicates $I_{n, d, t}$, independently of the particular problem. Thus one could foresee a set of tables being published listing the $I_{n, d, t}$. The compilation of the explicit formulas for these predicates would be a worthwhile project in itself.

When $f$ is bounded there is an explicit algorithm available for realization, as explained in Section~\ref{ch:5}. We have not included any discussion on numerical questions. In fact, the algorithm as presented is numerically unstable. It appears to be not at all difficult, however, to modify this algorithm in order to obtain a numerically stable one (at the cost of needing a slightly higher number of algebraic operations). This modification should be a direct analogue of that recently introduced by \citet{dejong1978numerical} to the corresponding linear system algorithms.

Various questions can be raised, however, regarding the suitability of a state-affine realization theory in the bounded case. Although boundedness implies state-affine realizability, lower-dimensional representations will in general result when more general classes of systems are considered. A trade-off between dimensionality and complexity of the defining maps is often involved. State-affine realizations have an obvious advantage from an analysis viewpoint; from a control-theoretic standpoint, however, they don't have desirable controllability properties. It is interesting to speculate on the impact of microprocessor technology, rendering attractive the idea of a parallel multiprocessor configuration calculating each state-variable via simple functions, as with state-affine systems.

Topological questions have been almost by definition omitted. There is a great number of such questions which are however of interest in realization theory. For example, questions of genericity and approximation: what type of observation algebras appear generically?; in what sense can be a finitely realizable $f$ be approximated by another $f$ with 'nice' $\mathcal{A}_{f}$?, etc. This area is almost completely open.

It is interesting to note that with $k$ finite \textit{every} response is polynomial, in fact bounded. This suggests applying the methods in this
\label{op:155}work, (modifying "polynomial" into "polynomial function") to the state-assignment problem for automata; this hasn't been tried yet. Another generalization deals with $k$ being a ring (e.g., the integers); preliminary results applicable to the internally-bilinear case are given by \citet{fliess1974matrices} and \citet{sontag1977anneaux}. Related to this point and previous ones, the effect of finite arithmetic is totally unexplored.

Perhaps one of the most interesting open problems is that of understanding the relationships between the discrete-time theory pursued here and the continuous-time theory developed by \citet{brockett1975volterra}, \citet{sussmann1976existence}, \citet{hermann1977nonlinear}, \citet{crouch1977finite}, and others. The results in the two theories have a few superficial similarities (e.g., the finiteness properties of the Lie algebra of a system have their parallels in properties of the observation space), but the tools and results are in general very different, due mainly to the nonreversibility of difference (as opposed to differential) equations (so that semigroups appear where groups appear in the continuous-time theory), and to the different algebraic properties of difference and differential operators. For example, the recent result of \citet{crouch1977finite} that a "finite" continuous-time map has its canonical state-space unconstrained is far from being true in the present context (cf. section~\ref{sec:28}).

In so far as we have attacked the realization problem using methods not standard in system theory, there arises the possibility of applying the same methods to the study of other system-theoretic questions. Two examples of this are the results in \citet{sontag1975discrete}, and a result stating that a generic input sequence is sufficient for the identification of a family of polynomial systems, proved in \citet{sontag1979observability}. Some parallel work, of a rather different type but also applying algebra-geometric tools in system theory, has been done by various authors; for example, \citet{hazewinkel1975invariants} (see also \citealp{byrnes1978moduli}) have studied the algebraic variety formed by the isomorphism classes of \textit{linear} systems of a given dimension, while \citet{hermann1977algebro} have applied tools from algebraic geometry to obtain interesting new
\label{op:156}derivations of results in linear system theory.

Finally, the use of other methods should be investigated, even for polynomial response maps. For example, an analytic realization of such a response $f$ may have 'nicer' properties than a polynomial or $k$-system realization. On a more abstract level, the arguments in section~\ref{sec:29} are very near the type of category-theoretic models suggested by \citet{arbib1974machines} and others; since our type of response does not seem to satisfy the hypothesis of any of the general approaches in the literature, it would be interesting to study what modifications are needed in the latter in order to have them include this case also.

\label{op:157}%

\section{Addendum: Connection to Some More Recent Work}\label{sec:addendum}

The work reproduced in the preceding sections was completed in 1976,
and its bibliography reflects that date (with corrected publication
details for ``to appear'' works).
Since then, several lines of research have developed that are closely related to the ideas presented there, often without reference to them. The purpose of this addendum is to point out some of these connections. We do not attempt a complete survey; rather, we indicate representative papers in each area, with a few words about how they relate to the results in the preceding sections. The selection is necessarily incomplete, and it reflects our own reading of the literature.

Throughout, we use the notations of the preceding sections. In particular, $f$ is a polynomial response map, $\mathcal{L}_{f}$, $\mathcal{A}_{f}$, and $\mathcal{Q}_{f}$ are its observation space, algebra, and field, and $\Sigma_{f}$ is its canonical realization.

\subsection{The Koopman approach and linearization by observables}

\subsubsection{The Koopman operator}

The basic idea of the Koopman approach, introduced by Koopman in 1931,
is to replace a nonlinear dynamical system by the linear operator that
it induces on functions defined on its state space. If
$x(t+1)=P(x(t))$, the Koopman operator acts on an observable $\varphi$
by $\varphi\mapsto\varphi\circ P$. This operator is linear, although
in general it acts on an infinite-dimensional space of functions. The
approach was revived by \citet{mezic2005spectral}, who related
spectral properties of the Koopman operator to model reduction, and it
has since become one of the main tools in the data-driven analysis of
dynamical systems; see for instance \citet{budisic2012applied},
\citet{mezic2013analysis}, and the surveys by
\citet{brunton2022modern}, \citet{otto2021koopman}, and
\citet{bevanda2021koopman}. The edited volume by
\citet{mauroy2020koopman}
collects many applications to systems and control. A widely used computational method, extended dynamic mode decomposition \citep{williams2015data}, approximates the Koopman operator by its compression to the span of a finite dictionary of observables.

The present work follows the same approach, but in an algebraic rather than a functional-analytic setting. The transpose $A(P)$ of the transition map acts on polynomial observables, and the canonical state space is reconstructed from the algebra of observables $\mathcal{A}_{f}$ as its set of $k$-points (Section~\ref{ch:3}).
In the Koopman literature, the observables typically form a Banach or Hilbert space, and questions of convergence and spectral approximation are central. Here, instead, the objects of interest are the vector space $\mathcal{L}_{f}$, the algebra $\mathcal{A}_{f}$, and the field $\mathcal{Q}_{f}$ generated by the output and its shifts, and the relevant finiteness properties are finite dimensionality, finite generation, and finite transcendence degree.

\subsubsection{Systems with inputs}

For systems with inputs, several extensions of the Koopman framework have been proposed. \citet{proctor2018generalizing} generalized the Koopman operator to systems with inputs and control, and \citet{korda2018linear} used lifted linear predictors of the form $z(t+1)=Az(t)+Bu(t)$ as the basis for model predictive control; see also \citet{korda2020optimal} and \citet{kaiser2021data} on the construction of eigenfunctions suited to prediction and control, \citet{mauroy2020lifting} on system identification in the lifted space of observables, \citet{surana2016koopman} on observer design, and \citet{lusch2018deep} on learning such embeddings with neural networks.

A recurring conclusion in this literature is that, once inputs are present, an exact finite-dimensional lifted model cannot in general be linear in the input, and bilinear models appear naturally. \citet{goswami2022bilinearization} used Koopman eigenfunctions to bilinearize control-affine systems and to study reachability and optimal control; \citet{bruder2021advantages} and \citet{otto2024learning} studied the modeling advantages of bilinear lifted models and how to learn them from data; and \citet{iacob2024koopman} showed that, in the presence of inputs, exact Koopman representations are in general at least bilinear in the lifted state and the input.

These questions are closely related to the results on bounded maps and state-affine systems in Sections~\ref{ch:4} and~\ref{ch:5}.
Observe that a state-affine system
\[
x(t+1)=F(u(t))\,x(t)+G(u(t)), \qquad y(t)=Hx(t),
\]
with $F$ and $G$ polynomial in $u$, is linear in the state and polynomial in the input, and that it contains bilinear systems as the special case in which $F$ and $G$ are affine in $u$. Section~\ref{ch:5} shows that a polynomial response map admits a finite-dimensional state-affine realization if and only if it is bounded and finitely realizable, that is, if and only if its observation space $\mathcal{L}_{f}$ is finite dimensional; in that case, the state space of the span-canonical state-affine realization is in a natural duality with $\mathcal{L}_{f}$. In Koopman terms, this says that the output together with all its shifts spans a finite-dimensional space of observables that is invariant under the dynamics, and that the lifted model is then exactly state-affine. This is precisely the type of lifted model obtained by \citet{iacob2025exact}, who construct exact finite-dimensional Koopman embeddings of block-oriented polynomial systems in which the lifted dynamics are linear in the lifted state and polynomial in the input, and bilinear under additional conditions.

\subsubsection{Exact finite-dimensional linear embeddings}

The problem of deciding when a nonlinear system can be embedded exactly into a finite-dimensional linear one has a long history in control theory. In the continuous-time setting, \citet{levine1986nonlinear} studied immersions of nonlinear systems into linear ones, with applications to observers and finite-dimensional filters, and \citet{jouan2003immersion} considered immersion into linear systems modulo output injection. Carleman linearization \citep{kowalski1991nonlinear} embeds a polynomial system into an infinite linear system acting on monomials, and its finite truncations have been analyzed recently by \citet{amini2025carleman}. \citet{brunton2016koopman} discussed Koopman-invariant subspaces and finite linear representations for control, \citet{haseli2022learning} gave algorithms that identify the largest Koopman-invariant subspace contained in the span of a given dictionary, and \citet{jungers2019nonlocal} characterized, through the notion of polyflows, when finitely many successive Lie derivatives of the state span a space that is invariant under the dynamics. Obstructions to the existence of such embeddings for systems with several limit sets, and their implications for learning Koopman embeddings, were studied by \citet{liu2025properties}; see also \citet{kvalheim2026linearizability} on the linearizability of flows by embeddings.

In the discrete-time polynomial setting of this work, the analogous question is whether $\mathcal{L}_{f}$ is finite dimensional, and the answer given there is in terms of the boundedness of the response map, the rank of its behavior matrix, and the existence of affine input/output difference equations.

\subsection{Algebraic realization theory after 1979}

The algebraic approach to nonlinear realization theory taken in this work was continued in several directions. For continuous-time polynomial systems, \citet{bartosiewicz1988minimal} developed a realization theory for systems defined on algebraic varieties, including a notion of minimality and a uniqueness theorem, and \citet{bartosiewicz1987rational} studied rational systems and their observation fields. \citet{grossman1992realization} gave a realization theory for input/output maps based on bialgebras. \citet{wang1992algebraic} related the existence of rational realizations of continuous-time input/output maps to algebraic input/output differential equations, and \citet{wang1995orders} showed, under analyticity assumptions, that there cannot be input/output equations of order smaller than the minimal dimension of an observable realization, which extends a well-known fact for linear systems. The role of spaces of observables in continuous-time nonlinear control, including the extension of several of the ideas of this work to that setting, was discussed in \citet{sontag1995spaces}; see also \citet{sontag1998mathematical} for a textbook treatment of realization and observability.

A systematic realization theory for rational systems was later developed by \citet{nemcova2009realization}, who gave necessary and sufficient conditions for the existence of rational realizations, and by \citet{nemcova2010minimal}, who characterized minimal rational realizations. For discrete-time nonlinear systems described by input/output equations, \citet{kotta2012realization} gave a realization procedure based on skew polynomial rings. More recently, \citet{pavlov2022realizing} and \citet{falkensteiner2025real} studied the realization of differential-algebraic input/output equations by rational dynamical systems, including the existence of observable and of real realizations; these papers use tools from algebraic geometry, such as parametrizations of the hypersurface defined by the input/output equation, that are close in spirit to those used here.

Algebraic methods have also been applied to the analysis of observability and reachability of polynomial systems. \citet{kawano2013observability} gave necessary and sufficient conditions for global and local observability at an initial state of polynomial systems, using commutative algebra, and \citet{kawano2013reachability} and \citet{kawano2016commutative} studied reachability and controllability of discrete-time polynomial systems by similar methods. These results complement the discussion of algebraic observability and quasi-reachability in Section~\ref{ch:3}.

Finally, several results of this work are stated only for algebraically closed fields, or for $k=\mathbb{R}$ with ad hoc arguments (see the Introduction). Real algebraic geometry, for which a standard modern reference is \citet{bochnak1998real}, now provides the tools needed to revisit them over the reals.

\subsection{Rational series, switched systems, and weighted automata}

The treatment of state-affine systems in Section~\ref{ch:5} is based on the theory of rational formal power series in noncommuting variables and on the rank of the associated behavior (Hankel) matrix. A modern account of that theory is given by \citet{berstel2011noncommutative}.

The same formalism has been used extensively in the realization theory of hybrid and parameter-varying systems. \citet{petreczky2010spaces} characterized classes of nonlinear and hybrid systems whose input/output maps can be represented by recognizable formal power series, building on the state-affine results of this work. \citet{petreczky2013realization} developed a realization theory for discrete-time linear switched systems, in which existence of realizations is characterized by a finite-rank condition on a generalized Hankel matrix and minimal realizations are unique up to isomorphism, and \citet{petreczky2012affine} obtained a similar theory for affine linear parameter-varying systems by reduction to switched systems.

In machine learning, the same Hankel-matrix ideas underlie spectral methods for learning weighted automata, which include hidden Markov models as special cases. \citet{balle2014spectral} give a unified presentation of spectral learning of weighted automata based on factorizations of Hankel matrices, and \citet{balle2015canonical} introduce a canonical form for weighted automata and apply it to approximate minimization.

\subsection{Identifiability and fields of observables}

The observation field $\mathcal{Q}_{f}$, and the question of which functions of the state are determined by input/output data, also play a central role in the study of parameter identifiability, where unknown parameters are regarded as additional constant states. Differential-algebraic methods for identifiability were introduced by \citet{ljung1994global}. \citet{nemcova2010structural} studied structural identifiability of polynomial and rational systems in terms of observation algebras and fields, and \citet{hong2020global} and \citet{ovchinnikov2023parameter} gave algorithms and theory for global identifiability based on input/output equations and on fields of identifiable functions.

Finally, we note that polynomial dynamical systems over finite fields have been used to model gene regulatory networks, and that computational algebra has been applied to their reverse engineering \citep{laubenbacher2004computational,jarrah2007reverse}. This work assumes that the field $k$ is infinite, so these results are outside its scope, but the algebraic viewpoint is similar.

\subsection{Concluding remarks}

Many of the questions treated in this work reappear in the literature discussed above in different language: the lifting of nonlinear dynamics to a linear action on observables, the finite dimensionality of spaces of observables, the passage from input/output equations to state-space realizations, and the characterization of realizability by rank conditions on Hankel matrices. We hope that making this connection explicit will be useful to readers coming from these areas. A more detailed discussion of these relations, and of their continuous-time counterparts, is left for future work.

\bibliographystyle{plainnat}
\renewcommand{\refname}{References for Original Work}
\newpage
\phantomsection\addcontentsline{toc}{section}{References for Original Work}

\newpage
\phantomsection\addcontentsline{toc}{section}{References for Addendum}
\renewcommand{\refname}{References for Addendum}

\end{document}